\documentclass[aps,prx,twocolumn, notitlepage,longbibliography,superscriptaddress,nofootinbib]{revtex4-2}

\usepackage[margin=1in]{geometry}
\usepackage{amsmath,amssymb,amsthm,mathtools}
\usepackage{hyperref}
\usepackage{enumitem}
\usepackage{tikz}
\usepackage{tikz-cd}
\usepackage{booktabs}
\usepackage{array}
\usepackage{stmaryrd}
\usepackage{physics}
\usepackage[ruled,vlined,linesnumbered]{algorithm2e}

\theoremstyle{plain}
\newtheorem{theorem}{Theorem}[section]
\newtheorem{lemma}[theorem]{Lemma}
\newtheorem{proposition}[theorem]{Proposition}
\newtheorem{corollary}[theorem]{Corollary}
\theoremstyle{definition}
\newtheorem{definition}[theorem]{Definition}
\newtheorem{example}[theorem]{Example}
\newtheorem{remark}[theorem]{Remark}
\newtheorem{construction}[theorem]{Construction}
\newtheorem{assumption}[theorem]{Assumption}
\newcommand{\Fq}{\mathbb{F}_q}
\newcommand{\Fp}{\mathbb{F}_p}

\newcommand{\rk}{\mathrm{rk}}
\newcommand{\im}{\mathrm{Im}}
\newcommand{\Short}{\mathrm{Short}}

\usepackage[colorinlistoftodos, color=green!40, prependcaption]{todonotes}

\usepackage{pst-all}
\usepackage{leftidx}

\def\H{\mathcal{H}}
\def\Hs{\mathsf{H}}

\def\Z{\mathbb{Z}}
\def\C{\mathcal{C}}

\def\B{\mathcal{B}}
\def\R{\mathcal{R}}
\def\N{\mathcal{N}}
\def\L{\mathcal{L}}
\def\T{\mathcal{T}}
\def\SS{\mathcal{S}}

\newcommand{\ZZ}{\mathbb{Z}}

\newcommand{\FF}
{\mathbb{F}}
\newcommand{\Fc}
{\mathcal{F}}
\newcommand{\I}{\mathbb{I}}
\newcommand{\CC}{\mathbb{C}}
\newcommand{\vc}[1]{\mathbf{#1}}

\newcommand{\bket}[2]{\langle \, #1 \,|\, #2 \, \rangle}
\newcommand{\boket}[3]{\langle\, #1 \,|\, #2 \,|\, #3 \,\rangle}

\newcommand{\be}{\begin{equation}}
\newcommand{\ee}{\end{equation}}
\newcommand{\bc}{\begin{center}}
\newcommand{\ec}{\end{center}}

\newcommand{\lo}{\overline}

\newcommand{\as}{\mathsf{a}}
\newcommand{\bs}{\mathsf{b}}
\newcommand{\cs}{\mathsf{c}}
\newcommand{\ds}{\mathsf{d}}

\newcommand{\fs}{\mathsf{f}}

\newcommand{\hs}{\mathsf{h}}
\newcommand{\ts}{\mathsf{\tau}}

\newcommand{\cC}{\mathcal{C}}
\newcommand{\cF}{\mathcal{F}}
\newcommand{\cG}{\mathcal{G}}
\newcommand{\cI}{\mathcal{I}}

\newcommand{\cS}{\mathcal{S}}\newcommand{\cT}{\mathcal{T}}
\newcommand{\cV}{\mathcal{V}}
\newcommand{\cX}{\mathcal{X}}
\newcommand{\cZ}{\mathcal{Z}}

\newcommand{\bvec}[1]{\boldsymbol{#1}}

\newcommand{\hv}{\bvec{h}}
\newcommand{\lv}{\bvec{\ell}}
\newcommand{\bone}{\bvec{1}}

\newcommand{\precdot}{\mathrel{\prec\!\cdot}}
\newcommand{\dotsucc}{\mathrel{\cdot\!\succ}}

\usepackage{graphicx}

\definecolor{mydarkgreen}{rgb}{0.0, 0.37, 0.1}

\renewcommand{\red}[1]{\textcolor{red}{#1}}
\renewcommand{\blue}[1]{\textcolor{blue}{#1}}
\renewcommand{\green}[1]{\textcolor{mydarkgreen}{#1}}

\newcommand{\rd}{\textcolor{red}{\texttt{r}}}
\newcommand{\bl}{\textcolor{blue}{\texttt{b}}}
\newcommand{\gr}{\textcolor{mydarkgreen}{\texttt{g}}}

\usetikzlibrary{arrows.meta}
\makeatletter
\newlength\lrvec@height
\newlength\lrvec@width
\newif\iflrvec@same@height
\def\lrvec{\@ifstar\slrvec@\lrvec@}
\newcommand{\slrvec@}[2][.4ex]{
  \lrvec@same@heighttrue
  \mathpalette\lrvec@@{{#1}{#2}}
}
\newcommand{\lrvec@}[2][.4ex]{
  \lrvec@same@heightfalse
  \mathpalette\lrvec@@{{#1}{#2}}
}
\def\lrvec@@#1#2{\lrvec@@@#1#2}
\def\lrvec@@@#1#2#3{%
  \iflrvec@same@height
    \settoheight{\lrvec@height}{$\m@th#1 \mathbf{T}#3$}
  \else
    \settoheight{\lrvec@height}{$\m@th#1#3$}
  \fi
  \settowidth{\lrvec@width}{$\m@th#1#3$}
  \kern.08em
  \raisebox{#2}{\raisebox{\lrvec@height}{\rlap{%
    \kern-.05em
    \begin{tikzpicture}[<-> /.tip={To[width=.4em, length=.2em]}]
      \draw [<->] (-.05em,0)--(\lrvec@width+.05em,0);
    \end{tikzpicture}%
  }}}%
  #3
  \kern.08em
}
\makeatother

\usepackage[bottom]{footmisc}

\begin{document}

\title{Sipser-Spielman meets Dijkgraaf-Witten: non-Abelian qLDPC codes via twisted sheaf gauge theory and almost-constant-overhead magic state fountain}
% \author{Guanyu Zhu, Louis Golowich, Shi Jie Samuel Tan, Ryohei Kobayashi, Po-Shen Hsin}

\author{Guanyu Zhu}
\affiliation{IBM Quantum, IBM T.J. Watson Research Center, Yorktown Heights, NY 10598 USA}

\author{Shi Jie Samuel Tan}
\affiliation{IBM Quantum, IBM T.J. Watson Research Center, Yorktown Heights, NY 10598 USA}
\affiliation{Joint Center for Quantum Information and Computer Science, University of Maryland, College Park, Maryland 20742, USA}

\author{Ryohei Kobayashi}
\affiliation{Department of Applied Physics, University of Tokyo, Bunkyo-ku, Tokyo 113-8656, Japan}

\author{Po-Shen Hsin}
\affiliation{Department of Mathematics, King’s College London, Strand, London WC2R 2LS, UK}

\begin{abstract}
Sipser-Spielman codes and Dijkgraaf-Witten twisted gauge theories are among the most profound ideas in the last few decades in the areas of computer science and mathematical physics respectively.  This work unifies them in the same framework using the language of sheaf cohomology, and produces a new framework of twisted sheaf gauge theories describing a class of quantum low-density parity-check (qLDPC) codes with non-Abelian $D_4$ topological order. 
The corresponding twisted Clifford-stabilizer qLDPC code can be obtained from gauging a new type of 0-form sub-complex symmetry from a qLDPC code defined on a general sheaf complex, and can be interpreted as a topological defect network of non-Abelian $D_4$ patches glued together with proper gapped interfaces.  As an application, one can use this to realize a \textit{magic state fountain} via the gauging measurement of the  addressable logical CZ gates as 0-form subcomplex symmetries in a 2D hypergraph-product (HGP) sheaf code.  This includes a scheme of subdividing an arbitrary constant-rate 2D HGP code with parameters $[[n,\Theta(n), \Omega(n^{1/2})]]$ into a quantum sheaf code which allows preparation of $\Theta(n^{1/2})$ logical CZ magic states in parallel, equivalent to the recent geometric construction using the code-to-manifold mapping in \cite{ZhuKobayashiHsin2026}.  Moreover, using the recent sheaf complex and algebraic code constructions by Golowich-Tamo-Zhu \cite{GolowichTamoZhu2026} with parameters $[[n,\Theta(n^{1-\epsilon}), \Omega(n^{(1-\epsilon)/2})]]$ for arbitrary small $\epsilon$, one can prepare $\Theta(n^{1-\epsilon})$ CZ magic states in parallel and hence achieve an almost-constant magic rate.    

\end{abstract}

\maketitle

\tableofcontents
\newpage

%======================================================================
\section{Introduction}
\label{sec:introduction}
%======================================================================

The rapid development of quantum computing and information science in recent years have stimulated the interaction between different scientific fields including quantum physics, computer science and mathematics.  One of the most influential ideas in theoretical computer science and coding theory in the past few decades is the notion of Sipser-Spielman code \cite{Sipser_Spielman}, which has been widely applied to the study of classical and quantum low-density parity-check (LDPC) codes, as well as classical and quantum complexity theory.  Meanwhile,   another  profound idea in the realm of mathematical physics is the Dijkgraaf-Witten twisted gauge theory \cite{Dijkgraaf:1989pz}, which has huge impact in the fields of topological quantum field theory (TQFT) and topological phases of matter.  In this work, we unify these two ideas in the mathematical language of sheaf cohomology, inspired by the recent development in quantum sheaf codes \cite{PanteleevKalachev2024, pkldpc22, lin2024transversal}.   This leads to the discovery of a new type of twisted sheaf gauge theory and the corresponding quantum low-density parity-check (qLDPC) codes with non-Abelian $D_4$ topological order.

One of the underlying motivations behind this work is the quest of low-overhead fault-tolerant quantum computation.  In recent years, a sustained
effort on qLDPC codes
\cite{fiberbundlecode21, hastingswr21, pkldpc22, Quantum_tanner, lh22,
guefficient22, dhlv23, lzdecoding23, gusingleshot23, Bravyi:2024wc} produced
asymptotically good families~\cite{pkldpc22}, which achieves constant encoding rate and linear distance and thereby bring the space overhead of a
quantum memory down to a constant. Meanwhile, ongoing efforts have been made to minimize the spacetime-overhead of fault-tolerant logical gates on the encoded quantum memory.

A large amount of recent studies investigate the logical Clifford gates, through
code surgery (see e.g.~\cite{cohen22, cross2024improved, williamson2024low,
swaroop2024universal}), homomorphic
measurement~\cite{huang2023homomorphic, xu2025fast}, fold-transversal
constructions~\cite{breuckmann2022fold}, and related techniques. Clifford gates
are not universal on their own, and supplying the missing non-Clifford resource
by magic state distillation reintroduces precisely the overhead that a high-rate
code was chosen to eliminate; a \emph{native} logical non-Clifford operation
would remove that cost at its source. A second difficulty is peculiar to the
high-rate setting. A single qLDPC block carries many logical qubits, so a gate
acting on all of them at once is of limited computational use: what one needs are
gates that are \emph{addressable}, acting on a chosen small set of logical
qubits, and \emph{parallelizable}, so that many disjoint gates can be applied in
one shot. In contrast, topological codes, holding one or a few logical qubits per patch, never
have to confront this.

The established route to a native non-Clifford gate is to raise the dimension. A
transversal gate at the $\ell^\text{th}$ level of the Clifford hierarchy is governed
by an $\ell$-fold cup product and therefore calls for a chain complex with $\ell$
degrees available above the qubits; accordingly, the recent constructions of
transversal non-Clifford gates on qLDPC codes are built on 3D product construction, i.e.\ on
4-term chain complexes (see \cite{zhu2023non, scruby2024quantum, golowich2024quantum,
lin2024transversal, breuckmann2024cups, Hsin2024:classifying, zhu2025topological,
zhu2025transversal, menon2025magic, golowich2025constant,
tan2025single, gulshen2025quantum, li2025poincar,  GolowichTamoZhu2026} etc.). Within this family,
Refs.~\cite{zhu2023non, zhu2025topological, zhu2025transversal} introduced the notion of 
`\textit{magic state fountain}', in which a transversal non-Clifford gate
fault-tolerantly prepares a large number of magic states in parallel in a single
shot and hence transfers the high encoding rate of the qLDPC code to a high magic rate.  

The dimension requirement is not an artifact of these particular
constructions: it echoes the Bravyi-K\"onig bound~\cite{Bravyi:2013dx}, which for
topological codes forces the spatial dimension to be at least $n$ for a
transversal gate at the $n^\text{th}$ level of the hierarchy, and analogues have
been sought for $n$-dimensional hypergraph-product codes as
well~\cite{fu2025no}. The extra dimension is paid for in connectivity: a 3D
product code demands longer-range couplings than a 2D one once laid out on a
planar chip, which is the binding constraint for most architectures, and the
resulting tension between universality and geometric locality is a basic open
question. Following Ref.~\cite{ZhuKobayashiHsin2026}, this paper shows that the
third dimension can be dispensed with --- provided one is willing to leave behind
transversal gates and Pauli stabilizer codes.

The ingredient that makes this possible has been available from the very
beginning, but has seen little use in the qLDPC context. In the same work that
introduced the toric code more than two decades ago, Kitaev also wrote down the
quantum double models~\cite{kitaev2003}: a large class of 2D \emph{non-Abelian}
topological codes whose anyons support universal quantum computation by braiding
alone. They achieve universality in two dimensions not by adding a third, but by
giving up the Abelian --- that is, Pauli stabilizer --- structure. Whether the
same bargain is available to qLDPC codes, which would pair non-Abelian
topological order with the low space overhead and large distance of an
expander-based construction, is the question that motivates this work and its predecessor \cite{ZhuKobayashiHsin2026}.

Several fundamental and technical barriers exist along the way. First of all,
non-Abelian codes are always associated with an underlying topological quantum
field theory (TQFT). However, a TQFT is typically defined on a manifold, with its
combinatorial version on the triangulation of a manifold. Moreover, the path
integral of a TQFT, also called a state sum (e.g.~\cite{walker2021universalstatesum}),
is a topological invariant that is only determined by the manifold topology and is
invariant under re-triangulation, or equivalently under local geometric
deformation. Some well-known examples of such manifold invariants include the
Turaev-Viro invariants~\cite{Turaev:1992hq} and the Dijkgraaf-Witten
invariants~\cite{Dijkgraaf:1989pz}. On the other hand, a qLDPC code is often
defined on a general chain complex based on expander graphs in order to boost the
rate and distance. The general chain complex is typically \emph{not} a
triangulation or a usual simplicial complex: its 1-cells are hyperedges and its 2-cells
are hyperfaces, and the cochain-level operations --- cup products, cap products,
Poincar\'e duality --- on which a twisted (non-Abelian) gauge theory is built simply do not exist there. 

The first construction of non-Abelian qLDPC codes, obtained in
Ref.~\cite{ZhuKobayashiHsin2026}, circumvents these barriers by \emph{leaving}
the chain complex of the code: a skeleton classical code is first mapped to a
high-dimensional manifold by a generalized~\cite{zhu2025topological} version of
the Freedman-Hastings mapping~\cite{freedman:2020_manifold_from_code}, then
deformation-retracted to a Poincar\'e CW complex, and only afterwards is the
twisted gauge theory of Ref.~\cite{Hsin2024_non-Abelian} written down, using the
cup product that the CW structure supplies; the resulting Clifford-stabilizer code and path integral
are finally pulled back to the 2D skeleton.  Another related work  \cite{christos2026non} later introduced a similar-style homological approach with a Clifford-stabilizer model based on cup product  applicable to a class of complexes satisfying a generalized notion of Poincar\'e duality with an instantiation of the non-Abelian fracton models on 3D Euclidean lattices.  Ref.~\cite{christos2026non} also introduced another approach based on  graph gauging on a broad class of qLDPC codes in the style of the gauging measurement in code surgery \cite{williamson2024low} with the generalization of Pauli measurement to Clifford measurement, which requires an additional ancillary system \footnote{ One of the main differences  of the construction in Ref.~\cite{ZhuKobayashiHsin2026} and the present paper from the surgery-type graph gauging in Ref.~\cite{christos2026non} is that the twisted (non-Abelian) and untwisted (Abelian) qLDPC codes constructed here share the same underlying complex $\L$ and hence have similar code distance scaling. See Sec.~\ref{sec:discussion} for more discussion.
}.   

In this work we take a more algebraic route, and obtain the cup product
\emph{intrinsically}, from the code itself. The observation that makes this
possible is due to Meshulam~\cite{Meshulam:2018}: a Sipser-Spielman (Tanner)
code~\cite{Sipser_Spielman} on a graph $G$ with a system of local codes is nothing but the first homology $H_1(G,\Fc)$ of a
\emph{sheaf} $\Fc$ --- a local coefficient system whose stalks are the local
check spaces.
We then show (Sec.~\ref{sec:classical}) that an \emph{arbitrary} classical code,
given by a parity-check matrix with no zero rows or columns and hence by a hypergraph rather than a
graph, can be \emph{subdivided} into a sheaf code on an honest graph, by splitting
each hyperedge into ordinary edges meeting at a new vertex that carries a
repetition local code (Fig.~\ref{fig:subdivision}). The subdivision preserves the
code dimension exactly and does not decrease the distance. The point is that the
subdivided object is a genuine one-dimensional simplicial complex with local
coefficients, so products of such objects are square complexes on which cup
products, cap products and the chain-cochain pairing are all defined at the
cochain level. 

More generally, a wide range of qLDPC codes correspond to quantum sheaf codes defined on a simplicial or more generally cell complex \cite{PanteleevKalachev2024} with sheaf structure, including the hypergraph product of classical Sipser-Spielman codes and the asymptotically good qLDPC code \cite{pkldpc22}.   The cohomology operations including cup products are well defined on such sheaf complexes and the twisted gauge theory can be well-defined. 

Equipped with this, we construct in Sec.~\ref{sec:non-Abelian_LDPC} the spacetime
path integral of a twisted $\ZZ_2^3$ \emph{sheaf gauge theory} on an arbitrary sheaf complex, corresponding to a twisted (non-Abelian) qLDPC codes. Its action
consists of three sheaf BF terms together with a type-III Dijkgraaf-Witten twist
$\pi\int_{\tilde\eta_3} a^1 \cup b^1 \cup c^1$ supported on a distinguished cycle
$\tilde\eta_3$, and we show that the resulting path integral is a cohomology
invariant, i.e.\ gauge invariant. The corresponding code is a \emph{Clifford}
stabilizer code~\cite{Hsin2024_non-Abelian, Davydova:2025ylx} rather than a Pauli
stabilizer code: its $X$-type generators are dressed by CZ gates
(Fig.~\ref{fig:dressed_stabilizer}), and they commute only on the zero-flux
subspace. We derive it in two independent ways --- by gauging the $0$-form
subcomplex symmetry of a sheaf cluster state and its associated
sheaf symmetry-protected topological (SPT) phase, and, in
Sec.~\ref{sec:gauging_untwisted}, by gauging the 0-form CZ symmetry directly in the
untwisted qLDPC code. We also observe that the gapped interfaces along which the
several faces of a branched square complex meet realize a \emph{topological defect
network}~\cite{Aasen2020} 
(Figs.~\ref{fig:gapped_boundary_picture} and~\ref{fig:gapped_boundary_dual}),
which gives a physical interpretation of the sheaf gauge theory.

A crucial question when generalizing gauging measurement to qLDPC codes which allows preparation of logical magic states is
addressability and parallelizability. There exists a folklore (see
e.g.~Refs.~\cite{JochymOConnor:2021ih, lin2024transversal, he2025quantum}) in the
quantum error correction community that in order to get an addressable (targeted)
transversal $\text{C}^{n-1}\text{Z}$ gate, one needs to conjugate a logical $X$
with a global logical $\text{C}^{n}\text{Z}$ gate, which would suggest that an
addressable logical CZ requires a 3D product code with a transversal CCZ.
Surprisingly, this statement is true only for topological codes defined on 
manifolds. A 2D manifold has a unique top-dimensional cycle, but a general 2D
chain complex has a large number of them: a single edge can be adjacent to many
faces (a `book-like' structure, Fig.~\ref{fig:gauging_region}), so there are many
inequivalent 2-cycles $\eta_2$ available. Each supports a symmetry operator, and
we call the resulting symmetries \textit{$0$-form subcomplex symmetries}: they act
on a codimension-$0$ subcomplex, rather than on a subdimensional one as for the
higher-form ($k$-form) symmetries responsible for addressable transversal gates in
the manifold case~\cite{zhu2023non, Hsin2024:classifying, ZhuKobayashiHsin2026}. In
the sheaf complex $\L$, the support of the transversal CZ labeled by
a sheaf $0$-cocycle $\green{\gamma^0}$ is the \emph{capped} $2$-chain
$\eta_2^{\gamma} = \green{\gamma^0} \frown \eta_2$ obtained from the cap product
(Fig.~\ref{fig:subcomplex_illustration_nested}). The number of
independently addressable transversal logical CZ gates is then
$\dim H^0(\L,\Fc)$, which in general grows with the system size.

These operators are what the protocol of Sec.~\ref{sec:logical_operation} measures.
Generalizing the gauging measurement of Ref.~\cite{Davydova:2025ylx} from
topological codes to sheaf qLDPC codes, we start from two decoupled untwisted
hypergraph-product sheaf codes, measure all dressed $X$-stabilizers of the third
colour for $O(d)$ rounds --- a single round of \emph{commuting}, weight-$O(1)$
checks, which is the whole point of gauging: a global transversal gate is
factorized into local Gauss's law operators --- and then ungauge. The product of
the outcomes along any basis $0$-cocycle returns the eigenvalue of the
corresponding transversal logical CZ, so all of the addressable logical CZ
measurements are obtained simultaneously, ancilla-free and without post-selection
coupling different pairs. With the initialization read off from the hypergraph of
the gate structure (Fig.~\ref{fig:interaction-hypergraph}), this realizes the
magic state fountain. The entire logical action is then rederived from first
principles in Sec.~\ref{sec:spacetime_path_integral} using the spacetime path
integral, in which the protocol appears as a slab of twisted theory bounded by two
spacetime domain walls (Fig.~\ref{fig:protocol}). The path integral is a powerful
formalism for studying fault-tolerant logical gates; its connection to error
correction in the context of topological codes was first established in
Refs.~\cite{Bauer:2023awl, Bauer:2024qpc, Bauer:2024alh} and applied to the
gauging measurement protocol for non-Abelian topological codes in
Ref.~\cite{Davydova:2025ylx}.

Two further generalizations occupy the second half of the paper. First, in
Sec.~\ref{sec:Fq_generalization} we lift the entire construction --- path
integral, twisted stabilizers, subcomplex symmetry and gauging --- from $\FF_2$ to
an arbitrary finite field $\Fq$ with $q=p^u$, using the trace pairing
$\mathrm{Tr}_{\Fq/\Fp}$ to define the generalized Pauli operators and the
$\Fq$-valued cup product. We emphasize that the resulting gauge theory is
\emph{not} a $\ZZ_q$ gauge theory: the additive group of $\Fq$ is
$(\ZZ_p)^u$, and what the field structure contributes is the Frobenius
multiplication tensor $T_{ijk}=\mathrm{Tr}(\theta_i\theta_j\theta_k)$ that
controls which of the $u^3$ component CZ gates appear in the twist. Second, in
Sec.~\ref{sec:high_yield} we address the rate at which magic states are produced.
Measuring the usefulness of the fountain  by the \emph{magic rate} $\varrho=s/n$ per physical qubit ($s$ being the disjoint magic
state count), the
construction of Sec.~\ref{sec:code_instantiation} gives only
$\varrho=\Theta(n^{-1/2})$, because one tensor factor of the sheaf is constant.
Removing that factor requires classical LDPC codes whose local codes have a
genuine multiplication property, i.e.\ local codes closed in the appropriate sense
under the Schur product. Such codes have recently been constructed by Golowich,
Tamo and Zhu~\cite{GolowichTamoZhu2026}, over an alphabet of large constant size
--- which is precisely why the $\Fq$ generalization was needed. Feeding them into
our construction yields an \emph{almost constant} magic rate,
$\varrho \ge n^{-\epsilon}$ for arbitrarily small $\epsilon>0$. Finally, the
construction is not tied to the hypergraph product: in Sec.~\ref{sec:balanced} we
show that the twist, the twisted Clifford stabilizer code and the gauging
measurement descend to $2$D \emph{balanced products}~\cite{9490244,PK:almost_linear}
of two sheaf codes with a common free group symmetry $H$ --- the quotient of the
hypergraph product by the diagonal action --- and, when $|H|$ is coprime to the
characteristic, we obtain an explicit K\"unneth basis of the logical qudits and a
closed formula for the logical CZ structure, in which the three $H$-characters of
a coupled triple must multiply to one. This sets up the problem of combining the
non-Abelian twist with the better distance of balanced- and lifted-product codes,
which we leave open.

Besides the application to the magic state fountain, exploring non-Abelian qLDPC
codes broadens the landscape of highly entangled quantum matter. qLDPC codes have
a deep connection to quantum complexity theory and provide a stronger form of
topologically ordered state: all states below a certain energy density are
non-trivial, i.e.\ there are no low-energy trivial states (NLTS), as conjectured in
Ref.~\cite{Freedman:2013zfj} and proven in Ref.~\cite{Anshu:2022hsn}. In contrast,
most topological codes only have non-trivial topological order in the ground
states. Recently, a ``no low-energy trivial magic'' (NLTM) conjecture was also
proposed in Ref.~\cite{wei2025long}, which is a stronger necessary condition for
quantum PCP than NLTS. The Clifford stabilizer qLDPC codes constructed here
naturally have non-trivial magic in their ground states due to the
non-stabilizerness of those states~\cite{parham2025quantum, wei2025long}. In
addition, the low-energy states of qLDPC codes give rise to topological quantum
spin glasses~\cite{Placke:2024wey}, where a constant rate gives rise to an
exponential number of stable, topologically ordered Gibbs states that persist at
finite temperature --- a quantum analogue of the classical spin glasses with
complex free-energy landscapes that provide a mechanism for long-lived
memory~\cite{ANDERSON1970549, PhysRevLett.43.1754, Binder:1986zz}. Studying the
non-Abelian generalization of these phases of matter is an interesting future
direction.

\subsection{Summary of results}
\label{sec:summary_of_results}

We summarize the main results and the flow of the paper as follows. The statements
below are informal; precise versions are given in the indicated places.

\smallskip
\noindent\textit{(i) An intrinsic simplicial model for an arbitrary classical
code.} Sections~\ref{sec:classical} and~\ref{sec:prelim} set up the sheaf
formalism. Propositions~\ref{Prop:sheaf_homology} and~\ref{Prop:sheaf_cohomology}
record Meshulam's dictionary, $\cT(G,\{\cC_v\})=H_1(G,\Fc)$ and
$\cT^*(G,\{\cC_v\})=H^0(G,\Fc)$, in the form we use throughout: a
\emph{global} homological condition is a \emph{local} codeword condition, one per
vertex. Subdivision then removes the restriction to graph-based codes.

\begin{theorem}[Subdivision into a sheaf code; informal version of
Proposition~\ref{prop:subdivision-params}]
\label{thm:informal_subdivision}
Let $\Hs\in\FF_2^{m\times n}$ have rank $r$ and no zero
rows or columns, and let $\C=\ker\Hs$.
Write $w_j=w_c(v'_j)$ for the weight of column $j$.
Replacing bit $j$ by $w_j$ incident edges, with repetition
constraints at the original bit vertices and SPC constraints
at the original check vertices, produces a sheaf code on a graph
satisfying
\[
H_1(G,\Fc)\cong\C,\qquad
H_1(G,\Fc^\perp)\cong\ker\Hs^T.
\]
Their dimensions are $n-r$ and $m-r$, respectively.
For nonzero $\C$, the subdivided edge-cycle distance is
\[
d_{\mathrm{sbd}}:=d_1(\Fc)
=\min_{0\ne x\in\C}\sum_{j:x_j=1}w_j
\ge w_c^{\min}d(\C),
\]
where $w_c^{\min}=\min_j w_j$.
The transpose construction uses the same graph with the
two local codes interchanged.
\end{theorem}

\smallskip
\noindent\textit{(ii) Twisted sheaf gauge theory and non-Abelian qLDPC codes.}
Section~\ref{sec:untwisted_code} builds the untwisted hypergraph-product sheaf
code, whose $X$- and $Z$-stabilizers are the columns of $d^0$ and $d^1$
(Fig.~\ref{fig:square_complex}), and observes that the branched structure of the
complex is a topological defect network of surface-code patches
(Figs.~\ref{fig:gapped_boundary_picture} and~\ref{fig:gapped_boundary_dual}).
Section~\ref{sec:non-Abelian_LDPC} then introduces the twist.

\begin{theorem}[Twisted sheaf gauge theory; informal version of
Lemma~\ref{lemma:invariant} and Sec.~\ref{sec:Clifford_stabilizer_codes}]
\label{thm:informal_twisted}
On an arbitrary (2+1)D spacetime sheaf complex $\tilde\L$,
the path integral of the twisted $\ZZ_2^3$ sheaf gauge theory
\be
\cZ[\tilde\L,\tilde\eta_3]=\sum_{\red{a^1},\blue{b^1},\green{c^1}\in H^1(\tilde\L)}
(-1)^{\int_{\tilde\eta_3}\red{a^1}\cup\blue{b^1}\cup\green{c^1}}
\ee
is a cohomology invariant, i.e.\ invariant under the gauge transformations
$\red{a^1} \mapsto \red{a^1}+d\eta^0$ and likewise for $\blue{b^1}$, $\green{c^1}$. Gauging the $0$-form
subcomplex symmetry of the corresponding sheaf cluster state and SPT --- or,
equivalently, of the untwisted qLDPC code itself
(Sec.~\ref{sec:gauging_untwisted}) --- produces a \emph{Clifford} stabilizer
qLDPC code $\tilde\C$ whose $X$-type generators $\tilde A_{v,s}$ are dressed by
CZ gates acting on the two other colors, with the dressing pattern determined by
the triple cup product and the choice of the spacetime 3-cycle $\tilde \eta_3$ or the corresponding spatial $2$-cycle $\eta_2$. The generators
commute on the zero-flux subspace, and the resulting code carries non-Abelian
topological order.
\end{theorem}

\smallskip
\noindent\textit{(iii) $0$-form subcomplex symmetry and addressable transversal
logical CZ.} Section~\ref{sec:0form_subcomplex} is the conceptual centre of the
paper: it identifies a single family of transversal operators that is a logical
gate in the untwisted code and a product of dressed stabilizers in the twisted one.

\begin{theorem}[Subcomplex symmetry, cap product and addressability; informal
version of Secs.~\ref{sec:transversal_CZ}--\ref{sec:charge_parity_operator}]
\label{thm:informal_subcomplex}
For every sheaf $0$-cocycle $\green{\gamma^0}\in H^0(\L,\Fc^{\gr})$ there is a transversal,
depth-one circuit of two-qubit CZ gates $\widetilde{\mathrm{CZ}}^{\rd,\bl}(\green{\gamma^0})$
supported on the capped codimension-$0$ subcomplex
$\eta_2^{\gamma}=\eta_2\frown\green{\gamma^0}$. In the untwisted code $\C$ it implements a
pattern $\Lambda^{\alpha\beta}(\green{\gamma^0})$ of \emph{addressable} logical CZ gates
between the \rd\ and \bl\ copies; in the twisted code $\tilde\C$ the very same
operator is the \gr-type charge parity operator, a product of twisted
$X$-stabilizers, and therefore acts as the identity. The number of independently
addressable such gates is $\dim H^0(\L,\Fc^{\gr})$, which grows with the system
size \footnote{No transversal CCZ gate is used, which gives a
counterexample to the folklore that addressable logical CZ requires a 3D product
code.}.
\end{theorem}

\noindent Section~\ref{sec:condition_cohomology} then determines when the
invariant is non-trivial: by the K\"unneth theorem the twist factorizes, and
non-triviality reduces to a \emph{local cycle condition} at each vertex
(Lemma~\ref{lemma:local_cycle}), which Sec.~\ref{sec:multiplication_property}
identifies with a \emph{multiplication property} of the local codes,
$\cC_v^{\rd,\perp}*\cC_v^{\bl,\perp}\subseteq\cC_v^{\gr}$, for the Schur product
$*$. This single algebraic condition is what drives both instantiations below.

\smallskip
\noindent\textit{(iv) A constant-rate instantiation and the magic state
fountain.} Section~\ref{sec:code_instantiation} instantiates the construction from
an \emph{arbitrary} asymptotically good classical LDPC code, using the constant
sheaf for one factor of the \bl\ and \gr\ copies so that the local cycle condition
holds for free, and Sec.~\ref{sec:logical_operation} runs the protocol on it.

\begin{theorem}[Code parameters and magic state fountain; informal version of
Theorems~\ref{thm:parameters} and~\ref{thm:fountain}]
\label{thm:informal_fountain}
From any asymptotically good classical LDPC code, Construction~\ref{con:2d-hgp}
produces untwisted and twisted $2$D hypergraph-product sheaf codes with
parameters $[[\,n,\ \Theta(n),\ \Theta(\sqrt{n})\,]]$, where $d$ is the
subsystem-code distance with respect to the protected logical qubits (the
short, ``spurious'' classes of the subdivided complex are treated as gauge
qubits). There is a gauging measurement protocol on the twisted code, consisting
of $O(d)$ rounds of weight-$O(1)$ commuting dressed $X$-stabilizer measurements
followed by ungauging, that prepares $\Theta(\sqrt{n})$ logical CZ magic states
on \emph{disjoint} pairs of logical qubits in a single run, protected by the same
distance $\Theta(\sqrt{n})$.
\end{theorem}

\noindent The logical action of the protocol is rederived from the spacetime path
integral in Sec.~\ref{sec:spacetime_path_integral}, where it appears as a slab of
twisted theory bounded by a gauging and an ungauging domain wall, confirming the
operator-level derivation and exhibiting the protocol as a spacetime object.

\smallskip
\noindent\textit{(v) Distance of the twisted codes: two routes.} For the
\emph{untwisted} CSS codes the distance is a standard homological computation
(Lemma~\ref{lemma:distance_untwisted} and Theorem~\ref{thm:high_yield_parameters}).
The \emph{twisted} codes are Clifford, not Pauli, stabilizer codes: their distance
is not $\min(d_X,d_Z)$, and their logical operators are not all Pauli operators.
We establish their distance in two complementary ways.

\emph{(1) The physical, TQFT route.} In the twisted gauge theory the electric
(Wilson) operators are untouched by the twist and the magnetic operators are only
dressed by diagonal CZ factors, so neither becomes lighter. The remaining
topological operators are \emph{condensation descendants} of these --- sums of
condensed operators over the support of a condensation defect --- which live on
regions of higher dimension and are therefore at least as heavy
(Assumption~\ref{assump:TQFT}, well motivated by the condensation-defect and
SymTFT picture of the gauge theory). Granting this, the twisted code inherits
the distance of the untwisted one; this is Lemma~\ref{lemma:distance_twisted}
for Construction~\ref{con:2d-hgp}, and the same shortcut is recorded for
Construction~\ref{con:high_yield} in Sec.~\ref{sec:high_yield_parameters}.

\emph{(2) The coding-theoretic route.} Independently of any assumption about
the operator content of the theory, we verify the (subsystem) Knill--Laflamme
condition \cite{Knill1997} directly from the Clifford stabilizer group: the code
space is written in an orbit basis labelled by the admissible logical labels, and
one shows that no operator below the claimed weight distinguishes two labels or
connects them. Two ingredients make this work: relative-distance bounds for the
protected K\"unneth sectors, obtained by contracting a (co)cycle against a
detecting (co)cycle of one graph factor, and a \emph{cleaning} step that moves the
logical representatives off the support of the error, after which every
label-dependent phase disappears.

Route 1 is the physicist's argument and is what one expects from the
topological order and TQFT; Route 2 is a coding-theoretic proof at the level of the stabilizer group. The
two agree wherever both apply, which we regard as a consistency check on the
condensation-descendant picture. The coding-theoretic distance proofs are given in
Appendices~\ref{app:twisted_KL} and~\ref{app:GTZ_KL}.

\smallskip
\noindent\textit{(vi) Generalization to $\Fq$.} Section~\ref{sec:Fq_generalization}
repeats the whole development over an arbitrary finite field.

\begin{theorem}[$\Fq$ sheaf gauge theory; informal version of
Sec.~\ref{sec:Fq_generalization}]
\label{thm:informal_Fq}
Let $q=p^u$ and let $\mathrm{Tr}=\mathrm{Tr}_{\Fq/\Fp}$. With generalized Pauli
operators defined through the trace pairing, the $\Fq$ path integral
$\cZ=\sum\omega^{\int_{\tilde\eta_3}\mathrm{Tr}(a^1\cup b^1\cup c^1)}$,
$\omega=e^{2\pi i/p}$, is again a cohomology invariant, and the twisted code is a
qudit Clifford stabilizer code whose $X$-generators are dressed by
$\mathrm{CZ}(\lambda)$ gates. The underlying gauge group is $(\ZZ_p)^{u}$ per
colour --- \emph{not} $\ZZ_q$ --- and the field structure enters only through the
Frobenius tensor $T_{ijk}=\mathrm{Tr}(\theta_i\theta_j\theta_k)$, which selects
which component gates appear.
\end{theorem}

\smallskip
\noindent\textit{(vii) Almost-constant magic rate.} Section~\ref{sec:high_yield}
defines the magic rate $\varrho=s/n$ (Definition~\ref{def:magic_rate}), sets up a
dictionary (Lemma~\ref{lemma:dictionary}) between our sheaves and the
\emph{systems of local codes} of Ref.~\cite{GolowichTamoZhu2026}, and uses their
classical codes with a multiplication property as the input.

\begin{theorem}[Almost-constant magic rate; informal version of
Theorems~\ref{thm:high_yield_parameters} and~\ref{thm:high_yield_fountain}]
\label{thm:informal_high_yield}
For every $\epsilon>0$ and every prime $p$ there is an infinite family of twisted
sheaf codes over $\Fq$, $q=p^{O_\epsilon(1)}$, with constant stabilizer weight
and parameters
$[[\,n,\ k\ge n^{1-\epsilon},\ d\ge n^{(1-\epsilon)/2}\,]]_q$,
on which a single run of the gauging measurement prepares $s\ge n^{1-\epsilon}$
logical CZ magic states on disjoint pairs of logical qudits --- an almost
constant magic rate $\varrho=s/n\ge n^{-\epsilon}$. The disjointness comes from a
\emph{perfect matching} structure of the logical CZ gates
(Lemma~\ref{lemma:matching} and Corollary~\ref{cor:disjoint_CZ}). 
\end{theorem}

\noindent Since each $\Fq$-qudit is a block of $u=O_\epsilon(1)$ qubits and each
$\Fq$-CZ is a depth-one circuit of ordinary qubit CZ gates
(Proposition~\ref{prop:cz_descent_app} of Appendix~\ref{app:alphabet}), the
statement descends to qubits with $\varrho$ changing only by the constant factor
$1/u$.

\smallskip
\noindent\textit{(viii) Balanced products.} Appendix~\ref{sec:balanced} shows that
the twisted sheaf gauge theory needs only a cellular cup product and a $2$-cycle
of the tensor sheaf, both of which descend from the hypergraph product
$G_x\times G_y$ to the balanced product $(G_x\times G_y)/H$ of two sheaf codes
with a common free group symmetry $H$ (Proposition~\ref{prop:twisted_balanced});
the whole construction, including the gauging-measurement protocol, therefore
extends to $2$D balanced-product codes, with $|H|$ times fewer qudits. When
$|H|$ is coprime to the characteristic there is an explicit K\"unneth basis of
the logical qudits labelled by pairs of $H$-characters $(\chi,\bar\chi)$
(Proposition~\ref{prop:balanced_kunneth}), and the logical CZ structure is that of
the hypergraph product restricted by the selection rule $\chi\psi\phi=1$ on the
three characters and rescaled by $|H|^{-1}$ (Proposition~\ref{prop:balanced_CZ});
for balanced products of subdivided codes it reproduces the cluster structure of
Construction~\ref{con:2d-hgp} on the $H$-invariant labels
(Corollary~\ref{cor:balanced_subdivided_CZ}).

\subsection{Organization}
\label{sec:organization}

The paper is organized as follows.
Section~\ref{sec:classical} develops the sheaf-theoretic description of classical
codes: Sipser-Spielman codes and their local codes, the identification of the
Tanner code with sheaf homology and of its transpose with sheaf cohomology, and
the subdivision of an arbitrary classical code into a sheaf code on a graph.
Section~\ref{sec:prelim} collects the algebraic-topological tools: graded
incidence posets, presheaves and sheaf codes, and the cup product, cap product and
chain-cochain pairing with local coefficients.
Section~\ref{sec:untwisted_code} constructs the untwisted qLDPC codes on the
hypergraph product of two sheaf codes, introduces the operator-valued sheaf
cochains used throughout, derives the $X$- and $Z$-stabilizers from $d^0$ and
$d^1$, and interprets the resulting model as a topological defect network.
Section~\ref{sec:non-Abelian_LDPC} contains the twisted theory: the spacetime path
integral and its cohomology invariance, the sheaf cluster state and $0$-form
subcomplex SPT, the gauging procedure producing the twisted qLDPC code, the
resulting Clifford stabilizer code with its commutation relations and logical
operators, and an alternative derivation by gauging directly from the untwisted
code.
Section~\ref{sec:0form_subcomplex} studies the $0$-form subcomplex symmetry: the
transversal CZ operators and their capped support, their action as addressable
logical CZ gates in the untwisted code, and their role as charge parity operators
in the twisted code.
Section~\ref{sec:Fq_generalization} generalizes everything from $\FF_2$ to an
arbitrary finite field $\Fq$.
Section~\ref{sec:condition_cohomology} determines when the cohomology invariant is
non-trivial, reducing this to the local cycle condition and thence to the
multiplication property of the local codes.
Section~\ref{sec:code_instantiation} gives the explicit constant-rate
instantiation from arbitrary good classical codes, together with its logical CZ
structure and its rate and distance.
Section~\ref{sec:logical_operation} presents the addressable and parallel gauging
measurement protocol, the magic state fountain, and the derivation of the logical
action from the spacetime path integral.
Section~\ref{sec:high_yield} upgrades the magic rate to almost constant using
local codes with a multiplication property.
Section~\ref{sec:discussion} concludes with a discussion and outlook.
Appendix~\ref{app:acyclic} proves the local acyclicity of the subdivided sheaf,
and Appendix~\ref{app:alphabet} treats alphabet reduction from $\Fq$ to $\FF_2$.
Appendix~\ref{sec:balanced} generalizes the construction from hypergraph products
to $2$D balanced products of sheaf codes with a common group symmetry, instantiates
it with balanced products of subdivided codes, and derives the K\"unneth basis and
the logical CZ structure when the order of the group is coprime to the
characteristic.

%======================================================================
\section{Generalized Sheaf-theoretical Approach for Classical Codes}
\label{sec:classical}
%======================================================================

A classical code defined on a graph, which we call a graph-based code, is a classical error-correcting code that imposes one even-parity constraint at each vertex. It is known to be a special case of a Tanner code. A constant-rate graph-based code cannot be a good code since the graph's girth which gives the code distance is upper bounded by a logarithmic function \cite{Meshulam:2018}. More general Tanner codes replace these parity constraints by local linear codes on the incident edges. The Sipser--Spielman construction uses bipartite expanders and suitable local codes to obtain asymptotically good LDPC codes~\cite{Sipser_Spielman}.   
We express these local constraints through stalk spaces and incidence maps. This allows the same graph structure to enter a chain complex and, later, a product complex. A product of two graphs naturally gives a square complex. Its cup product can be defined with the chosen cubical convention or a specified compatible simplicial construction.

% \subsection{Graphs and expansion}
% \label{subsec:graphs-and-expansion}
% Let $G = (V = L \sqcup R, E)$ be a $\Delta$-regular bipartite graph with $|L| = |R| = n$. For $v \in V$, define the set of incident edges as $E(v)=\{v \in V: e \in E\}$.  Thus, $|E| = \Delta n$. We use the unnormalized second singular value convention: $G$ has second singular value at most $\lambda$ if, for all $A \subseteq L$ and $B \subseteq R$,
% \begin{equation}
%     \left|E(A, B) - \frac{\Delta}{n}|A||B| \right| \leq \lambda\sqrt{|A||B|},
% \end{equation}
% where $E(A,B)$ is the number of edges between $A$ and $B$. For every vertex $v \in L \sqcup R$, fix a bijection $p_v\,:\, E(v) \mapsto [\Delta]$ which labels each edge adjacent to a vertex an index in $[\Delta]$.

\subsection{Sipser--Spielman code}
\label{subsec:sipser-spielman-code}

We now define Sipser--Spielman codes and their local codes, and recall that such
a code is the same object as a classical \emph{sheaf} code. Throughout, $G=(V,E)$
is a finite simple graph. Let $E(v) \subseteq E$ denote the set of edges incident to a vertex $v$.
For $\vc{z} \in \FF_2^{E}$, we let $\vc{z}|_{E(v)}$ denote its restriction to
those edges. Edges are denoted by a single letter $e=[v,v']$ where $v$ and $v'$ are the two vertices that form the endpoint of $e$.

Let $\cC_0 \subseteq \FF_2^{\Delta}$ be a binary linear code with parameters
$[\Delta,k_0,d_0]$ and fix a full-row-rank parity-check matrix
$\Hs_0 \in \FF_2^{m_0 \times \Delta}$ with $m_0 = \Delta-k_0$, so that
$\cC_0 = \ker \Hs_0$.

\begin{definition}[Sipser--Spielman code]
\label{def:sipser-spielman-code}
Let $G$ be a $\Delta$-regular graph and $\cC_0$ a
length-$\Delta$ local code. Fix a bijection
$\iota_v:[\Delta]\longrightarrow E(v)$ at every vertex.
For $\vc{z}\in\FF_2^E$, define its ordered local word by
\[
\vc{z}|_{E(v)}
=(z_{\iota_v(1)},\ldots,z_{\iota_v(\Delta)}).
\]
The associated \emph{Tanner code} is
\begin{align}
\cT(G,\cC_0) = \Big\{\vc{z} \in \FF_2^{E} \ :\ &
\vc{z}|_{E(v)} \in \cC_0 \cr
& \text{for every } v \in V \Big\} ,
\end{align}
equivalently the set of $\vc{z}$ with
$\Hs_0\,\big(\vc{z}|_{E(v)}\big) = 0$ at every vertex $v$.
\end{definition}

Bits sit on edges and checks on vertices, so $\cT(G,\cC_0)$ has $N=|E|=\Delta
|V|/2$ bits and at most $M=m_0|V|$ checks. Note that the checks need not be independent. In particular, \[\dim\cT(G,\cC_0)\ge N-m_0|V| =\left(\frac{2k_0}{\Delta}-1\right)N.\] This rank bound gives a positive rate lower bound when $k_0>\Delta/2$.
Writing $w_{\Hs}$ for the largest row weight of $\Hs_0$,
every global check row has weight at most $w_{\Hs}$ and every global column has weight at most $2m_0$. Thus, the code is LDPC when
$\Delta=O(1)$.

The collection of local codes attached to the vertices is also called a local
system, or a \emph{sheaf} $\Fc$. It is known that a Sipser--Spielman code is
computed by sheaf (co)homology. We recall the associated
sheaf chain and cochain descriptions~\cite{Meshulam:2018}.

\subsubsection{Relation to sheaf homology}\label{sec:sheaf-homology}

We allow a different local code at every vertex: each $v\in V$ carries a linear
subspace $\cC_v \subseteq \FF_2^{E(v)}$, presented as $\cC_v = \ker \Hs_v$ with
$\Hs_v : \FF_2^{E(v)} \to \FF_2^{m(v)}$ a parity-check matrix of $m(v)$ rows. We
write
\be\label{eq:local_col}
\lv_e \ \equiv\ \text{the column of } \Hs_v \text{ indexed by } e \in E(v) ,
\ee
suppressing the vertex on $\lv_e$, which is always the one under discussion. The
sheaf code is
\begin{align}
\C \;=\; \cT\big(G,\{\cC_v\}_{v\in V}\big) = \Big\{
\vc{z} = (z_e)_{e\in E} \in \FF_2^{E} : \cr
\vc{z}|_{E(v)} \in \cC_v \ \ \forall\, v \in V \Big\} ,
\end{align}
and the \emph{local codeword condition} at a vertex $v$ reads
\be\label{eq:local_code_condition_1}
\Hs_v\,\big(\vc{z}|_{E(v)}\big) = 0 ,
\ee
or, expanded in the columns \eqref{eq:local_col},
\be\label{eq:local_code_condition_2}
\sum_{e \in E(v)} z_e\, \lv_e \;=\; 0 .
\ee

The sheaf that computes this code assigns to each cell the space in which the
corresponding degree of freedom lives: an $\FF_2$-bit to each edge and the $m(v)$
local checks to each vertex,
\be\label{eq:local_system}
\Fc(\sigma) =
\begin{cases}
\FF_2^{m(v)} , & \sigma = v \in V \quad (\text{check}) , \\[2pt]
\FF_2 , & \sigma = e \in E \quad (\text{bit}) .
\end{cases}
\ee
Let $\Pi_e:\FF_2^{E(v)}\longrightarrow\FF_2$ be the
coordinate projection. The restriction map is its column
insertion followed by the local parity check:
\begin{equation}\label{eq:restriction_map}
\bar\Fc_{v\shortleftarrow e}
=\Hs_v \circ \Pi_e^T:\Fc(e)\longrightarrow\Fc(v).
\end{equation}
Consequently,
\begin{equation}\label{eq:restriction_map_column}
\bar\Fc_{v\shortleftarrow e}(z_e)=z_e\lv_e.
\end{equation}

 We now define a chain complex of sheaf $\Fc$ on the graph $G$ as
\begin{align}
C_1(G,\Fc)
:= \bigoplus_{e\in E}\Fc(e)
\xrightarrow{\;\partial_1\;}
C_0(G,\Fc)
:= \bigoplus_{v\in V}\Fc(v)
\end{align}
where, for a 1-chain $c_1 = \sum_{e \in E} z_e \cdot e \in C_1(G,\Fc)$, the component of its boundary at \(v\) is
\begin{align}
(\partial_1c_1)(v)
:=\sum_{e\in E(v)}
\bar{\Fc}_{v\shortleftarrow e}(z_e);
\end{align}
the signs are omitted because the coefficients are in $\FF_2$. Write $\Hs_\Fc$ for the matrix of $\partial_1$ in the fixed stalk bases. Its rows are indexed by a vertex and a local check coordinate; its columns are indexed by graph edges. The corresponding cochain differential is $d^0=\Hs_\Fc^T$.

\begin{proposition}\label{Prop:sheaf_homology}
(Meshulam \cite{Meshulam:2018}) The sheaf code is the first sheaf homology group,
\be
\cT\big(G,\{\cC_v\}_{v\in V}\big) \;=\; H_1(G,\Fc) .
\ee
\end{proposition}

\begin{proof}
A $1$-chain is $c_1 = \sum_{e \in E} z_e \cdot e \in C_1(G,\Fc)$, the signs being
irrelevant over $\FF_2$. Its boundary is a $0$-chain, and evaluating it at a
vertex gives, by the definition \eqref{eq:restriction_map} of the restriction
map,
\begin{align}\label{eq:boundary_at_vertex}
&\text{for each vertex } v: \cr
&(\partial_1 c_1)(v) = \sum_{e \in E(v)}
\bar{\Fc}_{v \shortleftarrow e}(z_e)
= \sum_{e \in E(v)} z_e\, \lv_e .
\end{align}
Hence the global $1$-cycle condition $\partial_1 c_1 = 0$ holds if and only if the
local codeword condition \eqref{eq:local_code_condition_2} holds at every vertex,
i.e.\ if and only if $(z_e)_{e\in E}$ is a codeword. Since $G$ has no $2$-cells,
$Z_1 = H_1$, and we obtain
$\cT(G,\{\cC_v\}_{v\in V}) \equiv \cT(G, \Fc)= H_1(G,\Fc)$. Accordingly we also call
$\cT(G,\{\cC_v\}_{v\in V})$ a \emph{$1$-cycle code}.
\end{proof}

Equation \eqref{eq:boundary_at_vertex} is the statement we shall use repeatedly:
\emph{global} homological conditions on a sheaf code are \emph{local} codeword
conditions, one per vertex. In the trivial case
$\lv_e = 1$ and $\Fc(\sigma)=\FF_2$ for all $\sigma$, it reduces to the familiar
parity constraint $\sum_{e \in E(v)} z_e = 0$, which in terms of the Pauli
operators $Z_e$ with eigenvalues $\pm1$ is the stabilizer condition
$\prod_{e \in E(v)} Z_e = 1$. We emphasize that in
Proposition~\ref{Prop:sheaf_homology} the local codes $\cC_v$ need not agree at
different vertices.

\subsubsection{Relation to sheaf cohomology}\label{sec:sheaf-cohomology}

Transposing the same data gives a second code, which will supply the $X$-type
checks of the quantum codes and the cocycles that label their logical operators.
The \emph{co-restriction map} is the transposed column,
\be\label{eq:co-restriction_map}
\Fc_{v \shortrightarrow e} : \Fc(v) \to \Fc(e) ,
\qquad
\Fc_{v \shortrightarrow e} = \Pi_{e} \circ \Hs_v^{\top} \equiv \lv_e^{\top} ,
\ee
i.e.\ the row vector dual to $\lv_e$, and on a $0$-cochain
$f(v) \in \Fc(v) = \FF_2^{m(v)}$ it acts as
\be\label{eq:co-restriction_map_column}
\Fc_{v \shortrightarrow e}\big(f(v)\big) = \lv_e^{\top} f(v) \ \in \FF_2 .
\ee
The sheaf $\Fc$ is the one of Eq.~\eqref{eq:local_system}: now the $m(v)$ bits
sit on the vertex $v$ and the single check sits on the edge $e$, so the roles of
bits and checks are exchanged relative to Sec.~\ref{sec:sheaf-homology}. The
resulting code is the transpose $\C^{\top}=\cT^*(G,\{\cC_v\}_{v\in V})$, with
global parity-check matrix $\Hs_\Fc^T$, presented by the length-one cochain
complex $C^{0} \xrightarrow{\ d^0 = \Hs_\Fc^T\ } C^{1}$.

\begin{proposition}\label{Prop:sheaf_cohomology}
The transposed sheaf code is the zeroth sheaf cohomology group,
\be
\cT^*\big(G,\{\cC_v\}_{v\in V}\big) \;=\; H^0(G,\Fc) .
\ee
\end{proposition}

\begin{proof}
A $0$-cochain is $f = \sum_{v\in V} f(v)\,\bar{v} \in C^0(G,\Fc)$, with
$\bar{v}$ the elementary $0$-cochain supported on $v$. Its coboundary is a
$1$-cochain, and evaluating it on an edge $e=[v,v']$ gives
\begin{align}\label{eq:sheaf_cohomology_condition}
&\text{for each edge } e=[v,v']: \cr
&(d^0 f)(e) = \Fc_{v \shortrightarrow e}\big(f(v)\big)
+ \Fc_{v' \shortrightarrow e}\big(f(v')\big) .
\end{align}
Hence $d^0 f = 0$ if and only if the two co-restrictions agree on every edge,
which is exactly the statement that the checks of $\C^{\top}$, one per edge, are
all satisfied. Since $G$ has no cells of degree $-1$, $Z^0=H^0$, and
$\cT^*(G,\{\cC_v\}_{v\in V}) = H^0(G,\Fc)$. We call it the
\emph{$0$-cocycle code}.
\end{proof}

\subsection{Subdividing Classical Codes into Sheaf Codes}\label{sec:subdivision}

\begin{figure*}
    \centering   \includegraphics[width=0.8\linewidth]{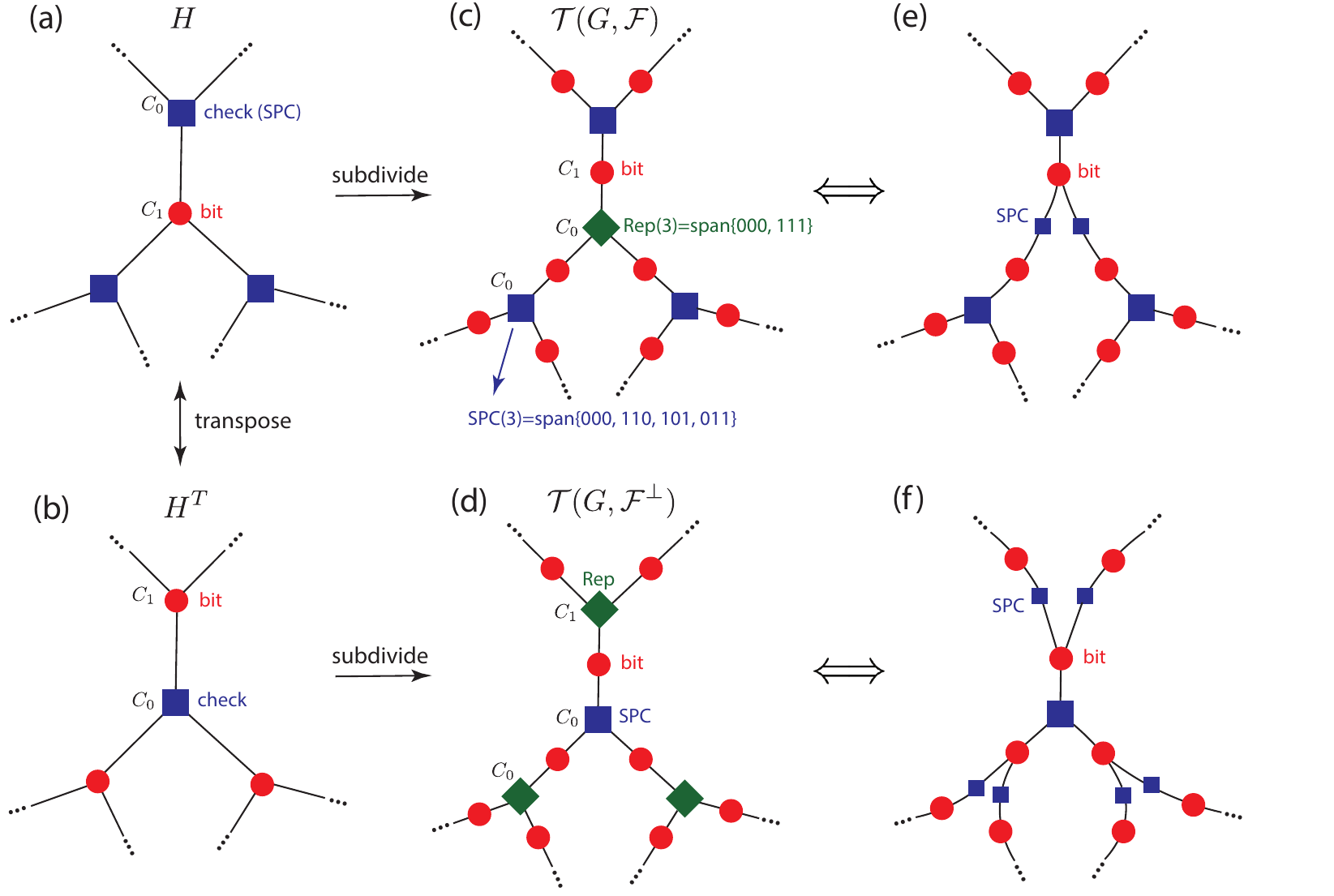}
    \caption{Subdivision of an arbitrary classical code into a sheaf
code. (a) The Tanner graph of $\C=\ker\Hs$, drawn as a hypergraph $G_h$: each
check is a vertex $v_h$ (blue square) carrying a single-parity-check (SPC)
constraint on $C_0$, and each bit is a hyperedge on $C_1$ (red dot) joining the
checks in which it participates. (b) The transposed code
$\C^{\top}=\ker\Hs^{\top}$ on the same hypergraph, with the roles of bits
(hyperedges) and checks (vertices) exchanged. (c) The subdivided sheaf code
$\T(G,\Fc)$: every hyperedge of (a) is split into $w_c$ ordinary edges meeting at
a new vertex $v'$ (green diamond), which carries the repetition local code
$\text{Rep}(w_c(v'))$ --- here $\text{Rep}(3)=\text{span}\{000,111\}$ --- while the
original check vertices (blue squares) carry
$\text{SPC}(w_r(v_h))=\text{Rep}(w_r(v_h))^{\perp}$, here
$\text{SPC}(3)=\text{span}\{000,110,101,011\}$. Bits now live on the edges, so $G$
is an honest graph and the code is a sheaf code. (d) The same subdivision applied
to $\C^{\top}$ produces $\T(G,\Fc^{\perp})$ on the \emph{identical} graph, with the
two families of local codes interchanged; since
$\text{Rep}=\text{SPC}^{\perp}$, this is exactly the dual sheaf. (e), (f)
Equivalent redrawings of (c) and (d) as bipartite graphs, making the
bit/check bipartition explicit.}
    \label{fig:subdivision}
\end{figure*}

We start with an arbitrary classical code $\C= \ker \ \Hs$, where $\Hs \in \FF_2^{m \times n}$ is the parity check matrix. In this work, we assume that $\Hs$
has no zero rows or columns, but allow arbitrary rank.
Full row rank is imposed separately where required.  We consider its associated Tanner graph, or equivalently a hypergraph $G_h=(V_h, E_h)$ with $|V_h| = m$ and $|E_h| = n$.  It corresponds to the following 2-term chain complex: 
\begin{align}
\C:   \qquad     C_1 & \xrightarrow[]{\partial_{1}=\mathsf{H}} C_{0} \ . \cr
            \text{bit}& \qquad \ \ \text{check}
\end{align}
Here, the check on $C_1$ corresponds to a vertex $v_h \in V_h$ (blue square) and the bit on $\cC_0$ corresponds to a hyperedge $ e \in E_h$ composed of several edges meeting at a junction (red dot), as illustrated in Fig.~\ref{fig:subdivision}(a). 
We will also consider the transposed classical code $\C^\top = \ker \ \Hs^\top$, where $\Hs^\top$ is the transpose of the original parity check matrix $\Hs$.  This 
corresponds to the following chain  complex:
\begin{align}
\C^\top:   \qquad     \cC_0 & \mathrel{\overset{\xrightarrow[]{d^0=\mathsf{H}}}{\underset{\xleftarrow[\partial_1=\mathsf{H}^\top]{}}{}}}      C_1 \ , \cr
           \ \text{check} & \qquad \quad \ \ \text{bit}
\end{align} 
where the roles of bits (hyperedges) and checks (vertices) are switched, as illustrated in Fig.~\ref{fig:subdivision}(b).

Our goal here is to transform the hypergraph $G_h$ associated with arbitrary classical code $\C$ into a graph $G$ with local coefficients. In other words, we are transforming the general 2-term chain complex associated with the hypergraph into a subdivided 2-term sheaved simplicial complex that now corresponds to a sheaf code. This would allow us to exploit the cup product defined for simplicial complexes. Note that the local cycle codes are allowed to vary by vertex, unlike the uniform local code in Definition~\ref{def:sipser-spielman-code}.

% Consider a hypergraph $G_h = (V_h, E_h)$ associated with the check matrix $H$ such that $|V_h| = m$ and $|E_h| = n$. Our convention is to let each vertex $v \in V_h$ to correspond to a check in $H$ and each hyperedge $e \in E_h$ correspond to a bit in $H$. 
The subdivision process proceeds in the following way.
For each bit $j \in [n]$, the hyperedge associated with $j$ in $G_h$ is subdivided into a set of $w_c$ edges, where $w_c$ is the column weight of $\Hs$, and a single new vertex $v'_j$ (green diamond) as illustrated in Fig.~\ref{fig:subdivision}(c). Each of the $w_c$ new edges connects $v'_j$ to one of the $w_c(v'_j)$ vertices corresponding to the checks adjacent to bit $j$ in the Tanner graph of $\Hs$.
Let $V'$ be the set of new vertices $\{v'\}$ introduced in this subdivision process.
After performing this subdivision for all $n$ bits, we obtain a subdivided graph $G = (V = V_h \sqcup V', E)$ with $|V| = n + m$ vertices and $|E|= \sum_{v'} w_c(v')$ edges.

We want the new classical code $\C'$ associated with $G$ to be similar to $\C=\ker \Hs$ and has similar codewords. Thus, all the new vertices $v'_j \in V'$ that we introduced have to impose the length-$w_j$ local repetition code on the $w_c(v')$ edges incident to each $v'$, whereas at $v_h$ it is the length-$t_h$ single-parity-check code. The repetition check has $w_j-1$ independent rows, so the associated syndrome stalk is $\FF_2^{w_j-1}$. This syndrome space is different from the length-$w_j$ space containing the incident edge word.Thus, we can define a sheaf structure on $G$:
\be
\Fc(\sigma )= \begin{cases}  \FF_2 & \sigma=e \in E, \\ 
 \FF_2 & \sigma=v_h \in V_h,  \\ 
 \FF_2^{w_c(v')-1} & \sigma=v' \in V', \end{cases}
\ee
where the stalk at each edge $e$ and original vertices $v$ is given a value in $\FF_2$, and every new vertex $v'$ is equipped with a subspace of $\FF_2^{w_c(v')-1}$. To be explicit, the vertex $v'$ with degree $w_c$ is only satisfied if its local coefficient lies in the local code space Rep($w_c$).  On the other hand, the check on each original vertex $v_h$  just imposes the usual single parity check constraint (SPC). We hence call the local code on $v_h$ the SPC code, which is denoted by SPC$(w_r(v_h))$ and has parameter $[w_r(v_h), w_r(v_h)-1, 2]$. Here,  $w_r(v_h)$ is the degree of the original vertex $v_h$ which equals the corresponding row weight of $\Hs$. Note that this is equivalent to the dual of the length-$w_r(v_h)$ repetition code. 
We denote this subdivided sheaf code by 
\be\label{eq:sub_Tanner_def}
\T(G,  \{C_v\}_{v\in V})\coloneqq \{c_1 \in \FF_2^{|E|}\,|\, \forall v \in V\,:\,c_1|_{v} \in C_v\},
\ee
where $c_1|_{v}$ is the restriction of the length-$|E|$ bitstring to the $\FF_2$-valued substring consisting of edges adjacent to the vertex $v$.

\begin{figure}[h]
    \centering   \includegraphics[width=1\linewidth]{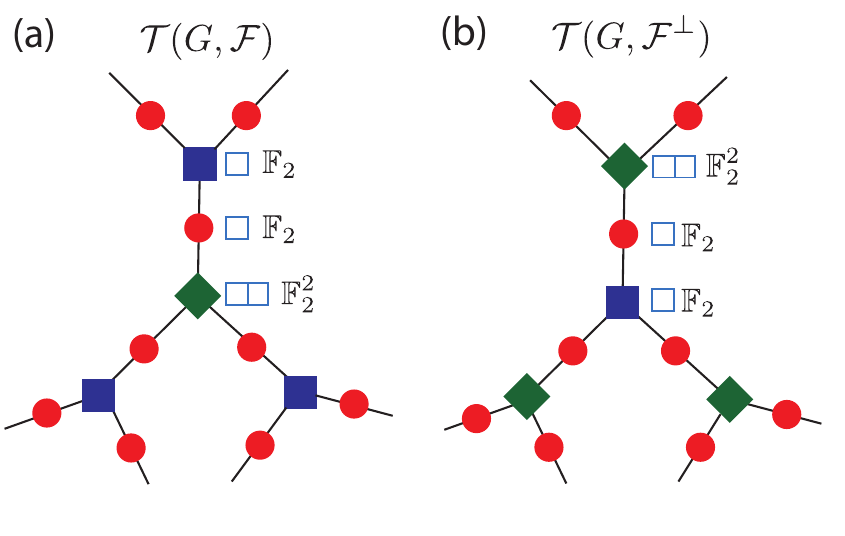}
    \caption{Stalk dimensions of the subdivided sheaf (blue
boxes), for (a) $\T(G,\Fc)$ and (b) $\T(G,\Fc^{\perp})$ on the same graph $G$ as
in Fig.~\ref{fig:subdivision}(c, d). Every edge carries the scalar stalk
$\Fc(e)=\FF_2$, holding the single bit on that edge. A vertex carrying the
single-parity-check code (blue square) imposes one constraint and hence has the
scalar stalk $\FF_2$, whereas a degree-$3$ vertex carrying the repetition code
(green diamond) imposes two independent constraints and hence has the
two-dimensional stalk $\FF_2^{2}$; in general
$\dim\Fc(v)=m(v)$ is the number of rows of the local parity-check matrix $\Hs_v$.
Passing from $\Fc$ to $\Fc^{\perp}$ exchanges the two local codes, and with them
the two stalk dimensions, on every vertex.}
\label{fig:stalk_distribution}
\end{figure}

\begin{figure*}
    \centering   \includegraphics[width=0.9\linewidth]{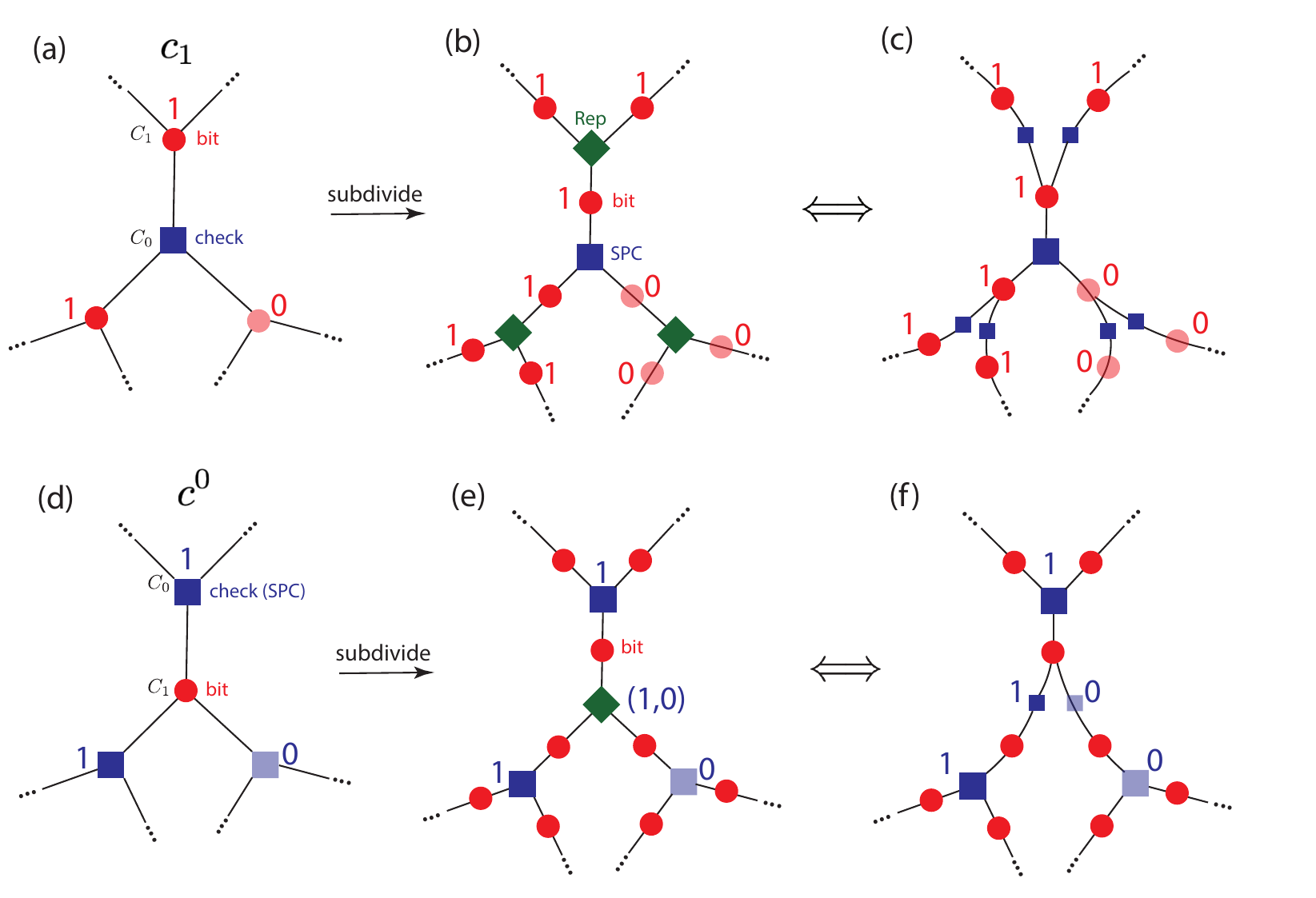}
    \caption{Codewords $c_1 \in \ker(\Hs)$ and
$c^0 \in \ker(\Hs^\top)$ being subdivided. (a) A codeword $c_1$ of the original
code, with bit values $1$ and $0$ on the hyperedges (red dots). (b) Its image
under subdivision: the repetition local code at each new vertex (green diamond)
forces all $w_c$ edges of a hyperedge to carry the common value of that bit, and
the SPC constraint at each original check vertex (blue square) is the original
parity check; the map is the isomorphism $\Phi$ of
Proposition~\ref{prop:subdivision-params}. (c) An equivalent redrawing of (b).
(d)--(f) The same for a codeword $c^0$ of the transposed code
$\C^{\top}=\ker\Hs^{\top}$, which subdivides on the same graph with the local
codes interchanged. Here the $0$-cochain is stalk-valued: the value on the new
vertex is the vector $(1,0) \in \FF_2^{2}$, not a single bit.}
    \label{fig:codeword_subdivided}
\end{figure*}

We can also do exactly the same subdivision to the transposed code $\C^\top = \ker \ \Hs^\top$, which leads to the same graph $G$ as shown in Fig.~\ref{fig:subdivision}(d). Note that the local code types (indicated by blue box and green diamond) on each vertex of the graph have been interchanged in the case of $\C$ and $\C^\top$. In fact, the two types of local codes on each vertex in the two cases are dual to each other: $\text{Rep}=\text{SPC}^{\perp}$, namely for any of the following two local codewords (1-chains) $a \in \text{Rep}(w)$ and $b \in \text{SPC}(w)$, $\sum_{i=1}^w a_i b_i=0$, which can be verified on Fig.~\ref{fig:subdivision}(c, d). 
Due to this dual relation of the local codes on each vertex, the subdivided sheaf code obtained from the transposed code $\C^\top = \ker \ \Hs^\top$ can be represented as $\T(G,  \{C^{\perp}_v\}_{v\in V})$, where $C_v^{\perp}$ denotes the dual code of $C_v$. We denote the coefficient system obtained from these dual local cycle codes by $\Fc^\perp$.

Another way to understand the relationship between the subdivided sheaf codes obtained from $\C$ and $\C^\top$ is from the Tanner graph picture. Instead of representing $\C$ as a hypergraph, we represent it with a bipartite graph with bits and checks on either side of the bipartition. Naturally, $\C^\top$ would just be a relabeling of bits and checks onto the opposite side of the bipartition. The subdivision process simply moves the bits from one half of the vertices onto the edges while turning those vertices into checks that impose repetition-code-like constraints. Because the edges and their connectivity do not change when we take the transpose of the classical code $\C$, we see that the number of bits in the subdivided code remains the same. The only difference between the subdivided sheaf code of $\C$ and $\C^\top$ is the swapping of local codes between the original bit vertices and the original check vertices. And because the repetition code and SPC are dual with respect to each other, it is equivalent to taking the dual of the sheaf.

\begin{proposition}[Parameters of the subdivided sheaf codes]
\label{prop:subdivision-params}
Let $\Hs\in\FF_2^{m\times n}$ have rank $r$ and no zero rows
or columns. Let $w_j=w_c(v'_j), t_h=w_r(v_h)$ and set
\begin{align*}
 \C=\ker\Hs&,\quad \C^\top=\ker\Hs^\top,\quad \\
 k=n-r&,\quad k^\top=m-r, \\
 \quad N=\sum_{j=1}^n w_j&=\sum_{h=1}^m t_h.
\end{align*}
Let $\Fc$ be the subdivision sheaf and let $\Fc^\perp$
use the dual local codes. Write $w_c^{\min}=\min_jw_j$, $w_c^{\max}=\max_jw_j$, $w_r^{\min}=\min_ht_h$, and $w_r^{\max}=\max_ht_h$. Both classical codes have block length $N$, and
\[
 H_1(G,\Fc)\cong\C,\qquad
 H_1(G,\Fc^\perp)\cong\C^\top.
\]
Their dimensions are $k$ and $k^\top$, respectively. Their minimum nonzero cycle weights are
\begin{align}
 d_1(\Fc)&=
 \min_{0\ne x\in\C}\sum_{j:x_j=1}w_j,\\
 d_1(\Fc^\perp)&=
 \min_{0\ne y\in\C^\top}\sum_{h:y_h=1}t_h.
\end{align}
In particular, writing $d=d(\C)$ and $d^\top=d(\C^\top)$,
\begin{align}
 w_c^{\min}d&\le d_1(\Fc)\le w_c^{\max}d,\\ w_r^{\min}d^\top&\le d_1(\Fc^\perp) \le w_r^{\max}d^\top.
\end{align}
\end{proposition}
\begin{proof}
Let $E=\{(h,j):\Hs_{h,j}=1\}$ be the Tanner-graph edge set. Both coefficient systems assign $\FF_2$ to each edge, so both edge codes lie in $\FF_2^E$ and have block length \[N=|E|=\sum_j w_j=\sum_h t_h.\] Since $G$ has no $2$-cells, first homology equals the cycle kernel. Define $\Phi:\C\to\FF_2^E$ by $\Phi(x)_{(h,j)}=x_j$. Its incident values satisfy the repetition constraint at each bit vertex, while the SPC constraint at check vertex $h$ is precisely $(\Hs x)_h=0$. Conversely, every cycle has a common value $x_j$ on the edges at bit vertex $j$. This value is well-defined because $w_j\ge1$. The SPC constraints give $\Hs x=0$, so reading these common values inverts $\Phi$. Thus $H_1(G,\Fc)\cong\C$ and its dimension is $n-r=k$. For $\Fc^\perp$, repetition constraints occur at the original check vertices and SPC constraints at the original bit vertices. The map $\Psi:\C^\top\to\FF_2^E$ defined by $\Psi(y)_{(h,j)}=y_h$ satisfies these constraints because $(\Hs^\top y)_j=0$. Reading the common value at each check vertex gives its inverse, using $t_h\ge1$. Hence $H_1(G,\Fc^\perp)\cong\C^\top$ and its dimension is $m-r=k^\top$. The maps copy each input bit to its incident edges, giving \[|\Phi(x)|=\sum_{j:x_j=1}w_j,\qquad |\Psi(y)|=\sum_{h:y_h=1}t_h.\] Minimizing over nonzero codewords proves the exact distance formulas. Bounding each degree by its minimum gives the lower bounds. Evaluating the weighted sum on a minimum-weight input codeword gives the upper bounds.
\end{proof}

With the notation and hypotheses of
Proposition~\ref{prop:subdivision-params}, fix the following local
parity-check presentations. For a length-$w$ repetition code use
$R_w=[I_{w-1}\mid\mathbf1]$, with $R_1$ the empty $0\times1$
matrix; for a length-$t$ SPC code use $(1,\ldots,1)$.
Use the same conventions for $\Fc^\perp$ after interchanging
the local cycle codes. These local check matrices have full row rank.

\begin{lemma}[The subdivision complex]
\label{lemma:subdivision-normal-form}
The subdivision chain complex is chain-isomorphic to
\[
 \bigl(\FF_2^n\xrightarrow{\Hs}\FF_2^m\bigr)
 \oplus
 \bigl(\FF_2^{N-n}\xrightarrow{I}\FF_2^{N-n}\bigr).
\]
In particular it is chain-homotopy equivalent to the original
complex. The assembled boundary $\Hs_\Fc$ has size
$(N-n+m)\times N$ and rank $N-n+r$. Consequently,
\[
\begin{array}{c|cccc}
 &\dim H_0&\dim H_1&\dim H^0&\dim H^1\\ \hline
 \Fc&m-r&n-r&m-r&n-r\\
 \Fc^\perp&n-r&m-r&n-r&m-r
\end{array}
\]
Here cochains use the transposed boundary matrices in the fixed
stalk bases. In particular, if $\Hs$ has full row rank, then
\begin{align}
 H^0(G,\Fc)&=0,& \dim H^1(G,\Fc)&=n-m,\\
 \dim H^0(G,\Fc^\perp)&=n-m,& H^1(G,\Fc^\perp)&=0.
\end{align}
\end{lemma}
\begin{proof}
Order the $w_j$ edges at each bit vertex and distinguish the
last one. An edge block has the unique form
\[
 z_j=x_j\mathbf1_{w_j}+(u_j,0),\qquad
 x_j=(z_j)_{w_j},\quad u_j=R_{w_j}z_j.
\]
Thus the repetition syndromes are precisely $u=(u_j)_j$.
The original-check syndromes are $\Hs x+Qu$, where $Q$ sends
each auxiliary coordinate to the check incident to its edge.
Therefore, in these domain coordinates,
\[
 \Hs_\Fc(x,u)=(\Hs x+Qu,u).
\]
The invertible codomain change $(s,u)\mapsto(s+Qu,u)$ gives
$\operatorname{diag}(\Hs,I_{N-n})$. The dimension and rank
claims follow. Applying the same construction to $\Hs^\top$
gives $\operatorname{diag}(\Hs^\top,I_{N-m})$ for $\Fc^\perp$.
The cohomology dimensions equal the corresponding homology
dimensions because the cochain matrices are transposes.
For completeness, write $Lx=(x_j\mathbf1_{w_j})_j$ and
$Su=((u_j,0))_j$, so that $Q=AS$ for the original-check
incidence map $A$. A chain inclusion and projection are
\begin{align*}
 i_1x&=Lx,\quad i_0s=(s,0),\qquad \\
 p_1z&=((z_j)_{w_j})_j,\quad p_0(s,u)=s+Qu.
\end{align*}
They satisfy $pi=I$. The map $h(s,u)=Su$ satisfies
$I-ip=\partial h+h\partial$, degree by degree, and hence
exhibits the chain-homotopy equivalence explicitly.
\end{proof}

All Hamming weights count scalar coordinates in the specified
stalk bases. Write $d^0$ for minimum nonzero zero-cocycle
weight. Let $\Delta=\max\{w_c^{\max},w_r^{\max}\}$.
The common-edge-value map gives isomorphisms
\[
 H^0(G,\Fc^\perp)\cong H_1(G,\Fc),\qquad
 H^0(G,\Fc)\cong H_1(G,\Fc^\perp).
\]
Indeed, at each vertex the transposed local check matrix for
$\Fc^\perp$ maps its stalk bijectively onto the local cycle
code for $\Fc$. The cocycle condition says that the endpoint
values agree on every edge. Conversely, every cycle has a
unique lift at every vertex. If $b$ lifts the edge word $c$,
each nonzero edge of $c$ has two nonzero endpoint stalks, so
\[
 2|c|\le\Delta|b|.
\]
Consequently,
\[
 d^0(\Fc^\perp)\ge\frac{2}{\Delta}d_1(\Fc),\qquad
 d^0(\Fc)\ge\frac{2}{\Delta}d_1(\Fc^\perp).
\]
For the systematic local checks fixed above, the lift of
$x\in\ker\Hs$ to $H^0(G,\Fc^\perp)$ retains the scalar $x_j$
at every original bit vertex. Its other coordinates are
selected incident copies of those bits. Hence also
\[
 d\le d^0(\Fc^\perp)\le(1+w_c^{\max})d.
\]
Similarly,
$ d^\top\le d^0(\Fc)\le(1+w_r^{\max})d^\top$.
Replacing a full-row-rank local check matrix by another basis
of the same row space preserves these cohomology spaces up
to an invertible stalk map. It need not preserve coordinate
weights exactly. For bounded $\Delta$, the change is
bounded by a constant.

%======================================================================
\section{Sheaves and Cup Products}
\label{sec:prelim}
%======================================================================

% Throughout, $\ZZ_2$ denotes a finite field with $q = 2^s$ elements (characteristic~$2$). The main results of this paper hold over $\ZZ_2$ ($s=1$); we state definitions for general $\ZZ_2$ but specialize to $\ZZ_2$ in all proofs. We recall the definitions from~\cite{LiShaoWeiLiLiu2026, LiLiLiuNguyen2025} adapted to our setting.

\subsection{Graded incidence posets and cochain complexes}

\begin{definition}[Graded incidence poset]\label{def:poset}
A finite \emph{graded poset} is a set $X$ equipped with a partial order $\preceq$ and a grading function $g: X \longrightarrow \mathbb Z_{\ge 0}$ such that $g(\sigma) \leq g(\tau)$ whenever $\sigma \preceq \tau$. We write $X(i) = \{\sigma \in X : g(\sigma) = i\}$ for the $i$-cells. The poset is a \emph{graded incidence poset} if for every $\sigma \prec \pi$ with $g(\pi) = g(\sigma) + 2$, there exists an \emph{even} number of $\tau$ with $\sigma \prec \tau \prec \pi$. The \emph{dimension} of $X$ is the maximum value of $g$.
\end{definition}

\begin{definition}[Combinatorial cochain complex]\label{def:cochain-scalar}
Given a graded incidence poset $X$ and a finite field $\FF_q$, define $C^i(X, \FF_q) = \FF_q^{|X(i)|}$ with coboundary operator $d^i: C^i \to C^{i+1}$ defined by
\begin{equation}\label{eq:coboundary-scalar}
d^i(\sigma) = \sum_{\tau \in X(i+1),\; \tau \dotsucc \sigma} \tau.
\end{equation}
where $\tau\dotsucc\sigma$ means $g(\tau) = g(\sigma) + 1$.  Note that for simplicity we consider $\FF_q$ with characteristic 2 such that we can ignore all the signs in the cobboundary and boundary maps. 
The even incidence condition ensures $d^{i+1} \circ d^i = 0$. 
\end{definition}

\subsection{Presheaves and sheaf codes}

\begin{definition}[Presheaf / local coefficient system]\label{def:presheaf}
Given a poset $X$, a \emph{system of local coefficients} (or presheaf) $\mathcal{F}$ consists of:
\begin{enumerate}[label=(\roman*)]
\item For each $\sigma \in X$, a finite-dimensional $\FF_q$-vector space $\mathcal{F}(\sigma) = V_\sigma$.
\item For each $\sigma \preceq \tau$, a co-restriction map $\mathcal{F}_{\sigma \shortrightarrow \tau}: V_\sigma \to V_\tau$ satisfying:
\begin{itemize}
\item $\mathcal{F}_{\sigma \shortrightarrow \sigma} = \mathrm{id}_{V_\sigma}$,
\item $\mathcal{F}_{\tau \shortrightarrow  \pi} \circ \mathcal{F}_{\sigma \shortrightarrow  \tau} = \mathcal{F}_{\sigma \shortrightarrow \pi}$ whenever $\sigma \preceq \tau \preceq \pi$ \quad (functoriality). 
\end{itemize}

\item For each $\sigma \preceq \tau$, a restriction map $\bar{\mathcal{F}}_{\sigma \shortleftarrow \tau}: V_\tau \to V_\sigma$ satisfying:
\begin{itemize}
\item $\bar{\mathcal{F}}_{\sigma \shortleftarrow \sigma} = \mathrm{id}_{V_\sigma}$,
\item $  \bar{\mathcal{F}}_{\sigma \shortleftarrow  \tau}  \circ \bar{\mathcal{F}}_{\tau \shortleftarrow  \pi} = \bar{\mathcal{F}}_{\sigma \shortleftarrow \pi}$ whenever $\sigma \preceq \tau \preceq \pi$. 
\end{itemize} \end{enumerate}
Fix a basis of each stalk. We require the two incidence maps to be adjoint for the coordinate pairings:
\[\bar\Fc_{\sigma\shortleftarrow\tau} =\bigl(\Fc_{\sigma\shortrightarrow\tau}\bigr)^T.\]
With matching incidence signs, this gives $\partial_i=(d^{i-1})^T$. 
\end{definition}

\begin{definition}[Cochain complex with local coefficients]\label{def:cochain-sheaf}
Given $X$ and $\mathcal{F}$, define
\begin{equation}\label{eq:Ci-sheaf}
C^i(X, \mathcal{F}) = \bigoplus_{\sigma \in X(i)} V_\sigma.
\end{equation}
An element $i$-cochain $c^i \in C^i(X, \mathcal{F})$ has a \emph{local component} $c^i(\sigma) \in V_\sigma$ for each $\sigma \in X(i)$. The coboundary operator $d^i: C^i(X, \mathcal{F}) \to C^{i+1}(X, \mathcal{F})$ is defined by
\begin{equation}\label{eq:coboundary-sheaf}
(d^i c^i)(\tau) = \sum_{\sigma \in X(i),\; \sigma \precdot \tau} \mathcal{F}_{\sigma \shortrightarrow \tau}(c^i(\sigma)).
\end{equation}
The functoriality condition and even incidence ensure $d^{i+1} \circ d^i = 0$. In matrix form, $d^i$ is a block matrix with entries $\mathcal{F}_{\sigma,\tau}$ for $\sigma \prec \tau$ and zero otherwise.
\end{definition}

The \emph{chain complex} $C_\bullet(X, \mathcal{F})$ has the same underlying spaces $C_i = \bigoplus_\sigma V_\sigma$. An element of the chain group is an $i$-chain $c_i \in C_i(X, \Fc)$ with its local component $c_i(\sigma) \in V_\sigma$ for each $\sigma \in X(i)$. The boundary operator is defined as
\begin{equation}\label{eq:boundary-sheaf}
(\partial_{i} c_i)(\sigma) = \sum_{\tau \in X(i),\; \tau \dotsucc \sigma} \bar{\mathcal{F}}_{\sigma \shortleftarrow \tau}(c_i(\tau)),
\end{equation}
so that $\partial_{i} = (d^{i-1})^T$ in matrix form.

\begin{definition}[Quantum sheaf code]\label{def:sheaf-code}
A \emph{quantum sheaf code} is a CSS code defined by extracting three consecutive terms from the cochain complex $C^\bullet(X, \mathcal{F})$:
\begin{equation}\label{eq:sheaf-code}
C^{p-1}(X, \mathcal{F}) \xrightarrow{d^{p-1}} C^p(X, \mathcal{F}) \xrightarrow{d^p} C^{p+1}(X, \mathcal{F}),
\end{equation}
with $H_Z = d^p$ and $H_X^T = d^{p-1}$.
\end{definition}

\subsection{The cup product}

We define the cup product first for simplicial complexes:

\begin{definition}[Cup product on simplicial complexes]\label{def:cup-simplicial}
Let $X$ be a simplicial complex with presheaves $\mathcal{F}$ and $\mathcal{G}$. The \emph{cup product}
\[
\cup: C^{p_1}(X, \mathcal{F}) \times C^{p_2}(X, \mathcal{G}) \to C^{p_1+p_2}(X, \mathcal{F} \otimes \mathcal{G})
\]
is defined as follows. Let $\sigma$ be a $(p_1+p_2)$-simplex with ordered vertices $v_0, v_1, \ldots, v_{p_1+p_2}$. Define the \emph{front $p_1$-face} $\sigma_{p_1} = [v_0, \ldots, v_{p_1}]$ and the \emph{back $p_2$-face} $\sigma_{p_2} = [v_{p_1}, \ldots, v_{p_1+p_2}]$. Then for $x_1 \in C^{p_1}(X, \mathcal{F})$ and $x_2 \in C^{p_2}(X, \mathcal{G})$:
\begin{equation}\label{eq:cup-simplicial}
(x_1 \cup x_2)(\sigma) = \mathcal{F}_{\sigma_{p_1} \shortrightarrow \sigma}(x_1(\sigma_{p_1})) \otimes \mathcal{G}_{\sigma_{p_2} \shortrightarrow \sigma}(x_2(\sigma_{p_2})).
\end{equation}
\end{definition}

\begin{proposition}[Properties of the cup product]\label{prop:cup-properties}
The cup product satisfies:
\begin{enumerate}[label=(\alph*)]
\item {Leibniz rule:} $d(x_1 \cup x_2) = (d x_1) \cup x_2 + x_1 \cup (d x_2)$.
\item {Associativity:} $(x_1 \cup x_2) \cup x_3 = x_1 \cup (x_2 \cup x_3)$.
\item {Descent to cohomology:} By the Leibniz rule, the cup product induces a well-defined map $H^{p_1}(X, \mathcal{F}) \times H^{p_2}(X, \mathcal{G}) \to H^{p_1+p_2}(X, \mathcal{F} \otimes \mathcal{G})$.
\end{enumerate}
\end{proposition}

\begin{proof}
These are standard results in algebraic topology. The Leibniz rule follows by direct computation using Definition~\ref{def:cup-simplicial} and the coboundary formula. Over characteristic~2, all signs are positive. Associativity follows from the fact that the front face of the front face is the front face, and similarly for back faces. Part~(c) follows immediately from (a): if $d x_1 = 0$ and $x_2 = d y_2$, then $x_1 \cup x_2 = x_1 \cup d y_2 = d(x_1 \cup y_2) + (d x_1) \cup y_2 = d(x_1 \cup y_2)$.
\end{proof}

\subsection{The cap product}
\label{subsec:cap-product}
We next introduce the cap product, which is the dual to the cup product and will be used later in the Poincar\'e duality map, following Ref.~\cite{li2025poincar}.

\begin{definition}[Cap product on simplicial complexes]
\label{def:cap-simplicial}
Let $X$ be a simplicial complex with presheaves $\Fc$ and $\mathcal{G}$. The \emph{cap product}
\[\frown\,:\,C^{p_1}(X, \Fc) \times C_{p_1 + p_2}(X, \Fc \otimes \mathcal{G}) \to C_{p_2}(X,\mathcal{G})\]
is defined as follows.
Let $\sigma$ be a $(p_1 + p_2)$-simplex with ordered vertices $v_0,\ldots,v_{p_1 + p_2}$. In addition, let the front $p_1$-face and back $p_2$-face be denoted as
\[\sigma_{p_1}=[v_0,\ldots,v_{p_1}],\qquad
\sigma_{p_2}=[v_{p_1},\ldots,v_{p_1+p_2}].
\]
For $x \in C^{p_1}(X,\Fc)$ and $y \in C_{p_1 + p_2}(X,\Fc \otimes \mathcal{G})$, let $y(\sigma) = \sum_{\ell} a_{\ell} \otimes b_{\ell}$ with $a_\ell \in \Fc(\sigma)$ and $b_\ell \in \mathcal{G}(\sigma)$. Then,
\begin{equation}
    (x \frown y)(\sigma_{p_2}) = \sum_{\ell} \langle \Fc_{\sigma_{p_1}\shortrightarrow\sigma}(x(\sigma_{p_1})),\,a_{\ell} \rangle\,\bar{\mathcal{G}}_{\sigma_{p_2} \shortleftarrow \sigma}(b_\ell).
\end{equation}
\end{definition}

\begin{proposition}[Properties of the cap product]
\label{prop:cap-properties}
    The cap product satisfies:
    \begin{enumerate}[label=(\alph*)]
        \item {Leibniz rule:} $\partial(x \frown y) = (d x) \frown y + x \frown (\partial y)$.
        \item {Compatibility with the cup product:} $(x_1 \cup x_2)\frown y = x_2 \frown (x_1 \frown y)$ whenever the degrees match.
        \item {Descent to (co)homology:} By the Leibniz rule, the cap product induces a well-defined map $H^{p_1}(X,\Fc) \times H_{p_1+p_2}(X,\mathcal{F}\otimes\mathcal{G}) \to H_{p_2}(X,\mathcal{G})$.
    \end{enumerate}
\end{proposition}

\begin{proof}
    These are standard properties dual to those of the cup product. The Leibniz rule follows from direct computation from Definition~\ref{def:cap-simplicial} together with the boundary and coboundary formulas. Compatibility with the cup product follows from the fact that taking successive front faces of a simplex agrees with taking the front face once and then capping the remaining back face. Part~(c) follows immediately from part~(a).
\end{proof}

% \subsection{The fundamental chain}

% \begin{definition}[Fundamental chain]\label{def:fundamental-class}
% The \emph{fundamental chain} of a $t$-dimensional cellular complex $X$ is
% \begin{equation}\label{eq:fundamental-class}
%   [X] := \sum_{\sigma \in X(t)} \sigma \;\in\; C_t(X, \mathcal{F}\otimes \Fc^{\perp}).
% \end{equation}
% \end{definition}
% \begin{remark}
%     The fundamental chain $[X]$ may not be a $t$-cycle.
% \end{remark}

\subsection{The pairing isomorphism}

\begin{definition}[Pairing]\label{def:pairing-sheaf}
The \emph{canonical pairing} $\langle \cdot, \cdot \rangle: C^i(X, \mathcal{F}) \times C_i(X, \mathcal{F}) \to \FF_q$ is defined by
\begin{equation}\label{eq:pairing-sheaf}
\langle x, y \rangle \equiv \int_x y = \sum_{\sigma \in X(i)} x(\sigma)^T y(\sigma),
\end{equation}
where $x(\sigma)^T y(\sigma)$ is the standard inner product on $V_\sigma$. The adjunction $\langle d^i x, z \rangle = \langle x, \partial_{i+1} z \rangle$ holds by $\partial_{i+1} = (d^i)^T$.
\end{definition}

\begin{proposition}[Non-degeneracy on (co)homology]\label{prop:nondegen-sheaf}
The induced pairing $\langle \cdot, \cdot \rangle: H^i(X, \mathcal{F}) \times H_i(X, \mathcal{F}) \to \FF_q$ is non-degenerate.
\end{proposition}

\begin{proof}
Define $\Phi: H^i(X, \mathcal{F}) \times H_i(X, \mathcal{F}) \to \FF_q$  by $\Phi([\eta])([x]) = \langle x, \eta \rangle$. 

This is well-defined and (co)boundary invariant: if $\eta' = \eta + \partial z$, then $\langle x, \eta' \rangle = \langle x, \eta \rangle + \langle x, \partial z \rangle = \langle x, \eta \rangle + \langle d x, z \rangle = \langle x, \eta \rangle$ (since $x \in \ker d$). Similarly for $x' = x + d w$.

To show $\Phi$ is injective: suppose $\langle x, \eta \rangle = 0$ for all $[x] \in H^i$, i.e., for all $x \in \ker d^i$. Then $\eta \in (\ker d^i)^\perp$. By the standard identity $(\ker A)^\perp = \im(A^T)$ applied to $A = d^i$, we get $\eta \in \im(\partial_{i+1})$, so $[\eta] = 0 \in H_i$. Since $\dim H_i = \dim H^i$ (both equal $\dim C^i - \rk d^i - \rk d^{i-1}$), $\Phi$ is an isomorphism.
\end{proof}

\begin{definition}[Invariant form]\label{def:invariant-form}
Given cocycles $x_{p_1} \in C^{p_1}(X, \mathcal{F}_1), \ldots, x_{p_\rho} \in C^{p_\rho}(X, \mathcal{F}_\rho)$ and a cycle $\eta \in C_{p_1+\cdots+p_\rho}(X, \mathcal{F}_1 \otimes \cdots \otimes \mathcal{F}_\rho)$, an \emph{invariant form}  is
\begin{equation}\label{eq:invariant-form}
 \langle x_{p_1} \cup \cdots \cup x_{p_\rho},\, \eta \rangle \equiv \int_{\eta} x_{p_1} \cup \cdots \cup x_{p_\rho}.
\end{equation}
By the Leibniz rule and the adjunction, the above expression descends to a multilinear form on cohomology classes.
\end{definition}

%======================================================================
\section{Untwisted qLDPC codes on sheaf complex}\label{sec:untwisted_code}
%======================================================================

Before turning to the twisted models, we set up the \emph{untwisted} quantum sheaf
code of Definition~\ref{def:sheaf-code} as an explicit Pauli stabilizer model, and
identify its $X$- and $Z$-stabilizers with the rows of the two sheaf coboundary
operators. Throughout this section we work with a \emph{single} system of local
coefficients; the three copies $\Fc^{\rd}$, $\Fc^{\bl}$, $\Fc^{\gr}$ needed to
define the twist are introduced only in Sec.~\ref{sec:non-Abelian_LDPC}.
Accordingly we drop the colour label and write $\Fc$ for the sheaf, and $A$, $B$
for the stabilizers.

Let $\L$ be a $\mathsf{d}=2$ dimensional cell (simplicial) complex with cell sets
$\L(0)=V$, $\L(1)=E$ and $\L(2)=F$, equipped with a system of local coefficients
$\Fc$ in the sense of Definition~\ref{def:presheaf}. The concrete example we have
in mind is the square complex obtained as the hypergraph product of two subdivided
classical sheaf codes, with the sheaf structure of
Eq.~\eqref{eq:local_system_2D}, but nothing below depends on that choice.
Extracting three consecutive terms of the cochain complex,
\be\label{eq:untwisted_complex}
C^{0}(\L, \Fc) \xrightarrow{\ d^{0}\ } C^{1}(\L, \Fc)
\xrightarrow{\ d^{1}\ } C^{2}(\L, \Fc) ,
\ee
gives a CSS code with $H_X^{T}=d^0$ and $H_Z=d^1$, whose physical qubits sit on the
$1$-cells, $X$-checks on the $0$-cells and $Z$-checks on the $2$-cells. Our task
is to write the corresponding stabilizer generators as Pauli operators.

\subsection{Hypergraph Product of Sheaf Codes}
Hypergraph product (HGP) codes are quantum low-density parity-check codes that can achieve constant rate~\cite{Tillich:2014_hyergraph_product}. These codes allow optimized QEC protocols~\cite{manes2025distance,tan2025effective,berthusen2025adaptive, liu2026achieving} and support many logical gates~\cite{quintavalle2023partitioning,xu2025fast,berthusen2025automorphism,golowich2025constant,tan2025single,li2025transversal,chang2026constant,blue2026full,koh2026achieving}.
Now we can subdivide an arbitrary HGP code into a quantum sheaf code. 
This can be simply achieved using the subdividing procedure presented in Sec.~\ref{sec:subdivision}: we first subdivide each classical code defined on a hypergraph $G_h$ into a sheaf code on a graph $G$, and then take the homological product of the subdivided codes. For input matrices of bounded row and column weight, the maps $i,p$ in Lemma~\ref{lemma:subdivision-normal-form}, and their transposes, send every scalar coordinate to $O(1)$ coordinates. Tensoring these maps, with the same degree and transpose conventions as the original hypergraph-product complex, therefore preserves its encoded dimension and changes its $X$- and $Z$-distances by at most constant factors. Indeed, a nonzero homology class projects to a nonzero original class, with at most constant weight increase, and inclusion gives the reverse case. The transposed maps prove the same statement for cohomology. The physical block length also changes by only a constant factor. These conclusions use the explicit chain maps, not the number of newly introduced cells alone.

For generality, we consider a 2D hypergraph product code formed by the tensor  product of two classical sheaf codes on  graphs: $\T (G_1, \{\cC_v\}_{v \in V_1})  $ and $ \T (G_2, \{\cC_{v'}\}_{{v'} \in V_2})$, where $G_1=(V_1, E_1)$ and $G_2=(V_2,E_2)$.  

Since we take a product of two graphs, i.e., $G_1$$\times$$G_2$, the cells $\tilde \sigma$ in the product complex  are composed of the product of cells in the graphs in the form of $\sigma_i \otimes \sigma_j$ where $i,j=0,1$. There are hence the following types of cells: 
\begin{enumerate}
\item 
Vertex: $ v \otimes v' \in \tilde V=V_1\otimes V_2$;
\item 
Horizontal edge: $e \otimes v' \in \tilde{E}_x = E_1 \otimes V_2$; 
\item
Vertical edge $v \otimes e' \in \tilde{E}_y = V_1 \otimes E_2$;
\item
Face: $ e$$\otimes$$e' \in \tilde{F}=E_1\otimes E_2$.
\end{enumerate}

The vertex stalk is the tensor product of the two local check spaces.
It does not impose separate horizontal and vertical cycle conditions.
Write a one-chain as $z=(z_x,z_y)$, with
\begin{align}
z_x\in C_1(G_1,\Fc_1)\otimes C_0(G_2,\Fc_2), \cr
z_y\in C_0(G_1,\Fc_1)\otimes C_1(G_2,\Fc_2).
\end{align}
Its cycle condition is
\begin{equation}\label{eq:product_one_cycle}
(\partial_1^{(1)}\otimes I)z_x
+(I\otimes\partial_1^{(2)})z_y=0.
\end{equation}
In particular, at $v\otimes v'$ it reads
\begin{align*}
&\sum_{e\succ v}
\bigl(\bar\Fc_{1,v\shortleftarrow e}\otimes I\bigr)
z_x(e,v')\\
&\quad+
\sum_{e'\succ v'}
\bigl(I\otimes\bar\Fc_{2,v'\shortleftarrow e'}\bigr)
z_y(v,e')=0 \cr
& \quad \text{in }\Fc_1(v)\otimes\Fc_2(v').
\end{align*}
Here $z_x(e,v')\in\Fc_2(v')$ and
$z_y(v,e')\in\Fc_1(v)$ when the graph edge stalks are scalar.
The two directional contributions may cancel. The stalks themselves are
therefore as follows.
\be\label{eq:local_system_2D}
\mathcal{F}(\tilde{\sigma}) = \begin{cases}
  \FF_2^{m(v)} \otimes \FF_2^{m(v')},   \quad  \tilde{\sigma} = v \otimes v' \in \tilde V \\      
  \FF_2 \otimes  \FF_2^{m(v')}=\FF_2^{m(v')},  \quad \tilde{\sigma} = e \otimes v' \in \tilde{E}_x\\
     \FF_2^{m(v)}  \otimes\FF_2=\FF_2^{m(v)},  \quad \ \tilde{\sigma} = v \otimes e' \in \tilde{E}_y \\
   \FF_2 \otimes \FF_2 \cong \FF_2,  \qquad \quad \ \ \tilde{\sigma} \in \tilde{F}.  
\end{cases}
\ee
Note that the degree of freedom associated with  $X$-check, qubit $Q$ and $Z$-check attached to vertex, edge and face corresponds to a rank-2 tensor, vector, and a scalar respectively.

\begin{figure*}
    \centering   \includegraphics[width=1\linewidth]{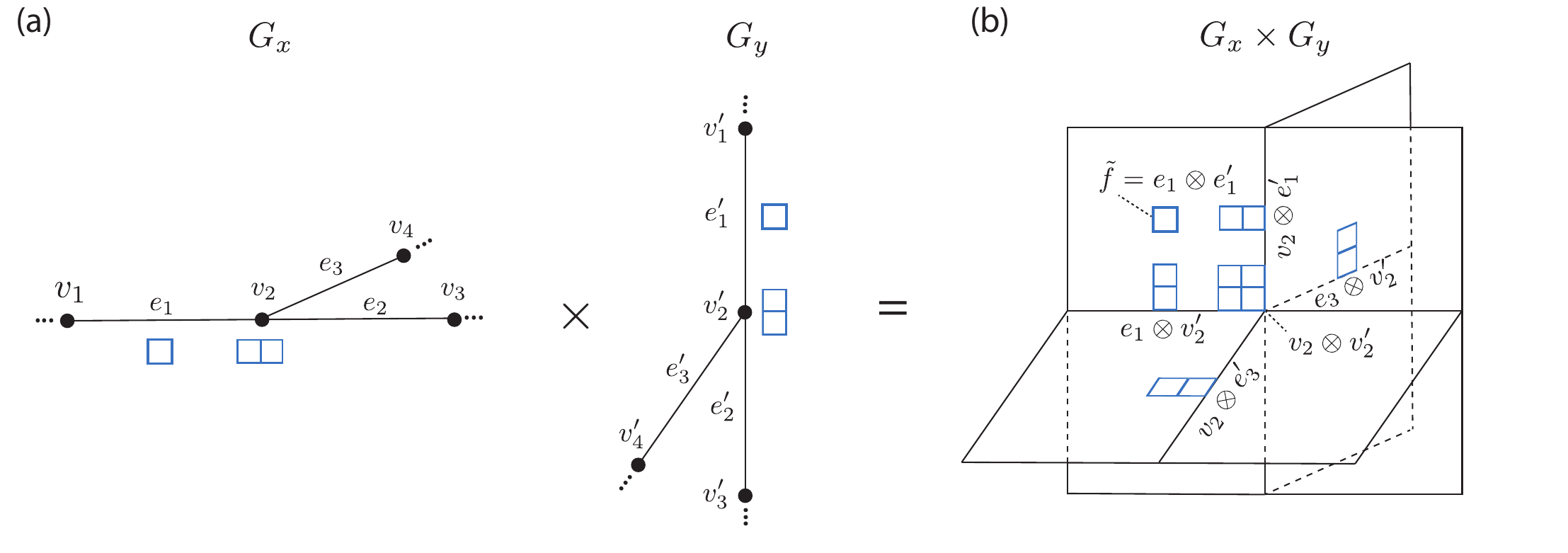}
    \caption{Hypergraph product of two sheaf codes. (a) The two
factor graphs $G_x$ (vertices $v_1,\dots,v_4$, edges $e_1,e_2,e_3$) and $G_y$
(vertices $v'_1,\dots,v'_4$, edges $e'_1,e'_2,e'_3$), each carrying a sheaf whose
stalk dimensions are indicated by the blue boxes: $\FF_2$ on an edge and
$\FF_2^{m(v)}$ on a vertex. (b) The product complex $\L=G_x \times G_y$, a
branched square complex whose cells are products of cells of the factors: faces
$\tilde f = e \otimes e'$, horizontal edges $e \otimes v'$, vertical edges
$v \otimes e'$, and vertices $v \otimes v'$. The stalks multiply accordingly
[Eq.~\eqref{eq:local_system_2D}]: a face carries the scalar
$\FF_2 \otimes \FF_2 \cong \FF_2$, an edge the vector space $\FF_2^{m(v')}$ or
$\FF_2^{m(v)}$, and a vertex the rank-two tensor
$\FF_2^{m(v)} \otimes \FF_2^{m(v')}$, drawn as a grid of boxes. Because a vertex
of $G_x$ may have degree larger than two, a single edge of $\L$ is contained in
more than two faces, so $\L$ is not the cellulation of a manifold.}
    \label{fig:HGP_illustration}
\end{figure*}

\subsection{Operator-valued sheaf cochains}\label{sec:op_valued_cochain}

Since the stalk $\Fc(e)$ on an edge is in general higher-dimensional, an edge
carries several qubits rather than one. We fix once and for all the standard basis
$\{\vec{e}_s\}_{s=1}^{m(\sigma)}$ of the stalk
$\Fc(\sigma)=\FF_2^{m(\sigma)}$, with $m(\sigma) \equiv \dim \Fc(\sigma)$ and
$\vec{e}_s$ the unit vector whose $s^\text{th}$ component equals $1$ and whose
other components vanish,
\be\label{eq:unit_vector_sheaf}
\vec{e}_s = \big[\,\underbrace{0, \ldots, 0}_{s-1},\, 1,\,
\underbrace{0, \ldots, 0}_{m(\sigma)-s}\,\big]^{T} .
\ee
The physical qubits are then labelled by pairs $(e,s)$ with $e \in \L(1)$ and
$s \in [\,m(e)\,]$, and we denote the associated Pauli operators by
$X_{e,s}, Y_{e,s}, Z_{e,s}$. We refer to $s$ as the \emph{vector index} of the
stalk; which stalk is meant is always fixed by the cell appearing alongside it.

Following the operator-valued cochain formalism of
Refs.~\cite{Hsin2024_non-Abelian, zhu2023non, Hsin2024:classifying,
zhu2025topological}, extended here to a system of local coefficients, we introduce
an operator-valued sheaf $k$-cochain $\hat{\lambda}^k \in C^k(\L, \Fc)$ which
assigns to each $k$-cell $\sigma$ a vector $\hat{\lambda}^k(\sigma) \in
\Fc(\sigma)$ whose components are commuting operators with eigenvalues $0$ or $1$.
For the $\ZZ_2$ gauge field on the edges, which is the case of interest here, we
write $\hat{a}^1 \in C^1(\L,\Fc)$ and relate its components to the Pauli-$Z$
operators by
\be\label{eq:Z_cochain_untwisted}
(-1)^{[\hat{a}^1(e)]_s}=Z_{e,s},
\qquad
[\hat{a}^1(e)]_s = \big[1-Z_{e,s}\big]/2 .
\ee
It is convenient to introduce the \emph{elementary sheaf cochain}
$\bar{\sigma}_s \in C^k(\L,\Fc)$, which takes the value $\vec{e}_s$ on the single
$k$-cell $\sigma$ and vanishes on all other $k$-cells; this is the sheaf
replacement for the indicator cochain $\bar{\sigma}$ of the scalar theory, and the
elementary cochains are in one-to-one correspondence with the physical qubits. Any
operator-valued cochain can then be expanded as
\be\label{eq:cochain_expansion_untwisted}
\hat{\lambda}^{k} = \sum_{\sigma \in \L(k)} \sum_{s=1}^{m(\sigma)}
\hat{N}_{\sigma,s}\, \bar{\sigma}_s ,
\qquad
(-1)^{\hat{N}_{\sigma,s}} = Z_{\sigma,s},
\ee
where $\hat{N}_{\sigma,s}$ is a quantum operator with eigenvalues
$N_{\sigma,s} \in \{0,1\}$ and $\bar{\sigma}_s$ is a classical variable. Recall
from Eq.~\eqref{eq:coboundary-sheaf} that the sheaf coboundary reads
\be\label{eq:coboundary_untwisted}
(d^{k}\hat{\lambda}^{k})(\tau)
= \sum_{\sigma \in \L(k),\ \sigma \precdot \tau}
\Fc_{\sigma \shortrightarrow \tau}\big(\hat{\lambda}^{k}(\sigma)\big),
\ee
so that, unlike the scalar case, every incidence $\sigma \precdot \tau$ is dressed
by the co-restriction map, e.g.\ by a column of the local parity-check matrix in
the hypergraph product complex.

Two elementary facts will be used repeatedly. First, the map $x \mapsto (-1)^x$ is
a homomorphism from $\FF_2$ to $\{\pm 1\}$, so a mod-$2$ sum of components
exponentiates into a \emph{product} of signs; at the operator level this requires
the components involved to commute, which they do, being simultaneously diagonal.
Second, an exponent taking values in $\{0,1\}$ acts as a selector on a Pauli
operator, $P^0=\I$ and $P^1=P$. Together these convert linear algebra over $\FF_2$
into products of Pauli operators.

\begin{figure}
    \centering   \includegraphics[width=1\linewidth]{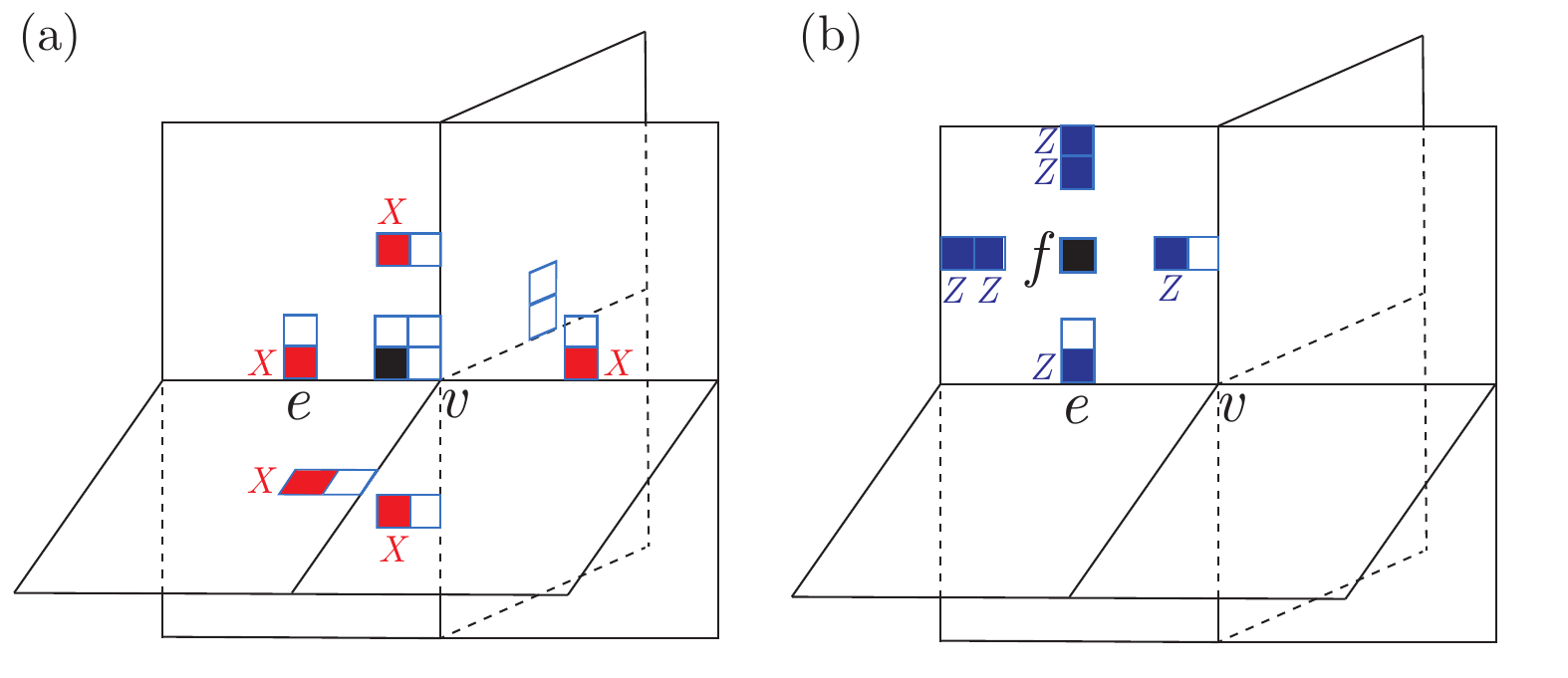}
    \caption{The two stabilizer types of the untwisted sheaf code on
a branched square complex, in which the edge $e$ is shared by three faces. Filled
boxes are stalk components; the qubits sit on the edge stalks. (a) The $X$-type
generator $A_{v,s}$ of Eq.~\eqref{eq:A_stabilizer_untwisted} is labelled by a
vertex $v$ together with a component $s$ of its stalk (black box at $v$), and
applies $X$ to those components of the incident edge stalks that are selected by
the co-restriction $\Fc_{v \shortrightarrow e}(\vec{e}_s)$ (red boxes). It is the
row $(v,s)$ of $H_X=(d^0)^{T}$. (b) The $Z$-type generator $B_f$ of
Eq.~\eqref{eq:B_stabilizer_untwisted} is labelled by a face $f$ alone, with no
vector index, because the stalk on a top cell is the scalar
$\Fc(f)\cong\FF_2$ (black box on $f$); it applies $Z$ to the edge stalk
components weighted by the functional $\Fc_{e \shortrightarrow f}$ (blue boxes),
and is the row $f$ of $H_Z=d^1$. Both weights are $O(1)$ for bounded vertex degree
and bounded stalk dimension, which is the qLDPC property.}
    \label{fig:square_complex}
\end{figure}

\subsection{$X$-stabilizers from $d^0$}\label{sec:X_stabilizer}

The $X$-type stabilizer generators are labelled by a vertex together with a vector
index of its stalk,
\be\label{eq:A_stabilizer_untwisted}
A_{v,s} = \prod_{e \dotsucc v} \prod_{s'}
\big(X_{e,s'}\big)^{\left[\Fc_{v \shortrightarrow e}(\vec{e}_s)\right]_{s'}} ,
\ee
where $s'$ runs over the stalk of the edge $e$. The structure of
Eq.~\eqref{eq:A_stabilizer_untwisted} is dictated by $d^0$. Evaluating the
coboundary \eqref{eq:coboundary_untwisted} on the elementary cochain $\bar{v}_s$, which
vanishes on every vertex other than $v$ where it takes the value $\vec{e}_s$, the
sum collapses to a single term,
\be\label{eq:A_coboundary}
\big(d^0 \bar{v}_s\big)(e) =
\begin{cases}
\Fc_{v \shortrightarrow e}(\vec{e}_s) \in \Fc(e), & e \dotsucc v,\\[2pt]
0, & \text{otherwise},
\end{cases}
\ee
which is why $A_{v,s}$ is supported on the edges incident to $v$, with the
exponents shown. Because $X_{e,s'}$ flips the $\FF_2$ variable
$[\hat{a}^1(e)]_{s'}$ diagonalized by $Z_{e,s'}$, the operator $A_{v,s}$
implements precisely the coboundary shift
$\hat{a}^1 \rightarrow \hat{a}^1 + d^0\bar{v}_s$:
\be\label{eq:A_conjugation}
A_{v,s}\, Z_{e,s'}\, A_{v,s}
= (-1)^{\left[(d^0 \bar{v}_s)(e)\right]_{s'}} Z_{e,s'} .
\ee
In matrix form, using
$[\Fc_{v \shortrightarrow e}(\vec{e}_s)]_{s'} =
[\Fc_{v \shortrightarrow e}]_{s' s} = [d^0]_{(e,s'),(v,s)}$, the operator
$A_{v,s}$ is the row $(v,s)$ of the $X$-type parity-check matrix
$H_X = (d^0)^{T}$. Its weight is bounded by
$\sum_{e \dotsucc v}\dim\Fc(e)=O(1)$ for bounded vertex degree and bounded stalk
dimension, which is one half of the qLDPC property.

\subsection{$Z$-stabilizers from $d^1$}\label{sec:Z_stabilizer}

The $Z$-type generators are labelled by a $2$-cell alone,
\be\label{eq:B_stabilizer_untwisted}
B_{f} = \prod_{e \precdot f} \prod_{s'}
\big(Z_{e,s'}\big)^{\left[\Fc_{e \shortrightarrow f}\right]_{s'}} ,
\ee
with \emph{no} vector index attached to $f$. The reason is that a $2$-cell is a
top cell of $\L$ and the stalk there is one-dimensional: for the hypergraph
product complex of Eq.~\eqref{eq:local_system_2D} one has
$\Fc(f)=\FF_2 \otimes \FF_2 \cong \FF_2$, i.e.\ the $Z$-check attached to a face
is a scalar rather than a vector. This is a structural feature of the sheaf codes
constructed in Sec.~\ref{sec:classical} rather than an accident: the local codes
are attached to the low-dimensional cells, so the stalk dimension decreases with
increasing degree, and the top cells of a graph code already carry $\FF_2$; taking
the product of two such graph codes gives $\FF_2 \otimes \FF_2$ on the $2$-cells.
Consequently the co-restriction $\Fc_{e \shortrightarrow f}: \Fc(e) \rightarrow
\FF_2$ is a linear \emph{functional}, i.e.\ a row vector carrying the single index
$s'$ labelling the stalk of the edge $e$, and there is exactly one $Z$-stabilizer
per face. (In $\mathsf{d}\ge 3$, where the $Z$-checks live on cells below the top
dimension, the stalk is in general not a scalar and the outer index reappears.)

The relation to the coboundary operator makes the content of
Eq.~\eqref{eq:B_stabilizer_untwisted} transparent. Evaluating
Eq.~\eqref{eq:coboundary_untwisted} on a $2$-cell and using $\Fc(f)\cong\FF_2$ gives the
single $\FF_2$-valued quantity
\begin{align}\label{eq:flux_coboundary}
\big(d\hat{a}^1\big)(f)
&= \sum_{e \precdot f} \Fc_{e \shortrightarrow f}\big(\hat{a}^1(e)\big) \cr
&= \sum_{e \precdot f} \sum_{s'}
\big[\Fc_{e \shortrightarrow f}\big]_{s'}\, \big[\hat{a}^1(e)\big]_{s'} .
\end{align}
Exponentiating the mod-$2$ sum into a product of signs and using
$Z_{e,s'} = (-1)^{[\hat{a}^1(e)]_{s'}}$ from
Eq.~\eqref{eq:Z_cochain_untwisted}, with the exponent
$[\Fc_{e \shortrightarrow f}]_{s'} \in \{0,1\}$ read as a selector, we obtain
\be\label{eq:B_is_flux}
B_{f} = (-1)^{\big(d\hat{a}^1\big)(f)} .
\ee
Hence $B_f$ is not merely correlated with flatness but literally \emph{measures}
it. It is a Hermitian Pauli operator with $B_f^2=\I$, and $B_f=+1$ if and only if
the flux vanishes on $f$. Equivalently, in matrix form $B_f$ is the row $f$ of the
sheaf coboundary matrix $d^1$, i.e.\ of the $Z$-type parity-check matrix
$H_Z=d^1$, and its weight is bounded by $\sum_{e \precdot f}\dim\Fc(e)=O(1)$.

The two stabilizer types therefore play complementary roles: $B_f$ \emph{measures}
a component of the flux, while $A_{v,s}$ \emph{generates} a gauge transformation.
This also accounts for an apparent asymmetry between
Eqs.~\eqref{eq:A_stabilizer_untwisted} and \eqref{eq:B_stabilizer_untwisted}: the
former involves a \emph{column} of $\Fc_{v \shortrightarrow e}$, because
$H_X=(d^0)^{T}$, whereas the latter involves a \emph{row} of
$\Fc_{e \shortrightarrow f}$, because $H_Z=d^1$.

\subsection{The untwisted stabilizer Hamiltonian}\label{sec:untwisted_Hamiltonian}

Collecting the two families, the untwisted sheaf code $\C$ is the Pauli stabilizer
model with Hamiltonian
\be\label{eq:H_untwisted}
H_0 = -\sum_{v \in \L(0)} \sum_{t} A_{v,s} \; - \sum_{f \in \L(2)} B_{f} ,
\ee
whose ground-state subspace is the code space of \eqref{eq:untwisted_complex}. Two
consistency properties are worth recording.

First, the stabilizers commute. The $A$'s commute among themselves, being
$X$-type, and likewise the $B$'s, being $Z$-type. For a mixed pair, $A_{v,s}$ and
$B_f$ overlap on the qubits $(e,s')$ with $e \dotsucc v$ and $e \precdot f$, and
their commutator vanishes if and only if that overlap has even parity, i.e.
\be\label{eq:commutation_d2}
\big[d^1 d^0\big]_{f,(v,s)} = \sum_{e}\sum_{s'}
\big[\Fc_{e \shortrightarrow f}\big]_{s'}
\big[\Fc_{v \shortrightarrow e}\big]_{s' s} = 0 ,
\ee
which is precisely $d^1 \circ d^0 = 0$. This holds by the functoriality of the
co-restriction maps together with the even-incidence condition of
Definition~\ref{def:poset}. In other words, the $X$- and $Z$-terms are compatible
exactly because the sheaf cochain sequence is a complex.

Second, on the ground-state subspace $B_f=+1$ for all $f$, so by
Eq.~\eqref{eq:B_is_flux} the gauge field is flat, $d\hat{a}^1=0$, and $\hat{a}^1$
is an operator-valued sheaf $1$-\emph{cocycle}; the $A_{v,s}$ then identify
configurations differing by a coboundary $d^0\bar{v}_s$, by
Eq.~\eqref{eq:A_conjugation}. The logical qubits are therefore counted by the
sheaf cohomology,
\be\label{eq:k_untwisted}
k = \dim H^1(\L,\Fc) = \dim C^1(\L,\Fc) - \rk\, d^0 - \rk\, d^1 .
\ee
Note that the generators listed in Eq.~\eqref{eq:H_untwisted} need not be
independent: the numbers of independent $X$- and $Z$-checks are $\rk\, d^0$ and
$\rk\, d^1$, which is what enters Eq.~\eqref{eq:k_untwisted}. Redundant generators
merely shift the ground-state energy by a constant and are harmless, but the
distinction matters when counting $k$.

Everything in Sec.~\ref{sec:non-Abelian_LDPC} is built on top of this model: the
twisted code is obtained by keeping the $Z$-stabilizers $B_f$ untouched and
dressing each $X$-stabilizer $A_{v,s}$ with a product of CZ gates acting on two
further copies of the code.

\subsection{Sheaf gauge theory as a topological defect network}\label{sec:sheaf_defect_network}

In this subsection we discuss the underlying physics of the quantum sheaf code and its associated sheaf gauge theory.
The sheaf structure naturally describes the physics of a topological defect network \cite{Aasen2020}.

\begin{figure*}[t]
    \centering   \includegraphics[width=0.8\linewidth]{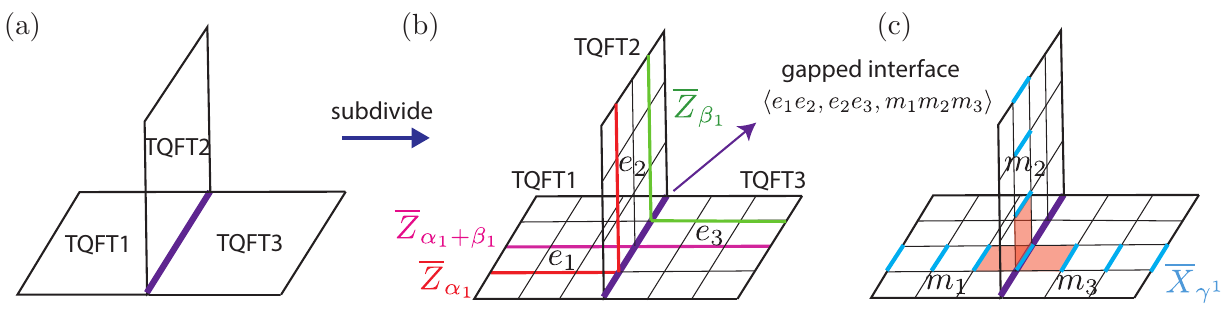}
    \caption{The sheaf gauge theory as a topological defect
network, for the constant sheaf. (a) On a general $2$D cell complex a single edge
(purple) can be adjacent to more than two faces, unlike in the cellulation of a
$2$-manifold. Each face is a TQFT patch --- here a $\ZZ_2$ topological order,
i.e.\ a surface code --- and the shared edge is a gapped interface between them.
(b) After subdividing each patch into $O(1)$ faces, each patch becomes a lattice
model. Only an even number of logical-$Z$ strings, supported on $1$-cycles, may
terminate on the interface without violating the $X$-stabilizers $A_v$ located
there: reading the $Z$-strings on patches $1,2,3$ as the worldlines of electric
anyons $e_1,e_2,e_3$, the interface condenses the pairs $e_1e_2$, $e_2e_3$,
$e_1e_3$, as illustrated by $\lo{Z}_{\alpha_1}$, $\lo{Z}_{\beta_1}$ and
$\lo{Z}_{\alpha_1+\beta_1}$ on the edges $e_1$, $e_2$, $e_3$. (c) Dually, three
logical-$X$ strings supported on a $1$-cocycle $\gamma^1$ must terminate together
at the interface in order to preserve the $Z$-stabilizers on the three adjacent
faces (shaded), so the interface condenses the triple $m_1m_2m_3$ of magnetic
anyons. The anyon condensation data is the Lagrangian subgroup
$L=\langle e_1e_2,\,e_2e_3,\,m_1m_2m_3 \rangle$.}
    \label{fig:gapped_boundary_picture}
\end{figure*}

This comes naturally from the geometric property of a general simplicial or cell complex. As illustrated in Fig.~\ref{fig:gapped_boundary_picture}(a), on a general 2D simplicial or cell complex, a single edge (purple) can be adjacent to more than two faces, in contrast to the case of the cellulation of a 2-manifold where a single edge is always adjacent to exactly two faces.  In the language of topological quantum field theory (TQFT), we can view such an edge as a gapped interface, while each of the adjacent face is a TQFT patch corresponding to a $\ZZ_2$-topological order in the case of the usual (untwisted) qLDPC code.   To increase the clarity, we can imagine subdivide each face (patch) into multiple faces and a single edge into multiple edges, as illustrated in Fig.~\ref{fig:gapped_boundary_picture}(b), such that each original face becomes a patch of lattice models described by the corresponding TQFT and in this case equivalent to a surface code patch.  We note that one should still consider subdividing each face into only $O(1)$ faces, such that the global code parameter scaling will still be preserved.  

As the simplest example, we consider the cell (simplicial) complex with a constant sheaf (i.e., trivial local codes), namely the ordinary cell (simplicial) complex $C^{\bullet}(\L; \ZZ_2)$, as illustrated in Fig.~\ref{fig:gapped_boundary_picture}.  As we can see from Fig.~\ref{fig:gapped_boundary_picture}(b), at the gapped interface (purple), only even number of logical-$Z$ strings (supported on the 1-cycles) can terminate on it to preserve the $X$-stabilizer $A_v$ defined on the vertices $v$ located on the interface.  In the TQFT language, we consider the $Z$-strings appearing on TQFT patch 1, 2 and 3 as the worldlines of three types of electric anyons $e_1$, $e_2$ and $e_3$ respectively.  We hence say the interface condenses pairs of anyons $e_1 e_2$, $e_2 e_3$ and $e_1 e_3$. Similarly, the $X$-strings on TQFT patch 1, 2 and 3 are considered as the wordlines of three types of magnetic anyons  $m_1$, $m_2$ and $m_3$ respectively. As we can see from Fig.~\ref{fig:gapped_boundary_picture}(c), three logical-X strings (supported on the 1-cocycle $\gamma^1$) has to terminate together at the interface to preserve the $Z$-stabilizers supported on the three faces (red) adjacent to the interface.  The interface hence condense a triple of magnetic anyons: $m_1 m_2 m_3$.  Therefore, we can represent the Lagrangian subgroup containing the anyon condensation data of the interface by its generators as:
\be\label{eq:Lagrangian-subgroup}
L=\langle e_1e_2, e_2e_3, m_1m_2m_3 \rangle.
\ee

We then consider an example with non-constant sheaf (non-trivial local codes). As illustrated in Fig.~\ref{fig:gapped_boundary_dual}(a),  we start with the dual complex $\L^*$ (dahsed lines) of the original complex  in $\L$ Fig.~\ref{fig:gapped_boundary_picture} (solid lines). Here, the duality corresponds to the map $\mathcal{D}: \L(i) \rightarrow \L^*(2-i)$, which exchanges the $i$-cell with a $(2-i)$-cell. The original and dual 1-cells intersect with each other.   
In this case, the edges in the original gapped interface becomes hyperedges (green) which are adjacent to 3 vertices, and those between the neighboring hyperedges are hyperfaces (red).  Due to the presence of hyperedges and hyperfaces, $\L^*$ is not a cell complex. Nevertheless, one can still define a meaningful stabilizer code with small modification along the interface, where the hyperface is associated with a $Z$-stabilizers containing qubits defined on the adjacent edges and hyperedges.  There are also $X$-stabilizers on the vertices on these hyperfaces which have qubits supported on the adjacent edges and hyperedges. With this definition,  a logical-$X$ operator (blue and yellow) supported on 1-cocycle of $\L^*$ correponding to the worldline of $m$ anyon can turn at the interface, while the logical-$Z$ operator (red) corresponding to the worldline of $e$ anyon has to branch into all the adjacent patches.  The corresponding gapped interface has the anyon condensation data described by the following Lagrangian subgroup:
\be
L'=\langle m_1 m_2, m_2 m_3, e_1 e_2 e_3 \rangle,
\ee
which exhibits the $e$-$m$ duality when compared to the condensation data Eq.~\eqref{eq:Lagrangian-subgroup} in the original complex $\L$.

Now we want to subdivide the general chain complex $\L^*$ into a cell complex $\L^*_c$ using our prescription in Sec.~\ref{sec:classical}. We first represent $\L^*$ as the  product of two 1D complexes  $G_{x, h} \times G_y$, i.e., a hypergraph $G_{x, h}$ and and a graph $G_y$ as shown in Fig.~\ref{fig:gapped_boundary_dual}(c), where $G_{x,h}$ contains a hyperedge (green) at the junction.  One can then subdivide the  hypergrpah $G_{x,h}$ into a graph $G_x$ by splitting the hyperedge into 3 edges and attaching a repetition local code at the middle as illustrated in Fig.~\ref{fig:gapped_boundary_dual}(d).  This gives rise to a classical sheaf code where the edge stalk is the scalar $\FF_2$, and all the vertices (blue squares) associated with the single parity check (SPC) also has a scalar stalk $\FF_2$ while the new vertex (green dimond) associated with the Rep local code has a vector stalk $\FF^2_2$.  We can then decompose the blue 1-cocycle in (b) into a tensor product of the sheaf 0-cocycle (blue) on $G_x$ and a single-edge 1-cocycle on $G_y$ as illustrated in (d).  Note that the value on the new vertex of the 0-cocycle is a vector $(1, 0)$
in this case.

\begin{figure}[t]
    \centering   \includegraphics[width=1\linewidth]{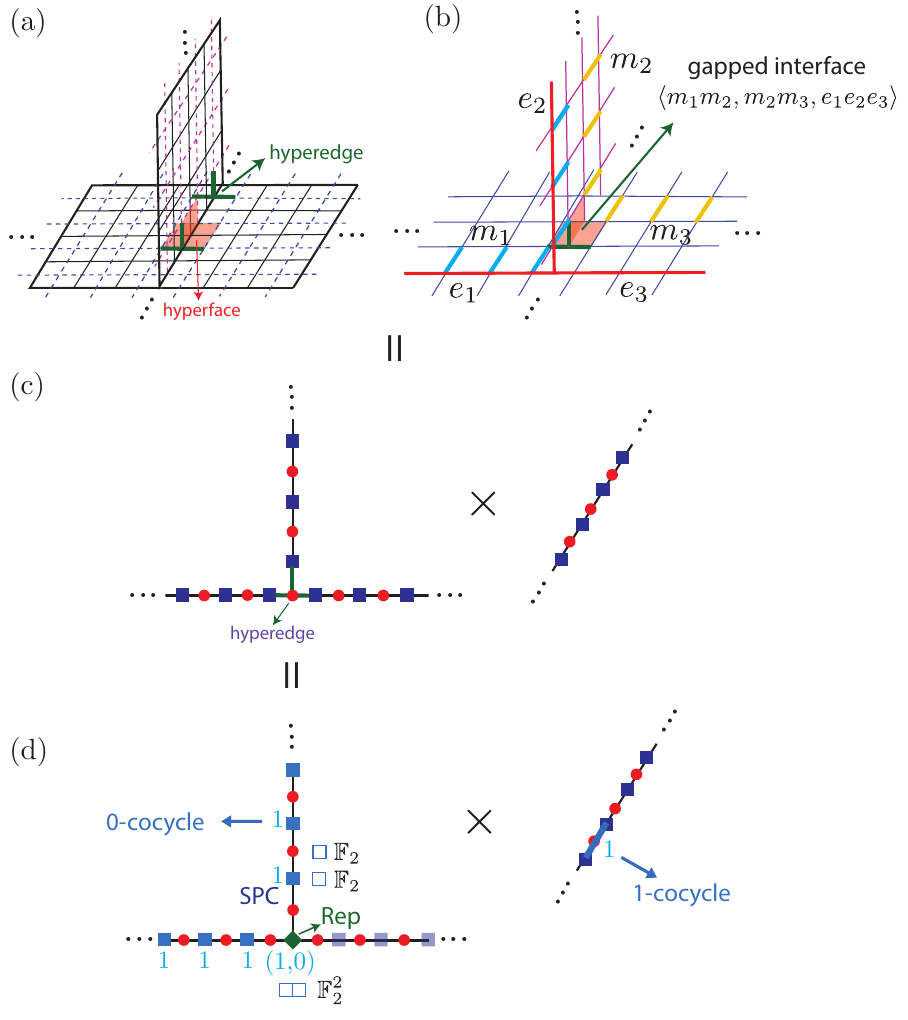}
    \caption{The same defect network for a non-constant sheaf,
obtained by dualizing Fig.~\ref{fig:gapped_boundary_picture}. (a) The dual complex
$\L^{*}$ (dashed) of the complex $\L$ (solid) under
$\mathcal{D}:\L(i)\rightarrow\L^{*}(2-i)$: the junction carries a hyperedge
(green) adjacent to three vertices, and the cells between neighbouring hyperedges
are hyperfaces (red), so $\L^{*}$ is not a cell complex. (b) The corresponding
gapped interface, with Lagrangian subgroup
$L'=\langle m_1m_2,\,m_2m_3,\,e_1e_2e_3 \rangle$; comparing with
Fig.~\ref{fig:gapped_boundary_picture}(c) exhibits the $e$--$m$ duality. A
logical-$X$ string (blue, yellow) on a $1$-cocycle of $\L^{*}$ can turn at the
interface, while a logical-$Z$ string (red) must branch into all adjacent patches.
(c) $\L^{*}$ written as a product $G_{x,h}\times G_y$ of a hypergraph $G_{x,h}$,
carrying the hyperedge (green) at the junction, with a graph $G_y$. (d) After
subdividing $G_{x,h}$ into the graph $G_x$ by splitting the hyperedge into three
edges and attaching a repetition local code at the new vertex (green diamond,
stalk $\FF_2^{2}$), while the SPC vertices (blue squares) and the edges keep the
scalar stalk $\FF_2$. The blue $1$-cocycle of (b) then factorizes as a tensor
product of a sheaf $0$-cocycle on $G_x$, whose value on the new vertex is the
vector $(1,0)$, with a single-edge $1$-cocycle on $G_y$.}
    \label{fig:gapped_boundary_dual}
\end{figure}

%======================================================================
\section{Non-Abelian qLDPC codes as twisted sheaf gauge theories}\label{sec:non-Abelian_LDPC}
%======================================================================

\subsection{Spacetime path integral as a cohomology invariant}\label{sec:Feynman_path_integral}

We consider the twisted $\Z_2^3$ sheaf gauge theory with the following action:
\begin{align}\label{eq:non-Abelian_action}
S=& \pi \int_{\tilde{\L}} \red{ a^1} \cup \red{d \check{a}^1 } +  \pi \int_{\tilde{\L}} \blue{b^1} \cup \blue{ d\check{b}^1} + \pi \int_{\tilde{\L}} \green{c^1} \cup \green{d\check{c}^1} \cr
&+ \pi \int_{\tilde\eta_3} \red{a^1} \cup \blue{b^1} \cup \green{c^1},
\end{align}
where $\red{ a^1}$$\in$$C^1(\tilde \L, \Fc^{\rd} )$, $\blue{b^1}$$\in$$C^1(\tilde \L, \Fc^{\bl})$, $\green{c^1}$$\in$$ C^1(\tilde \L, \Fc^{\gr})$ are sheaf 1-cochains physically corresponding to the electric gauge fields, and $\red{ \check a^1} $$\in$$C^1(\tilde \L, \Fc^{\rd} )$, $\blue{\check b^1} $$\in$$C^1(\tilde \L, \Fc^{\bl} )$, $\green{\check c^1} $$\in$$C^1(\tilde \L, \Fc^{\gr} )$  are dual magnetic gauge fields. Here, $\tilde \L$ represents the entire (2+1)D-spacetime sheaf complex, while $\tilde \eta_3 \in C^1(\tilde \L, \Fc^{\rd} \otimes \Fc^{\bl} \otimes \Fc^{\gr} )$ represents a 3-cycle in spacetime complex equipped with tensor sheaf $\Fc^{\rd} \otimes \Fc^{\bl} \otimes \Fc^{\gr}$. 
The first three terms are the sheaf generalization of the BF theory, while adding the last term gives rise to the sheaf generalization of the Dijkgraaf-Witten twisted gauge theory.

\begin{figure}[hbt]
    \centering   \includegraphics[width=1\linewidth]{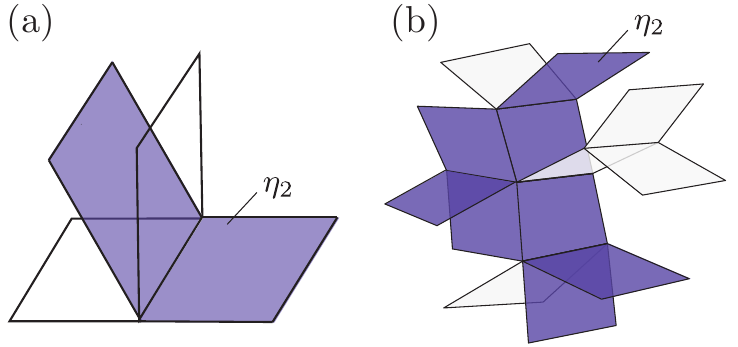}
    \caption{The $2$-cycle $\eta_2$ (purple) supporting the
Dijkgraaf-Witten twist $\int_{\tilde\eta_3}a^1\cup b^1\cup c^1$, shown on two
branched square complexes: (a) a small `book-like' complex in which one edge is
shared by three faces, and (b) a larger, irregular complex with several
generations of branching. Unlike a closed $2$-manifold, which has a unique
top-dimensional cycle, such a complex carries a large number of inequivalent
$2$-cycles; each of them defines a different twist, and, after gauging, a
different family of $0$-form subcomplex symmetries. The spacetime cycle
$\tilde\eta_3 = \eta_2 \otimes I_t$ is the corresponding region of the
$(2+1)$D complex $\tilde\L=\L\otimes I_t$.}
    \label{fig:gauging_region}
\end{figure}

We then switch to the study of the \textit{path integral} $\mathcal{Z}[\tilde{\L}, \eta_3]$, also called the \textit{partition function}, which is related to the action by
\be
\mathcal{Z}[\tilde{\L}, \eta_3] = \frac{1}{\N} \sum_{\red{a^1}, \blue{b^1}, \green{c^1}, \red{\tilde{a}^1}, \blue{\tilde{b}^1}, \green{\tilde{c}^1} \in {C^1}(\tilde{\L})} e^{iS[\red{a^1}, \blue{b^1}, \green{c^1}, \red{\tilde{a}^1}, \blue{\tilde{b}^1}, \green{\tilde{c}^1}]},
\ee
where $\N$ is a normalization constant.  Note that the cochains with flavor $\rd$, $\gr$ and $\bl$ belongs to the sheaf cochain groups $C^1(\tilde \L, \Fc^{\rd} )$, $C^1(\tilde \L, \Fc^{\bl})$ and $C^1(\tilde \L, \Fc^{\gr})$ respectively, where $C^1(\tilde \L)$ above is just a short-hand notation without specifying the specific sheaf (local systems of coefficients).  Summing over the Lagrange-multiplier fields in the first three BF terms of $S$ [Eq.~\eqref{eq:non-Abelian_action}] produces an overall factor of $1$ while enforcing the cocycle condition on the sheaf cochains:
\be
\red{da^1}=\blue{db^1}=\green{dc^1}=0,
\ee
i.e., $\red{a^1}, \blue{b^1}, \green{c^1}$ are constrained to be flat gauge fields (cocycles). The path integral therefore reduces to
\be\label{eq:path_inegral_cup}
\mathcal{Z}[\tilde{\L}, \tilde\eta_3] = \sum_{\red{a^1}, \blue{b^1}, \green{c^1} \in {H^1(\tilde{\L})}} (-1)^{\int_{\tilde \eta_3} \red{a^1} \cup \blue{b^1} \cup \green{c^1}},
\ee
where we have set the normalization constant to $\N=1$, as it plays no role in any of the results below.

We now present the following lemma:
\begin{lemma}\label{lemma:invariant} 
The path integral $\mathcal{Z}[\tilde{\L}, \eta_3]$ defined on the simplicial complex $\tilde{\L}$ is a cohomology invariant (physically corresponds to the gauge invariance).  
\end{lemma}
\begin{proof}

We apply the following coboundary transofrmation or equivalently the gauge transformation on the gauge fields:  $\red{a^1} \rightarrow \red{a^1} + d\tilde\eta^{0}$, $\blue{b^1} \rightarrow \blue{b^1} + d \tilde \xi^{0}$, $\green{c^1} \rightarrow \green{c^1} + d\tilde\zeta^{0}$, and need to show that $\mathcal{Z}[\tilde{\L}, \eta_3]$ is invariant under this transformation.

The triple cup product induces the following trilinear map on the cohomology group:
\begin{align}\label{eq:trilinear_map}
\cup :&  H^{1}(\tilde \L, \Fc^{\rd}) \times H^{1}(\tilde \L, \Fc^{\bl}) \times H^{1}(\tilde \L, \Fc^{\gr})  \cr
& \longrightarrow  H^{3}(\tilde \L, \Fc^{\rd} \otimes \Fc^{\bl} \otimes \Fc^{\gr}).
\end{align}
We hence obtain 
\begin{align}
 & (\red{a^1} + d\tilde\eta^0 )\cup (\blue{b^1} +d\tilde \xi^0 )\cup (\green{c^1} + d\tilde\zeta^0 ) \cr
=&  \red{a^1} \cup \blue{b^1} \cup \green{c^1} + d\omega^2 \in   H^{3}(\tilde \L, \Fc^{\rd} \otimes \Fc^{\bl} \otimes \Fc^{\gr}) 
\end{align}
The exponent in $\mathcal{Z}[\tilde{\L}]$ hence becomes:
\begin{align}
& \int_{\tilde \eta_3} {(\red{a^1} + d\tilde\eta^0 )\cup (\blue{b^1} +d\tilde \xi^0 )\cup (\green{c^1} + d\tilde\zeta^0 )}   \cr
=& \int_{\tilde \eta_3}  \red{a^1} \cup \blue{b^1} \cup \green{c^1} +  \int_{\tilde \eta_3} d\omega^2 \cr
=& \int_{\tilde \eta_3}  \red{a^1} \cup \blue{b^1} \cup \green{c^1},
\end{align}
where  we have used the Stoke's theorem in the second equality: 
\be
\int_{\tilde \eta_3} d\omega^{2} = \int_{ \partial 
\tilde\eta_3 }\omega^{2}=0,
\ee
since $\tilde{\eta}_3$ has no boundary, i.e., $\tilde{\eta}_3=0$.
\end{proof}

\begin{figure}[t]
    \centering   \includegraphics[width=0.8\linewidth]{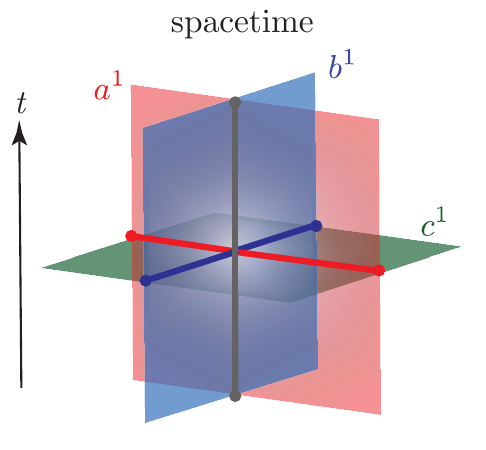}
    \caption{Triple intersection of the three gauge-field
worldsheets in $(2+1)$D spacetime, with time $t$ running upward. The flat
$1$-cocycles $\red{a^1}$ (red), $\blue{b^1}$ (blue) and $\green{c^1}$ (green) are worldsheets of the magnetic excitations of the three
colors $m^{\rd}, m^{\bl}, m^{\gr}$; the twisted term $\pi\int_{\tilde\eta_3}\red{a^1}\cup\blue{b^1}\cup
\green{c^1}$ in the action \eqref{eq:non-Abelian_action} counts the points at
which all three sheets meet. A non-vanishing triple intersection is exactly what
makes the path integral $\cZ[\tilde\L,\tilde\eta_3]$ a non-trivial cohomology
invariant, and it is the mechanism behind both the non-Abelian braiding statistics
of the twisted code and the non-trivial logical action of the gauging measurement
protocol.}
    \label{fig:triple_intersection}
\end{figure}

\subsection{Sheaf cluster state and sub-complex $\Z_2^3$ SPT}
\label{sec:sheaf_SPT}

We now construct the microscopic lattice model whose topological response is the
twisted term of the action in Eq.~\eqref{eq:non-Abelian_action}.  The resulting state is a sheaf generalization of the
$2$D cluster state, dubbed as a
\emph{sheaf cluster state}, with one essential new feature: it is supported not on the
whole spatial complex but on a \emph{codimension-$0$ subcomplex} determined by a
$2$-cycle $\eta_2$.  As we will see in Sec.~\ref{sec:subcomplex_symmetry}, this sheaf cluster state is protected by a $\red{\ZZ_2} \times \blue{\ZZ_2} \times \green{\ZZ_2} \equiv \Z_2^3$ symmetry supported on this codimension-0 subcomplex, or equivalnetly the $2$-cycle $\eta_2$. We will hence call it a \emph{$0$-form subcomplex
$\red{\ZZ_2} \times \blue{\ZZ_2} \times \green{\ZZ_2}$   symmetry-protected topological} (\emph{SPT}) \emph{state} (or abbreviated as $\Z_2^3$ SPT).  Note that by `\emph{$k$-form}' symmetry we mean the symmetry operator is defined on a codimension-$k$ subcomplex (with dimension $\ds-k$, where $\ds$ is the total dimension of the spatial complex $\L$), which is a generalization with the usual definition in terms of codimension-$k$ submanifold in the previous literature.

Let $\L$ be the $\mathsf{d}=2$ dimensional spatial cell  complex, with cell sets
$\L(0)=V$, $\L(1)=E$ and $\L(2)=F$; for example, $\L$ can be a 2D simplicial complex, or the square complex obtained
as the hypergraph product of two subdivided classical sheaf codes constructed in
Sec.~\ref{sec:subdivision}. The spacetime complex is
$\tilde{\L} \equiv \L \times S^1_t$, and the $3$-cycle appearing in the twist term
is $\tilde{\eta}_3 = \eta_2 \otimes S^1_t$, where
$\eta_2 \in Z_2(\L, \Fc^{\rd} \otimes \Fc^{\bl} \otimes \Fc^{\gr})$ is a spatial $2$-\emph{cycle}
carrying the tensor sheaf.

\begin{remark}[\emph{Why the SPT is supported on a $2$-cycle $\eta_2$ and not the whole 2-complex $\L$}]
\label{rem:why_cycle}
In the conventional construction of $k$-form SPTs one integrates the response
action over the entire spatial complex, $\int_{\L}$, which implicitly uses the
\textit{fundamental chain} $[\L]=\sum_{\sigma \in \L(2)}\sigma$ as the integration cycle, also called the \textit{fundamental class}.
This is legitimate only when the fundamental class $[\L]$ has no boundary, i.e., $\partial [\L]=0$, such as in the case when $\L$ is the cellulation of a manifold. Cell (simplicial) complexes, such as the square complexes produced by
the hypergraph product code of Sec.~\ref{sec:subdivision} are in general \emph{not} cellulation of  manifolds:
their fundamental chain generally has $\partial[\L]\neq 0$, and moreover with
system of local coefficients $[\L]$ need not even be a chain in
$C_2(\L,\Fc^{\rd}\otimes\Fc^{\bl}\otimes\Fc^{\gr})$. One must therefore replace $[\L]$ by a
genuine tensor-sheaf $2$-cycle $\eta_2 \in Z_2(\L, \Fc^{\rd} \otimes \Fc^{\bl} \otimes \Fc^{\gr})$. The consequence is that the
resulting SPT is supported on the codimension-$0$ subcomplex
\be\label{eq:subcomplex}
\L_{\eta_2} \coloneqq \overline{\mathrm{supp}\,\eta_2} \subseteq \L
\ee
(the closure of the support of $\eta_2$, i.e.\ all $2$-cells on which $\eta_2$ is
nonvanishing together with all their faces) rather than on all of $\L$. The
condition $\partial \eta_2 = 0$ is exactly what makes the response action gauge
invariant, cf.\ Lemma~\ref{lemma:invariant}.
\end{remark}

We start with three independent copies of qubit lattice models defined on the same 2D cell (simplicial complex) $\L$, denoted
\red{red} (\rd), \blue{blue} (\bl) and \green{green} (\gr), each carrying its own
system of local coefficients $\Fc^{\rd}$, $\Fc^{\bl}$ and $\Fc^{\gr}$ respectively, in the sense
of Definition~\ref{def:presheaf}. The colour of a copy is thus carried by a
superscript, and whenever a statement holds for each of the three copies
separately we write it once using the generic colour label $i$, i.e.\ $\Fc^{i}$
with $i \in \{\rd, \bl, \gr\}$. The matter qubits of all
three copies sit on the \emph{vertices} $v \in V$ of $\L$. The crucial difference
from the scalar (constant-sheaf) case is that a vertex no longer carries a single
qubit: the stalk $\Fc^{\rd}(v)=\FF_2^{m(v)}$ carries $m(v) \equiv \dim \Fc^{\rd}(v)$ qubits.
We therefore label the physical qubits of the \rd, \bl \ and \gr \ copies by pairs
\be\label{eq:qubit_labels}
(v,s), \qquad (v,r), \qquad (v,t),
\ee
where $v \in V$ and the stalk indices run over $s \in [\dim \Fc^{\rd}(v)]$,
$r \in [\dim \Fc^{\bl}(v)]$ and $t \in [\dim \Fc^{\gr}(v)]$. The corresponding Pauli
operators are denoted $\{\tilde X^{\rd}_{v,s}, \tilde Y^{\rd}_{v,s},
\tilde Z^{\rd}_{v,s}\}$ and likewise for \bl \ and \gr, where the tilde
distinguishes the operators of the ungauged SPT model from those of the gauged
model introduced later in Sec.~\ref{sec:gauging}.

\subsubsection{Operator-valued sheaf cochains}

We now extend the operator-valued cochain formalism with scalar coefficients of
Refs.~\cite{Hsin2024_non-Abelian, zhu2023non, Hsin2024:classifying,
zhu2025topological} to system of local coefficients (sheaf). An operator-valued sheaf $k$-cochain
$\hat{\lambda}^k_i \in C^k(\L, \Fc^{i})$ assigns to each $k$-cell $\sigma$ a vector
$\hat{\lambda}^k_i(\sigma) \in \Fc^{i}(\sigma)$ whose components are commuting
operators with eigenvalues $0$ or $1$, where $i \in \{\rd, \bl, \gr\}$; physically these are the $\FF_2$-valued
matter fields transformed under the $\ZZ_2$ global symmetry. Fixing the
standard basis $\{\vec{e}_s\}_{s=1}^{m_i(\sigma)}$ of the stalk
$\Fc^{i}(\sigma)=\FF_2^{m_i(\sigma)}$, with $m_i(\sigma) \equiv \dim
\Fc^{i}(\sigma)$ and $\vec{e}_s$ the unit vector whose $s^\text{th}$ component
equals $1$ and whose other components vanish,
\be\label{eq:unit_vector}
\vec{e}_s = \big[\,\underbrace{0, \ldots, 0}_{s-1},\, 1,\,
\underbrace{0, \ldots, 0}_{m_i(\sigma)-s}\,\big]^{T} ,
\ee
the components are related to the Pauli-$Z$ operators by
\be\label{eq:Pauli_relation}
(-1)^{[\hat{\lambda}^k_i(\sigma)]_s}=\tilde Z^{i}_{\sigma,s},
\qquad
[\hat{\lambda}^k_i(\sigma)]_s = \big[1-\tilde Z^{i}_{\sigma,s}\big]/2 .
\ee
It is convenient to introduce the \emph{elementary sheaf cochain}
$\bar{\sigma}_s \in C^k(\L,\Fc^{i})$, which takes the value $\vec{e}_s$ on the
single $k$-cell $\sigma$ and vanishes on all other $k$-cells. This is the sheaf
replacement for the indicator cochain $\bar{\sigma}$ of the scalar theory,
and the elementary cochains are in one-to-one correspondence with the physical
qubits. Any operator-valued cochain can then be expanded as
\be\label{eq:cochain_expansion}
\hat{\lambda}^{k}_i = \sum_{\sigma \in \L(k)} \sum_{s=1}^{\dim \Fc^{i}(\sigma)}
\hat{N}^{i}_{\sigma,s}\, \bar{\sigma}_s ,
\qquad
(-1)^{\hat{N}^{i}_{\sigma,s}} = \tilde{Z}^{i}_{\sigma,s},
\ee
where $\hat{N}^{i}_{\sigma,s}$ is a quantum operator with eigenvalues
$N^{i}_{\sigma,s} \in \{0,1\}$, and $\bar{\sigma}_s$ is a classical 
cochain (i.e., a classical variable). In what follows we write the stalk index as $s$, $r$ or $t$ according to
whether the cochain belongs to the \rd, \bl \ or \gr \ copy. Recall from
Eq.~\eqref{eq:coboundary-sheaf} that the coboundary of a sheaf cochain is
\be\label{eq:coboundary}
(d^{k}\hat{\lambda}^{k}_i)(\tau)
= \sum_{\sigma \in X(k),\ \sigma \precdot \tau}
\Fc^{i}_{\sigma \shortrightarrow \tau}\big(\hat{\lambda}^{k}_i(\sigma)\big),
\ee
so that, unlike the scalar case, each incidence is dressed by the co-restriction
map, e.g.\ by a column of the local parity-check matrix in the HGP complex.

\subsubsection{Topological  response action}

The topological response of the SPT is the triple cup product of the background
gauge fields, integrated over the spacetime $3$-cycle:
\be\label{eq:SPT_action}
S_\text{SPT}
= \pi \int_{\tilde{\eta}_3} \red{A^1} \cup \blue{B^1} \cup \green{C^1},
\ee
where $\red{A^1}\in C^1(\tilde\L,\Fc^{\rd})$, $\blue{B^1}\in C^1(\tilde\L,\Fc^{\bl})$ and
$\green{C^1}\in C^1(\tilde\L,\Fc^{\gr})$ are classical (static) background gauge fields. As in
Definition~\ref{def:invariant-form}, the sum (integral) over a cycle is equivalent to 
the canonical pairing of Definition~\ref{def:pairing-sheaf},
\be\label{eq:integral_notation}
\int_{\tilde{\eta}_3} x^1 \cup y^1 \cup z^1
\;\equiv\; \big\langle x^1 \cup y^1 \cup z^1, \, \tilde{\eta}_3 \big\rangle
\in \FF_2 ,
\ee
which reduces to the usual sum $\int_{\tilde\L}$ over the whole complex when
the sheaf is constant and $\tilde{\eta}_3 = [\tilde{\L}]$ is the fundamental
chain. Eq.~\eqref{eq:SPT_action} is the sheaf generalization of the ``type-III
cocycle'' SPT \cite{deWildPropitius:1995cf, Yoshida:2015cia} of the
group-cohomology classification \cite{Chen:2013foa}; setting
$\Fc^{\rd}=\Fc^{\bl}=\Fc^{\gr}=\FF_2$ and $\eta_2 = [\L]$ recovers the ordinary
$\ZZ_2^3$ type-III SPT in $(2+1)$D, i.e.\ the $2$D cluster state.

Written out explicitly on the spatial complex, with
$\eta_2(\sigma) = \sum_{\ell} a_\ell \otimes b_\ell \otimes c_\ell$ where
$a_\ell \in \Fc^{\rd}(\sigma)$, $b_\ell \in \Fc^{\bl}(\sigma)$, $c_\ell \in \Fc^{\gr}(\sigma)$,
the degree-$(0,1,1)$ cup product on a $2$-cell $\sigma = [v_0,v_1,v_2]$ gives
\begin{align}\label{eq:integral_explicit}
\int_{\eta_2} x^0 \cup y^1 \cup z^1 =& \sum_{\sigma \in \L(2)} \sum_{\ell}
\big\langle \Fc^{\rd}_{[v_0]\shortrightarrow \sigma}\, x^0([v_0]),\, a_\ell \big\rangle
\cr
&\times
\big\langle \Fc^{\bl}_{[v_0 v_1]\shortrightarrow \sigma}\, y^1([v_0 v_1]),\, b_\ell
\big\rangle \cr
&\times
\big\langle \Fc^{\gr}_{[v_1 v_2]\shortrightarrow \sigma}\, z^1([v_1 v_2]),\, c_\ell
\big\rangle ,
\end{align}
following Definition~\ref{def:cup-simplicial}, and analogously for the splittings
$(1,0,1)$ and $(1,1,0)$, for which the front and back faces are
$\{[v_0v_1],[v_1],[v_1v_2]\}$ and $\{[v_0v_1],[v_1v_2],[v_2]\}$ respectively.
(For the cubical complex obtained from the hypergraph product one uses the
corresponding cubical cup product; only the combinatorics of the front/back faces
changes.) Note that the sum over $2$-cells is effectively restricted to
$\L_{\eta_2}$, since $\eta_2(\sigma)=0$ off the subcomplex.

\subsubsection{The CCZ Entangler}

Following the gauging strategy of Refs.~\cite{Yoshida:2015cia,
barkeshli2023codimension}, we start from the trivial paramagnet
\begin{align}\label{eq:trivial_H}
H^0 = -\sum_{v \in V}\sum_{t} \tilde X^{\rd}_{v,s}
      -\sum_{v \in V}\sum_{r} \tilde X^{\bl}_{v,r}
      -\sum_{v \in V}\sum_{t} \tilde X^{\gr}_{v,t},
\end{align}
where $s$, $r$ and $t$ run over bases of $\Fc^{\rd}(v)$, $\Fc^{\bl}(v)$ and $\Fc^{\gr}(v)$. The
SPT entangler is the unitary
\begin{equation} \label{eq:entangler}
U = (-1)^{\int_{\eta_2} \hat{\lambda}^{0}_{\rd} \cup
d\hat{\lambda}^{0}_{\bl} \cup d\hat{\lambda}^{0}_{\gr}},
\end{equation}
whose exponent has total degree $0+1+1=2=\mathsf{d}$, as required for it to pair
with the spatial $2$-cycle $\eta_2$. Inserting the expansion
\eqref{eq:cochain_expansion} and using multilinearity of the cup product and of
the pairing, we obtain
\begin{align}\label{eq:entangler_re-express}
 U =& (-1)^{\int_{\eta_2} \hat{N}^{\rd}_{u,s}\,\bar{u}_s \cup
 \hat{N}^{\bl}_{v,r}\, d\bar{v}_r \cup
 \hat{N}^{\gr}_{w,t}\, d\bar{w}_t} \cr
 =& \prod_{(u,s),(v,r),(w,t)}
 \Big[\text{CCZ}^{\rd,\bl,\gr}_{(u,s),(v,r),(w,t)}\Big]^{
 \int_{\eta_2} \bar{u}_s \cup d\bar{v}_r \cup d\bar{w}_t} \cr
\equiv & \prod_{\substack{(u,s),(v,r),(w,t):\\
 \int_{\eta_2} \bar{u}_s \cup d\bar{v}_r \cup d\bar{w}_t \neq 0}}
 \text{CCZ}^{\rd,\bl,\gr}_{(u,s),(v,r),(w,t)} ,
\end{align}
where the summation over repeated cell and stalk labels in the exponent of the
first line is implicit, and we have used the identity
\be
(-1)^{\hat{N}^{\rd}_{u,s}\hat{N}^{\bl}_{v,r}\hat{N}^{\gr}_{w,t}}
\equiv \text{CCZ}^{\rd,\bl,\gr}_{(u,s),(v,r),(w,t)},
\ee
which follows from
\begin{align}
&\text{CCZ}^{\rd,\bl,\gr}_{(u,s),(v,r),(w,t)}
\ket{N^{\rd}_{u,s}, N^{\bl}_{v,r}, N^{\gr}_{w,t}} \cr
=& (-1)^{N^{\rd}_{u,s} N^{\bl}_{v,r} N^{\gr}_{w,t}}
\ket{N^{\rd}_{u,s}, N^{\bl}_{v,r}, N^{\gr}_{w,t}} .
\end{align}
As in the scalar theory, the evaluation of the cup-product integral in the
exponent of the second line of Eq.~\eqref{eq:entangler_re-express} takes the
value $1$ or $0$ and therefore determines whether the CCZ gate is applied,
leading to the condition in the product in the third line: the circuit consists
of one CCZ gate on each qubit triple $\{(u,s),(v,r),(w,t)\}$ satisfying the
sheaf-weighted triple-intersection condition
\be\label{eq:triple_intersection}
\int_{\eta_2} \bar{u}_s \cup d\bar{v}_r \cup d\bar{w}_t = 1 .
\ee
Two properties of Eq.~\eqref{eq:triple_intersection} deserve emphasis. First, it
can only be satisfied when $u$, $v$, $w$ are mutually incident to a common
$2$-cell of $\L$, so the circuit is geometrically local and of constant depth.
Second, that common $2$-cell must lie in $\L_{\eta_2}$: \emph{the entangler acts
as the identity outside the subcomplex}. The sheaf cluster state
\be
\ket{\Psi} = U \ket{+}^{\otimes n}
\ee
is therefore a product of a genuine cluster
state on $\L_{\eta_2}$ with unentangled $\ket{+}$ states everywhere else.

\subsubsection{The Sheaf cluster-state and SPT Hamiltonian}

Conjugating the trivial paramagnet by the entangler gives the Hamiltonian of sheaf SPT (cluster state), 
\begin{align}\label{eq:SPT_Hamiltonian}
 H_\text{SPT} =& \, U H^0 U^\dag  \cr
=&  -\sum_{u,s} \tilde X^{\rd}_{u,s}\,
(-1)^{\int_{\eta_2} \bar{u}_s \cup d\hat{\lambda}^{0}_{\bl} \cup
d\hat{\lambda}^{0}_{\gr}} \cr
&-\sum_{v,r} \tilde X^{\bl}_{v,r}\,
(-1)^{\int_{\eta_2} d\hat{\lambda}^{0}_{\rd} \cup \bar{v}_r \cup
d\hat{\lambda}^{0}_{\gr}} \cr
&-\sum_{w,t} \tilde X^{\gr}_{w,t}\,
(-1)^{\int_{\eta_2} d\hat{\lambda}^{0}_{\rd} \cup d\hat{\lambda}^{0}_{\bl} \cup
\bar{w}_t} \cr
=&  -\sum_{u,s} \tilde X^{\rd}_{u,s} \prod_{(v,r),(w,t)}
{\text{CZ}^{\bl,\gr}_{(v,r),(w,t)}}^{\int_{\eta_2} \bar{u}_s \cup d\bar{v}_r
\cup d\bar{w}_t} \cr
&-\sum_{v,r} \tilde X^{\bl}_{v,r} \prod_{(u,s),(w,t)}
{\text{CZ}^{\rd,\gr}_{(u,s),(w,t)}}^{\int_{\eta_2} d\bar{u}_s \cup \bar{v}_r
\cup d\bar{w}_t} \cr
&-\sum_{w,t} \tilde X^{\gr}_{w,t} \prod_{(u,s),(v,r)}
{\text{CZ}^{\rd,\bl}_{(u,s),(v,r)}}^{\int_{\eta_2} d\bar{u}_s \cup d\bar{v}_r
\cup \bar{w}_t}, \cr
\end{align}
where the last equality follows by re-expanding the remaining operator-valued
cochains as in Eq.~\eqref{eq:cochain_expansion} together with the identity
$(-1)^{\hat{N}^{i}_{\sigma,s}\hat{N}^{i'}_{\sigma',p'}} $$\equiv $$
\text{CZ}^{i,i'}_{(\sigma,s),(\sigma',p')}$ with $i,i' \in \{\rd,\bl,\gr\}$.
Taking the exponents as conditions in the products as before, we arrive at the
final form of the sheaf SPT Hamiltonian:
\begin{align}\label{eq:SPT_Hamiltonian_final}
 H_\text{SPT} =& -\sum_{u,s} \tilde X^{\rd}_{u,s}
 \prod_{\substack{(v,r),(w,t):\\
 \int_{\eta_2} \bar{u}_s \cup d\bar{v}_r \cup d\bar{w}_t \neq 0}}
 \text{CZ}^{\bl,\gr}_{(v,r),(w,t)} \cr
 &-\sum_{v,r} \tilde X^{\bl}_{v,r}
 \prod_{\substack{(u,s),(w,t):\\
 \int_{\eta_2} d\bar{u}_s \cup \bar{v}_r \cup d\bar{w}_t \neq 0}}
 \text{CZ}^{\rd,\gr}_{(u,s),(w,t)} \cr
 &-\sum_{w,t} \tilde X^{\gr}_{w,t}
 \prod_{\substack{(u,s),(v,r):\\
 \int_{\eta_2} d\bar{u}_s \cup d\bar{v}_r \cup \bar{w}_t \neq 0}}
 \text{CZ}^{\rd,\bl}_{(u,s),(v,r)} .
\end{align}
This is a Clifford-stabilizer Hamiltonian built from \emph{dressed}
$X$-stabilizers, i.e.\ single-site $X$ operators dressed by a product of CZ
gates, and the sheaf cluster state $\ket{\Psi}=U\ket{+}^{\otimes n}$ is its unique ground state.

\begin{figure}[t]
    \centering   \includegraphics[width=1\linewidth]{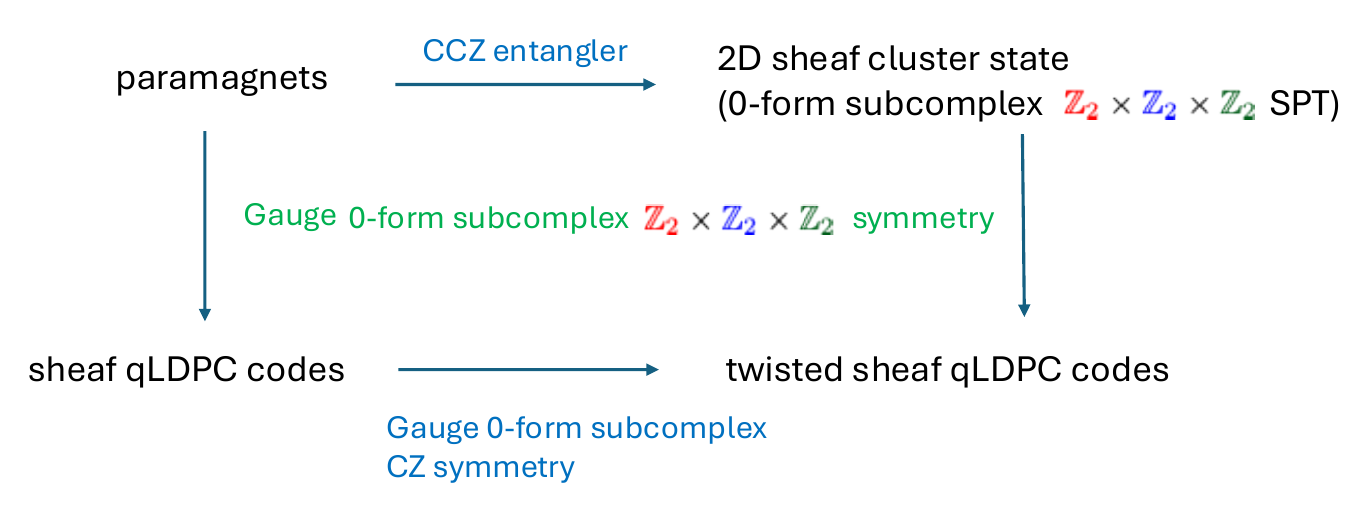}
    \caption{The summary of the two routes of gauging procedure in this section and later in Sec.~\ref{sec:gauging_untwisted}. Route 1 persued in this section is to first apply the CCZ entangler on the paramagnetic state on the sheaf complex to produce a 2D sheaf cluster/SPT state protected by the 0-form subcomplex $\ZZ_2^3$ symmetry and then gauge this symmetry to obtain the twisted $\ZZ_2^3$ sheaf qLDPC code. The alternative route in Sec.~\ref{sec:gauging_untwisted} is to first gauge the subcomplex $\ZZ_2^3$ symmetry on the paramagents to produce a sheaf qLDPC code and then further gauge the 0-form subcomplex CZ symmetry to produce the same twisted sheaf qLDPC code.  }
    \label{fig:gauging_process}
\end{figure}

\subsubsection{Contrast with the ordinary cluster state}
\label{sec:contrast}

It is worth spelling out, term by term, how the sheaf cluster state differs from
the familiar $2$D cluster state.

\emph{(i) Number of qubits per site.} In the ordinary cluster state each vertex of
each colour hosts exactly one qubit, and the total qubit count is $3|V|$. In the
sheaf cluster state the stalk $\Fc^{i}(v)$ hosts $\dim\Fc^{i}(v)$ qubits, so a single
vertex carries a small register. Equivalently, the ordinary cluster state is the
special case $\Fc^{\rd}=\Fc^{\bl}=\Fc^{\gr}=\FF_2$ of the constant sheaf.

\emph{(ii) Weighted incidences.} The ordinary cluster-state stabilizer
$X_u \prod_{v\sim u}Z_v$ treats all incidences democratically. Here every
incidence $u \precdot e$ is dressed by the co-restriction map
$\Fc^{\rd}_{u\shortrightarrow e}$, i.e.\ by a column of the local parity-check matrix
$\mathsf{H}_v$ [Eq.~\eqref{eq:restriction_map}]. The local code
attached to each vertex is what allows the underlying classical code to be good,
which is impossible for a bare graph code because the girth of a graph is
logarithmically bounded (Sec.~\ref{sec:sheaf-homology}).

\emph{(iii) Support: whole complex versus subcomplex.} This is the essential
structural difference. The ordinary cluster state is built on a closed surface and
its response action uses the fundamental class $[\L]$, which covers every $2$-cell
exactly once. On the non-manifold square complexes of Sec.~\ref{sec:subdivision}
no such class exists (Remark~\ref{rem:why_cycle}), and one instead selects a
$2$-cycle $\eta_2$. Since $\dim \eta_2 = 2 = \mathsf{d}$, the subcomplex
$\L_{\eta_2}$ has codimension $0$, consistent with the symmetry being of $0$-form
type; but it is a \emph{proper} subcomplex, and outside it the state is a trivial
product state. For the hypergraph product one may take $\eta_2 = z\otimes z'$ with
$z,z'$ classical codewords, in which case $\L_{\eta_2}$ is the sparse ``sub-grid''
consisting of the $2$-cells $e\otimes e'$ with $e \in \mathrm{supp}\,z$ and
$e' \in \mathrm{supp}\,z'$.

% \emph{(iv) Size of the symmetry group.} The ordinary cluster state realizes a
% $\ZZ_2\times\ZZ_2\times\ZZ_2$ SPT with exactly three generators. Here, as we show
% next, the group has rank $k+k'+k''$ growing linearly with the number of qubits,
% and its generators are in bijection with codewords of the classical sheaf codes.
% This extensive symmetry is what will supply the extensively many logical qubits of
% the gauged model, and hence the constant rate of the resulting non-Abelian qLDPC
% code.  

\subsubsection{The $0$-form subcomplex symmetry of SPT}
\label{sec:subcomplex_symmetry}

The model is protected by
$\red{\Z_2^{(0)}} \times \blue{\Z_2^{(0)}} \times \green{\Z_2^{(0)}}$ $0$-form
symmetries,where `$(0)$' emphasizes it is 0-form and will be omitted sometimes. Because the coefficients are local, the symmetry generators are not
simply the uniform products $\prod_{v \in \L(0)} \tilde X_u$ as in the usual 2D cluster state ($\Z_2^3$ SPT); instead they are labelled by
sheaf $0$-\emph{cocycles}. Explicitly, for $\mu \in Z^0(\L,\Fc^{\rd})$,
$\nu \in Z^0(\L,\Fc^{\bl})$ and $\rho \in Z^0(\L,\Fc^{\gr})$ define
\begin{align}
    \begin{split}
        \red{\Z_2^{(0)}}: & \quad
        V_{\rd}(\mu) = \prod_{v \in V}\prod_{s}
        \big(\tilde X^{\rd}_{v,s}\big)^{[\mu(v)]_s}, \\
        \blue{\Z_2^{(0)}}: & \quad
        V_{\bl}(\nu) = \prod_{v \in V}\prod_{r}
        \big(\tilde X^{\bl}_{v,r}\big)^{[\nu(v)]_r}, \\
        \green{\Z_2^{(0)}}: & \quad
        V_{\gr}(\rho) = \prod_{v \in V}\prod_{t}
        \big(\tilde X^{\gr}_{v,t}\big)^{[\rho(v)]_t} .
        \label{eq:globalsym}
    \end{split}
\end{align}
These operators implement the shifts $\hat{\lambda}^0_{\rd} \to
\hat{\lambda}^0_{\rd}+\mu$, etc. Since $d\mu = d\nu = d\rho = 0$, the
coboundaries $d\hat{\lambda}^0_i$ appearing in Eq.~\eqref{eq:SPT_Hamiltonian} are
left invariant, and one verifies directly that $H_\text{SPT}$ commutes with all
three families of generators.

Two features distinguish this from the usual $0$-form symmetry of a cluster state:

\emph{The symmetry group is the classical sheaf code.} Because there are no
$(-1)$-cells, $Z^0(\L,\Fc^{\rd}) = \ker d^0 = H^0(\L,\Fc^{\rd})$, which by
Proposition~\ref{Prop:sheaf_cohomology} is precisely the $0$-cocycle classical
sheaf code $\T^*(G,\{C_u\}_{v\in V})$. The symmetry group of the sheaf SPT is
therefore
\be\label{eq:symmetry_group}
\red{H^0(\L,\Fc^{\rd})} \times \blue{H^0(\L,\Fc^{\bl})} \times \green{H^0(\L,\Fc^{\gr})}
\cong \ZZ_2^{k} \times \ZZ_2^{k'} \times \ZZ_2^{k''} ,
\ee
i.e.\ symmetry generators are in one-to-one correspondence with \emph{codewords}
of the underlying classical sheaf codes, and the number of independent generators
is the classical code dimension. In the case of usual cluster state defined on the cellulation of a manifold, 
one has $H^0(\L, \FF_2)=\FF_2$ which recovers the single uniform generator
$\prod_{v}\tilde{X}_u$.

\emph{The symmetry is a $0$-form subcomplex symmetry.} Conventionally a $k$-form
symmetry is generated by an operator supported on a codimension-$k$ submanifold,
so a $0$-form symmetry acts on all of space. Here, following the generalization to
codimension-$k$ \emph{subcomplexes} $\L^*_k \subseteq \L^*$, the relevant
generators act on the codimension-$0$ subcomplex
$\L^*_{\eta_2} \subseteq \L^*$ dual to $\L_{\eta_2}$. Concretely, decompose any
cocycle as $\mu = \mu|_{\L_{\eta_2}} + \mu|_{\L \setminus \L_{\eta_2}}$. The
second piece acts on qubits that the entangler never touches, so it acts on a
trivial paramagnet and carries no SPT data; all of the protected structure resides
in the restrictions
\begin{align}\label{eq:subcomplex_sym}
    \begin{split}
        \red{\Z_2^{(0)}[\eta_2]}: & \quad
        V_{\rd}\big(\mu|_{\L_{\eta_2}}\big), \\
        \blue{\Z_2^{(0)}[\eta_2]}: & \quad
        V_{\bl}\big(\nu|_{\L_{\eta_2}}\big), \\
        \green{\Z_2^{(0)}[\eta_2]}: & \quad
        V_{\gr}\big(\rho|_{\L_{\eta_2}}\big).
    \end{split}
\end{align}
Accordingly we refer to Eq.~\eqref{eq:SPT_Hamiltonian_final} as a \emph{$0$-form
subcomplex $\red{\Z_2^{(0)}} \times \blue{\Z_2^{(0)}} \times \green{\Z_2^{(0)}}$ SPT}: the symmetry is $0$-form in the
sense that its generators have codimension $0$, but they are supported on the
proper subcomplex $\L_{\eta_2}$ picked out by the $2$-cycle $\eta_2$, rather than
on the entire spatial complex. Note that it is precisely the cycle condition
$\partial\eta_2=0$ that makes
$\L_{\eta_2}$ a legitimate support for a $0$-form symmetry: it is the
combinatorial substitute for ``closed and codimension $0$'', and it is what
guarantees the gauge invariance of the response action in
Lemma~\ref{lemma:invariant}.

%======================================================================
\subsection{Constructing twisted qLDPC codes by gauging the $0$-form
subcomplex symmetry}
\label{sec:gauging}
%======================================================================

We now gauge the $0$-form subcomplex $\red{\Z_2^{(0)}} \times \blue{\Z_2^{(0)}} \times \green{\Z_2^{(0)}}$ symmetry
identified in Eqs.~\eqref{eq:symmetry_group} and \eqref{eq:subcomplex_sym},
following Refs.~\cite{Yoshida:2015cia, barkeshli2023codimension}. There are two
equivalent notions of gauging:
\begin{enumerate}
\item
Promote the static background gauge fields $A$ to dynamical gauge fields $a$ with
quantum fluctuations.
\item
Promote the global symmetries to local gauge symmetries.
\end{enumerate}
Notion 1 amounts to promoting the background sheaf $1$-cochains $\red{A^1}$,
$\blue{B^1}$, $\green{C^1}$ in $S_\text{SPT}$ of Eq.~\eqref{eq:SPT_action} to
dynamical sheaf gauge fields $\red{a^1}$, $\blue{b^1}$, $\green{c^1}$ in the
twisted sheaf gauge theory of Eq.~\eqref{eq:non-Abelian_action}:
\begin{align}
& S_\text{SPT} = \pi \int_{\tilde{\eta}_3}
\red{A^1} \cup \blue{B^1} \cup \green{C^1} \cr
&\rightarrow \ S = \pi \int_{\tilde{\eta}_3}
\red{a^1} \cup \blue{b^1} \cup \green{c^1} + \text{sheaf BF terms}. \cr
\end{align}
Notion 2 is the one we pursue to construct the Clifford stabilizer models of the
non-Abelian qLDPC codes: we promote the $0$-form subcomplex symmetries of
Eq.~\eqref{eq:globalsym} to local gauge symmetries by introducing new degrees of
freedom---$\ZZ_2$ sheaf gauge fields---and imposing a Gauss law.

It is important to be precise about what is being gauged where. The Gauss law is
imposed on \emph{every} vertex of $\L$, so that the gauge fields occupy all edges
and the resulting code is defined on the whole complex; what is localized to the
subcomplex is the \emph{twist}. Off $\L_{\eta_2}$ the procedure gauges a trivial
paramagnet and produces the untwisted sheaf code of
Eq.~\eqref{eq:H_untwisted}; on $\L_{\eta_2}$ it gauges the nontrivial
$0$-form subcomplex $\ZZ_2^3$ SPT and produces the CZ-dressed stabilizers that
make the code non-Abelian. In this sense the $2$-cycle $\eta_2$ is a tunable
resource: it selects which region of the code carries the twist, and different
homology classes $[\eta_2]$ give inequivalent twisted codes.

\subsubsection{Gauge fields on edges}

The gauge qubits live on the $1$-cells (edges) of $\L$, with
$\dim \Fc^{\rd}(e)$ gauge qubits per edge in copy \rd, and similarly for \bl \ and \gr.
Recall that in the example hypergraph product complex of Eq.~\eqref{eq:local_system_2D} the
edge stalks $\FF_2 \otimes L^y_{v'}$ and $L^x_v \otimes \FF_2$ are typically 
higher-dimensional, so an edge in general carries several gauge qubits. Following the
convention of Eq.~\eqref{eq:qubit_labels}, we label the gauge qubits by $(e,s)$,
$(e,r)$ and $(e,t)$ in the \rd, \bl \ and \gr \ copies, and denote the associated
Pauli operators by $X'^{\rd}_{e,s}, Y'^{\rd}_{e,s}, Z'^{\rd}_{e,s}$ and so on.
The three types of $\ZZ_2$ sheaf gauge fields $\red{\hat{a}'^1} \in C^1(\L,\Fc^{\rd})$,
$\blue{\hat{b}'^1} \in C^1(\L,\Fc^{\bl})$, $\green{\hat{c}'^1} \in C^1(\L,\Fc^{\gr})$ are
operator-valued $1$-cochains related to the Pauli-$Z$ operators by
\begin{align}
    \begin{split}
        (-1)^{[\red{\hat{a}'^1}(e)]_s} = Z'^{\rd}_{e,s}, \quad
        (-1)^{[\blue{\hat{b}'^1}(e)]_r} = Z'^{\bl}_{e,r}, \\
        (-1)^{[\green{\hat{c}'^1}(e)]_t} = Z'^{\gr}_{e,t}.
    \end{split}
\end{align}

The key idea of gauging in notion 2 is that the matter fields $\hat{\lambda}^0_i$
of Eq.~\eqref{eq:Pauli_relation}, previously transformed only under the $0$-form
subcomplex symmetries \eqref{eq:globalsym} labelled by \emph{cocycles}, are now
transformed under local gauge symmetries labelled by arbitrary sheaf
$0$-\emph{cochains} $\mu \in C^0(\L,\Fc^{\rd})$, $\nu \in C^0(\L,\Fc^{\bl})$,
$\rho \in C^0(\L,\Fc^{\gr})$. Under a gauge transformation the matter fields shift as
\begin{align}
    \hat{\lambda}_{\rd}^{0} &\rightarrow \hat{\lambda}_{\rd}^{0}-\mu ~,\cr
    \hat{\lambda}_{\bl}^{0} &\rightarrow \hat{\lambda}_{\bl}^{0}-\nu ~,\cr
    \hat{\lambda}_{\gr}^{0} &\rightarrow \hat{\lambda}_{\gr}^{0}-\rho ~,
    \label{eq:gaugetrans1}
\end{align}
while the gauge fields shift by the corresponding sheaf coboundaries,
\begin{align}
    \red{\hat{a}'^1} &\rightarrow \red{\hat{a}'^1}+d\mu ~,\cr
    \blue{\hat{b}'^1} &\rightarrow \blue{\hat{b}'^1}+d\nu ~,\cr
    \green{\hat{c}'^1} &\rightarrow \green{\hat{c}'^1}+d\rho ~,
    \label{eq:gaugetrans2}
\end{align}
with $d$ the sheaf coboundary of Eq.~\eqref{eq:coboundary-sheaf}.

\subsubsection{Gauss's law}

We choose as generators of the gauge group the elementary cochains
$\mu = \bar{v}_s$, $\nu = \bar{v}_r$, $\rho = \bar{v}_t$, each of which acts on a
single matter qubit. The corresponding Gauss's law operators are
\begin{align}\label{eq:Gauss_law}
    G^{\rd}(v,s) = & \ \tilde X^{\rd}_{v,s}
    \prod_{e \dotsucc v} \prod_{s'}
    \big(X'^{\rd}_{e,s'}\big)^{\left[\Fc^{\rd}_{u \shortrightarrow e}
    (\vec{e}_s)\right]_{s'}} ~,\cr
    G^{\bl}(v,r) = & \ \tilde X^{\bl}_{v,r}
    \prod_{e \dotsucc v} \prod_{r'}
    \big(X'^{\bl}_{e,r'}\big)^{\left[\Fc^{\bl}_{u \shortrightarrow e}
    (\vec{e}_r)\right]_{r'}} ~,\cr
    G^{\gr}(v,t) = & \ \tilde X^{\gr}_{v,t}
    \prod_{e \dotsucc v} \prod_{t'}
    \big(X'^{\gr}_{e,t'}\big)^{\left[\Fc^{\gr}_{u \shortrightarrow e}
    (\vec{e}_t)\right]_{t'}} ~,
\end{align}
where $s'$, $r'$, $t'$ index the edge stalks $\Fc^{\rd}(e)$, $\Fc^{\bl}(e)$, $\Fc^{\gr}(e)$.
Compared to the scalar case, the product over incident edges is now weighted by
the co-restriction map, i.e.\ by the corresponding column of the local
parity-check matrix $\mathsf{H}_v$ of Eq.~\eqref{eq:restriction_map}.
Equivalently, the $X'$ part of
$G^{\rd}(v,s)$ is exactly the row $(v,s)$ of the sheaf coboundary matrix $d^0$,
i.e.\ the $X$-stabilizer $A^{\rd}_{v,s}$ of the quantum sheaf code of
Definition~\ref{def:sheaf-code} at degree $1$. One checks directly that
Eq.~\eqref{eq:Gauss_law} generates the transformations
\eqref{eq:gaugetrans1}--\eqref{eq:gaugetrans2}: for instance $G^{\rd}(v,s)$ acts
by $\hat{\lambda}_{\rd}^{0}\to \hat{\lambda}_{\rd}^{0}-\bar{v}_s$ and
$\red{\hat{a}'^1}\to \red{\hat{a}'^1} + d\bar{v}_s$.

Gauging is carried out by imposing $G^{\rd}=G^{\bl}=G^{\gr}=1$ on every vertex
qubit, which defines the gauge-invariant physical Hilbert space
$\mathcal{H}_{\text{phys}}$ spanned by states with
$G^{i}\ket{\text{phys}} = \ket{\text{phys}}$. Note that the local generators
$\bar{v}_s$ are \emph{not} cocycles; the $0$-form subcomplex symmetry
\eqref{eq:subcomplex_sym} we started from is recovered as the subgroup of gauge
transformations with $d\mu=0$, which act trivially on the gauge fields and
therefore survive as global (logical) operations after gauging.

\subsubsection{Minimal coupling and the gauged Hamiltonian}

Next we minimally couple $H_{\text{SPT}}$ to the sheaf gauge fields so that it
commutes with the Gauss's law operators and acts within
$\mathcal{H}_{\text{phys}}$. This is achieved by the replacement
\begin{align}
d\hat{\lambda}_{\rd}^{0} \rightarrow d\hat{\lambda}_{\rd}^{0}+\red{\hat{a}'^1}
\equiv \red{\hat{a}^1}, \cr
d\hat{\lambda}_{\bl}^{0} \rightarrow d\hat{\lambda}_{\bl}^{0}+\blue{\hat{b}'^1}
\equiv \blue{\hat{b}^1}, \cr
d\hat{\lambda}_{\gr}^{0} \rightarrow d\hat{\lambda}_{\gr}^{0}+\green{\hat{c}'^1}
\equiv \green{\hat{c}^1}.
\end{align}
The new gauge fields $\red{\hat{a}^1}, \blue{\hat{b}^1}, \green{\hat{c}^1}$ are
related to the old ones by the gauge transformation \eqref{eq:gaugetrans2} with
$\mu = \hat{\lambda}^0_{\rd}$, $\nu = \hat{\lambda}^0_{\bl}$,
$\rho = \hat{\lambda}^0_{\gr}$, which eliminates the matter fields entirely.
Accordingly we define the gauge-invariant Pauli operators on each edge qubit:
\begin{align}
    \begin{split}
        X^{\rd}_{e,s} &= X'^{\rd}_{e,s}~, \\
        Z^{\rd}_{e,s} &= (-1)^{[\red{\hat{a}^1}(e)]_s}
        = (-1)^{[(d\hat{\lambda}^0_{\rd} + \red{\hat{a}'^1})(e)]_s}~, \\
        X^{\bl}_{e,r} &= X'^{\bl}_{e,r}~, \\
        Z^{\bl}_{e,r} &= (-1)^{[\blue{\hat{b}^1}(e)]_r}
        = (-1)^{[(d\hat{\lambda}^0_{\bl} + \blue{\hat{b}'^1})(e)]_r}~, \\
        X^{\gr}_{e,t} &= X'^{\gr}_{e,t}~, \\
        Z^{\gr}_{e,t} &= (-1)^{[\green{\hat{c}^1}(e)]_t}
        = (-1)^{[(d\hat{\lambda}^0_{\gr} + \green{\hat{c}'^1})(e)]_t}~.
        \label{eq:gaugeinvariant_paulis}
    \end{split}
\end{align}
We further impose flatness of the sheaf gauge fields, or equivalently the cocycle
condition $\red{d\hat{a}^1}=\blue{d\hat{b}^1}=\green{d\hat{c}^1}=0$, so that
$\red{\hat{a}^1}, \blue{\hat{b}^1}, \green{\hat{c}^1}$ become operator-valued
sheaf $1$-cocycles rather than general $1$-cochains. The gauged SPT Hamiltonian
is then
\begin{align}
    \tilde{H}&=H_\text{Gauss}+H_\text{Flux}~,
\end{align}
where the Gauss's law term reads
\begin{align}\label{eq:Gauss_law_term}
    H_\text{Gauss}
    =& -\sum_{v,s} A^{\rd}_{v,s}\,
    (-1)^{\int_{\eta_2} \bar{v}_s \cup \blue{\hat{b}^1} \cup
    \green{\hat{c}^1}} \cr
    &-\sum_{v,r} A^{\bl}_{v,r}\,
    (-1)^{\int_{\eta_2} \red{\hat{a}^1} \cup \bar{v}_r \cup
    \green{\hat{c}^1}} \cr
    &-\sum_{v,t} A^{\gr}_{v,t}\,
    (-1)^{\int_{\eta_2} \red{\hat{a}^1} \cup \blue{\hat{b}^1} \cup
    \bar{v}_t} ~,
\end{align}
with the bare $X$-stabilizers
\begin{align}\label{eq:A_stabilizer}
A^{\rd}_{v,s} &= \prod_{e \dotsucc v} \prod_{s'}
\big(X^{\rd}_{e,s'}\big)^{\left[\Fc^{\rd}_{u \shortrightarrow e}
(\vec{e}_s)\right]_{s'}} , \cr
A^{\bl}_{v,r} &= \prod_{e \dotsucc v} \prod_{r'}
\big(X^{\bl}_{e,r'}\big)^{\left[\Fc^{\bl}_{u \shortrightarrow e}
(\vec{e}_r)\right]_{r'}} , \cr
A^{\gr}_{v,t} &= \prod_{e \dotsucc v} \prod_{t'}
\big(X^{\gr}_{e,t'}\big)^{\left[\Fc^{\gr}_{u \shortrightarrow e}
(\vec{e}_t)\right]_{t'}} .
\end{align}
Here we have re-expressed the single-site $\tilde X$ operators of $H_\text{SPT}$
as products of edge $X$ operators using the Gauss's law constraints $G^i=1$
according to Eq.~\eqref{eq:Gauss_law}, together with the identification of $X'$
and $X$ in Eq.~\eqref{eq:gaugeinvariant_paulis}. The substitution is exact on the
physical Hilbert space, since $G^{\rd}(v,s)=1$ and
$\big(A^{\rd}_{v,s}\big)^2=\I$ give
$\tilde X^{\rd}_{v,s}\big|_{\mathcal{H}_\text{phys}} =
A^{\rd}_{v,s}\big|_{\mathcal{H}_\text{phys}}$, which is what allows the matter
qubits to be eliminated altogether. The $A^{i}_{v,s}$ are the colour-resolved
versions of the $X$-stabilizers of Eq.~\eqref{eq:A_stabilizer_untwisted}, i.e.\ the
rows of $H_X=(d^0)^T$; see Sec.~\ref{sec:X_stabilizer} for their derivation from
$d^0$ and their interpretation as generators of a coboundary shift.

Note that in the second and
third lines of Eq.~\eqref{eq:Gauss_law_term} the elementary cochain sits in the
middle and last cup slot respectively, so the relevant front/back faces are those
of the $(1,0,1)$ and $(1,1,0)$ splittings. For vertices $u \notin \L_{\eta_2}$ all
three integrals vanish identically and the dressing factor is $1$.

The flux term
\begin{align}\label{eq:Hflux}
    H_\text{Flux} = -\sum_{f \in \L(2)}
    \Big( B^{\rd}_{f} + B^{\bl}_{f} + B^{\gr}_{f} \Big) ~,
\end{align}
with
\begin{align}\label{eq:B_stabilizer}
B^{\rd}_{f} &= \prod_{e \precdot f} \prod_{s'}
\big(Z^{\rd}_{e,s'}\big)^{\left[\Fc^{\rd}_{e \shortrightarrow f}\right]_{s'}} , \cr
B^{\bl}_{f} &= \prod_{e \precdot f} \prod_{r'}
\big(Z^{\bl}_{e,r'}\big)^{\left[\Fc^{\bl}_{e \shortrightarrow f}\right]_{r'}} , \cr
B^{\gr}_{f} &= \prod_{e \precdot f} \prod_{t'}
\big(Z^{\gr}_{e,t'}\big)^{\left[\Fc^{\gr}_{e \shortrightarrow f}\right]_{t'}} ,
\end{align}
energetically enforces the flat sheaf gauge field (cocycle) condition
$\red{d\hat{a}^1}=0, \ \blue{d\hat{b}^1}=0, \ \green{d\hat{c}^1}=0$ on every
$2$-cell $f$ in the ground-state subspace of $\tilde{H}$, in direct analogy with
the plaquette term of the standard ($1$-form) toric code.

The $B^{i}_{f}$ are the colour-resolved versions of the $Z$-stabilizers of
Eq.~\eqref{eq:B_stabilizer_untwisted}; in particular
$B^{i}_{f}=(-1)^{(d\hat{a}^1)(f)}$ is the row $f$ of $H_Z=d^1$, and it carries no
vector index because the stalk on a top cell is a scalar. See
Sec.~\ref{sec:Z_stabilizer} for the derivation, and
Sec.~\ref{sec:untwisted_Hamiltonian} for the commutation
$[A^{i}_{v,s},B^{i}_{f}]=0$, which is equivalent to $d^1 d^0=0$.
Unlike the Gauss term, the flux term is imposed on all of $\L$ and is untouched by
the twist.

\subsubsection{The twisted sheaf code}\label{sec:twisted_code}

Finally we re-express the Gauss's law terms in Eq.~\eqref{eq:Gauss_law_term} by
expanding the gauge fields in elementary cochains,
$\red{\hat{a}^1} = \sum_{e,s} \hat{N}^{\rd}_{e,s}\, \bar{e}_s$,
$\blue{\hat{b}^1} = \sum_{e,r} \hat{N}^{\bl}_{e,r}\, \bar{e}_r$,
$\green{\hat{c}^1} = \sum_{e,t} \hat{N}^{\gr}_{e,t}\, \bar{e}_t$, along with the
identity $(-1)^{\hat{N}^{i}_{e,s}\hat{N}^{i'}_{e',s'}} =
\text{CZ}^{i,i'}_{(e,s),(e',s')}$, which turns the terms containing gauge fields
into CZ operators. The rewritten gauged SPT Hamiltonian
$\tilde{H}=H_\text{Gauss}+H_\text{Flux}$ then gives rise to the Clifford
stabilizer generators of the twisted sheaf code $\tilde{\C}$:
\begin{align}\label{eq:stabilizer_summary}
\tilde{A}^{\rd}_{v,s} =& A^{\rd}_{v,s}
\prod_{\substack{(e,r),(e',t):\\
\int_{\eta_2} \bar{v}_s \cup \bar{e}_r \cup \bar{e}'_t \neq 0}}
\text{CZ}^{\bl,\gr}_{(e,r),(e',t)} , \cr
\tilde{A}^{\bl}_{v,r} =& A^{\bl}_{v,r}
\prod_{\substack{(e,s),(e',t):\\
\int_{\eta_2} \bar{e}_s \cup \bar{v}_r \cup \bar{e}'_t \neq 0}}
\text{CZ}^{\rd,\gr}_{(e,s),(e',t)} , \cr
\tilde{A}^{\gr}_{v,t} =& A^{\gr}_{v,t}
\prod_{\substack{(e,s),(e',r):\\
\int_{\eta_2} \bar{e}_s \cup \bar{e}'_r \cup \bar{v}_t \neq 0}}
\text{CZ}^{\rd,\bl}_{(e,s),(e',r)} , \cr
B^{\rd}_{f}, \quad & B^{\bl}_{f}, \quad B^{\gr}_{f}
\quad \text{as in Eq.~\eqref{eq:B_stabilizer}} .
\end{align}
The $Z$-type generators $B^i$ are unchanged by the twist, while each $X$-type
generator is dressed by CZ gates acting on the two \emph{other} colors, and only
for vertices lying in the subcomplex $\L_{\eta_2}$. The set of qubit pairs
entering each dressing is controlled entirely by the triple cup product together
with the choice of $2$-cycle $\eta_2$, i.e.\ by the invariant form of
Definition~\ref{def:invariant-form}.

\begin{figure}[t]
    \centering   \includegraphics[width=1\linewidth]{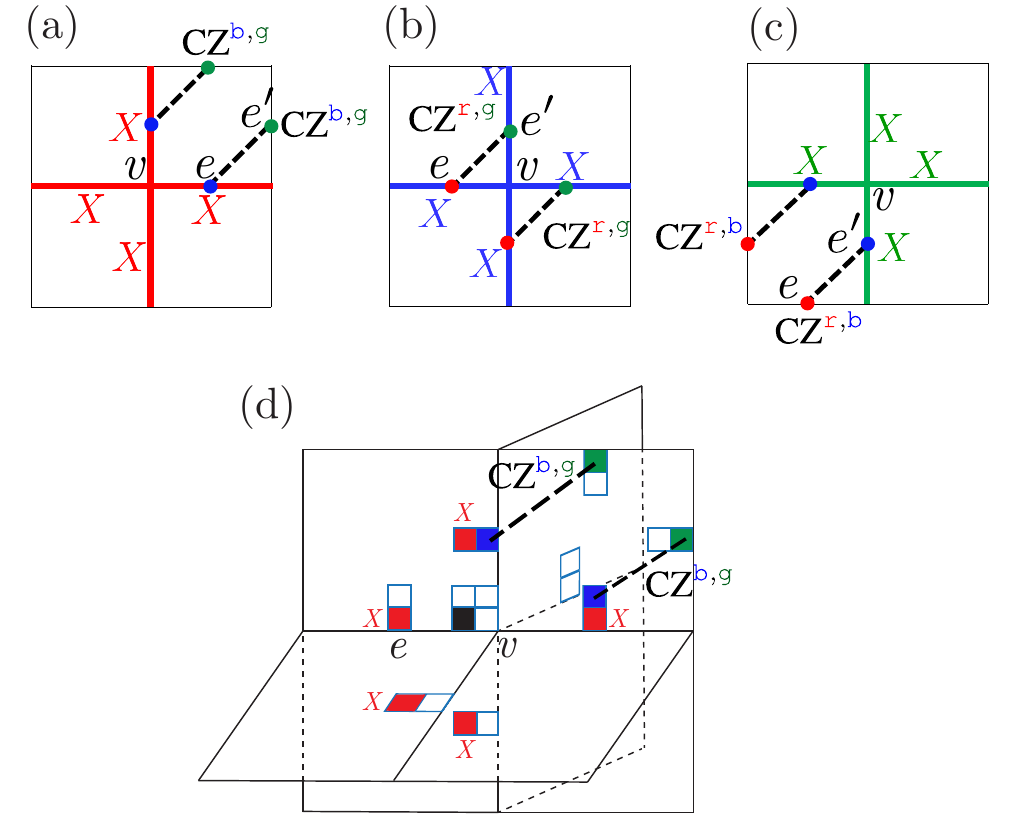}
    \caption{The twisted ($X$-type) stabilizer generators of the
Clifford stabilizer code $\tilde{\C}$, Eq.~\eqref{eq:stabilizer_summary}. The
$Z$-type generators $B_f$ are unchanged by the twist. (a-c) is an illustration with a simplifed case with constant sheaf (trivial local code). (a) The \rd\ generator: the
Pauli part applies $X$ to the \rd\ qubits selected at the vertex $v$ exactly as in
Fig.~\ref{fig:square_complex}(a), and the dressing applies CZ gates between the
\bl\ and \gr\ qubits on the pairs of edges  $e$, $e'$
picked out by the triple cup product on $\eta_2$. (b), (c) The \bl\ and \gr\
generators, obtained by cyclically permuting the colours, so that each generator is
dressed by CZ gates acting on the two \emph{other} copies. (d) The same
generators drawn on the $2$D branched square complex. The dressing is supported on
bounded-size clusters of cells, so the code remains qLDPC; because the dressing is
non-Pauli, the generators commute only on the zero-flux subspace
(Lemma~\ref{lemma:stabilizer_commutation}).}
    \label{fig:dressed_stabilizer}
\end{figure}

Note that the untwisted code $\C$ can be
obtained by directly gauging the trivial paramagnetic Hamiltonian $H^0$ of
Eq.~\eqref{eq:trivial_H} instead of the SPT Hamiltonian $H_\text{SPT}$, with
completely the same gauging procedure discussed above. The CZ terms in the
Gauss's law term $A$ go away because they were absent in $H^0$. The resulting
code is precisely the quantum sheaf code of Definition~\ref{def:sheaf-code} at
degree $1$,
\be
C^{0}(\L, \Fc^{i}) \xrightarrow{\ d^{0}\ } C^{1}(\L, \Fc^{i})
\xrightarrow{\ d^{1}\ } C^{2}(\L, \Fc^{i}),
\ee
with $H_X^T = d^0$ and $H_Z = d^1$, i.e.\ three decoupled copies of the
hypergraph-product sheaf code of Sec.~\ref{sec:classical}. Equivalently, the
twisted code $\tilde{\C}$ agrees with $\C$ outside $\L_{\eta_2}$ and differs from
it only by the CZ dressings inside. The twist by the triple cup product couples
the three copies and, as we show in the next section, promotes the Abelian
$\ZZ_2^3$ topological order to a non-Abelian one while preserving the qLDPC
property, since the cup product and $\eta_2$ are supported on bounded-size
clusters of cells.

%======================================================================
\subsection{Clifford stabilizer codes on sheaf complexes}
\label{sec:Clifford_stabilizer_codes}
%======================================================================

Having derived the twisted stabilizer generators by gauging, we now study the
resulting family of codes in its own right. This subsection is the sheaf
translation of the corresponding analysis for higher-form twisted gauge theories
on Poincar\'e complexes \cite{Hsin2024_non-Abelian, ZhuKobayashiHsin2026},
specialized throughout to the case in which all three form degrees equal one,
i.e.\ to the $0$-form subcomplex
$\red{\Z_2^{(0)}} \times \blue{\Z_2^{(0)}} \times \green{\Z_2^{(0)}}$ symmetry of
Sec.~\ref{sec:subcomplex_symmetry}. Two features distinguish the discussion from
the manifold case. First, all cells carry local coefficients, so every operator
label is a pair (cell, vector index) and every incidence is weighted by a
co-restriction map. Second, the dual complex and Poincar\'e duality are not
needed: on a sheaf complex the $Z$-type logical operators are naturally labelled
by sheaf $1$-\emph{chains} and the $X$-type ones by sheaf $1$-\emph{cochains},
which are paired non-degenerately by Proposition~\ref{prop:nondegen-sheaf}.

\subsubsection{Construction of the codes}\label{sec:code_construction}

The twisted code $\tilde{\C}$ is a \emph{Clifford} stabilizer code: its stabilizer
group $\tilde{\SS}$ is a subgroup of the $n$-qubit Clifford group $Cl_n$ rather
than of the Pauli group, and is generated by the operators of
Eq.~\eqref{eq:stabilizer_summary},
\be\label{eq:Clifford_group_sheaf}
\tilde{\SS} = \big\langle\, \{\,
\tilde{A}^{\rd}_{v,s},\ \tilde{A}^{\bl}_{v,r},\ \tilde{A}^{\gr}_{v,t},\
B^{\rd}_{f},\ B^{\bl}_{f},\ B^{\gr}_{f} \,\}\, \big\rangle ,
\ee
where $v \in \L(0)$ and $f \in \L(2)$, the indices $s$, $r$, $t$ run over the
vertex stalks $\Fc^{\rd}(v)$, $\Fc^{\bl}(v)$, $\Fc^{\gr}(v)$, and the $B$'s carry
no vector index because the stalk on a top cell is a scalar
(Sec.~\ref{sec:Z_stabilizer}). The code space is defined as in the Pauli case,
\be\label{eq:code_space_Clifford}
\tilde{\C}(\tilde{\SS}) = \big\{ \ket{\psi} \in (\CC^2)^{\otimes n} :\;
g \ket{\psi} = \ket{\psi} \ \ \forall g \in \tilde{\SS} \big\} .
\ee
Setting all CZ dressings to the identity returns the untwisted Pauli stabilizer
group
\be\label{eq:stabilizer_untwisted_sheaf}
\SS = \big\langle\, \{\,
A^{\rd}_{v,s},\ A^{\bl}_{v,r},\ A^{\gr}_{v,t},\
B^{\rd}_{f},\ B^{\bl}_{f},\ B^{\gr}_{f} \,\}\, \big\rangle < \mathcal{P}_n ,
\ee
which is three decoupled copies of the untwisted sheaf code of
Sec.~\ref{sec:untwisted_code}, with $A^{i}_{v,s}$ the rows of $H_X=(d^0)^T$ and
$B^{i}_{f}$ the rows of $H_Z=d^1$. The $Z$-stabilizers are identical in
$\tilde{\SS}$ and $\SS$; only the $X$-stabilizers are dressed.

\subsubsection{Constraints on the code space}\label{sec:code_space_constraint}

On a manifold the analogous constraints follow from the equations of motion of the
twisted action. On a sheaf complex they follow directly from the stabilizer group,
which is both simpler and more general.

Let $\mu^0 \in Z^0(\L,\Fc^{\rd})$ be a sheaf $0$-\emph{cocycle}, i.e.\
$d^0\mu^0=0$, and expand it in elementary cochains as
$\mu^0=\sum_{v,s}\mu_{v,s}\,\bar{v}_s$. Consider the corresponding product of red
twisted stabilizers. Because $A^{\rd}_{v,s}$ acts only on the \rd\ copy while its
dressing $\text{CZ}^{\bl,\gr}$ acts only on the \bl\ and \gr\ copies, the two
factors commute and the product splits exactly,
\begin{align}\label{eq:constraint_derivation}
\prod_{v,s} \big(\tilde{A}^{\rd}_{v,s}\big)^{\mu_{v,s}}
=& \prod_{v,s} \big(A^{\rd}_{v,s}\big)^{\mu_{v,s}} \cdot
(-1)^{\int_{\eta_2}\mu^0 \cup \blue{\hat{b}^1} \cup \green{\hat{c}^1}} \cr
=& \; (-1)^{\int_{\eta_2}\mu^0 \cup \blue{\hat{b}^1} \cup \green{\hat{c}^1}} ,
\end{align}
where the $X$-part collapses to the identity because, by
Eq.~\eqref{eq:A_conjugation} applied to each factor, it implements the shift
$\red{\hat{a}^1} \rightarrow \red{\hat{a}^1} + d^0\mu^0 = \red{\hat{a}^1}$. Since
every generator equals $+1$ on $\tilde{\C}$, we obtain the operator constraint
\be\label{eq:constraint_operator_sheaf}
\int_{\eta_2} \mu^0 \cup \blue{\hat{b}^1} \cup \green{\hat{c}^1} = 0
\qquad \text{on } \tilde{\C}, \quad \forall\, \mu^0 \in Z^0(\L,\Fc^{\rd}),
\ee
and, permuting the colours,
\be\label{eq:constraint_operator_sheaf2}
\int_{\eta_2} \red{\hat{a}^1} \cup \nu^0 \cup \green{\hat{c}^1} = 0,
\qquad
\int_{\eta_2} \red{\hat{a}^1} \cup \blue{\hat{b}^1} \cup \green{\gamma^0} = 0,
\ee
for all $\nu^0 \in Z^0(\L,\Fc^{\bl})$ and $\green{\gamma^0} \in Z^0(\L,\Fc^{\gr})$. Note
that the cocycles $\mu^0,\nu^0,\green{\gamma^0}$ are precisely the labels of the $0$-form
subcomplex symmetry generators of Eq.~\eqref{eq:globalsym}, so
Eqs.~\eqref{eq:constraint_operator_sheaf}--\eqref{eq:constraint_operator_sheaf2}
are one constraint per symmetry generator.

Expanding the gauge fields in elementary cochains as in
Sec.~\ref{sec:twisted_code} and using
$(-1)^{\hat{N}^{i}_{e,s}\hat{N}^{i'}_{e',s'}} \equiv
\text{CZ}^{i,i'}_{(e,s),(e',s')}$, the same statement takes the form of a
transversal CZ operator being trivial on the code space,
\be\label{eq:CZ_identity_sheaf}
\widetilde{\text{CZ}}^{\bl,\gr}(\mu^0) \equiv
\prod_{\substack{(e,r),(e',t):\\ \int_{\eta_2}\mu^0\cup\bar{e}_r\cup\bar{e}'_s\neq 0}}
\text{CZ}^{\bl,\gr}_{(e,r),(e',t)} = \I \quad \text{on } \tilde{\C},
\ee
together with $\widetilde{\text{CZ}}^{\rd,\gr}(\nu^0)=\I$ and
$\widetilde{\text{CZ}}^{\rd,\bl}(\green{\gamma^0})=\I$.

\subsubsection{Stabilizer commutation relations}
\label{sec:stabilizer_commutation_sheaf}

The parent Hamiltonian of the twisted code,
\begin{align}\label{eq:H_twisted_sheaf}
\tilde{H} = &-\sum_{v,s}\tilde{A}^{\rd}_{v,s}
-\sum_{v,r}\tilde{A}^{\bl}_{v,r}
-\sum_{v,t}\tilde{A}^{\gr}_{v,t} \cr
&-\sum_{f\in\L(2)}\Big(B^{\rd}_{f}+B^{\bl}_{f}+B^{\gr}_{f}\Big),
\end{align}
is frustration free, but $\tilde{\SS}$ is a \emph{non-commuting} stabilizer group:
its generators fail to commute on the full Hilbert space. They do commute on the
zero-flux subspace,
\be\label{eq:commute_on_PB}
P_B\,[g,g']\,P_B = \I , \qquad \forall\, g,g' \in \tilde{\SS},
\ee
where $[U,V]:=UVU^{-1}V^{-1}$ is the group commutator and $P_B$ projects onto the
$+1$ eigenspace of all $Z$-stabilizers $\{B^{\rd}_f,B^{\bl}_f,B^{\gr}_f\}$.

\begin{lemma}\label{lemma:stabilizer_commutation}
The commutator of two twisted $X$-stabilizers of different colours is a product of
$Z$-stabilizers of the third colour. All other pairs of generators commute
exactly. Hence $\tilde{\SS}$ is commuting on the zero-flux subspace $B=1$.
\end{lemma}

\begin{proof}
The $B$'s commute among themselves, being $Z$-type. Each $B^{i}_{f}$ commutes
exactly with every $\tilde{A}^{i'}_{v,s}$: for $i=i'$ this is the untwisted
statement $d^1 d^0=0$ of Eq.~\eqref{eq:commutation_d2}, for $i\neq i'$ the two act
on different copies, and the CZ dressing is diagonal and therefore commutes with
any product of $Z$'s. Two twisted $X$-stabilizers of the \emph{same} colour also
commute exactly, since $A^{\rd}_{v,s}$ and $A^{\rd}_{v',s'}$ are both $X$-type
while the dressings are diagonal and act on the other two copies.

It remains to compute the commutator of two twisted $X$-stabilizers of different
colours. We use the identity, valid for a gauge field $a$ diagonalized by
$Z=(-1)^a$ and for arbitrary $\FF_2$-valued functions $f,g$,
\begin{align}\label{eq:commutator_identity_sheaf}
&\big[\, X_{(\sigma,s)}(-1)^{f(a)},\; X_{(\sigma',p')}(-1)^{g(a)} \,\big] \cr
&\qquad = (-1)^{f(a+\bar{\sigma}'_{p'})-f(a)}\,
(-1)^{g(a+\bar{\sigma}_{s})-g(a)} ,
\end{align}
which follows from $[X_{(\sigma,s)},(-1)^{f(a)}]
=(-1)^{f(a+\bar{\sigma}_s)-f(a)}$ \cite{Hsin2024_non-Abelian}, now with
$\bar{\sigma}_s$ the elementary \emph{sheaf} cochain. Since $A^{\bl}_{v,r}$
implements $\blue{\hat{b}^1}\to\blue{\hat{b}^1}+d^0\bar{v}_r$ and $A^{\rd}_{u,s}$
implements $\red{\hat{a}^1}\to\red{\hat{a}^1}+d^0\bar{u}_s$
[Eq.~\eqref{eq:A_conjugation}], Eq.~\eqref{eq:commutator_identity_sheaf} gives
\begin{align}\label{eq:AA_commutator}
\big[\tilde{A}^{\rd}_{u,s},\, \tilde{A}^{\bl}_{v,r}\big]
&= (-1)^{\int_{\eta_2}\left[\,\bar{u}_s \cup d\bar{v}_r
+ d\bar{u}_s \cup \bar{v}_r\,\right]\cup \green{\hat{c}^1}} \cr
&= (-1)^{\int_{\eta_2} d\left(\bar{u}_s \cup \bar{v}_r\right)
\cup \green{\hat{c}^1}} \cr
&= (-1)^{\int_{\eta_2} \left(\bar{u}_s \cup \bar{v}_r\right)
\cup \green{d\hat{c}^1}} ,
\end{align}
where the second line is the Leibniz rule of Proposition~\ref{prop:cup-properties}
in characteristic $2$, and the third line uses the Leibniz rule once more together
with $\langle d\omega, \eta_2\rangle = \langle \omega, \partial \eta_2\rangle = 0$,
which holds because $\eta_2$ is a \emph{cycle}. The cycle condition
$\partial\eta_2=0$ is therefore not only what makes the response action gauge
invariant, but also what makes the twisted stabilizer group consistent.

Finally we convert the last line of Eq.~\eqref{eq:AA_commutator} into $Z$-type
stabilizers. The degrees are $0+0+2=2$, so the cup product is evaluated on
$2$-cells; since $\green{(d\hat{c}^1)}(f)\in\Fc^{\gr}(f)\cong\FF_2$ is a scalar,
the pairing with $\eta_2(f)$ factorizes and
\be\label{eq:Theta_def}
\int_{\eta_2} \left(\bar{u}_s \cup \bar{v}_r\right) \cup \green{d\hat{c}^1}
= \sum_{f \in \L(2)} \Theta^{(u,s),(v,r)}(f)\ \green{(d\hat{c}^1)}(f) ,
\ee
where $\Theta^{(u,s),(v,r)} \in \FF_2^{|\L(2)|}$ collects the co-restriction
weights and the local components of $\eta_2$. Using
$B^{\gr}_{f}=(-1)^{(\green{d\hat{c}^1})(f)}$ from
Eq.~\eqref{eq:B_stabilizer_untwisted}, we obtain
\be\label{eq:AA_commutator_B}
\big[\tilde{A}^{\rd}_{u,s},\, \tilde{A}^{\bl}_{v,r}\big]
= \prod_{f:\ \Theta^{(u,s),(v,r)}(f)\,=\,1} B^{\gr}_{f} ,
\ee
and similarly
$[\tilde{A}^{\rd}_{u,s},\tilde{A}^{\gr}_{w,t}]$ is a product of $B^{\bl}_f$ and
$[\tilde{A}^{\bl}_{v,r},\tilde{A}^{\gr}_{w,t}]$ a product of $B^{\rd}_f$. All
these equal $\I$ on the subspace $B^{\rd}=B^{\bl}=B^{\gr}=1$.
\end{proof}

\subsubsection{Logical operators}\label{sec:logical_operators_sheaf}

\paragraph*{Electric (logical-$Z$) operators.} These are unchanged by the twist
and are the sheaf Wilson operators corresponding to the worldlines of the Abelian electric anyons with three flavors (colors) $e^{\rd}$, $e^{\bl}$ and $e^{\gr}$ . For a sheaf $1$-chain
$\eta_1 \in C_1(\L,\Fc^{\rd})$ define
\be\label{eq:logical_Z_sheaf}
\widetilde{Z}^{\rd}_{\eta_1} \equiv W^{\rd}(\eta_1)
= (-1)^{\int_{\eta_1} \red{\hat{a}^1}}
= \prod_{(e,s)} \big(Z^{\rd}_{e,s}\big)^{[\eta_1(e)]_s} ,
\ee
which is the wordline of $e^{\rd}$,
and likewise $\widetilde{Z}^{\bl}_{\eta_1}$, $\widetilde{Z}^{\gr}_{\eta_1}$ with
$\eta_1 \in C_1(\L,\Fc^{\bl})$, $C_1(\L,\Fc^{\gr})$ for $e^{\bl}$ and $e^{\gr}$ respectively. Here
$\int_{\eta_1}(\cdot)\equiv\langle\cdot,\eta_1\rangle$ is the canonical pairing of
Definition~\ref{def:pairing-sheaf}. Such an operator commutes with all $B$'s
trivially, and by Eq.~\eqref{eq:A_conjugation} together with the adjunction
$\langle d^0 x, \eta_1\rangle = \langle x, \partial_1 \eta_1\rangle$ its
commutator with $A^{\rd}_{v,s}$ is $(-1)^{\langle \bar{v}_s, \partial_1\eta_1
\rangle}$. Hence
\be\label{eq:Z_logical_condition}
\big[\widetilde{Z}^{\rd}_{\eta_1},\, \tilde{A}^{\rd}_{v,s}\big] = \I
\quad \Longleftrightarrow \quad \partial_1 \eta_1 = 0 ,
\ee
i.e.\ the Wilson operator is a logical operator precisely when $\eta_1$ is a sheaf
$1$-cycle, and it acts trivially when $\eta_1$ is a boundary. The logical-$Z$
operators are therefore labelled by the sheaf homology classes
$[\eta_1] \in H_1(\L,\Fc^{i})$, and by Proposition~\ref{prop:nondegen-sheaf}
their number matches $\dim H^1(\L,\Fc^{i})=k$. Choosing a basis
$\{\red{\alpha_1}\}$, $\{\blue{\beta_1}\}$, $\{\green{\gamma_1}\}$ of cycle
representatives, we write the logical operators of individual logical qubits as
\be\label{eq:logical_Z_basis_sheaf}
\lo{Z}^{\rd}(\red{\alpha}) \equiv \widetilde{Z}^{\rd}_{\red{\alpha_1}} ,
\quad
\lo{Z}^{\bl}(\blue{\beta}) \equiv \widetilde{Z}^{\bl}_{\blue{\beta_1}} ,
\quad
\lo{Z}^{\gr}(\green{\gamma}) \equiv \widetilde{Z}^{\gr}_{\green{\gamma_1}} .
\ee
Note that no dual complex or Poincar\'e duality is invoked as in the manifold case \cite{Hsin:2024pdi, ZhuKobayashiHsin2026}: the electric operators
live on sheaf $1$-chains and the magnetic ones, below, on sheaf $1$-cochains.

\paragraph*{Magnetic (dressed-$X$) operators.}
The magnetic operators correspond to the worldlines of three flavors (colors) of non-Abelian magnetic anyons $m^{\rd}$, $m^{\bl}$ and $m^{\gr}$. An $X$-type operator on the \rd\
copy is labelled by a sheaf $1$-cochain $\zeta^1 \in C^1(\L,\Fc^{\rd})$ through
$\prod_{(e,s)}(X^{\rd}_{e,s})^{[\zeta^1(e)]_s}$. Its commutator with $B^{\rd}_f$
is $(-1)^{(d^1\zeta^1)(f)}$, so it commutes with all $Z$-stabilizers precisely
when $\zeta^1$ is a $1$-\emph{cocycle}. In the twisted code such an operator must
additionally be dressed. When $\zeta^1=d^0\nu^0$ is \emph{exact}, the dressed
operator is
\begin{align}\label{eq:magnetic_contractible_sheaf}
\widetilde{\cX}^{\rd}(\nu^0)
=& \prod_{(e,s)}\big(X^{\rd}_{e,s}\big)^{[(d^0\nu^0)(e)]_s} \cdot
(-1)^{\int_{\eta_2}\nu^0\cup\blue{\hat{b}^1}\cup\green{\hat{c}^1}} \cr
=& \prod_{v,s} \big(\tilde{A}^{\rd}_{v,s}\big)^{\nu_{v,s}} ,
\end{align}
i.e.\ it factorizes into a product of twisted stabilizers and therefore acts
trivially on $\tilde{\C}$. This is the sheaf statement that contractible magnetic
operators carry no logical action. It also shows that the dressing is
well defined: if $\nu^0$ and $\nu'^0$ satisfy $d^0\nu^0=d^0\nu'^0=\zeta^1$, then
$\nu^0-\nu'^0$ is a $0$-cocycle and the two dressings differ by
$(-1)^{\int_{\eta_2}(\nu^0-\nu'^0)\cup\blue{\hat{b}^1}\cup\green{\hat{c}^1}}=\I$
on $\tilde{\C}$ by Eq.~\eqref{eq:constraint_operator_sheaf}.

Genuine logical operators correspond to cocycles with $[\zeta^1]\neq 0$ in
$H^1(\L,\Fc^{\rd})$, for which no global $\nu^0$ with $d^0\nu^0=\zeta^1$ exists
and the dressing of Eq.~\eqref{eq:magnetic_contractible_sheaf} is unavailable.
Following Ref.~\cite{Hsin2024_non-Abelian}, one instead decorates the bare $X$
string with projectors enforcing the constraint
\eqref{eq:constraint_operator_sheaf}:
\begin{align}\label{eq:magnetic_projected_sheaf}
\cX^{\rd}(\zeta^1) &= P^{\rd}\bigg[ 
\prod_{(e,s)}\big(X^{\rd}_{e,s}\big)^{[\zeta^1(e)]_s} \cr
& \cdot \prod_{\substack{(e,r),\,(e',t):\\ e,e' \in
\mathrm{supp}\,\zeta^1}}
\big(\I+Z^{\bl}_{e,r}\big)\big(\I+Z^{\gr}_{e',t}\big) \bigg] P^{\rd}, \cr
\end{align}
and similarly for $\cX^{\bl}$ and $\cX^{\gr}$ by permuting colours. Here
$P^{\rd}$ projects onto the $+1$ eigenspace of the twisted $X$-stabilizers
$\tilde A^{\bl}_{v,r}$ and $\tilde A^{\gr}_{v,t}$ of the two other copies, and analogously for
$P^{\bl}$, $P^{\gr}$. These operators commute with $\tilde{\SS}$ on the code
space: the $Z$-type generators are handled as in the untwisted code, $\cX^{\rd}$
commutes with $\tilde{A}^{\rd}$ by construction, and for the other two colors
\begin{align}\label{eq:projector_absorption_sheaf}
&P^{\rd}A^{\bl}_{v,r}=P^{\rd}A^{\gr}_{v,t}=P^{\rd}
=A^{\bl}_{v,r}P^{\rd}=A^{\gr}_{v,t}P^{\rd} \cr
&\qquad \Longrightarrow \quad
\cX^{\rd}A^{\bl}_{v,r}=\cX^{\rd}=A^{\bl}_{v,r}\cX^{\rd} ,
\end{align}
and likewise for $A^{\gr}_{v,t}$. It is Eq.~\eqref{eq:magnetic_projected_sheaf}
that underlies the non-Abelian fusion rules of the twisted sheaf code, in direct
analogy with the manifold case \cite{Hsin2024_non-Abelian}.

\subsection{Non-Abelian topological defect network}

Here, we discuss the topological defect network picture similar to the Abelian (untwisted) case in Sec.~\ref{sec:sheaf_defect_network}. 

The twisted $\ZZ_2^3$ gauge theory is known to be equivalent to a non-Abelian $D_4$ topological order \cite{deWildPropitius:1995cf}.  Therefore, each face on our 2D sheaf complex can be viewed as a patch of $D_4$ topological order, as illustrated in Fig.~\ref{fig:gapped_boundary_picture_D4}(b).  Several $D_4$ patches can meet at a gapped interface with a proper condensation condition specified by the underlying  sheaf structure (local system of coefficients) as illustrated before in Fig.~\ref{fig:gapped_boundary_picture}.      

Within each $D_4$ patch (face), if one subdivide it into enough number of faces as shown in Fig.~\ref{fig:gapped_boundary_picture}(b),  the anyons generated within each patch satisfy the fusion rules of the $D_4$ topological order as follows:
\begin{align}
m^{\rd} \times e^{\gr} =& m^{\rd}, \quad  m^{\rd} \times e^{\bl} = m^{\rd}, \cr
m^{\rd} \times m^{\rd} =& 1 + e^{\gr} + e^{\bl} + e^{\bl}e^{\gr},
\end{align}
and similarly for other color permutations.
The last one above is a non-Abelian fusion rule where the fusion outcome is not unique,  due to the fact that the magnetic anyon $m^{\rd}$ is a non-Abelian anyon.
Similarly, the anyon braiding statics within each patch satisfies the $D_4$ braiding statistics.

We note that it is important to only subdivide each patch into $O(1)$ faces to not change the code parameter scaling of the family of codes, which can eventually determine the complexity/entanglement properties of the corresponding topological defect networks, such as the NLTS and NLTM properties.

\begin{figure}[t]
    \centering   \includegraphics[width=1\linewidth]{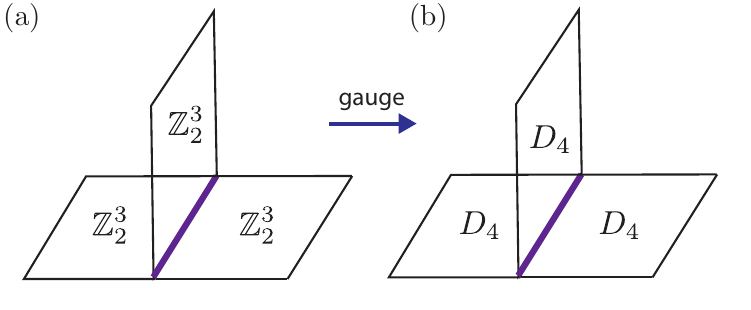}
    \caption{The topological defect network of the Abelian  $\ZZ_2^3$ gauge theory (a) corresponding to the untwisted qLDPC code is gauged to a  topological defet network corresponding to the twisted qLDPC codes (b). In (b), each face can be viewed as a patch of $D_4$ topological order glued along a gapped interface with proper condensation condition specified by the underlying sheaf structure.}
    \label{fig:gapped_boundary_picture_D4}
\end{figure}

%======================================================================
\section{$0$-form subcomplex symmetry in untwisted and twisted qLDPC codes}
\label{sec:0form_subcomplex}
%======================================================================

We now turn to the generalized symmetries of the sheaf codes constructed above.
The central object is a family of transversal CZ operators, one for each sheaf
$0$-cocycle. We first show that in the \emph{untwisted} code $\C$ such an operator
is an addressable transversal \emph{logical CZ gate}, and then show that in the
\emph{twisted} code $\tilde{\C}$ the very same operator is a charge parity
operator, i.e.\ a product of twisted stabilizer generators. Twisting therefore
absorbs the logical gate into the stabilizer group, which is the operator-level
counterpart of gauging the symmetry.

Throughout this section we illustrate everything with the transversal CZ gate
between the \red{red} and \blue{blue} copies, $\widetilde{\text{CZ}}^{\rd,\bl}$,
which by Eq.~\eqref{eq:stabilizer_summary} is the dressing carried by the
\green{green} $X$-stabilizers and therefore measures the \gr-type charge. The
operators of the other two colour pairings are obtained by permuting colours
throughout. This is the case we will use in the explicit construction of the
addressable gauging measurement of Sec.~\ref{sec:logical_operation}.

On a manifold the analogous operators are higher-form ($k$-form) symmetries
supported on submanifolds of codimension $k$; here, we consider instead the
$0$-form symmetries, i.e.\ of codimension $0$, but supported on a \emph{proper}
subcomplex rather than on all of $\L$. This is the same phenomenon already met for
the SPT in Sec.~\ref{sec:subcomplex_symmetry}, now realized at the level of the
codes. Following Refs.~\cite{ZhuKobayashiHsin2026, zhu2025topological} we
call such operators $0$-form \emph{subcomplex} symmetries; they have no
counterpart in conventional gauge theory on a manifold, where the top-dimensional
cycle is unique.

\subsection{Transversal CZ operators and their support}
\label{sec:transversal_CZ}

\begin{figure}[t]
    \centering   \includegraphics[width=1\linewidth]{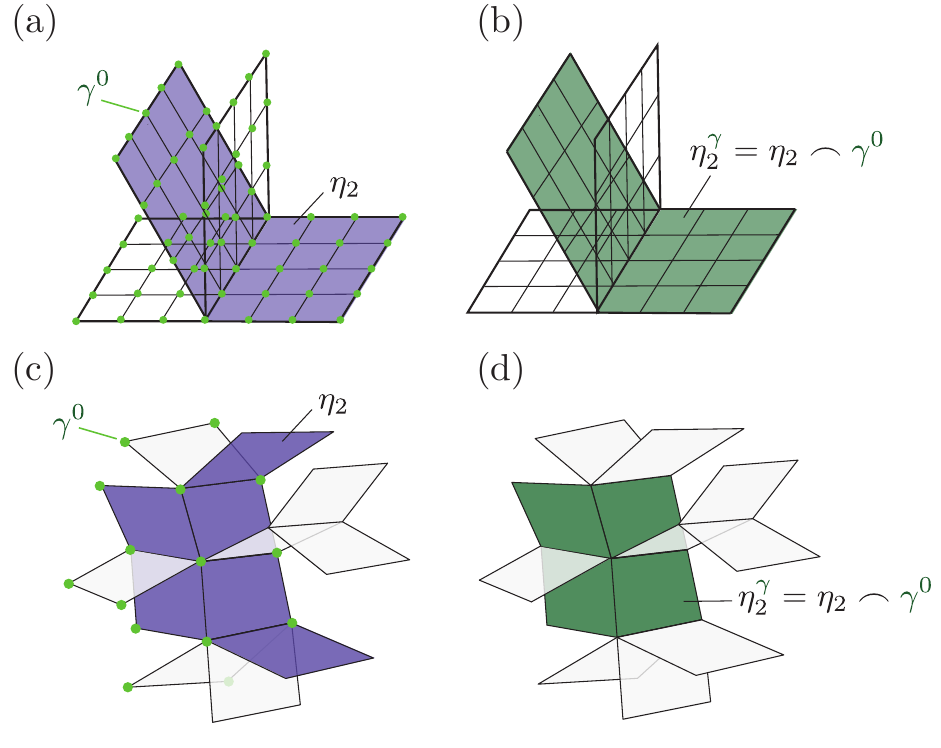}
    \caption{Capping a $2$-cycle with a
$0$-cocycle, which is the sheaf replacement for ``supported on the Poincar\'e dual
cycle''. (a) The $2$-cycle $\eta_2$ (purple) carrying the twist, together with a
sheaf $0$-cocycle $\gamma^0 \in Z^0(\L,\Fc^{\gr})$ whose values are marked at the
vertices (green). (b) The capped $2$-chain $\eta_2^{\gamma}=\eta_2\frown\gamma^0$
(green), a codimension-$0$ subcomplex of $\L$ on which the corresponding
transversal CZ operator $\widetilde{\text{CZ}}^{\rd,\bl}(\gamma^0)$ is supported.
(c), (d) The same construction on a larger and more irregular complex. Since
$\partial(\eta_2\frown\gamma^0)=(\partial\eta_2)\frown\gamma^0+
\eta_2\frown(d\gamma^0)=0$, the capped chain is again a cycle, and distinct
classes $[\gamma^0]\in H^0(\L,\Fc^{\gr})$ give distinct subcomplexes --- which is
the source of addressability.}
\label{fig:subcomplex_illustration_nested}
\end{figure}

Let $\green{\gamma^0} \in Z^0(\L,\Fc^{\gr}) = H^0(\L,\Fc^{\gr})$ be a sheaf $0$-cocycle of
the \gr\ copy, expanded in elementary cochains as
$\green{\gamma^0}=\sum_{v,t}\gamma_{v,t}\,\bar{v}_t$. Define the associated \emph{transversal
CZ operator} acting on the \rd\ and \bl\ copies,
\begin{align}\label{eq:transversal_CZ_def}
\widetilde{\text{CZ}}^{\rd,\bl}(\green{\gamma^0})
\equiv& \;(-1)^{\int_{\eta_2}\red{\hat{a}^1} \cup \blue{\hat{b}^1} \cup \green{\gamma^0}}
\cr
=& \prod_{\substack{(e,s),(e',r):\\
\int_{\eta_2}\bar{e}_s\cup\bar{e}'_q\cup\green{\gamma^0}\neq 0}}
\text{CZ}^{\rd,\bl}_{(e,s),(e',r)} ,
\end{align}
where the second line follows by expanding the gauge fields in elementary
cochains, $\red{\hat{a}^1}=\sum_{e,s}\hat{N}^{\rd}_{e,s}\bar{e}_s$ and
$\blue{\hat{b}^1}=\sum_{e,r}\hat{N}^{\bl}_{e,r}\bar{e}_r$, and using
$(-1)^{\hat{N}^{\rd}_{e,s}\hat{N}^{\bl}_{e',r}} \equiv
\text{CZ}^{\rd,\bl}_{(e,s),(e',r)}$. Note that the label $\green{\gamma^0}$ occupies the
last cup slot, matching the position of the elementary cochain $\bar{v}_t$ in the
\gr-type stabilizer $\tilde{A}^{\gr}_{v,t}$ of Eq.~\eqref{eq:stabilizer_summary},
and that it is exactly a label of the SPT symmetry generators of
Eq.~\eqref{eq:globalsym}. The operator is manifestly \emph{transversal}: it is a
depth-one circuit of two-qubit CZ gates between two distinct code blocks.

\medskip
\noindent\textit{Operator support.}
On a manifold one identifies the support of such an operator with the Poincar\'e
dual of the label cocycle. On a sheaf complex Poincar\'e duality is not available
in general, but the \emph{cap product} of Definition~\ref{def:cap-simplicial} does
exactly the required job and is intrinsic. Since $\green{\gamma^0}$ sits in the last cup
slot, we use the mirror version of Definition~\ref{def:cap-simplicial} in which
the cochain caps the \emph{back} face and the surviving chain sits on the front
face; writing $\eta_2(f)=\sum_\ell a_\ell\otimes b_\ell\otimes c_\ell$ on a
$2$-cell $f=[v_0v_1v_2]$, it reads
\be\label{eq:back_cap}
\big(\eta_2 \frown \green{\gamma^0}\big)(f) = \sum_{\ell}
\Big\langle \Fc^{\gr}_{[v_2]\shortrightarrow f}\big(\green{\gamma^0}([v_2])\big),\,
c_\ell \Big\rangle\; a_\ell \otimes b_\ell .
\ee
By construction this satisfies
\be
\big\langle x^1 \cup y^1,\, \eta_2 \frown \green{\gamma^0} \big\rangle = \big\langle x^1 \cup y^1 \cup \green{\gamma^0},\, \eta_2  \big\rangle
\equiv \int_{\eta_2} x^1 \cup y^1 \cup \green{\gamma^0}
\ee
for all $1$-cochains $x^1$, $y^1$, so
\begin{align}\label{eq:cap_support}
\widetilde{\text{CZ}}^{\rd,\bl}(\green{\gamma^0})
&= (-1)^{\int_{\eta_2^{\gamma}} \red{\hat{a}^1} \cup \blue{\hat{b}^1}} , \cr
\eta_2^{\gamma} &\equiv \eta_2 \frown \green{\gamma^0}
\;\in\; C_2\big(\L, \Fc^{\rd}\!\otimes\!\Fc^{\bl}\big) .
\end{align}
The operator is therefore a pairing between the operator-valued $2$-cochain
$\red{\hat{a}^1}\cup\blue{\hat{b}^1}$ and the $2$-chain $\eta_2^{\gamma}$ obtained
by capping the twist cycle $\eta_2$ with the symmetry label $\green{\gamma^0}$. This is the
sheaf replacement for ``supported on the Poincar\'e dual cycle''.

Two properties follow immediately. First, $\eta_2^{\gamma}$ is a $2$-\emph{cycle}: by
the Leibniz rule for the cap product,
Proposition~\ref{prop:cap-properties}(a),
\be\label{eq:capped_cycle_closed}
\partial\big(\eta_2 \frown \green{\gamma^0}\big)
= \big(\partial \eta_2\big) \frown \green{\gamma^0} + \eta_2 \frown \big(d\green{\gamma^0}\big) = 0 ,
\ee
using $\partial\eta_2=0$ and $d\green{\gamma^0}=0$. The cycle condition on the twist and the
cocycle condition on the symmetry label are thus exactly the two inputs needed for
the support of the operator to be closed. Second, since
$\dim\eta_2^{\gamma} = 2 = \mathsf{d}$, the closure of its support
\be\label{eq:subcomplex_support}
\L_{\eta_2^{\gamma}} \equiv \overline{\mathrm{supp}\,\eta_2^{\gamma}} \subseteq \L
\ee
is a codimension-$0$ subcomplex of $\L$, and in general a \emph{proper} one:
$\eta_2^{\gamma}$ vanishes on every $2$-cell where either $\eta_2$ vanishes or
$\green{\gamma^0}$ does. 

Therefore we have now justified  the name of the \textit{$0$-form subcomplex symmetry}: the operator is $0$-form because its support has
codimension $0$, and subcomplex because that support is not all of $\L$.

\subsection{Transversal logical CZ in the untwisted code}
\label{sec:logical_CZ_untwisted}

We first work in the untwisted code $\C$ of
Eq.~\eqref{eq:stabilizer_untwisted_sheaf}, i.e.\ three decoupled copies of the
sheaf code of Sec.~\ref{sec:untwisted_code}.

\medskip
\noindent\textit{The operator preserves the code space.}
Being diagonal, $\widetilde{\text{CZ}}^{\rd,\bl}(\green{\gamma^0})$ commutes exactly with
every $Z$-stabilizer $B^{i}_{f}$ and with the $X$-stabilizers of the \gr\ copy,
which act on different qubits. For the remaining case, $A^{\rd}_{v,s}$ implements
$\red{\hat{a}^1}\to\red{\hat{a}^1}+d^0\bar{v}_s$ by
Eq.~\eqref{eq:A_conjugation}, so
\begin{align}\label{eq:CZ_commutator}
\big[\widetilde{\text{CZ}}^{\rd,\bl}(\green{\gamma^0}),\, A^{\rd}_{v,s}\big]
&= (-1)^{\int_{\eta_2} d\bar{v}_s \cup \blue{\hat{b}^1} \cup \green{\gamma^0}} \cr
&= (-1)^{\int_{\eta_2} d\left(\bar{v}_s \cup \blue{\hat{b}^1} \cup
\green{\gamma^0}\right)} \cr
&\quad \times (-1)^{\int_{\eta_2} \bar{v}_s \cup \blue{d\hat{b}^1}
\cup \green{\gamma^0}} \cr
&= (-1)^{\int_{\eta_2} \bar{v}_s \cup \blue{d\hat{b}^1} \cup \green{\gamma^0}} ,
\end{align}
where the second equality is the Leibniz rule of
Proposition~\ref{prop:cup-properties} together with $d\green{\gamma^0}=0$, and the third
uses $\langle d\omega,\eta_2\rangle=\langle\omega,\partial\eta_2\rangle=0$. As in
Eq.~\eqref{eq:Theta_def}, the surviving factor is a product of $B^{\bl}_{f}$
operators and therefore equals $\I$ on the zero-flux subspace. Hence
\be\label{eq:CZ_is_symmetry}
P_B\,\big[\widetilde{\text{CZ}}^{\rd,\bl}(\green{\gamma^0}),\, g\big]\,P_B = \I ,
\qquad \forall\, g \in \SS ,
\ee
so the operator maps $\C$ to itself. Note that $d\green{\gamma^0}=0$ is indispensable: for a
general $0$-cochain the Leibniz rule leaves an extra term
$\int_{\eta_2}\bar{v}_s\cup\blue{\hat{b}^1}\cup d\green{\gamma^0}$ which is not a product of
$Z$-stabilizers, and the code space is not preserved. Note also that
$\widetilde{\text{CZ}}^{\rd,\bl}(\green{\gamma^0})$ is \emph{not} an element of the
untwisted stabilizer group $\SS$, which contains only Pauli operators; it is a
genuine logical gate, as we now verify.

\medskip
\noindent\textit{The logical action is a CZ.}
Recall from Sec.~\ref{sec:logical_operators_sheaf} that the logical-$X$ operators
of a single copy are labelled by sheaf $1$-cocycles,
$\lo{X}^{\rd}(\red{\alpha^1})
=\prod_{(e,s)}(X^{\rd}_{e,s})^{[\red{\alpha^1}(e)]_s}$ with
$\red{\alpha^1} \in Z^1(\L,\Fc^{\rd})$, and the logical-$Z$ operators by sheaf
$1$-cycles. Since $\lo{X}^{\rd}(\red{\alpha^1})$ implements
$\red{\hat{a}^1}\to\red{\hat{a}^1}+\red{\alpha^1}$, conjugation gives
\begin{align}\label{eq:CZ_conjugation_X}
&\widetilde{\text{CZ}}^{\rd,\bl}(\green{\gamma^0})\ \lo{X}^{\rd}(\red{\alpha^1})\
\widetilde{\text{CZ}}^{\rd,\bl}(\green{\gamma^0}) \cr
&\qquad = \lo{X}^{\rd}(\red{\alpha^1})\cdot
(-1)^{\int_{\eta_2}\red{\alpha^1} \cup \blue{\hat{b}^1} \cup \green{\gamma^0}} ,
\end{align}
while the logical-$Z$ operators are left invariant because the operator is
diagonal. The factor generated on the right-hand side is diagonal and linear in
$\blue{\hat{b}^1}$, hence it is a $Z$-type operator on the \bl\ copy,
\be\label{eq:CZ_generated_Z}
(-1)^{\int_{\eta_2}\red{\alpha^1} \cup \blue{\hat{b}^1} \cup \green{\gamma^0}}
= \widetilde{Z}^{\bl}_{\zeta_1} ,
\qquad
\big[\zeta_1(e)\big]_r \equiv
\int_{\eta_2}\red{\alpha^1} \cup \bar{e}_r \cup \green{\gamma^0} ,
\ee
with $\zeta_1 \in C_1(\L,\Fc^{\bl})$. Moreover $\zeta_1$ is a $1$-\emph{cycle}:
for any $0$-cochain $x^0 \in C^0(\L,\Fc^{\bl})$, the Leibniz rule together with
$d\red{\alpha^1}=0$ and $d\green{\gamma^0}=0$ gives
$d(\red{\alpha^1}\cup x^0\cup\green{\gamma^0})=\red{\alpha^1}\cup d^0x^0\cup\green{\gamma^0}$, whence
\begin{align}\label{eq:zeta_is_cycle}
\big\langle x^0, \partial_1 \zeta_1 \big\rangle
&= \big\langle d^0 x^0, \zeta_1 \big\rangle
= \int_{\eta_2}\red{\alpha^1} \cup d^0 x^0 \cup \green{\gamma^0} \cr
&= \big\langle \cdot\,, \partial\eta_2 \big\rangle = 0 ,
\end{align}
for all $x^0$, so $\partial_1\zeta_1=0$. Hence
$\widetilde{Z}^{\bl}_{\zeta_1}$ is a logical-$Z$ operator of the \bl\ copy, and
Eq.~\eqref{eq:CZ_conjugation_X} is precisely the defining action of a logical CZ
gate: $\lo{X} \mapsto \lo{X}\,\lo{Z}$ on the other block, with $\lo{Z}$'s fixed.

Labelling the logical qubits of the \rd\ and \bl\ copies by cocycle bases
$\{\red{\alpha^1}\}$ and $\{\blue{\beta^1}\}$ as in
Sec.~\ref{sec:logical_operators_sheaf}, the coefficient matrix of the induced
logical gate is the invariant form of Definition~\ref{def:invariant-form}
evaluated on the capped cycle,
\be\label{eq:logical_CZ_coefficient}
\Lambda^{\red{\alpha}\blue{\beta}}(\green{\gamma^0})
= T_{\eta_2^{\gamma}}\big(\red{\alpha^1},\blue{\beta^1}\big)
= \int_{\eta_2}\red{\alpha^1} \cup \blue{\beta^1} \cup \green{\gamma^0} \in \FF_2 ,
\ee
so that
\be\label{eq:logical_CZ_action}
\widetilde{\text{CZ}}^{\rd,\bl}(\green{\gamma^0}) \;\widehat{=}\;
\prod_{\red{\alpha},\blue{\beta}}
\big[\lo{\text{CZ}}^{\rd,\bl}_{\red{\alpha},\blue{\beta}}
\big]^{\Lambda^{\red{\alpha}\blue{\beta}}(\green{\gamma^0})}
\ee
on the code space. This is the promised statement: a depth-one physical CZ circuit
between the \rd\ and \bl\ blocks implements a product of logical CZ gates, with the
pattern of logical gates determined by the triple cup product. By
Proposition~\ref{prop:cup-properties}(c) the coefficient
\eqref{eq:logical_CZ_coefficient} depends only on the cohomology classes
$[\red{\alpha^1}]$, $[\blue{\beta^1}]$, $[\green{\gamma^0}]$, so the logical action is well
defined.

\medskip
\noindent\textit{Addressability.}
There is one such operator for each class $[\green{\gamma^0}]\in H^0(\L,\Fc^{\gr})$, and by
Eq.~\eqref{eq:subcomplex_support} each is supported on its own codimension-$0$
subcomplex $\L_{\eta_2^{\gamma}}$. Distinct labels therefore give distinct patterns
$\Lambda^{\red{\alpha}\blue{\beta}}(\green{\gamma^0})$ of logical CZ's, which can be applied
independently. By Proposition~\ref{Prop:sheaf_cohomology}, $H^0(\L,\Fc^{\gr})$ is
the $0$-cocycle classical sheaf code, of dimension $k_{\gr}$ growing linearly with
system size, so the number of independently addressable transversal logical CZ
gates grows linearly as well. This is the practical payoff of the subcomplex
formulation, and it is what replaces the higher-form mechanism available on
manifolds.

\subsection{Charge parity operators in the twisted code}
\label{sec:charge_parity_operator}

We now pass to the twisted code $\tilde{\C}$ and show that the same transversal
operator reappears there in an entirely different role.

For $\green{\gamma^0} \in Z^0(\L,\Fc^{\gr})$ define the \gr-type \emph{charge parity
operator} as the corresponding product of twisted $X$-stabilizers,
\be\label{eq:charge_parity_sheaf}
\mathsf{C}^{\gr}(\green{\gamma^0}) \equiv
\prod_{v,t}\big(\tilde{A}^{\gr}_{v,t}\big)^{\gamma_{v,t}} \equiv  \prod_{(v,t)\in \green{\gamma^0}} \tilde{A}^{\gr}_{v,t} ,
\ee
and analogously $\mathsf{C}^{\rd}(\mu^0)$, $\mathsf{C}^{\bl}(\nu^0)$ for the other
two colours. The key point is that all Pauli-$X$ content cancels. In the untwisted
code the $X$-stabilizers obey the redundancy relations
\be\label{eq:X_redundancy_sheaf}
\prod_{v,t}\big(A^{\gr}_{v,t}\big)^{\gamma_{v,t}} \equiv \prod_{(v,t)\in \green{\gamma^0}} {A}^{\gr}_{v,t} = \I ,
\qquad \forall\, \green{\gamma^0} \in Z^0(\L,\Fc^{\gr}) ,
\ee
which follow from Eq.~\eqref{eq:A_conjugation}: the product implements the shift
$\green{\hat{c}^1}\to\green{\hat{c}^1}+d^0\green{\gamma^0}$, and $d^0\green{\gamma^0}=0$. Since the
$A$'s and the CZ dressings act on different copies and therefore commute, the
product in Eq.~\eqref{eq:charge_parity_sheaf} reduces to the dressing alone, which
is exactly the operator of Eq.~\eqref{eq:transversal_CZ_def}:
\be\label{eq:charge_is_CZ_sheaf}
 \mathsf{C}^{\gr}(\green{\gamma^0})
= (-1)^{\int_{\eta_2}\red{\hat{a}^1} \cup \blue{\hat{b}^1} \cup \green{\gamma^0}}
= \widetilde{\text{CZ}}^{\rd,\bl}(\green{\gamma^0}) \ 
\ee
and by permuting colours
\be\label{eq:charge_is_CZ_sheaf2}
\mathsf{C}^{\rd}(\mu^0) = \widetilde{\text{CZ}}^{\bl,\gr}(\mu^0) ,
\qquad
\mathsf{C}^{\bl}(\nu^0) = \widetilde{\text{CZ}}^{\rd,\gr}(\nu^0) .
\ee
Thus the charge parity operator of one colour \emph{is} the transversal CZ
operator acting on the other two colours --- for the \gr-type charge, the very
operator that implements a logical CZ gate between the \rd\ and \bl\ copies of the
untwisted code, Eq.~\eqref{eq:logical_CZ_action}.

The consequence is that its status changes completely. In $\tilde{\C}$ the
operator belongs to the stabilizer group by construction, being a product of
generators of $\tilde{\SS}$, so by Eq.~\eqref{eq:constraint_operator_sheaf2} it
acts as the identity,
\be\label{eq:charge_parity_plus_one}
\mathsf{C}^{\gr}(\green{\gamma^0}) = \mathsf{C}^{\rd}(\mu^0) = \mathsf{C}^{\bl}(\nu^0) = \I
\qquad \text{on } \tilde{\C} .
\ee
The same transversal circuit that acts as a nontrivial logical CZ on $\C$ acts
trivially on $\tilde{\C}$. This is the operator-level statement of gauging: the
CZ-dressing of the $X$-stabilizers in Eq.~\eqref{eq:stabilizer_summary} absorbs
the symmetry generated by $\widetilde{\text{CZ}}^{\rd,\bl}$ into the stabilizer
group, so that what was a logical gate becomes a constraint. It is also the reason
the twisted code is non-Abelian: the price of promoting a Clifford logical gate to
a stabilizer is that the resulting stabilizer group is no longer Pauli, and its
generators fail to commute off the zero-flux subspace
(Lemma~\ref{lemma:stabilizer_commutation}).

\begin{figure}[t]
    \centering   \includegraphics[width=1\linewidth]{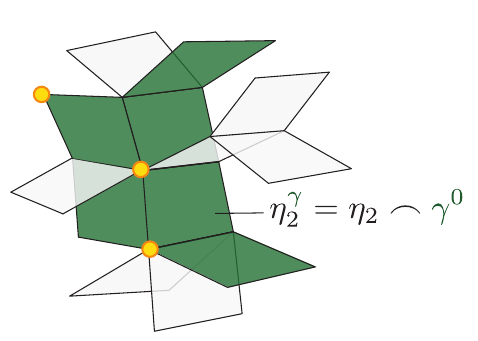}
    \caption{What the charge parity operator
measures. The capped subcomplex $\eta_2^{\gamma}=\eta_2\frown\gamma^0$ (green)
encloses a set of \gr-type gauge charges (orange dots), and the operator
$\mathsf{C}^{\gr}(\gamma^0)=\prod_{v,t}(\tilde{A}^{\gr}_{v,t})^{\gamma_{v,t}}$
returns their parity. In the twisted code $\tilde{\C}$ this operator is a product
of stabilizer generators and hence acts as the identity, whereas in the untwisted
code $\C$ the very same transversal circuit
$\widetilde{\text{CZ}}^{\rd,\bl}(\gamma^0)$ is a non-trivial addressable logical
CZ gate: gauging has absorbed the logical gate into the stabilizer group. On a
manifold the analogous operator would be supported on the Poincar\'e dual of the
symmetry label; here the cap product does the same job intrinsically.}
\label{fig:subcomplex_illustration_charge}
\end{figure}

\subsection{Charge interpretation and contrast with manifolds}
\label{sec:charge_interpretation}

The name ``charge parity'' comes from the following reading of
Eq.~\eqref{eq:charge_parity_sheaf}. An electric charge of colour \gr\ is a
violated $X$-stabilizer, $A^{\gr}_{v,t}=-1$, located at the vertex $v$; it is
created at the endpoints of an open sheaf $Z$-string. Indeed for a $1$-chain
$\zeta_1\in C_1(\L,\Fc^{\gr})$ the Wilson operator $\widetilde{Z}^{\gr}_{\zeta_1}$
of Eq.~\eqref{eq:logical_Z_sheaf} anticommutes with $A^{\gr}_{v,t}$ precisely when
$\langle \bar{v}_t, \partial_1\zeta_1\rangle=1$, so the charge configuration is
the $0$-chain $\partial_1\zeta_1$. Consequently
\be\label{eq:charge_measured}
\mathsf{C}^{\gr}(\green{\gamma^0})\ \widetilde{Z}^{\gr}_{\zeta_1}\ket{\psi}
= (-1)^{\langle \green{\gamma^0},\, \partial_1\zeta_1\rangle}\
\widetilde{Z}^{\gr}_{\zeta_1}\ket{\psi} , \qquad \ket{\psi}\in\tilde{\C},
\ee
i.e.\ $\mathsf{C}^{\gr}(\green{\gamma^0})$ measures the total colour-\gr\ charge parity
weighted by $\green{\gamma^0}$, or equivalently the total charge contained in the subcomplex
$\L_{\eta_2^{\gamma}}$. By the adjunction
$\langle \green{\gamma^0},\partial_1\zeta_1\rangle = \langle d^0\green{\gamma^0},\zeta_1\rangle = 0$
this parity is unaffected by any Pauli-$Z$ error, which is the operator statement
that Eq.~\eqref{eq:X_redundancy_sheaf} is a redundancy among the $X$-checks: the
charge parity operators are `global checks' of the code, and their violation signals a
state in a different super-selection sector outside the code space rather than a correctable error.

Finally we record the contrast with the manifold case, which is the reason these
symmetries are new. On a triangulated closed $\mathsf{d}$-manifold the top
homology and bottom cohomology are one-dimensional, $b_{\mathsf{d}}=b_0=1$, so
there is a \emph{unique} codimension-$0$ cycle, namely the manifold itself, and
correspondingly a single global transversal CZ operator, measuring the total charge
of the whole system and implementing a single fixed pattern of logical CZ's. A
general sheaf $2$-complex has no such uniqueness: by
Proposition~\ref{Prop:sheaf_cohomology}, $H^0(\L,\Fc^{\gr})$ is the $0$-cocycle
classical sheaf code, whose dimension $k_{\gr}$ grows linearly with the system
size. There are therefore extensively many independent codimension-$0$ cycles
$\eta_2^{\gamma}$, one per basis class of $[\green{\gamma^0}]$, and extensively many
independent transversal CZ operators. The point-like charges of colour \gr\ can be
supported on any basis class of $H_0(\L,\Fc^{\gr})$, and the parity within each
class is measured separately. It is precisely this multiplicity, absent on a
manifold, that makes the transversal logical CZ gates individually addressable,
and it is the code analogue of the extensive symmetry group of the sheaf cluster
state found in Eq.~\eqref{eq:symmetry_group}.

\subsection{Alternative route: gauging the subcomplex symmetry from the untwisted qLDPC codes}
\label{sec:gauging_untwisted}
%======================================================================

In Sec.~\ref{sec:gauging} we obtained the twisted code by gauging the $0$-form
subcomplex symmetry of the sheaf cluster state, i.e.\ by starting from an SPT and
gauging the symmetry that protects it. There is a second and logically independent
route to the same code, which starts one step further along: one may begin
directly with the \emph{untwisted} qLDPC codes of Sec.~\ref{sec:untwisted_code}
--- two decoupled $\ZZ_2$ sheaf gauge theories --- and gauge the transversal CZ
symmetry acting between them. We develop this route here in the path-integral
language. It is worth doing carefully for three reasons: it makes the appearance
of the twisted action \eqref{eq:non-Abelian_action} inevitable rather than
postulated, it exhibits the code-space constraints of
Sec.~\ref{sec:code_space_constraint} as an equation of motion, and it is the form
in which the gauging is actually carried out as a fault-tolerant measurement
protocol later on.

Throughout, $\tilde{\L}$ is the $(2+1)$D spacetime sheaf complex and
$\tilde{\eta}_3 \in Z_3(\tilde{\L},\Fc^{\rd}\otimes\Fc^{\bl}\otimes\Fc^{\gr})$ the
twist cycle of Sec.~\ref{sec:Feynman_path_integral}. When a time direction needs
to be made explicit we take the product form $\tilde{\L}=\L\otimes S^1_t$, with
$I_t$ the chain complex of a $1$D circle carrying the constant sheaf, so that
$\tilde{\eta}_3=\eta_2\otimes S^1_t$ with $\eta_2$ the spatial $2$-cycle of
Sec.~\ref{sec:sheaf_SPT}; this decomposition is used systematically in
Sec.~\ref{sec:condition_cohomology}.

\medskip
\noindent\textit{The ungauged theory.}
The starting point is the untwisted theory of the \rd\ and \bl\ copies, whose
action consists of the first two BF terms of Eq.~\eqref{eq:non-Abelian_action}
with the twisted term absent. Summing over the magnetic Lagrange multipliers
$\red{\check{a}^1}$, $\blue{\check{b}^1}$ imposes flatness, so that
\be\label{eq:Z_untwisted}
\cZ_0[\tilde{\L}] = \sum_{\substack{\red{a^1}\in H^1(\tilde{\L},\Fc^{\rd}) \\
\blue{b^1}\in H^1(\tilde{\L},\Fc^{\bl})}} 1 ,
\ee
an unweighted sum over flat sheaf gauge fields. This is the spacetime description
of two decoupled copies of the untwisted sheaf code of
Sec.~\ref{sec:untwisted_code}, $\C = \C^{\rd}\otimes\C^{\bl}$, whose stabilizer
group is the \rd\ and \bl\ part of Eq.~\eqref{eq:stabilizer_untwisted_sheaf}.

\medskip
\noindent\textit{The symmetry and its defects.}
Fix a third system of local coefficients $\Fc^{\gr}$ on the same complex. At this
stage it is only a bookkeeping device --- there is no \gr\ gauge field in
Eq.~\eqref{eq:Z_untwisted} --- and it becomes the coefficient system of a genuine
third gauge field only after gauging. Let
$\{\green{\gamma^{0}_{k}}\}$ be a basis of $Z^0(\L,\Fc^{\gr})=H^0(\L,\Fc^{\gr})$ and
let
\be\label{eq:spacetime_label}
\green{\tilde{\gamma}^{1}_{k}} = \green{\gamma^{0}_{k}} \otimes {\ts}^1
\;\in\; Z^1\big(\tilde{\L},\Fc^{\gr}\big) ,
\ee
with ${\ts}^1$ the generator of $H^1(S^1_t)$, be the corresponding spacetime
labels. The associated symmetry defect operator of the theory
\eqref{eq:Z_untwisted} is
\begin{align}\label{eq:symmetry_defect}
U\big(\green{\tilde{\gamma}^{1}_{k}}\big)
=&\ (-1)^{\int_{\tilde{\eta}_3}\red{a^1}\cup\blue{b^1}\cup
\green{\tilde{\gamma}^{1}_{k}}} \cr
=&\ (-1)^{\int_{\tilde{\eta}_3^{\,k}}\red{a^1}\cup\blue{b^1}} ,
\qquad
\tilde{\eta}_3^{\,k} \equiv \tilde{\eta}_3 \frown
\green{\tilde{\gamma}^{1}_{k}},   \cr
\end{align}
the second line using the cap product of
Definition~\ref{def:cap-simplicial} to convert the label into a support. Since
$\dim\tilde{\eta}_3^{\,k}=3-1=2$ while $D=3$, the defect is supported on a
\emph{codimension-one} locus of spacetime, as a $0$-form symmetry generator must
be. Restricting it to a single time slice replaces
$\green{\tilde{\gamma}^{1}_{k}}$ by its spatial part $\green{\gamma^{0}_{k}}$ and
returns the equal-time operator
\be\label{eq:defect_equal_time}
U\big(\green{\tilde{\gamma}^{1}_{k}}\big)\Big|_{t\ \text{fixed}}
= (-1)^{\int_{\eta_2}\red{\hat{a}^1}\cup\blue{\hat{b}^1}\cup
\green{\gamma^{0}_{k}}}
= \widetilde{\text{CZ}}^{\rd,\bl}\big(\green{\gamma^{0}_{k}}\big) ,
\ee
the transversal CZ operator already met in Eq.~\eqref{eq:CZ_identity_sheaf}: a
codimension-zero operator in space, supported on a proper subcomplex of $\L$. That
$\widetilde{\text{CZ}}^{\rd,\bl}(\green{\gamma^{0}})$ commutes with the untwisted
stabilizer group, and hence really generates a symmetry of $\C$, is verified
directly at the operator level in Sec.~\ref{sec:logical_CZ_untwisted}; in the path
integral it follows from the topological invariance of the defect, i.e.\ from
Lemma~\ref{lemma:invariant} together with $\partial\tilde{\eta}_3=0$.

Because
$U(\green{\tilde{\gamma}^{1}})\,U(\green{\tilde{\gamma}'^{1}})
=U(\green{\tilde{\gamma}^{1}}+\green{\tilde{\gamma}'^{1}})$, the symmetry group being
gauged is
\be\label{eq:symmetry_group_gauged}
\Big\{U\big(\green{\tilde{\gamma}^{1}}\big)\Big\}
\;\cong\; H^0\big(\L,\Fc^{\gr}\big) \;\cong\; \ZZ_2^{\,k_{\gr}} ,
\ee
which by Proposition~\ref{Prop:sheaf_cohomology} is the $0$-cocycle classical
sheaf code, of dimension $k_{\gr}$ growing linearly with the system size. This is
the same extensive symmetry group as Eq.~\eqref{eq:symmetry_group}, now realized
on the codes rather than on the cluster state, and it is to be contrasted with the
single $\ZZ_2$ available on a closed manifold, where $b_0=1$.

\medskip
\noindent\textit{Coupling to a background gauge field.}
Gauging means summing over insertions of the defect network
\eqref{eq:symmetry_defect}. The first step is to couple the theory to a
\emph{background} gauge field for the symmetry group
\eqref{eq:symmetry_group_gauged}. Such a background is a flat sheaf $1$-cochain
$\green{C^1}\in Z^1(\tilde{\L},\Fc^{\gr})$, and inserting the defect network it
labels multiplies the weight of each configuration by
\be\label{eq:background_coupling}
(-1)^{\int_{\tilde{\eta}_3}\red{a^1}\cup\blue{b^1}\cup\green{C^1}} .
\ee
Two features of Eq.~\eqref{eq:background_coupling} deserve emphasis, since neither
is present in the manifold case. First, the coefficient system is forced: the
symmetry labels live in $Z^0(\L,\Fc^{\gr})$, so the background must be
$\Fc^{\gr}$-valued in order that the triple cup product take values in
$\Fc^{\rd}\otimes\Fc^{\bl}\otimes\Fc^{\gr}$, which is the sheaf against which
$\tilde{\eta}_3$ is paired [Eq.~\eqref{eq:trilinear_map}]. The third copy of the
code therefore enters carrying its own local system, and not as a scalar $\FF_2$
field. Second, the coupling is non-trivial only if the tensor sheaf
$\Fc^{\rd}\otimes\Fc^{\bl}\otimes\Fc^{\gr}$ admits a non-zero $1$-cycle on the
underlying classical codes; this is the \emph{local cycle condition}, analyzed in
Sec.~\ref{sec:condition_cohomology}, and it has no counterpart for constant
coefficients. Gauge invariance of Eq.~\eqref{eq:background_coupling} under
$\red{a^1}\to\red{a^1}+d\lambda^0$ and its analogues is
Lemma~\ref{lemma:invariant}, whose proof uses $\partial\tilde{\eta}_3=0$: the
cycle condition on the twist is what makes the symmetry gaugeable in the first
place.

\medskip
\noindent\textit{Making the background dynamical.}
Gauging is completed by promoting $\green{C^1}$ to a dynamical field
$\green{c^1}\in C^1(\tilde{\L},\Fc^{\gr})$, summing over it, and adding the BF
term $\pi\int_{\tilde{\L}}\green{c^1}\cup\green{d\check{c}^1}$ that enforces its
flatness. The total action is then exactly Eq.~\eqref{eq:non-Abelian_action}, and
the partition function is
\be\label{eq:gauged_partition}
\cZ\big[\tilde{\L},\tilde{\eta}_3\big] =
\sum_{\red{a^1},\blue{b^1},\green{c^1}\in H^1(\tilde{\L})}
(-1)^{\int_{\tilde{\eta}_3}\red{a^1}\cup\blue{b^1}\cup\green{c^1}} ,
\ee
which is Eq.~\eqref{eq:path_inegral_cup}. In words: \emph{gauging the $0$-form
subcomplex CZ symmetry of the untwisted sheaf code produces the twisted sheaf
code.} At the operator level this is the passage from
Eq.~\eqref{eq:stabilizer_untwisted_sheaf} to
Eq.~\eqref{eq:Clifford_group_sheaf}, the CZ dressing of
Eq.~\eqref{eq:stabilizer_summary} being the minimal coupling of the symmetry
current to $\green{c^1}$ --- the same replacement performed microscopically in
Sec.~\ref{sec:gauging}, where the corresponding Gauss law is
Eq.~\eqref{eq:Gauss_law}.

\medskip
\noindent\textit{The equation of motion for $\green{c^1}$.}
The field $\green{c^1}$ enters Eq.~\eqref{eq:gauged_partition} linearly, so it
does not propagate: it is a Lagrange multiplier. Expanding it in the basis
\eqref{eq:spacetime_label}, $\green{c^1}=\sum_{k}\green{l_{k}}\,
\green{\tilde{\gamma}^{1}_{k}}$ with $\green{l_{k}}\in\FF_2$, the sum factorizes and
each component gives
\begin{align}\label{eq:c_eom}
&\sum_{\green{l_{k}}\in\FF_2}
(-1)^{\green{l_{k}}\int_{\tilde{\eta}_3}\red{a^1}\cup\blue{b^1}\cup
\green{\tilde{\gamma}^{1}_{k}}} \cr
&\qquad = 2\,\delta\Big[\int_{\tilde{\eta}_3}\red{a^1}\cup\blue{b^1}\cup
\green{\tilde{\gamma}^{1}_{k}} = 0\Big] .
\end{align}
The gauged theory therefore retains only those configurations of the \rd\ and \bl\
gauge fields that are invariant under every symmetry generator. This is exactly
the code-space constraint \eqref{eq:constraint_operator_sheaf2}, equivalently the
statement $\widetilde{\text{CZ}}^{\rd,\bl}(\green{\gamma^{0}})=\I$ on $\tilde{\C}$
of Eq.~\eqref{eq:CZ_identity_sheaf}: gauging a symmetry projects onto its
invariant sector, and in code language each projection removes one logical qubit.
The number of such projections is the rank $k_{\gr}$ of
Eq.~\eqref{eq:symmetry_group_gauged}, and the resulting reduction of the logical
dimension is computed for the explicit construction in
Lemma~\ref{lemma:rate_twisted}.

\medskip
\noindent\textit{Wilson operators and twisted sectors.}
The theory \eqref{eq:gauged_partition} projects onto the invariant sector, but the
gauging can equally well land in any other sector, and this too is visible in the
path integral. Let $\green{\tilde{\gamma}_{1,k}}$ be a spacetime $1$-cycle pairing
non-trivially with $\green{\tilde{\gamma}^{1}_{k}}$, and insert the \gr\ Wilson
operator
\be\label{eq:wilson_insertion}
W_{\green{w_{k}}} =
(-1)^{\green{w_{k}}\int_{\green{\tilde{\gamma}_{1,k}}}\green{c^1}}
= (-1)^{\green{w_{k}}\green{l_{k}}} ,
\qquad \green{w_{k}}\in\{0,1\} .
\ee
Repeating the sum \eqref{eq:c_eom} with this extra sign replaces the constraint by
\be\label{eq:constraint_shifted}
\int_{\tilde{\eta}_3}\red{a^1}\cup\blue{b^1}\cup
\green{\tilde{\gamma}^{1}_{k}} = \green{w_{k}} ,
\ee
i.e.\ the value of the transversal CZ operator \eqref{eq:defect_equal_time} is
pinned to $(-1)^{\green{w_{k}}}$. The sectors of the twisted theory labelled by
$\green{\vec{w}}=\{\green{w_{k}}\}$ are thus exactly the joint eigenspaces of the
operators $\widetilde{\text{CZ}}^{\rd,\bl}(\green{\gamma^{0}_{k}})$, and
$\green{w_{k}}$ counts the parity of \gr\ charge worldlines linking
$\green{\tilde{\gamma}_{1,k}}$. This is what makes the gauging useful as a logical
operation: performing it by measurement rather than by post-selection produces a
random $\green{\vec{w}}$, and the recorded outcomes tell us onto which joint
eigenspace of a collection of transversal CZ operators the state has been
projected. That protocol is the subject of
Sec.~\ref{sec:logical_operation}.

\begin{remark}[Two routes to the twisted code]\label{rem:two_routes}
The twisted sheaf code has now been obtained twice. In
Secs.~\ref{sec:sheaf_SPT}--\ref{sec:gauging} we began with the sheaf cluster
state, a $\red{\Z_2^{(0)}}\times\blue{\Z_2^{(0)}}\times\green{\Z_2^{(0)}}$ SPT,
and gauged its full protecting symmetry, all three colours at once; the twisted
term arose as the topological response \eqref{eq:globalsym} of the SPT. Here we
began instead with two copies of an ordinary untwisted qLDPC code and gauged only
the single $\ZZ_2$ symmetry generated by the transversal CZ operator between them;
the twisted term arose as the minimal coupling of that symmetry to a third gauge
field. The two constructions agree because gauging is associative: gauging
$\red{\Z_2}\times\blue{\Z_2}$ in the SPT first produces
Eq.~\eqref{eq:Z_untwisted}, with the residual $\green{\Z_2}$ symmetry acting as
Eq.~\eqref{eq:defect_equal_time}, and gauging that residual symmetry is the step
carried out above. The second route is the more useful one operationally: its
input, the untwisted code, is a standard CSS qLDPC code rather than a cluster
state, and the symmetry being gauged is generated by a depth-one transversal
circuit.
\end{remark}

%======================================================================
\section{Generalization to $\Fq$: qudit path integral and Clifford stabilizer
codes}
\label{sec:Fq_generalization}
%======================================================================

The untwisted and twisted qLDPC codes and the associated path integrals and symmetry operators have so far been written for $\FF_2$, i.e.\ for qubits and for sheaves
of $\FF_2$-vector spaces. On the other hand, the sheaves and the corresponding  cup and cap products of
Sec.~\ref{sec:prelim}, and the local cycle condition of
Sec.~\ref{sec:condition_cohomology} have already been developed over an arbitrary finite field $\Fq$. In this section we
lift the
spacetime path integral, the untwisted qLDPC code, and the twisted Clifford
stabilizer code --- from $\FF_2$ to $\Fq$.

There are two reasons to do this. The first is structural: the twisted term
$\int_{\tilde\eta_3}\red{a^1}\cup\blue{b^1}\cup\green{c^1}$ is a type-III
Dijkgraaf--Witten cocycle for the three copies, and nothing in the derivation
used $q=2$ beyond the replacement of $\pm 1$ phases by $p$-th roots of unity. We
stress at the outset that the gauge group is \emph{not} cyclic of order $q$: the
additive group of $\Fq$ is $(\Z_p)^{u}$, so what we obtain is a $\Z_p$ gauge
theory with several flavours, carrying a twist built from the field's
multiplication. Remark~\ref{rem:which_gauge_theory} makes this precise. The
second reason is practical, and is the subject of
Sec.~\ref{sec:high_yield}: the classical codes with a non-trivial multiplication
property that we shall need there to achieve almost-constant magic rate are algebraic, and exist only over alphabets of
size $q = p^{u}$ with $u$ large enough. Working over $\Fq$ is therefore not a
generalization for its own sake but a prerequisite for the construction of magic state fountain with better rate.

\subsection{Qudits, generalized Paulis, and the trace pairing}
\label{sec:qudits}

Fix a prime $p$ and a power $q=p^{u}$. A physical \emph{$\Fq$-qudit} is a
$q$-dimensional system with computational basis $\{\ket{x}\}_{x\in\Fq}$; concretely
it is a block of $u$ qudits of dimension $p$. Write
\begin{align}\label{eq:omega_trace}
\omega &= e^{2\pi i/p} ,
\qquad
\mathrm{Tr} \equiv \mathrm{Tr}_{\Fq/\FF_p} : \Fq \to \FF_p , \cr
\mathrm{Tr}(x) &= \sum_{l=0}^{u-1} x^{p^{l}} .
\end{align}
The trace is $\FF_p$-linear and the bilinear form
$(\alpha,\beta)\mapsto\mathrm{Tr}(\alpha\beta)$ is non-degenerate on $\Fq$; this is
the only property of $\mathrm{Tr}$ we use, and it is what converts an
$\Fq$-valued cohomology invariant into a $p$-th root of unity. Two consequences
will be used repeatedly:
\be\label{eq:trace_delta}
\sum_{\alpha\in\Fq}\omega^{\mathrm{Tr}(\alpha\beta)} = q\,\delta_{\beta,0} ,
\qquad
\omega^{\mathrm{Tr}(\alpha+\beta)} = \omega^{\mathrm{Tr}\alpha}\,
\omega^{\mathrm{Tr}\beta} .
\ee
The first is the $\Fq$ replacement for $\sum_{l\in\FF_2}(-1)^{lx}=2\delta_{x,0}$
and will produce the equations of motion; the second is the homomorphism property
that turns sums of cochains into products of operators, exactly as in
Sec.~\ref{sec:op_valued_cochain}.

The generalized Pauli operators on a single $\Fq$-qudit are
\begin{align}\label{eq:qudit_pauli}
X(\alpha)\ket{x} &= \ket{x+\alpha} , \cr
Z(\beta)\ket{x} &= \omega^{\mathrm{Tr}(\beta x)}\ket{x} ,
\qquad \alpha,\beta\in\Fq ,
\end{align}
which satisfy $X(\alpha)X(\alpha')=X(\alpha+\alpha')$,
$Z(\beta)Z(\beta')=Z(\beta+\beta')$ and the commutation relation
\be\label{eq:qudit_commutation}
Z(\beta)\,X(\alpha) = \omega^{\mathrm{Tr}(\alpha\beta)}\, X(\alpha)\,Z(\beta) .
\ee
For $q=2$ these reduce to $X$ and $Z$, and Eq.~\eqref{eq:qudit_commutation} to
anticommutation. The two-qudit generalized CZ gate is the diagonal Clifford gate
\be\label{eq:qudit_CZ}
\text{CZ}(\lambda)\ket{x,y} = \omega^{\mathrm{Tr}(\lambda x y)}\ket{x,y} ,
\qquad \lambda\in\Fq ,
\ee
and similarly $\text{CCZ}(\lambda)\ket{x,y,z}=\omega^{\mathrm{Tr}(\lambda
xyz)}\ket{x,y,z}$. Regarded as a gate on two blocks of $u$ elementary
$\FF_p$-qudits, $\text{CZ}(\lambda)$ is a depth-one circuit of $u$ elementary CZ
gates, by Proposition~\ref{prop:cz_descent_app} of
Appendix~\ref{app:alphabet}; the analogous statement fails for
$\text{CCZ}(\lambda)$, which is the entire content of alphabet reduction and is
discussed there. Conjugation by $\text{CZ}(\lambda)$ maps
$X_1(\alpha)\mapsto X_1(\alpha)Z_2(\lambda\alpha)$, so
$\text{CZ}(\lambda)$ is Clifford but not Pauli, and $\text{CCZ}(\lambda)$ is
non-Clifford; these are the $\Fq$ statements of the two facts on which the whole
construction rests.

The physical qudits of the code sit, as before, on the pairs $(e,s)$ with
$e\in\L(1)$ and $s$ a stalk index, and we write $X^{i}_{e,s}(\alpha)$,
$Z^{i}_{e,s}(\beta)$ for the corresponding operators of colour
$i\in\{\rd,\bl,\gr\}$. Operator-valued sheaf cochains are defined exactly as in
Sec.~\ref{sec:op_valued_cochain}, with Eq.~\eqref{eq:Z_cochain_untwisted} replaced
by
\be\label{eq:Z_cochain_Fq}
Z_{e,s}(\beta) = \omega^{\mathrm{Tr}\left(\beta\,[\hat{a}^1(e)]_s\right)} ,
\ee
so that $\hat{a}^1 \in C^1(\L,\Fc)$ now has components valued in $\Fq$ rather than
in $\FF_2$, and the elementary cochain $\bar{\sigma}_s$ in 
Eq.~\eqref{eq:cochain_expansion_untwisted} is unchanged.

\begin{remark}[Signs and orientations]\label{rem:signs}
For $p=2$ every sign is $+1$ and orientations may be ignored, which is why they
were invisible above. For odd $p$ the coboundary of
Eq.~\eqref{eq:coboundary-sheaf} carries the incidence signs of the complex, and
the Leibniz rule of Proposition~\ref{prop:cup-properties} takes its general form
\be\label{eq:leibniz_signed}
d\big(x^{i}\cup y^{j}\big) = d x^{i} \cup y^{j}
+ (-1)^{i}\, x^{i} \cup d y^{j} .
\ee
All the manipulations below involve a $0$-cochain in the leftmost slot, for which
$(-1)^{i}=+1$, so the signs never obstruct the argument; we indicate the one place
where a global minus sign survives. Throughout this section the graphs $G_\nu$ are
the bipartite Tanner graphs of Sec.~\ref{sec:classical}, oriented from bit
vertices to check vertices, so that each vertex sees edges of a single orientation
and the local co-restrictions $\Fc_{v\shortrightarrow e}$ carry a common sign; this
is the standard hypothesis under which Tanner codes and sheaf cocycle codes agree
over a general field;  Alternatively, we can consider more general graph $G_\nu$ but restrict to the case where $\FF_q$ has characteristic 2, i.e., $p=2$.
\end{remark}

\subsection{The $\Fq$ path integral}
\label{sec:Fq_path_integral}

The twisted $\Fq$ sheaf gauge theory has action
\begin{align}\label{eq:action_Fq}
S =& \frac{2\pi}{p}\,\mathrm{Tr}\bigg[
\int_{\tilde{\L}} \red{a^1}\cup\red{d\check{a}^1}
+ \int_{\tilde{\L}} \blue{b^1}\cup\blue{d\check{b}^1}
+ \int_{\tilde{\L}} \green{c^1}\cup\green{d\check{c}^1} \bigg] \cr
&+ \frac{2\pi}{p}\,\mathrm{Tr}
\int_{\tilde{\eta}_3} \red{a^1}\cup\blue{b^1}\cup\green{c^1} ,
\end{align}
with $\red{a^1}\in C^1(\tilde{\L},\Fc^{\rd})$ and its analogues now
$\Fq$-valued sheaf cochains. Summing over the magnetic Lagrange multipliers
imposes $\red{da^1}=\blue{db^1}=\green{dc^1}=0$ by the first identity of
Eq.~\eqref{eq:trace_delta}, and the partition function reduces to
\be\label{eq:path_integral_Fq}
\cZ\big[\tilde{\L},\tilde{\eta}_3\big] =
\sum_{\red{a^1},\blue{b^1},\green{c^1}\in H^1(\tilde{\L})}
\omega^{\,\mathrm{Tr}\int_{\tilde{\eta}_3}
\red{a^1}\cup\blue{b^1}\cup\green{c^1}} ,
\ee
which reduces to Eq.~\eqref{eq:path_inegral_cup} at $q=2$. The exponent is the
$\Fq$-valued invariant already studied in
Sec.~\ref{sec:condition_cohomology}: the triple cup product lands in
$H^3(\tilde{\L},\Fc^{\rd}\otimes\Fc^{\bl}\otimes\Fc^{\gr})$ by
Eq.~\eqref{eq:trilinear_map}, and $\tilde{\eta}_3$ is a cycle of the tensor sheaf.

\begin{lemma}[Gauge invariance over $\Fq$]\label{lemma:invariant_Fq}
$\cZ[\tilde{\L},\tilde{\eta}_3]$ of Eq.~\eqref{eq:path_integral_Fq} is a
cohomology invariant.
\end{lemma}

\begin{proof}
Identical to Lemma~\ref{lemma:invariant}. Under
$\red{a^1}\to\red{a^1}+d\tilde{\eta}^0$ and its analogues, the signed Leibniz rule
\eqref{eq:leibniz_signed} gives
$\red{a^1}\cup\blue{b^1}\cup\green{c^1}\to
\red{a^1}\cup\blue{b^1}\cup\green{c^1}+d\varpi^2$ for some $\varpi^2$, and
$\int_{\tilde{\eta}_3}d\varpi^2=\int_{\partial\tilde{\eta}_3}\varpi^2=0$. The
exponentiation to a phase then uses only the $\FF_p$-linearity of $\mathrm{Tr}$,
Eq.~\eqref{eq:trace_delta}.
\end{proof}

\begin{remark}[Which gauge theory is this?]\label{rem:which_gauge_theory}
The gauge group is elementary abelian, not cyclic. Gauge transformations act by
$\red{a^1}\to\red{a^1}+d^0\lambda^0$ with
$\lambda^0\in C^0(\L,\Fc^{\rd})$, so the group of gauge transformations of one
colour is $\bigoplus_{v}\Fc^{\rd}(v)\cong\Z_p^{\,u\sum_v m(v)}$, since the
additive group of $\Fq$ is $(\Z_p)^{u}$ and not $\Z_{p^{u}}$. The theory is
therefore a type-III twisted $\Z_p$ gauge theory with many flavours; it is a
$\Z_q$ gauge theory only in the prime case $u=1$.

What the field structure supplies is not the gauge group but the \emph{twist}.
Fixing an $\FF_p$-basis $\{\theta_1,\ldots,\theta_u\}$ of $\Fq$ and writing
$\red{a^1}=\sum_i \red{a^1_i}\theta_i$ and similarly for the other colours,
\begin{align}\label{eq:T_ijk}
\mathrm{Tr}\big(\red{a}\,\blue{b}\,\green{c}\big) &=
\sum_{i,j,k} T_{ijk}\ \red{a_i}\,\blue{b_j}\,\green{c_k} , \cr
T_{ijk} &= \mathrm{Tr}\big(\theta_i\theta_j\theta_k\big) \in \FF_p ,
\end{align}
a totally symmetric $\FF_p$-valued three-tensor --- the structure constants of
$\Fq$ regarded as a commutative Frobenius algebra over $\FF_p$, the trace form
being the Frobenius pairing. The twisted term of Eq.~\eqref{eq:action_Fq} is thus
a specific superposition of $u^{3}$ ordinary type-III $\Z_p^{3}$
Dijkgraaf--Witten cocycles, with the field's multiplication table as
coefficients. It is non-trivial: taking $\theta_1=1$ gives
$T_{1jk}=\mathrm{Tr}(\theta_j\theta_k)$, which is the non-degenerate trace form.

Two consequences are worth recording. First, the classical input genuinely needs
a ring: flatness and Gauss's law involve only the additive group, whereas the
multiplication property of Sec.~\ref{sec:condition_cohomology} is a statement
about Schur products and therefore about the field's multiplication. As an
$\FF_p$-sheaf, an $\Fq$-sheaf code is merely an $\FF_p$-sheaf code with stalks
$u$ times larger; the $\Fq$-module structure is extra structure on top, invisible
to the gauge group, and it is exactly what makes the twist available. Second,
for $p=2$ the physical systems remain \emph{qubits}: each $\Fq$-qudit is a block
of $u$ qubits and, by Proposition~\ref{prop:cz_descent_app} of
Appendix~\ref{app:alphabet}, the CZ dressings descend to depth-one qubit CZ
circuits, so no qudit hardware is implied
(Remark~\ref{rem:qubit_realization}).

Finally, if one did want a genuinely cyclic $\Z_{p^{u}}$ gauge theory, one would
work over the ring $\Z/p^{u}\Z$ rather than over a field. The cochain complexes,
cup products and twisted stabilizer construction are all well defined over any
commutative ring; what fails is the classical coding theory, since the
algebraic-geometry codes of Sec.~\ref{sec:high_yield} require a field.
\end{remark}

\begin{remark}[Which $q$ makes the twist non-trivial]\label{rem:nontrivial_q}
The classification of the invariant in Sec.~\ref{sec:condition_cohomology} is
unchanged: by Lemma~\ref{lemma:factorization} the invariant factorizes over the
tensor factors, and by Lemma~\ref{lemma:necessary_condition} it is non-zero only
if the tensor sheaf admits a non-zero $1$-cycle on each spatial factor, i.e.\ only
if the local cycle condition of Lemma~\ref{lemma:local_cycle} holds. That
condition, and its reformulation as the multiplication property
\eqref{eq:multiplication_property_local}, were already stated over $\Fq$. What
changes with $q$ is only how easy the condition is to satisfy \emph{non-trivially},
and this is precisely where large $q$ helps.
\end{remark}

\subsection{The untwisted qudit sheaf code}
\label{sec:untwisted_Fq}

The untwisted code is the CSS code over $\Fq$ attached to
$C^0(\L,\Fc)\xrightarrow{d^0}C^1(\L,\Fc)\xrightarrow{d^1}C^2(\L,\Fc)$, with
$H_X^{T}=d^0$ and $H_Z=d^1$ now matrices over $\Fq$. It is convenient to label
the stabilizer generators by cochains rather than by cells. For
$\lambda^0\in C^0(\L,\Fc)$ and $\beta\in\Fq$, $f\in\L(2)$, set
\begin{align}\label{eq:AB_Fq}
A(\lambda^0) &= \prod_{e,s} X_{e,s}
\Big(\big[(d^0\lambda^0)(e)\big]_s\Big) , \cr
B_{f}(\beta) &= \omega^{\mathrm{Tr}\left[\beta\,
(d^1\hat{a}^1)(f)\right]} ,
\end{align}
the generators of Eqs.~\eqref{eq:A_stabilizer_untwisted}
and~\eqref{eq:B_stabilizer_untwisted} being recovered as
$A_{v,s}(\alpha)=A(\alpha\,\bar{v}_s)$ and $B_f(\beta)$; as before the $B$'s carry
no stalk index because $\Fc(f)$ is one-dimensional. The map
$\lambda^0\mapsto A(\lambda^0)$ is a homomorphism from $C^0(\L,\Fc)$ into the
Pauli group, since $X(\alpha)X(\alpha')=X(\alpha+\alpha')$ and $d^0$ is linear. By
Eq.~\eqref{eq:qudit_commutation}, $A(\lambda^0)$ implements the shift
\begin{align}\label{eq:A_conjugation_Fq}
&A(\lambda^0)\, Z_{e,s}(\beta)\, A(\lambda^0)^{-1} \cr
&\qquad = \omega^{-\mathrm{Tr}\left(\beta\,[(d^0\lambda^0)(e)]_s\right)}
Z_{e,s}(\beta) , \cr
&\qquad\qquad
\hat{a}^1 \ \longrightarrow\ \hat{a}^1 + d^0\lambda^0 ,
\end{align}
the $\Fq$ form of Eq.~\eqref{eq:A_conjugation}, and
$[A(\lambda^0),B_f(\beta)]=\I$ reduces as in
Eq.~\eqref{eq:commutation_d2} to $d^1d^0=0$. The logical operators are labelled as
in Sec.~\ref{sec:logical_operators_sheaf}, by $H^1(\L,\Fc)$ for the $X$-type and
$H_1(\L,\Fc)$ for the $Z$-type, the two being paired non-degenerately by
Proposition~\ref{prop:nondegen-sheaf}, which holds over any field. Nothing in
Sec.~\ref{sec:untwisted_code} used $q=2$.

\subsection{The twisted Clifford stabilizer code over $\Fq$}
\label{sec:Clifford_Fq}

The twisted generators are obtained, as in Eq.~\eqref{eq:stabilizer_summary}, by
dressing each $X$-type generator with the diagonal operator built from the twisted
term of Eq.~\eqref{eq:action_Fq}. Writing the dressing for colour \rd\ as
\be\label{eq:dressing_Fq}
U^{\rd}(\lambda^0) = \omega^{\,\mathrm{Tr}\int_{\eta_2}
\lambda^0 \cup \blue{\hat{b}^1} \cup \green{\hat{c}^1}} ,
\ee
which by Eq.~\eqref{eq:qudit_CZ} is a depth-one circuit of two-qudit
$\text{CZ}^{\bl,\gr}(\lambda)$ gates, the twisted stabilizer generators are
\be\label{eq:twisted_generators_Fq}
\tilde{A}^{\rd}(\lambda^0) = A^{\rd}(\lambda^0)\, U^{\rd}(\lambda^0) \ \text{and} \ 
\qquad B^{\rd}_{f}(\beta) ,
\ee
together with their colour permutations, and
\begin{align}\label{eq:Clifford_group_Fq}
\tilde{\SS} = \big\langle\, & \tilde{A}^{i}(\lambda^0),\ B^{i}_{f}(\beta) \cr
& :\ i\in\{\rd,\bl,\gr\},\ \lambda^0\in C^0(\L,\Fc^{i}),\ \beta\in\Fq
\,\big\rangle.   \cr
\end{align}
Since $A^{\rd}$ acts on the \rd\ qudits and $U^{\rd}$ on the \bl\ and \gr\ ones,
the two factors commute and
$\tilde{A}^{\rd}(\lambda^0)\tilde{A}^{\rd}(\lambda'^0)=
\tilde{A}^{\rd}(\lambda^0+\lambda'^0)$: same-colour twisted generators commute
exactly and form a group isomorphic to $C^0(\L,\Fc^{\rd})$. As over $\FF_2$,
$\tilde{\SS}$ is a subgroup of the $\Fq$ Clifford group and not of the Pauli
group.

\begin{lemma}[Commutation over $\Fq$]\label{lemma:commutation_Fq}
For $\lambda^0\in C^0(\L,\Fc^{\rd})$ and $\kappa^0\in C^0(\L,\Fc^{\bl})$,
\be\label{eq:AA_commutator_Fq}
\big[\tilde{A}^{\rd}(\lambda^0),\, \tilde{A}^{\bl}(\kappa^0)\big]
= \omega^{-\mathrm{Tr}\int_{\eta_2}
\left(\lambda^0\cup\kappa^0\right)\cup \green{d\hat{c}^1}} ,
\ee
which is a product of $B^{\gr}_{f}$ operators and therefore equals $\I$ on the
zero-flux subspace. All other pairs of generators of $\tilde{\SS}$ commute
exactly, so Lemma~\ref{lemma:stabilizer_commutation} holds verbatim over $\Fq$.
\end{lemma}

\begin{proof}
$A^{\rd}(\lambda^0)$ shifts $\red{\hat{a}^1}\to\red{\hat{a}^1}+d^0\lambda^0$ by
Eq.~\eqref{eq:A_conjugation_Fq} and leaves $\blue{\hat{b}^1},\green{\hat{c}^1}$
alone, so
$A^{\rd}(\lambda^0)\,U^{\bl}(\kappa^0)\,A^{\rd}(\lambda^0)^{-1}
= U^{\bl}(\kappa^0)\,\omega^{\mathrm{Tr}\int_{\eta_2}d^0\lambda^0\cup\kappa^0
\cup\green{\hat{c}^1}}$, and symmetrically for the other order. Multiplying the
two phases and using the Leibniz rule \eqref{eq:leibniz_signed} for the two
$0$-cochains, where the sign is $+1$,
\be\label{eq:leibniz_zero_cochains}
d^0\lambda^0\cup\kappa^0 + \lambda^0\cup d^0\kappa^0
= d\big(\lambda^0\cup\kappa^0\big) ,
\ee
gives the exponent $\mathrm{Tr}\int_{\eta_2}
d(\lambda^0\cup\kappa^0)\cup\green{\hat{c}^1}$. Applying
Eq.~\eqref{eq:leibniz_signed} once more to
$d[(\lambda^0\cup\kappa^0)\cup\green{\hat{c}^1}]$, again with $+1$ sign since
$\lambda^0\cup\kappa^0$ has degree $0$, and using
$\langle d\varpi,\eta_2\rangle=\langle\varpi,\partial\eta_2\rangle=0$, one obtains
Eq.~\eqref{eq:AA_commutator_Fq}; the surviving minus sign comes from moving the
second Leibniz term across the equality and is invisible at $q=2$. Since
$(\green{d\hat{c}^1})(f)\in\Fc^{\gr}(f)\cong\Fq$ is a scalar, the exponent
factorizes over $2$-cells exactly as in Eq.~\eqref{eq:Theta_def} and the result is
a product of $B^{\gr}_{f}(\beta)$'s.
\end{proof}

\begin{lemma}[Code-space constraints over $\Fq$]\label{lemma:constraints_Fq}
Let $\mu^0\in Z^0(\L,\Fc^{\rd})$ be a sheaf $0$-cocycle. Then
$\tilde{A}^{\rd}(\mu^0)=U^{\rd}(\mu^0)$, and hence on the twisted code space
$\tilde{\C}$
\be\label{eq:constraint_Fq}
\int_{\eta_2}\mu^0\cup\blue{\hat{b}^1}\cup\green{\hat{c}^1} = 0 \in \Fq ,
\ee
together with the two colour-permuted statements. Equivalently, the transversal
$\Fq$ CZ operators
\be\label{eq:CZ_Fq}
\widetilde{\text{CZ}}^{\rd,\bl}(\green{\gamma^0}) \equiv
\omega^{\,\mathrm{Tr}\int_{\eta_2}\red{\hat{a}^1}\cup\blue{\hat{b}^1}\cup\green{\gamma^0}} ,
\qquad \green{\gamma^0} \in Z^0\big(\L,\Fc^{\gr}\big) ,
\ee
act as the identity on $\tilde{\C}$.
\end{lemma}

\begin{proof}
The $X$-part of $\tilde{A}^{\rd}(\mu^0)$ is $A^{\rd}(\mu^0)$, which implements the
shift by $d^0\mu^0=0$ and is therefore $\I$, so the generator reduces to its
dressing. Every generator equals $+1$ on $\tilde{\C}$, so
$\omega^{\mathrm{Tr}[\,\cdot\,]}=1$ for the exponent of
Eq.~\eqref{eq:dressing_Fq}; since the same holds for $\mu^0$ replaced by
$\alpha\mu^0$ for every $\alpha\in\Fq$, the non-degeneracy of the trace form,
Eq.~\eqref{eq:trace_delta}, upgrades the statement from
$\mathrm{Tr}[\,\cdot\,]=0$ in $\FF_p$ to Eq.~\eqref{eq:constraint_Fq} in $\Fq$.
\end{proof}

Lemma~\ref{lemma:constraints_Fq} is the $\Fq$ form of
Eqs.~\eqref{eq:constraint_operator_sheaf}--\eqref{eq:CZ_identity_sheaf}, and the
step through $\alpha\mu^0$ is worth noting: over $\Fq$ the code-space constraint is
an equation in $\Fq$, not merely a $\FF_p$ phase condition, and it is the
non-degeneracy of the trace pairing that makes the two equivalent. The remaining
statements of Sec.~\ref{sec:Clifford_stabilizer_codes} transfer without change:
the electric logical operators are the sheaf Wilson operators
\begin{align}\label{eq:logical_Z_Fq}
\widetilde{Z}^{\rd}_{\eta_1}(\beta) &=
\omega^{\mathrm{Tr}\left(\beta\int_{\eta_1}\red{\hat{a}^1}\right)} \cr
&= \prod_{(e,s)} Z^{\rd}_{e,s}\big(\beta\,[\eta_1(e)]_s\big) ,
\qquad \eta_1\in Z_1(\L,\Fc^{\rd}), \cr
\end{align}
the contractible magnetic operators factorize into twisted stabilizers as in
Eq.~\eqref{eq:magnetic_contractible_sheaf}, and the non-contractible ones are
dressed with projectors as in Eq.~\eqref{eq:magnetic_projected_sheaf}, with
$\I+Z$ replaced by $\sum_{\beta\in\Fq}Z(\beta)$.

\subsection{Subcomplex symmetry and gauging over $\Fq$}
\label{sec:subcomplex_Fq}

The operator $\widetilde{\text{CZ}}^{\rd,\bl}(\green{\gamma^0})$ of Eq.~\eqref{eq:CZ_Fq} is
again a $0$-form subcomplex symmetry: its support is the capped cycle
$\eta_2^{\gamma}=\eta_2\frown\green{\gamma^0}$ of Eq.~\eqref{eq:cap_support}, the cap product
of Definition~\ref{def:cap-simplicial} being defined over any field, and the whole
of Sec.~\ref{sec:0form_subcomplex} goes through with $(-1)^{x}$ replaced by
$\omega^{\mathrm{Tr}(x)}$. Two points are specific to $\Fq$ and worth recording.

First, the symmetry group is now $\Fq$-linear rather than $\FF_2$-linear:
\begin{align}\label{eq:symmetry_group_Fq}
&\Big\{\widetilde{\text{CZ}}^{\rd,\bl}(\green{\gamma^0})\ :\
\green{\gamma^0}\in Z^0(\L,\Fc^{\gr})\Big\} \cr
&\qquad \cong\; H^0\big(\L,\Fc^{\gr}\big) \;\cong\; \Fq^{\,k_{\gr}} ,
\end{align}
so each of the $k_{\gr}$ independent generators carries an $\Fq$-valued, rather
than binary, label; correspondingly the logical action of
Eq.~\eqref{eq:logical_CZ_action} becomes
\begin{align}\label{eq:logical_CZ_Fq}
\widetilde{\text{CZ}}^{\rd,\bl}(\green{\gamma^0}) \;\widehat{=}\;&
\prod_{\red{\alpha},\blue{\beta}}
\lo{\text{CZ}}^{\rd,\bl}_{\red{\alpha},\blue{\beta}}
\Big(\Lambda^{\red{\alpha}\blue{\beta}}(\green{\gamma^0})\Big) , \cr
\Lambda^{\red{\alpha}\blue{\beta}}(\green{\gamma^0}) =&\
\int_{\eta_2}\red{\alpha^1}\cup\blue{\beta^1}\cup\green{\gamma^0} \ \in\ \Fq ,
\end{align}
a product of $\Fq$ logical CZ gates with multipliers given by the invariant form.
Second, gauging proceeds as in Sec.~\ref{sec:gauging_untwisted}: promoting the
background to a dynamical $\green{c^1}$ and summing, the equation of motion
\eqref{eq:c_eom} becomes, by the first identity of Eq.~\eqref{eq:trace_delta},
\be\label{eq:c_eom_Fq}
\sum_{\green{l}\in\Fq}
\omega^{\,\mathrm{Tr}\left[\green{l}\int_{\tilde{\eta}_3}\red{a^1}\cup\blue{b^1}
\cup\green{\tilde{\gamma}^{1}}\right]}
= q\,\delta\Big[\int_{\tilde{\eta}_3}\red{a^1}\cup\blue{b^1}
\cup\green{\tilde{\gamma}^{1}} = 0\Big] ,
\ee
and inserting a \gr\ Wilson operator of charge $\green{w}\in\Fq$ shifts the
right-hand constraint to $\green{w}$. The measurement outcomes of the gauging
protocol of Sec.~\ref{sec:logical_operation} therefore take values in $\Fq$ rather
than in $\FF_2$, and the projectors in Eq.~\eqref{eq:output_logical_state} project
onto the corresponding $\Fq$ eigenspaces. Everything else in that protocol ---
including the fact that the $\gr$-type twisted stabilizers are simultaneously
measurable, Lemma~\ref{lemma:commutation_Fq} --- is unchanged.

% \begin{remark}[What is gained]\label{rem:why_Fq}
% No qualitatively new phenomenon appears at $q>2$: the $\Fq$ theory is the same
% Dijkgraaf--Witten twist with $\ZZ_2$ replaced by $\Z_p$ throughout. The gain is
% that the \emph{classical} input is now allowed to be an algebraic code over a
% large alphabet. The local cycle condition
% \eqref{eq:multiplication_property_local} asks for local codes whose Schur product
% stays inside a prescribed code, and over $\FF_2$ the only easy solutions are the
% degenerate ones --- the repetition/single-parity-check pair used in
% Construction~\ref{con:2d-hgp}, which costs a factor of $\sqrt{n}$ in the number of
% available symmetry generators. Over $\Fq$ with $q$ a large constant, Reed--Solomon
% and algebraic-geometry codes satisfy a multiplication property with both factors
% of large dimension and distance, and this is what we exploit next.
% \end{remark}

%======================================================================
\section{Condition for non-trivial cohomology invariants in hypergraph-product constructions}
\label{sec:condition_cohomology}
%======================================================================

The twisted theory of Sec.~\ref{sec:non-Abelian_LDPC} is non-trivial only if the
cohomology invariant appearing in the path integral of
Eq.~\eqref{eq:path_inegral_cup} actually evaluates to a non-zero value for some
choice of cocycles. In this section we determine a \emph{necessary} condition for
this to happen when the spatial complex is a hypergraph product of classical sheaf
codes. The answer turns out to be surprisingly concrete. The spacetime invariant
factorizes, by the K\"unneth theorem, into a product of contributions from the two
spatial tensor factors and one from the temporal factor. The temporal contribution
is automatically non-trivial, so the whole condition is carried by the spatial
factors, and within each of them the triple cup product must be paired against a
non-trivial $1$-cycle of the \emph{tensor sheaf}. Requiring such a cycle to exist
gives a purely local linear-algebraic constraint on the local codes, which we call
the \emph{local cycle condition}. Specializing the cycle to the fundamental class
turns this condition into the \emph{multiplication property} of the local codes
with respect to the Schur product.

Two conventions are used throughout this section. First, we work over a general
finite field $\Fq$ rather than only $\FF_2$; all statements below hold
verbatim for any $q$, and $q=2$ recovers the qubit case of the previous sections.
Since the invariant is $\Fq$-valued, the phase in
Eq.~\eqref{eq:path_inegral_cup} is written with a non-degenerate additive
character $\chi:\Fq \to \CC^{\times}$, e.g.\ $\chi(x)=\omega_p^{\,\mathrm{tr}(x)}$
with $\omega_p=e^{2\pi i/p}$ and $\mathrm{tr}$ the field trace for $q=p^{s}$; the
statement ``the invariant is non-trivial'' always means that the $\Fq$-valued
pairing itself is non-zero. Second, following the convention of
Ref.~\cite{ZhuKobayashiHsin2026} we use a
tilde to distinguish \emph{spacetime} (co)cycles from \emph{spatial} ones:
$\red{\tilde{\alpha}^1}$,
$\blue{\tilde{\beta}^1}$, $\green{\tilde{\gamma}^1}$ are spacetime cocycles, while
$\red{\alpha^1}$, $\blue{\beta^1}$, $\green{\gamma^0}$ denote their spatial parts.

% Third, we label edges directly by $e$ rather than by their endpoints $e$, so
% that the restriction map of Eq.~\eqref{eq:restriction_map} is written
% \be\label{eq:restriction_concise}
% \bar{\Fc}_{v \shortleftarrow e} = \lv_e ,
% \ee
% the $e^\text{th}$ column of the local parity-check matrix $\mathsf{H}_v$ at the
% vertex $v$. This is more compact and is the notation we use for the rest of the
% section.

\subsection{Setup: spatial hypergraph product and the temporal factor}
\label{sec:triple_HGP}

The spatial complex is the $2$D square complex as the Cartesian product of \emph{two} graphs, as in Sec.~\ref{sec:classical}
\be\label{eq:spatial_HGP}
\L = G_x \times G_y , \qquad G_\nu = (V_\nu, E_\nu) \quad (\nu = x,y) ,
\ee
and the spacetime complex adjoins one further  factor, the temporal complex
$S^1_t$, which is the chain complex of a repetition code, i.e.\ a $1$D line with
time slices on $0$-cells and time steps on $1$-cells,
\be\label{eq:spacetime_complex}
\tilde{\L} = \L \times S^1_t = G_x \times G_y \times S^1_t .
\ee
Thus $\mathsf{d}=2$, $D=3$, in agreement with the all-degrees-one case of
Sec.~\ref{sec:sheaf_SPT}. Each colour $i \in \{\rd,\bl,\gr\}$ carries its own
system of local coefficients on each spatial factor, denoted $\Fc^{i}_{\nu}$, so
that the cochain complex of colour $i$ is
\be\label{eq:colour_complex}
\mathcal{A}^{i} = C^\bullet\big(G_x,\Fc^{i}_x\big) \otimes
C^\bullet\big(G_y,\Fc^{i}_y\big) \otimes C^\bullet\big(S^1_t,\Fq\big) ,
\ee
the temporal factor carrying the constant sheaf $\Fq$ since the repetition code
has trivial local code. On each spatial factor the sheaf has the structure of
Eq.~\eqref{eq:local_system}, generalized to $\Fq$,
\be\label{eq:stalk_Fq}
\Fc^{i}_{\nu}(\sigma) =
\begin{cases}
\Fq^{\,m^{i}_{\nu}(v)} , & \sigma = v \in V_\nu \quad (\text{check}), \\[2pt]
\Fq , & \sigma = e \in E_\nu \quad (\text{bit}) ,
\end{cases}
\ee
with restriction maps $\bar{\Fc}^{i}_{\nu,\,v \shortleftarrow e} = \lv^{i}_{e}$
given by the columns of the local parity-check matrix $\mathsf{H}^{i}_{v}$, whose
kernel is the local code $\cC^{i}_{v} = \ker \mathsf{H}^{i}_{v}$. We suppress the
factor label $\nu$ on $\lv$ and $\mathsf{H}$ whenever a single factor is under
discussion. Throughout, $\nu\in\{x,y\}$ denotes the spatial factor index, while the
subscripted symbols $i_1,i_2,i_3$ and $j_1,j_2,j_3$ introduced below are
K\"unneth basis labels of the $x$- and $y$-factor respectively.

Recall from Proposition~\ref{Prop:sheaf_homology} that the classical sheaf code on
a graph is the first sheaf homology, and note that the graph chain complex has only
two terms, $C_1(G,\Fc) \xrightarrow{\partial_1} C_0(G,\Fc)$. Consequently
\be\label{eq:H1_is_Z1}
H_1(G,\Fc) = \ker \partial_1 = Z_1(G,\Fc) ,
\ee
i.e.\ on a graph every non-zero $1$-cycle is automatically non-trivial in
homology, there being no $2$-chains to bound it. This will let us replace
``non-trivial in $H_1$'' by ``non-zero'' below.

\subsection{Factorization of the invariant by the K\"unneth theorem}
\label{sec:kunneth_factorization}

The three spacetime gauge fields entering Eq.~\eqref{eq:path_inegral_cup} are
$1$-cocycles, so their triple cup product is a $3$-cochain, matching
$D=3$. By the K\"unneth theorem each spacetime $1$-cocycle decomposes into
components labelled by how the total degree $1$ is distributed among the three
tensor factors of Eq.~\eqref{eq:spacetime_complex}. The only assignment producing
a non-vanishing triple product is the one in which each colour carries its
degree-$1$ piece on a \emph{different} factor, so that every factor receives
exactly one degree-$1$ and two degree-$0$ cochains. We therefore take
\begin{align}\label{eq:cocycle_decomposition}
\red{\tilde{\alpha}^{1}_{\vc{i_1}}} &= \red{\as^1_{i_1}} \otimes
\red{\bs^0_{j_1}} \otimes \ts^0
\;\equiv\; \red{\alpha^{1}_{\vc{i_1}}} \otimes \ts^0 , \cr
\blue{\tilde{\beta}^{1}_{\vc{i_2}}} &= \blue{\bs^0_{i_2}} \otimes
\blue{\as^1_{j_2}} \otimes \ts^0
\;\equiv\; \blue{\beta^{1}_{\vc{i_2}}} \otimes \ts^0 , \cr
\green{\tilde{\gamma}^{1}_{\vc{i_3}}} &= \green{\bs^0_{i_3}} \otimes
\green{\bs^0_{j_3}} \otimes {\ts}^1
\;\equiv\; \green{\gamma^{0}_{\vc{i_3}}} \otimes {\ts}^1 ,
\end{align}
where we have fixed a K\"unneth basis and labelled its elements as follows. On the
$x$-factor, $\as^1_{i} \in Z^1(G_x,\Fc^{i}_x)$ and
$\bs^0_{i} \in Z^0(G_x,\Fc^{i}_x)$ are basis degree-$1$ and degree-$0$ cocycles
carrying the scalar labels $i_1$, $i_2$, $i_3$ for the \rd, \bl \ and \gr \ copy
respectively; on the $y$-factor, $\as^1_{j} \in Z^1(G_y,\Fc^{i}_y)$ and
$\bs^0_{j} \in Z^0(G_y,\Fc^{i}_y)$ are the corresponding basis cocycles, carrying
the scalar labels $j_1$, $j_2$, $j_3$. No prime is needed on the cocycles: the
index family alone already records the factor, an $i$-family subscript meaning the
$x$-factor and a $j$-family subscript the $y$-factor, while the colour records
which copy the cocycle belongs to. The composite spacetime basis cocycles then
carry the \emph{multi-indices}
\be\label{eq:multi_index}
\vc{i_1} = (i_1, j_1) , \qquad
\vc{i_2} = (i_2, j_2) , \qquad
\vc{i_3} = (i_3, j_3) ,
\ee
written in boldface to distinguish them from the scalar labels of the individual
factors; no temporal label is needed since $H^0(S^1_t;\Fq)$ and $H^1(S^1_t;\Fq)$ are
one-dimensional, with generators $\ts^0$ and ${\ts}^1$. The spatial parts are
therefore
\begin{align}\label{eq:spatial_parts}
\red{\alpha^{1}_{\vc{i_1}}} &= \red{\as^1_{i_1}} \otimes \red{\bs^0_{j_1}} ,
\quad
\blue{\beta^{1}_{\vc{i_2}}} = \blue{\bs^0_{i_2}} \otimes \blue{\as^1_{j_2}} ,
\cr
\green{\gamma^{0}_{\vc{i_3}}} &= \green{\bs^0_{i_3}} \otimes
\green{\bs^0_{j_3}} ,
\end{align}
of spatial degrees $1$, $1$ and $0$ respectively, adding up to
$\mathsf{d}=2$. It is worth noting that $\green{\gamma^0_{\vc{i_3}}}$ is exactly
the type of object that labels the $0$-form subcomplex symmetry of
Sec.~\ref{sec:transversal_CZ}: the green $0$-cocycle $\green{\gamma^0_{\vc{i_3}}}$ appearing in
$\widetilde{\text{CZ}}^{\rd,\bl}(\green{\gamma^0_{\vc{i_3}}})$ is the spatial part of the green
spacetime cocycle $\green{\tilde \gamma^0_{\vc{i_3}}}$, the temporal direction having been absorbed into ${\ts}^1$.

Correspondingly the spacetime $3$-cycle is the product of the spatial $2$-cycle of
Sec.~\ref{sec:sheaf_SPT} with the temporal $1$-cycle, and the spatial $2$-cycle
itself factorizes by K\"unneth into a pair of tensor-sheaf $1$-cycles, one per
spatial factor,
\begin{align}\label{eq:eta_product}
\tilde{\eta}_3 &= \eta_2 \otimes S^1_t , \qquad
\eta_2 = \zeta \otimes \zeta' , \cr
\zeta &\in H_1\big(G_x,\ \Fc^{\rd}_x \otimes \Fc^{\bl}_x \otimes
\Fc^{\gr}_x\big) , \cr
\zeta' &\in H_1\big(G_y,\ \Fc^{\rd}_y \otimes \Fc^{\bl}_y \otimes
\Fc^{\gr}_y\big) ,
\end{align}
where, having fixed a cycle basis once and for all, we suppress the basis label on
the cycles and simply write $\zeta$ for the chosen $1$-cycle on the $x$-factor and
$\zeta'$ for the one on the $y$-factor; here the prime distinguish the two
factors. Note that both are
$1$-cycles of the \emph{tensor sheaf} on the corresponding graph, since they must
pair with cochains valued in
$\Fc^{\rd}_\nu\otimes\Fc^{\bl}_\nu\otimes\Fc^{\gr}_\nu$.

We first dispose of the temporal factor.

\begin{lemma}[The temporal factor is automatically non-trivial]
\label{lemma:temporal_factor}
With $\ts^0$ and ${\ts}^1$ the generators of $H^0(S^1_t;\Fq)$ and
$H^1(S^1_t;\Fq)$,
\be\label{eq:temporal_triple}
\ts^0 \cup \ts^0 \cup {\ts}^1 \neq 0 \in H^1(S^1_t;\Fq) ,
\qquad
\int_{S^1_t} \ts^0 \cup \ts^0 \cup {\ts}^1 = 1 .
\ee
\end{lemma}

\begin{proof}
The temporal complex carries the constant sheaf, so $\ts^0$ is the unit of the
cohomology ring, $\ts^0 \cup x = x$ for any $x$. Hence
$\ts^0\cup\ts^0\cup{\ts}^1 = {\ts}^1$, and
$\int_{S^1_t}{\ts}^1 = \langle {\ts}^1, S^1_t\rangle = 1$ by the duality between
the generator of $H^1(S^1_t;\Fq)$ and the fundamental $1$-cycle of $S^1_t$.
\end{proof}

\begin{lemma}[Factorization of the cup product]\label{lemma:factorization}
With the decomposition \eqref{eq:cocycle_decomposition} and the product cycle
\eqref{eq:eta_product}, the triple cup product factorizes over the tensor
factors,
\begin{align}\label{eq:cup_factorization}
\red{\tilde{\alpha}^1_{\vc{i_1}}} \cup \blue{\tilde{\beta}^1_{\vc{i_2}}} \cup
\green{\tilde{\gamma}^1_{\vc{i_3}}}
=& \big(\red{\as^1_{i_1}} \cup \blue{\bs^0_{i_2}} \cup
\green{\bs^0_{i_3}}\big) \otimes \big(\red{\bs^0_{j_1}} \cup \blue{\as^1_{j_2}} \cup
\green{\bs^0_{j_3}}\big) \cr
&\otimes \big(\ts^0 \cup \ts^0 \cup {\ts}^1\big) ,
\end{align}
the first two factors being $1$-cocycles of the tensor sheaf on the corresponding
graph, and the invariant factorizes into a product of $\Fq$-valued pairings,
\begin{align}\label{eq:invariant_factorization}
&\int_{\tilde{\eta}_3} \red{\tilde{\alpha}^1_{\vc{i_1}}} \cup
\blue{\tilde{\beta}^1_{\vc{i_2}}} \cup \green{\tilde{\gamma}^1_{\vc{i_3}}} \cr
=&\ \int_\zeta \red{\as^1_{i_1}} \cup \blue{\bs^0_{i_2}} \cup
\green{\bs^0_{i_3}} 
\cdot \int_{\zeta'} \red{\bs^0_{j_1}} \cup \blue{\as^1_{j_2}}
\cup \green{\bs^0_{j_3}}  \cr
&\cdot \int_{S^1_t} \ts^0 \cup \ts^0 \cup {\ts}^1, \cr
\end{align}
where the last factor $\int_{S^1_t} \ts^0 \cup \ts^0 \cup {\ts}^1=1$ using Lemma~\ref{lemma:temporal_factor}. Equivalently, the
spacetime invariant equals the spatial one,
\begin{align}\label{eq:spacetime_equals_spatial}
&\int_{\tilde{\eta}_3} \red{\tilde{\alpha}^1_{\vc{i_1}}} \cup
\blue{\tilde{\beta}^1_{\vc{i_2}}} \cup \green{\tilde{\gamma}^1_{\vc{i_3}}} 
= \int_{\eta_2} \red{\alpha^1_{\vc{i_1}}} \cup
\blue{\beta^1_{\vc{i_2}}} \cup \green{\gamma^0_{\vc{i_3}}} .
\end{align}
\end{lemma}

\begin{proof}
Equation \eqref{eq:cup_factorization} is the statement that the cup product on a
tensor product of cochain complexes is the tensor product of the cup products on
the factors, applied twice using associativity
(Proposition~\ref{prop:cup-properties}(b)); the degrees work out because each
factor receives $1+0+0=1$. Equation \eqref{eq:invariant_factorization} then
follows because the canonical pairing of Definition~\ref{def:pairing-sheaf} on a
tensor product of complexes is the product of the pairings on the factors, and
both the cochain \eqref{eq:cup_factorization} and the chain
\eqref{eq:eta_product} are product elements. Equation
\eqref{eq:spacetime_equals_spatial} follows by regrouping the first two pairings
using $\eta_2=\zeta\otimes\zeta'$ and
Eq.~\eqref{eq:spatial_parts}.
\end{proof}

\begin{lemma}[Necessary condition]\label{lemma:necessary_condition}
The invariant
$\int_{\tilde{\eta}_3}\red{\tilde{\alpha}^1_{\vc{i_1}}}\cup
\blue{\tilde{\beta}^1_{\vc{i_2}}}\cup\green{\tilde{\gamma}^1_{\vc{i_3}}}$ is
non-zero only if the tensor-sheaf $1$-cycles on \emph{both} spatial factors are
non-zero,
\begin{align}\label{eq:necessary_condition}
\zeta \;\in\;&  H_1\big(G_x,\ \Fc^{\rd}_x \otimes \Fc^{\bl}_x \otimes
\Fc^{\gr}_x\big) \setminus \{0\} , \cr
\zeta' \;\in\;&  H_1\big(G_y,\ \Fc^{\rd}_y \otimes \Fc^{\bl}_y \otimes
\Fc^{\gr}_y\big) \setminus \{0\} .
\end{align}
In particular, if the tensor sheaf on either spatial factor admits no non-zero
$1$-cycle, then the invariant vanishes identically and the twisted theory is
trivial. No condition arises from the temporal factor.
\end{lemma}

\begin{proof}
By Eq.~\eqref{eq:invariant_factorization} the invariant is a product of two
scalars in $\Fq$, hence non-zero only if each is non-zero; a pairing against a
vanishing cycle vanishes. The temporal factor contributes $1$ by
Lemma~\ref{lemma:temporal_factor}.
\end{proof}

This reduces the problem to a question internal to a \emph{single} classical sheaf
code: when does the tensor sheaf
$\Fc^{\rd}\otimes\Fc^{\bl}\otimes\Fc^{\gr}$ on a graph $G$ support a non-zero
$1$-cycle? We answer this next.

\subsection{The local cycle condition}
\label{sec:local_cycle_condition}

Fix one spatial factor and suppress the factor label $\nu$. A $1$-chain of the
tensor sheaf is then
\be\label{eq:tensor_1chain}
\zeta = \sum_{e \in E} z_e \cdot e
\ \in\ C_1\big(G, \Fc^{\rd}\otimes\Fc^{\bl}\otimes\Fc^{\gr}\big) ,
\qquad z_e \in \Fq ,
\ee
where we used that the stalk on an edge is
$\Fq \otimes \Fq \otimes \Fq \cong \Fq$ by Eq.~\eqref{eq:stalk_Fq}, so that a
single scalar $z_e$ suffices per edge. The vertex stalks, by contrast, are the
tensor products $\Fq^{m^{\rd}(v)} \otimes \Fq^{m^{\bl}(v)} \otimes
\Fq^{m^{\gr}(v)}$, and the restriction map of the tensor sheaf is the tensor
product of the three restriction maps.

\begin{lemma}[Local cycle condition]\label{lemma:local_cycle}
The $1$-chain \eqref{eq:tensor_1chain} is a cycle, $\partial_1 \zeta = 0$, if and
only if for every vertex $v \in V$
\be\label{eq:local_cycle_condition}
\ \sum_{e \in E(v)} z_e\; \lv^{\rd}_{e} \otimes \lv^{\bl}_{e}
\otimes \lv^{\gr}_{e} \;=\; 0 \ 
\ee
in $\Fq^{m^{\rd}(v)} \otimes \Fq^{m^{\bl}(v)} \otimes \Fq^{m^{\gr}(v)}$.
Equivalently, in components, for every vertex $v$ and every triple of indices
$(i,j,l)$,
\be\label{eq:local_cycle_components}
\sum_{e \in E(v)} z_e\,
\mathsf{H}^{\rd}_{i,e}\,\mathsf{H}^{\bl}_{j,e}\,\mathsf{H}^{\gr}_{l,e}
= \Big\langle \vc{z}\big|_{E(v)},\
\hv^{\rd}_{i} * \hv^{\bl}_{j} * \hv^{\gr}_{l} \Big\rangle = 0 ,
\ee
where $\hv^{\rd}_{i}$, $\hv^{\bl}_{j}$ and $\hv^{\gr}_{l}$ denote the
$i^\text{th}$, $j^\text{th}$ and $l^\text{th}$ rows of the local parity-check
matrices $\mathsf{H}^{\rd}_{v}$, $\mathsf{H}^{\bl}_{v}$, $\mathsf{H}^{\gr}_{v}$,
and $*$ is the Schur (component-wise) product.
\end{lemma}

\begin{proof}
By Eq.~\eqref{eq:boundary-sheaf} the boundary of $\zeta$ evaluated at a vertex $v$
is 
\be
(\partial_1\zeta)(v) = \sum_{e \in E(v)}
\big(\bar\Fc^{\rd}\otimes\bar\Fc^{\bl}\otimes\bar\Fc^{\gr}\big)_{v
\shortleftarrow e}(z_e),
\ee
and the restriction map of a tensor sheaf is the tensor
product of the factor restriction maps
[Eq.~\eqref{eq:restriction_map_column}], giving
Eq.~\eqref{eq:local_cycle_condition}. Taking the $(i,j,l)$ component of a
tensor product of vectors multiplies the corresponding components, so
\be
\big[\lv^{\rd}_e \otimes \lv^{\bl}_e \otimes \lv^{\gr}_e\big]_{ijl}
= \mathsf{H}^{\rd}_{i,e}\mathsf{H}^{\bl}_{j,e}\mathsf{H}^{\gr}_{l,e},
\ee
using that
$\lv_e$ is the $e^\text{th}$ column of $\mathsf{H}_v$, i.e.\
$[\lv_e]_i = \mathsf{H}_{i,e}$. Collecting the three factors into a Schur product
of rows gives Eq.~\eqref{eq:local_cycle_components}.
\end{proof}

\begin{remark}\label{rem:restrictive}
For a \emph{single} sheaf the cycle condition at a vertex is
$\sum_{e\in E(v)} z_e \lv_e = 0$, i.e.\ $m(v)$ scalar equations, which is exactly
the local codeword condition \eqref{eq:local_code_condition_2} and is satisfied by
a space of dimension $\dim \cC_v$. For the tensor sheaf the condition
\eqref{eq:local_cycle_condition} imposes instead
$m^{\rd}(v)\,m^{\bl}(v)\,m^{\gr}(v)$ scalar equations on the same
$|E(v)|=\Delta$ unknowns $\{z_e\}$. The number of constraints therefore grows as
the product of the three local check numbers while the number of variables stays
fixed, so generic local codes admit no non-zero solution and the invariant
vanishes. Non-triviality of the twisted theory is thus a genuine and rather
restrictive design requirement on the local codes, not something one gets for
free.
\end{remark}

\subsection{The fundamental class and the multiplication property}
\label{sec:fundamental_class_multiplication}

A natural and maximally symmetric choice of cycle is the fundamental class, i.e.\
the all-ones chain.

\begin{lemma}[Fundamental class]\label{lemma:fundamental_class}
Let $\zeta = [G] \equiv \sum_{e \in E} e$, i.e.\ $\vc{z} = \bone$. Then $\zeta$ is
a $1$-cycle of the tensor sheaf if and only if for every vertex $v$ and every
triple of indices $(i,j,l)$,
\be\label{eq:fundamental_class_condition}
\sum_{e \in E(v)}
\Big[\hv^{\rd}_{i} * \hv^{\bl}_{j} * \hv^{\gr}_{l}\Big]_{e} = 0 .
\ee
Equivalently, writing $\cC^{i,\perp}_{v} = \operatorname{rowspan}
\mathsf{H}^{i}_{v}$ for the dual local code of colour $i$ at $v$,
\be\label{eq:triple_orthogonality}
\sum_{e \in E(v)} \big(\vc{a} * \vc{b} * \vc{c}\big)_e = 0
\quad \forall\,
\vc{a} \in \cC^{\rd,\perp}_{v},\
\vc{b} \in \cC^{\bl,\perp}_{v},\
\vc{c} \in \cC^{\gr,\perp}_{v} .
\ee
\end{lemma}

\begin{proof}
Setting $z_e=1$ in Eq.~\eqref{eq:local_cycle_components} gives
Eq.~\eqref{eq:fundamental_class_condition} directly. For the second form, note
that $\hv^{\rd}_{i}$ ranges over the rows of $\mathsf{H}^{\rd}_{v}$ as $i$ varies,
and these rows span $\cC^{\rd,\perp}_{v}$; since
$\vc{a}\mapsto \sum_e (\vc{a}*\vc{b}*\vc{c})_e$ is linear in each argument
separately, the condition on all rows is equivalent to the condition on all
elements of the spans.
\end{proof}

Condition \eqref{eq:triple_orthogonality} is precisely a \emph{multiplication
property} of the three dual local codes, as we now make explicit.

\begin{proposition}[Local cycle condition as a multiplication property]
\label{prop:multiplication_property_equivalence}
The fundamental class $[G]$ is a $1$-cycle of the tensor sheaf
$\Fc^{\rd}\otimes\Fc^{\bl}\otimes\Fc^{\gr}$ if and only if, at every vertex $v$,
\be\label{eq:multiplication_property_local}
\cC^{\rd,\perp}_{v} * \cC^{\bl,\perp}_{v} \ \subseteq\ \cC^{\gr}_{v} ,
\ee
and by the symmetry of Eq.~\eqref{eq:triple_orthogonality} in the three colours,
equivalently
$\cC^{\bl,\perp}_{v} * \cC^{\gr,\perp}_{v} \subseteq \cC^{\rd}_{v}$ or
$\cC^{\rd,\perp}_{v} * \cC^{\gr,\perp}_{v} \subseteq \cC^{\bl}_{v}$.
\end{proposition}

\begin{proof}
For any $\vc{a},\vc{b},\vc{c}\in\Fq^{\Delta}$ one has
$\sum_e(\vc{a}*\vc{b}*\vc{c})_e = \langle \vc{a}*\vc{b},\, \vc{c}\rangle$. Hence
Eq.~\eqref{eq:triple_orthogonality} says that $\vc{a}*\vc{b}$ is orthogonal to
every $\vc{c} \in \cC^{\gr,\perp}_{v}$, i.e.\
$\vc{a}*\vc{b} \in \big(\cC^{\gr,\perp}_{v}\big)^{\perp} = \cC^{\gr}_{v}$, for all
$\vc{a}\in\cC^{\rd,\perp}_{v}$ and $\vc{b}\in\cC^{\bl,\perp}_{v}$. Since the
Schur products $\vc{a}*\vc{b}$ span $\cC^{\rd,\perp}_{v} * \cC^{\bl,\perp}_{v}$ by
definition, this is Eq.~\eqref{eq:multiplication_property_local}. The three
colour-permuted forms are equivalent because
Eq.~\eqref{eq:triple_orthogonality} is symmetric under permuting
$(\vc{a},\vc{b},\vc{c})$.
\end{proof}

\begin{remark}
Proposition~\ref{prop:multiplication_property_equivalence} converts a homological
requirement on the spacetime complex --- existence of a $3$-cycle supporting the
triple cup product --- into a condition on a \emph{single} local code of length
$\Delta=O(1)$. This is what makes the requirement checkable: one need only verify
a multiplication property for the $O(1)$-sized local codes of each spatial factor,
and the corresponding $1$-cycles $\zeta,\zeta'$, hence the spatial
$2$-cycle $\eta_2=\zeta\otimes\zeta'$ and the spacetime
$3$-cycle $\tilde\eta_3=\eta_2\otimes S^1_t$, then exist globally by
Lemma~\ref{lemma:local_cycle}. For a Sipser-Spielman code with a uniform local
code (Definition~\ref{def:sipser-spielman-code}) the condition is a single
statement about $\cC_0$, independent of the vertex.
\end{remark}

\subsection{Multiplication property of classical codes}
\label{sec:multiplication_property}

For completeness we recall the notion invoked above. Throughout, $A,B,C \subseteq
\Fq^{n}$ are linear codes of the same length.

\begin{definition}[Schur product]\label{def:schur_product}
The \emph{Schur} (or component-wise, or Hadamard) product of two vectors
$\vc{a},\vc{b}\in\Fq^{n}$ is
\be\label{eq:schur_vectors}
\vc{a} * \vc{b} = \big(a_1 b_1,\ a_2 b_2,\ \ldots,\ a_n b_n\big) ,
\ee
and the Schur product of two codes is the span of the pairwise products,
\be\label{eq:schur_codes}
A * B = \operatorname{span}_{\Fq}
\big\{ \vc{a} * \vc{b} \ :\ \vc{a}\in A,\ \vc{b}\in B \big\} .
\ee
\end{definition}

\begin{definition}[Multiplication property]\label{def:multiplication_property}
The triple of codes $(A,B,C)$ has the \emph{multiplication property} if
\be\label{eq:multiplication_property}
A * B \ \subseteq\ C .
\ee
\end{definition}

\begin{lemma}[Dual formulation]\label{lemma:multiplication_dual}
$(A,B,C)$ has the multiplication property if and only if
\be\label{eq:multiplication_dual}
\sum_{i=1}^{n} \big(\vc{a} * \vc{b} * \vc{c}\big)_i = 0
\qquad \forall\,
\vc{a}\in A,\ \vc{b}\in B,\ \vc{c}\in C^{\perp} .
\ee
\end{lemma}

\begin{proof}
$\sum_i(\vc{a}*\vc{b}*\vc{c})_i = \langle \vc{a}*\vc{b},\vc{c}\rangle$, so
Eq.~\eqref{eq:multiplication_dual} states that $\vc{a}*\vc{b} \in
(C^{\perp})^{\perp} = C$ for all $\vc{a}\in A$, $\vc{b}\in B$, which is
Eq.~\eqref{eq:multiplication_property} since such products span $A*B$.
\end{proof}

Equation \eqref{eq:multiplication_dual} makes the symmetry of the condition
manifest: it is invariant under permuting the three arguments, so
$A*B\subseteq C$, $A*C^{\perp}\subseteq B^{\perp}$ and
$B*C^{\perp}\subseteq A^{\perp}$ are all the same statement. It is this symmetric
form that appears as the local cycle condition
\eqref{eq:triple_orthogonality}, with
$(A,B,C^{\perp}) = \big(\cC^{\rd,\perp}_{v}, \cC^{\bl,\perp}_{v},
\cC^{\gr,\perp}_{v}\big)$.

\begin{remark}[Examples and remarks]\label{rem:multiplication_examples}
Multiplication properties are restrictive but far from vacuous. The extreme cases
are instructive: for the repetition code $\mathrm{Rep}(n)$ one has
$\mathrm{Rep}*\mathrm{Rep}=\mathrm{Rep}$, so
$(\mathrm{Rep},\mathrm{Rep},\mathrm{Rep})$ has the multiplication property --- this
is why the temporal factor of Eq.~\eqref{eq:spacetime_complex}, being a repetition
code, imposes no condition, consistent with
Lemma~\ref{lemma:temporal_factor}. At the other extreme, for the full space
$A=B=\Fq^{n}$ one has $A*B=\Fq^{n}$ and only $C=\Fq^{n}$ works. More usefully,
Reed--Solomon codes satisfy
$\mathrm{RS}_{k_1} * \mathrm{RS}_{k_2} \subseteq \mathrm{RS}_{k_1+k_2-1}$, since
multiplying polynomials adds degrees, and Reed--Muller codes satisfy the analogous
$\mathrm{RM}(r_1,m)*\mathrm{RM}(r_2,m)\subseteq\mathrm{RM}(r_1+r_2,m)$. Both
families therefore furnish local codes for which the fundamental class is a
tensor-sheaf $1$-cycle, at the cost of a rate penalty: the product code has
strictly larger dimension than its factors, so demanding
$A*B\subseteq C$ forces $C$ to be large, hence $C^{\perp}$ small; this tension is
quantified by the product Singleton bound of
Ref.~\cite{randriambololona2013upper} and its equality cases
\cite{mirandola2015critical}, and for random codes the square $A*A$ is generically
as large as possible \cite{cascudo2015squares}. A systematic account of products
and powers of linear codes under componentwise multiplication is given in
Ref.~\cite{randriambololona2015products}. Multiplication properties were
introduced in cryptography, where they underlie multiplicative linear secret
sharing and secure multi-party computation \cite{cramer2000general}; in the
quantum-code setting the same condition appears as the triorthogonality of
magic-state distillation codes \cite{bravyi2012magic} and, closest to the present
construction, in recent qLDPC codes with transversal non-Clifford gates built from
products of algebraic codes and from sheaves
\cite{golowich2024quantum, lin2024transversal}.
\end{remark}

We close by summarizing the logical structure of this section. Non-triviality of
the cohomology invariant in Eq.~\eqref{eq:path_inegral_cup} requires, by
Lemmas~\ref{lemma:temporal_factor}--\ref{lemma:necessary_condition}, a non-zero
tensor-sheaf $1$-cycle on each of the two \emph{spatial} tensor factors of the
hypergraph product, the temporal factor being automatically non-trivial; by
Lemma~\ref{lemma:local_cycle} the existence of such a cycle is equivalent to the
local cycle condition \eqref{eq:local_cycle_condition}, a finite linear system at
each vertex; and by
Proposition~\ref{prop:multiplication_property_equivalence} the natural solution
$\zeta=[G_x]$ and $\zeta'=[G_y]$ exists precisely when the dual local
codes satisfy the multiplication
property \eqref{eq:multiplication_property_local}. Choosing local codes with a
multiplication property is therefore the concrete design principle for obtaining a
non-trivial twist, and hence a genuinely non-Abelian qLDPC code, from a
hypergraph-product construction.

\section{Instantiation: Twisted hypergraph-product code with constant-rate and non-trivial cohomology invariants from arbitrary input good classical codes}
\label{sec:code_instantiation}

\subsection{Code construction}

In the following construction, we can start from an arbitrary 2D HGP code as the tensor product of arbitrary input good classical codes $\C$ and $\C'$ and use the generalized sheaf formalism in Sec.~\ref{sec:subdivision} to subdivde the HGP code into generalized quantum sheaf codes with non-uniform local codes.   We then design another two copies of 2D HGP codes and then gauge (twist) them into the corresponding non-Abelian 2D HGP codes.  

We start with the following constructions of the underlying 2D sheaf complexes for both the untwisted and twisted 2D HGP codes:

\begin{construction}\label{con:2d-hgp}(\emph{2D subdivided HGP sheaf complexes}) Consider an arbitrary input 2D hypergraph-product code $HGP(\bar\C, {\bar{\C}}'^\top)$ with $\bar\C= \ker (\bar\Hs)$ and $\bar\C'= \ker (\bar\Hs')$ being arbitrary asymptotically good classical codes, with ${\bar\C}'^\top = \ker {\bar\Hs}'^\top$ being the transposed code of $\bar\C'$. We then subdivide the classical code $\bar\C$ and its transposed code $\bar\C^\top=\ker \bar\Hs^\top$  into sheaf complexes $X$ and $X^*$ respectively, and similarly  $\bar\C'$ and ${\bar\C}'^\top = \ker {\bar\Hs}'^\top$ into $X'$ and ${X'}^*$ respectively, with the local codes having the alternating pattern of single-parity check (SPC) and repetition (Rep) codes   [see Fig.~\ref{fig:subdivision}]. 
Let $G_x, G_y$ be the graphs for each subdivided classical sheaf codes in the two factors of the product construction. We then consider the following three copies of untwisted subdivided hypergraph product codes and there associated subdivided sheaf complexes:
\begin{align}\label{eq:three_copies}
     \mathcal{A}^{\rd} =& X \otimes {X'}^* \equiv  C^\bullet\big(G_x,\Fc_x\big) \otimes
C^\bullet\big(G_y,\Fc_y^{\perp}\big) \cr  
     \mathcal{A}^{\bl} =& X_R \otimes X' \equiv C^\bullet\big(G_x,\Fc^0\big) \otimes
C^\bullet\big(G_y,\Fc_y\big) \cr
     \mathcal{A}^{\gr} =& X^* \otimes X' \equiv C^\bullet\big(G_x,\Fc_x^{\perp}\big) \otimes
C^\bullet\big(G_y,\Fc^0\big), \cr
\end{align}
where $\Fc_\nu^{\perp}$ is the dual sheaf of $\Fc_\nu$ ($\nu=x,y$), meaning their local codes on each vertex are dual to each other (i.e., SPC $= \text{Rep}^{\perp}$), and $\Fc^0= \FF_2$ is the constant sheaf corresponding to the usual single-parity-check (SPC) local code on each vertex which is dual to the repetition code (Rep).   Note that both the subdivided sheaf complexes $X$ obtained from $\bar\C=\ker \bar\Hs$ and $X^*$ from $\bar\C^\top=\ker \bar\Hs^\top$ corresponds to the same graph $G_x$ [see Fig.~\ref{fig:subdivision}], and the same graph $G_x$ is used for $X_R=C^\bullet\big(G_x,\Fc^0\big)$ where all the local codes are chosen to be SPC; similarly the three sheaf complexes are defined on the same graph $G_y$ in $y$-factor.   The three different HGP codes are hence all defined on the same square complex: 
\be
\L =  G_x \otimes G_y,
\ee
while equipped with different sheaf structure specified above.  The construction includes the special case that the graphs and sheafs in the two classical code factor are identical $G_x=G_y$, $\Fc_x=\Fc_y=\Fc$, and $X=X'$, i.e., 
\begin{align}\label{eq:three_copies_identical_factors}
     \mathcal{A}^{\rd} =& X \otimes X^* \equiv C^\bullet\big(G_x,\Fc\big) \otimes
C^\bullet\big(G_y,\Fc^{\perp}\big) \cr  
     \mathcal{A}^{\bl} =&   X^* \otimes X_R \equiv C^\bullet\big(G_x,\FF_2\big) \otimes
C^\bullet\big(G_y, \Fc\big) \cr
     \mathcal{A}^{\gr} =&  X_R \otimes X \equiv  C^\bullet\big(G_x, \Fc^{\perp} \big) \otimes
C^\bullet\big(G_y,\FF_2\big). \cr
\end{align}

\label{construction}
\end{construction}

Equivalently, we can also re-write the above construction in terms of  hypergraph-product of  classical sheaf codes of the $x$- and $y$-factors  as: 
\begin{align}
\C^{\rd} =& \text{HGP}(\T^{\rd} (G_x, \{\cC^x_v\}_{v \in V_x}),  \T^{\rd} (G_y, \{\cC^{y\perp}_v\}_{v \in V_x}))   \cr
\C^{\bl} =& \text{HGP}(\T^{\bl} (G_x,  \{\text{Rep}^{\perp}\}),  \T^{\bl} (G_y, \{\cC^y_v\}_{v \in V_y})  \cr
\C^{\gr} =& \text{HGP}(\T^{\gr} (G_x, \{\cC^{x\perp}_v\}_{v \in V_x} ),  \T^{\gr} (G_y, \{\text{Rep}^{\perp}\})).   \cr
\end{align}
We note that in the general sheaf code the local codes $\C^x_v$, $\C^y_v$ and their duals can be different on different vertices $v$.
Again in the special case where the two graphs are identical $G_x=G_y=G=(V, E)$,  we have $\C^x_v=\C^y_v=\C_v$, and the above setup can be simplified as
\begin{align}
\C^{\rd} =& \text{HGP}(\T^{\rd} (G, \{\cC_v\}_{v \in V}),  \T^{\rd} (G, \{\cC^{\perp}_v\}_{v \in V}))   \cr
\C^{\bl} =& \text{HGP}(\T^{\bl} (G,  \{\text{Rep}^{\perp}\}),  \T^{\bl} (G, \{\cC_v\}_{v \in V})  \cr
\C^{\gr} =& \text{HGP}(\T^{\gr} (G, \{\cC^{\perp}_v\}_{v \in V} ),  \T^{\gr} (G, \{\text{Rep}^{\perp}\})).   \cr
\end{align}

We note that Construction \ref{con:2d-hgp} has not only defined the families of untwisted 2D HGP codes, but also the corresponding 2D sheaf complexes where the twisted 2D HGP codes are supported.   In particular, with the corresponding sheaf structure $\Fc_\nu, \Fc_\nu^{\perp}$ and $\Fc^0=\FF_2$, we can hence  precisely define the twisted (non-Abelian) 2D HGP codes via the general Clifford stabilizer construction in Eq.~\eqref{eq:stabilizer_summary}, and with the corresponding path-integral description given in Eq.~\eqref{eq:path_inegral_cup}. 

Now if we pick the spatial 2-cycle as the fundamental class of the spatial square complex $\L$, i.e.,
\be
\eta_2=\zeta \otimes \zeta'=[\L]=[G_x] \otimes [G_y],
\ee
meaning each of its classical code factor is the fundamental class on the graph, i.e., $\zeta=[G_x]$ and $\zeta'=[G_y]$, we know that it clearly satisfies the \emph{local cycle condition} in Lemma \ref{lemma:local_cycle} and more specifically Lemma \ref{lemma:fundamental_class} according to the multiplication property in Proposition \ref{prop:multiplication_property_equivalence}, i.e., $\cC^{\rd,\perp}_{v} * \cC^{\bl,\perp}_{v} \ \subseteq\ \cC^{\gr}_{v}$.
For the $x$-factor, the condition becomes
\be
\cC^{x \perp}_v * \cC^x_v  \subseteq \text{Rep}^{\perp}.
\ee 
As stated before, this  can be proven by re-writing it into Schur product as
\be
\sum_{e \in E(v)} (\vc{a}* \vc{b}*\vc{1})_e= \sum_{e \in E(v)} (\vc{a}* \vc{b})_e=\langle \vc{a}, \vc{b} \rangle=\vc{0}
\ee
with $\vc{a} \in \cC^{x \perp}_v, \  \vc{b} \in \cC^x_v, \ \vc{1} \in \text{Rep}$ (representing the all-1 repetition codeword), where the third equality directly follows from the dual code condition. 
The $y$-factor follows similarly.

Finally, we also consider the following alternative but more specific construction which will give the same code parameter scaling: 
\begin{construction}\label{con:2d-hgp-Sipser-Spielman}(\emph{2D HGP from standard Siper-Spielman codes})
Although in the above we have been focused on subdividng arbitrary HGP codes into sheaf codes, we do not have to do the subdivision procedure if the input good classical codes are already the standard Sipser-Spielman code with uniform local code $\C^{\nu}$ defined on a graph $G_\nu$.  In that case, we can still construct three copies of sheaf complexes in exactly the same manner as Eq.~\eqref{eq:three_copies}.  The only difference is that the local codes $\C^{\nu}$ corresponding to the sheaf $\Fc^{\nu}$ are no longer alternating patterns of SPC and Rep, arbitrary  $\C^{\nu}$ suffices as long we pick the dual local code ${\C^{\nu}}^{\perp}$ for the dual sheaf $\Fc_{\nu}^{\perp}$.
\end{construction}

\subsection{Subcomplex symmtery and logical CZ gate structure via Poincar\'e duality of sheaf}

We now consider the 0-form subcomplex symmetries in the untwisted qLDPC codes defined on the product sheaf complex in Eq.~\eqref{eq:three_copies} and the corresponding logical CZ gate structure.

\subsubsection{From Poincar\'e duality to the non-degenerate pairing}\label{ss:PD-pairing}

We start by restating the Poincar\'e duality result proven in Ref.~\cite{li2025poincar}.

\begin{theorem}[{Poincar\'e duality on sheaf complexes~\cite[Theorem 3.17]{li2025poincar}}]
\label{thm:poincare-duality-sheaf-codes}
    Suppose $X$ is a $t$-dimensional sparse cell complex with a locally acyclic sheaf $\Fc$, then we have a duality between $\Fc$ and $\Fc^\perp$. To be precise, for any $0 \leq i \leq t$, there is an isomorphism:
   \begin{equation}\label{eq:PD-t1}
  \text{PD}: H^i(X, \mathcal{F}^\perp) \xrightarrow{\;\cong\;} H_{t-i}(X, \mathcal{F}).
\end{equation}
\end{theorem}

% Now that we have shown that our product sheaf from taking the 3-fold tensor product of three subdivided sheaf codes satisfies the local acyclicity property for the product sheaf of the 3-dimensional cubical complex, we can apply the Poincar\'e duality on sheaf complexes from Theorem~\ref{thm:poincare-duality-sheaf-codes}.
% Because of the subdivision performed to our cochain complexes before we took the 3-fold tensor product of them, $X$ is a cell complex. This admits an explicit isomorphism between the cocycles and cycles of $\widetilde{X}$ as defined in the following theorem restated from Ref.~\cite{li2025poincar}.

\begin{theorem}[{Cap product induces Poincar\'e duality map~\cite[Theorem~4.23]{li2025poincar}}]
\label{thm:cap-product-poincare-duality-map}
    Suppose $X$ is a cell complex that admits a simplicial approximation. Let $[X] \in C_t(X, \Fc^\perp \otimes \Fc)$ be the fundamental class defined by $[X] = \sum_{\sigma \in X(t)} \sigma$, then for each $0 \leq i \leq t$, there is an isomorphism 
    \begin{equation}
        \mathrm{PD}\,:\,H^i(X, \Fc^\perp) \to H_{t-i}(X, \Fc),
    \end{equation}
    given by $\mathrm{PD}[\mathbf{\alpha}] = [\mathbf{\alpha}] \frown [X].$
\end{theorem}
% \guanyu{(Add cap product definition in the preliminary)}

In the following we will use Kunneth theorem to factorize the cup-product structure into each tensor factor corresponding to the subdivided classical sheaf codes defined on the subdivided graph $G$ according to Lemma \ref{lemma:factorization}.
For the above two theorems to apply, we have to show that the sheaf defined on the subdivided graph $G$ with is \emph{locally acyclic}, see Appendix \ref{app:acyclic} for the proof.

Next, we introduce a non-degenerate pairing on the subdivided classical sheaf codes that we have constructed. The required bilinear form is obtained simply by composing Poincar\'e duality with the canonical pairing between cohomology and homology.

The canonical pairing induces a bilinear map:
\begin{align}
& \langle \cdot, \cdot \rangle\,:\, H^1(G, \Fc) \times H_1(G, \Fc) \to \FF_2,\qquad \cr
& \langle [\as^1], [\eta_1] \rangle = \sum_{e \in E} \as^1(e)\eta_1(e),
\end{align}
and this pairing is non-degenerate. Note that this is merely Proposition~\ref{prop:nondegen-sheaf} specialized to the case where $i = 1$.
For $t=1$, together with the canonical pairing, the Poincar\'e duality isomorphism
\begin{align}
& \mathrm{PD}: H^0(G,\mathcal{F}^{\perp}) \xrightarrow{\;\sim\;} H_1(G,\mathcal{F}), \cr  
& \mathrm{PD}[\bs^0]=[\bs^0 \frown [G]],
\end{align}
induces a bilinear form
\begin{align}\label{eq:P-def}
&\mathcal{P}:  H^1(G, \mathcal{F}) \times H^0(G, \mathcal{F}^\perp) \to \FF_2,
\end{align}
such that we have
\begin{align}
\mathcal{P}([\as^1],[\bs^0]) &:= \langle [\as^1], \mathrm{PD}[\bs^0]\rangle \nonumber\\
&= \langle [\as^1], [\bs^0 \frown [G]]\rangle \nonumber\\
&= \langle [\bs^0] \cup [\as^1], [G]\rangle. \nonumber
\end{align}
Most importantly, $\mathcal{P}$ is non-degenerate.

We then have the following corollary:
\begin{corollary}[Kronecker duality]\label{cor:kronecker-dual}
The  bases $\{\as^1_i\}$ of $H^1(G, \Fc)$ and its $\mathcal{P}$-dual basis $\{\bs^0_j\}$ of $H_0(G, \Fc^{\perp})$
satisfy $\mathcal{P}([\as^1],[\bs^0]) = \delta_{ij}$.
\end{corollary}

\begin{proof}
We have the following equality
\be
\mathcal{P}([\as^1_i], [\bs^0_j]) = \langle [\as^1_i], \text{PD}[\bs^0_j]\rangle = \langle [\as^1_i], [\eta_1^j]\rangle= \delta_{ij},
\ee
where $\eta^j_1=\text{PD}[\bs^0_j] \in H_1(G, \Fc)$ and the last equality comes from the non-degenerate cyclce-cocycle pairing according to Proposition~\ref{prop:nondegen-sheaf}.
\end{proof}

\subsubsection{Logical CZ structure}

Now we choose a subset of spatial basis cocycles $\{\red{\alpha^{1}_{\vc{i_1}}}\}$, $\{\blue{\beta^{1}_{\vc{i_2}}}\}$ and $\{\green{\gamma^{0}_{\vc{i_3}}}\}$ according to the Kunneth decomposition as in Eq.~\eqref{eq:cocycle_decomposition} and summarized below:
\begin{align}\label{eq:spatial_cocycle_decomposition}
\red{{\alpha}^{1}_{\vc{i_1}}} &= \red{\as^1_{i_1}} \otimes
\red{\bs^0_{j_1}}, \cr
\blue{{\beta}^{1}_{\vc{i_2}}} &= \blue{\cs^0} \otimes
\blue{\as^1_{j_2}}, \cr
\green{{\gamma}^{0}_{\vc{i_3}}} &= \green{\bs^0_{i_3}} \otimes
\green{\cs'^0}.
\end{align}
% \begin{align}\label{eq:spatial_cocycle_decomposition}
% \red{\alpha^{1}_{\vc{i_1}}} &= \red{\as^1_{i_1}} \otimes
% \red{\bs^0_{j_1}}, \cr
% \blue{\beta^{1}_{\vc{i_2}}} &= \blue{\bs^0_{i_2}} \otimes
% \blue{\as^1_{j_2}}, \cr
% \green{\gamma^{0}_{\vc{i_3}}} &= \green{\bs^0_{i_3}} \otimes
% \green{\bs^0_{j_3}}.
% \end{align}
This specific basis is chosen with the interest of gauging the transversal CZ gates between the $\rd$ and $\bl$ copies.   Here $\blue{\cs^0}$ is the unique 0-cocycle in the classical sheaf code 1-complex $C^\bullet\big(G_x, \FF_2\big)$ with the constant sheaf $\FF_2$ or equivalently  $\T^{\bl} (G_x, \text{Rep}^{\perp})\equiv \T^{\bl} (G_x, \text{SPC})$ with the trivial single-parity check (SPC) local code, which is also essentially just the unique 0-cocycle on the graph $G_x$.  Similarly,  $\green{\cs'^0}$ is the unique 0-cocycle on the sheaf complex $C^\bullet\big(G_y, \FF_2\big)$ or equivalently $\T^{\gr} (G_y, \text{Rep}^{\perp})\equiv \T^{\gr} (G_y, \text{SPC})$, and essentialy the graph unique 0-cocycle on the graph $G_y$.  

We choose this specific subset of basis cocycles to consider the transversal logical CZ gates between the $\rd$ and $\bl$ copies. According to Eq.~\eqref{eq:logical_CZ_action},  we have
\be\label{eq:logical_CZ_basis}
\widetilde{\text{CZ}}^{\rd,\bl}(\green{\gamma^{0}_{\vc{i_3}}}) \;\widehat{=}\;
\prod_{\red{\vc{i_1}},\blue{\vc{i_2}}}
\big[\lo{\text{CZ}}^{\rd,\bl}_{\red{\vc{i_1}},\blue{\vc{i_2}}}
\big]^{\Lambda
(\red{\vc{i_1}}, \blue{\vc{i_2}}, \green{\vc{i_3}})}.
\ee
Here, the logical CZ gate structure is completely determined by 
\begin{align}
\Lambda
(\red{\vc{i_1}}, \blue{\vc{i_2}}, \green{\vc{i_3}})=&\int_{[\L]} \red{\alpha^1_{\vc{i_1}}} \cup \blue{\beta^1_{\vc{i_2}}} \cup \green{\gamma^0_{\vc{i_3}}} \cr
=& \int_{[\L] \frown \green{\gamma^0_{\vc{i_3}}}} \red{\alpha^1_{\vc{i_1}}} \cup \blue{\beta^1_{\vc{i_2}}} \equiv \int_{\green{\gamma^*_{2,\vc{i_3}}}} \red{\alpha^1_{\vc{i_1}}} \cup \blue{\beta^1_{\vc{i_2}}}, \cr
\end{align}
where we have used the fundamental class in the tensor sheaf $[\L] \in H_2(\L, \Fc^{\rd} \otimes  \Fc^{\bl} \otimes  \Fc^{\gr} ) $ and the sheaf 2-cycle $\green{\gamma^*_{2,\vc{i_3}}} $$\equiv$$  [\L] \frown \green{\gamma^0_{\vc{i_3}}} \in  H_2(\L, \Fc^{\rd} \otimes  \Fc^{\bl})$.  

We then apply the K\"unneth decomposition [see Lemma \ref{lemma:factorization}] of the triple cup product sum:
\begin{align}
& \int_{[\L]} \red{\alpha^1_{\vc{i_1}}} \cup \blue{\beta^1_{\vc{i_2}}} \cup \green{\gamma^0_{\vc{i_3}}}  \cr
=& \int_{[G_x]} \red{\as^1_{i_1}} \cup \blue{\cs^0} \cup \green{\bs^0_{i_3}} \cdot \int_{[G_y]} \red{\bs^0_{j_1}} \cup \blue{\as^1_{j_2}} \cup \green{\cs'^0}. \cr
\end{align}

\begin{figure}[t]
    \centering   \includegraphics[width=0.8\linewidth]{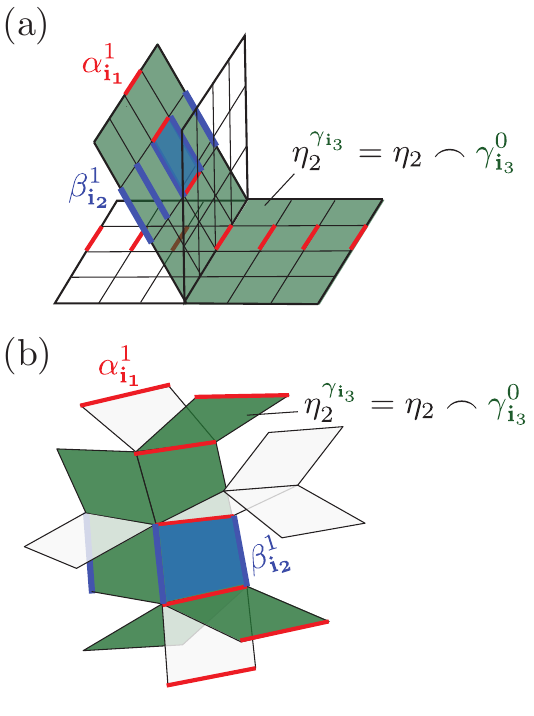}
    \caption{Addressability of the transversal logical
CZ gates. The capped codimension-$0$ subcomplex
$\eta_2^{\gamma_{i_3}}=\eta_2\frown\gamma^0_{\vc{i_3}}$ (green) supports the
symmetry operator labelled by $\green{\vc{i_3}}$; drawn on it are a \rd\
$1$-cocycle $\red{\alpha^1_{\vc{i_1}}}$ and a \bl\ $1$-cocycle
$\blue{\beta^1_{\vc{i_2}}}$ representing the logical operators of the two copies.
The logical CZ structure is determined by  the triple intersection
$\Lambda(\red{\vc{i_1}},\blue{\vc{i_2}},\green{\vc{i_3}})=
\int_{\eta_2}\red{\alpha^1_{\vc{i_1}}}\cup\blue{\beta^1_{\vc{i_2}}}\cup
\green{\gamma^0_{\vc{i_3}}}$, which in the instantiation of
Sec.~\ref{sec:code_instantiation} is non-zero precisely when the two cocycles meet
in a single face inside the subcomplex (blue square). (a) Illustration with each patch being subdivided. (b) Illustration on a typical 2D square complex. Different labels
$\green{\vc{i_3}}$ select different subcomplexes and hence different, independently
applicable patterns of logical CZ gates.}
\label{fig:subcomplex_illustration_cup}
\end{figure}

Now by using the Poincar\'e duality of the classical sheaf code,  we have the following non-degenerate pairing for the $x$-factor:
\be
\mathcal{P}([\red{\as^1_{i_1}}], [\green{\bs^0_{i_3}}])_x \equiv \int_{[G_x]} \green{\bs^0_{i_3}} \cup \red{\as^1_{i_1}}  = \delta_{\red{i_1}, \green{i_3}}. 
\ee
Since $\blue{\cs^0} = \vc{1}$ is the all-1 vector, which acts as the identity of the cup product,  we have the following triple cup product evaluation 
\be
\int_{[G_x]} \red{\as^1_{i_1}} \cup \blue{\cs^0} \cup \green{\bs^0_{i_3}} = \int_{[G_x]} \green{\bs^0_{i_3}} \cup \red{\as^1_{i_1}} = \delta_{\red{i_1}, \green{i_3}}.
\ee
Similarly, for the $y$-factor, we have 
\be
\mathcal{P}([\blue{\as^1_{j_2}}],[\red{\bs^0_{j_1}}])_y = \int_{[G_y]} \red{\bs^0_{j_1}}  \cup \blue{\as^1_{j_2}} =\delta_{\red{j_1}, \blue{j_2}},
\ee
which leads to the triple cup product sum:
\be
\int_{[G_y]} \red{\bs^0_{j_1}} \cup \blue{\as^1_{j_2}} \cup \green{\cs'^0}  =  \int_{[G_y]} \red{\bs^0_{j_1}}  \cup \blue{\as^1_{j_2}} = \delta_{\red{j_1}, \blue{j_2}}.
\ee
Therefore, the logical CZ gate structure is given by 
\be\label{eq:intersection_structure}
\Lambda
(\red{\vc{i_1}}, \blue{\vc{i_2}}, \green{\vc{i_3}})=\int_{[\L]} \red{\alpha^1_{\vc{i_1}}} \cup \blue{\beta^1_{\vc{i_2}}} \cup \green{\gamma^0_{\vc{i_3}}} = \delta_{\red{i_1}, \green{i_3}}  \delta_{\red{j_1}, \blue{j_2}}.
\ee
The above logical CZ gate structure will determine the magic rate of the magic state fountain protocol in Sec.~\ref{sec:fountain}.

\subsection{Encoding rate and distance}
\label{sec:rate_and_distance}

We now determine the parameters $[[n,k,d]]$ of the codes of
Construction~\ref{con:2d-hgp}, in both the untwisted and the twisted case. The
conclusion is that they are the same as those obtained in
Ref.~\cite{ZhuKobayashiHsin2026} for the corresponding \emph{skeleton} codes, namely
constant rate $k=\Theta(n)$ and subsystem-code distance $d=\Omega(\sqrt{n})$. This
is not a coincidence: the three sheaf complexes of Eq.~\eqref{eq:three_copies} are
precisely subdivided from the skeleton chain complexes of Ref.~\cite{ZhuKobayashiHsin2026}.
The difference is only one of route: there, the classical code was first mapped to
a high-dimensional Poincar\'e CW complex and then truncated back to a $2$D skeleton;
here, the classical code is subdivided directly into a sheaf code
(Sec.~\ref{sec:subdivision}) and the sheaf carries the local structure that the
extra dimensions used to carry. The sheaf route has the advantage that the
parameters can be read off from those of the underlying classical code without any
intermediate manifold, which is what we do below.

\subsubsection{Inputs from the classical code}
\label{sec:rate_inputs}

Fix an asymptotically good classical LDPC code $\bar{\C}$ with parameters
$[\bar n,\bar k,\bar d]$,
\be\label{eq:good_code}
\bar{k}=\Theta(\bar n) , \qquad \bar{d}=\Theta(\bar n) ,
\ee
with parity-check matrix $\Hs$, and subdivide it into a sheaf code on
a graph $G$ as in Sec.~\ref{sec:subdivision}, obtaining the sheaf $\Fc$; write
$\Fc^{\perp}$ for the sheaf with the dual local codes $\{\cC^{\perp}_v\}$, and
$\Fc^0=\FF_2$ for the constant sheaf, whose associated classical code is the
repetition code $\text{Rep}$. By Propositions~\ref{Prop:sheaf_homology}
and~\ref{Prop:sheaf_cohomology} the two classical codes attached to a graph factor
are $H_1(G,\Fc)$ and $H^0(G,\Fc)$, and by
Proposition~\ref{prop:subdivision-params} the subdivision preserves the dimension
and does not decrease the distance,
\begin{align}\label{eq:subdiv_params}
&\dim H_1(G,\Fc) = \bar k = \Theta(\bar n) , \cr
&\min\big\{|\eta_1| : 0 \neq [\eta_1] \in H_1(G,\Fc)\big\} \cr
&\qquad \ge\ w_c^{\min}\,\bar d = \Theta(\bar n) .
\end{align}
Poincar\'e duality (Theorem~\ref{thm:poincare-duality-sheaf-codes}) transfers
Eq.~\eqref{eq:subdiv_params} to the $0$-cocycles of the dual sheaf: the map
$\bs^0\mapsto\bs^0\frown[G]$ is an isomorphism $H^0(G,\Fc^{\perp})\to H_1(G,\Fc)$,
and the $1$-chain $\bs^0\frown[G]$ is supported on the edges to which $\bs^0$
restricts non-trivially. Writing $\Delta=O(1)$ for the maximum vertex degree (and
stalk dimension) of $G$, a non-zero $0$-cocycle must therefore be non-zero on at
least $d_1(\Fc)/\Delta$ vertices,
\begin{align}\label{eq:dual_cocycle_weight}
\min\big\{|\bs^0| : 0\neq \bs^0\in H^0(G,\Fc^{\perp})\big\}
&\;\ge\; \frac{w_c^{\min}\,\bar d}{\Delta} \cr
&\;=\; \Theta(\bar n) .
\end{align}
The corresponding statements for the transposed code
$\bar\C^{\top}=\ker\Hs^{\top}$, i.e.\ for
$H^0(G,\Fc)\cong H_1(G,\Fc^{\perp})\cong\bar\C^{\top}$, involve the quantities
$\bar k^{\top}$ and $\bar d^{\top}$ left open in
Proposition~\ref{prop:subdivision-params}. Throughout this subsection we assume
that the transposed code is either \emph{trivial}, i.e.\ $\Hs$ has full row rank
so that $\bar\C^{\top}=0$ and $H^0(G,\Fc)=H_0(G,\Fc)=0$, or \emph{likewise good},
$\bar k^{\top}=\Theta(\bar n)$ and $\bar d^{\top}=\Theta(\bar n)$, in which case
Eq.~\eqref{eq:dual_cocycle_weight} also holds with $\Fc$ and $\Fc^{\perp}$
interchanged and $w_c^{\min}\bar d$ replaced by $w_r^{\min}\bar d^{\top}$. The
distance argument below covers both cases; the only difference is whether the
second K\"unneth summand of the \rd\ copy in Eq.~\eqref{eq:kunneth_H1} is present.

The constant-sheaf factor behaves differently from the two good factors, and this
asymmetry drives everything below. For the repetition code on a connected graph,
\begin{align}\label{eq:rep_factor}
\dim H^0(G,\Fc^0) &= 1 , \qquad |\cs^0| = \Theta(\bar n) , \cr
\dim H_0(G,\Fc^0) &= 1 , \qquad |\cs_0| = 1 ,
\end{align}
where $\cs^0$ is the all-ones $0$-cocycle, supported on \emph{every} vertex, and
$\cs_0$ is a single vertex. So the repetition factor contributes only a single
homology class, but that class is heavy as a cocycle and light as a cycle.

Finally, since the qubits of each copy sit on the $1$-cells of
$\L=G_x\times G_y$ with $O(1)$-dimensional stalks, and $|E_\nu|=\Theta(\bar n)$,
\be\label{eq:n_scaling}
n = \Theta(\bar n^{2}) , \qquad\text{equivalently}\qquad
\bar n = \Theta(\sqrt{n}) .
\ee

\subsubsection{Encoding rate}
\label{sec:encoding_rate}

\begin{lemma}[Rate of the untwisted code]\label{lemma:rate_untwisted}
The untwisted $[[n,k,d]]$ code $\C = \C^{\rd}\otimes\C^{\bl}\otimes\C^{\gr}$ of
Construction~\ref{con:2d-hgp} has constant rate, $k=\Theta(n)$.
\end{lemma}

\begin{proof}
By Eq.~\eqref{eq:k_untwisted} the number of logical qubits of each copy is
$k^{i}=\dim H^1(\L,\Fc^{i})$, and by the K\"unneth theorem, applied to
$\L=G_x\times G_y$ (with $G_x=G_y=G$ in this case) exactly as in Sec.~\ref{sec:kunneth_factorization},
\begin{align}\label{eq:kunneth_H1}
H^1\big(\L,\Fc^{i}\big) =&
\Big[H^1\big(G_x,\Fc^{i}_x\big)\otimes H^0\big(G_y,\Fc^{i}_y\big)\Big] \cr
&\oplus \Big[H^0\big(G_x,\Fc^{i}_x\big)\otimes
H^1\big(G_y,\Fc^{i}_y\big)\Big].   \cr
\end{align}
For the \rd\ copy both factors are good, so by
Eqs.~\eqref{eq:subdiv_params} and~\eqref{eq:n_scaling}
\be\label{eq:k_red}
k^{\rd} \;\ge\; \Theta(\bar n)\cdot\Theta(\bar n) \;=\; \Theta(\bar n^{2})
\;=\; \Theta(n) ,
\ee
and since linear dimension is optimal, $k^{\rd}=\Theta(n)$. For the \bl\ and \gr\
copies one factor is the repetition code, which by Eq.~\eqref{eq:rep_factor}
contributes a single class, so the corresponding product in
Eq.~\eqref{eq:kunneth_H1} contributes
\be\label{eq:k_blue}
k^{\bl},\,k^{\gr} \;\ge\; \Theta(\bar n)\cdot 1 \;=\; \Omega(\bar n)
\;=\; \Omega(\sqrt{n}) .
\ee
The total number of logical qubit is hence 
\be
 k=k^{\rd}+k^{\bl}+k^{\gr}=\Theta(n).
\ee
\end{proof}

\begin{remark}\label{rem:asymmetry_rate}
The asymmetry between Eqs.~\eqref{eq:k_red} and~\eqref{eq:k_blue} is the sheaf
counterpart of the observation in Ref.~\cite{ZhuKobayashiHsin2026} that the secondary
copy carries only $\Omega(\sqrt n)$ logical qubits because the zeroth Betti number
of the corresponding complex is $1$. Here it is visible directly:
$\dim H^0(G,\Fc^0)=1$ for the constant sheaf. The \rd\ copy therefore carries all
of the constant-rate logical information, and the \bl\ and \gr\ copies are
sub-extensive; the \gr\ copy in particular is used for the gauging measurement of
the transversal CZ rather than for storage, with all its encoded logicals being gauge qubits in the subsystem code description as will disccused below.
\end{remark}

\begin{lemma}[Rate of the twisted code]\label{lemma:rate_twisted}
The twisted code $\tilde{\C}$ with parameters
$[[\tilde n = \Theta(n),\tilde k,\tilde d]]$ has constant rate,
$\tilde k=\Theta(\tilde n)=\Theta(n)$.
\end{lemma}

\begin{proof}
The twisted code has the same physical qubits, so $\tilde n = n$. Its logical
qubits are labelled by the same $Z$-type operators
$\widetilde{Z}^{i}_{\eta_1}$ of Eq.~\eqref{eq:logical_Z_sheaf} as in the untwisted
code, so the only effect of the twist on the count is that some untwisted logical
states are projected out by the code-space constraints of
Eqs.~\eqref{eq:constraint_operator_sheaf}
and~\eqref{eq:constraint_operator_sheaf2}. Equivalently,
by Eq.~\eqref{eq:charge_is_CZ_sheaf} each independent charge parity operator
$\mathsf{C}^{i}$ that is not already trivial imposes one $\FF_2$ constraint.

By Eq.~\eqref{eq:cocycle_decomposition} the constraints labelled by
$\red{\alpha^1_{\vc{i_1}}}$ and $\blue{\beta^1_{\vc{i_2}}}$ are automatically
satisfied, because the corresponding transversal CZ operators
$\widetilde{\text{CZ}}^{\bl,\gr}$ and $\widetilde{\text{CZ}}^{\rd,\gr}$ act
trivially on the code space; only the constraints labelled by the \gr-type
$0$-cocycles $\green{\gamma^0_{\vc{i_3}}}=\green{\cs^0} \otimes
\green{\bs^0_{j_3}}$ survive. One of the two tensor factors of
$\green{\gamma^0_{\vc{i_3}}}$ is the constant sheaf, contributing a single class by
Eq.~\eqref{eq:rep_factor}, so
\be\label{eq:number_of_constraints}
\big|\{\green{\gamma^0_{\vc{i_3}}}\}\big| = 1\cdot\Theta(\bar n) + \Theta(\bar n)\cdot 1
= \Theta(\bar n) = \Theta(\sqrt{n}) ,
\ee
and therefore
\be\label{eq:k_twisted}
\tilde k \;=\; k - \Theta(\sqrt{n}) \;=\; \Theta(n) \;=\; \Theta(\tilde n) ,
\ee
using Lemma~\ref{lemma:rate_untwisted}. The twist removes only a sub-extensive
number of logical qubits and the constant rate survives.
\end{proof}

\subsubsection{Subsystem-code distance}
\label{sec:distance}

The subdivided sheaf complexes contain \emph{spurious} short cycles, inherited
from the graph cycles of $G$ carried by the constant sheaf $\FF_2$ rather than
from the input classical code. In the convention of Eq.~\eqref{eq:three_copies}
they live in the \bl\ and \gr\ copies:
\begin{align}\label{eq:spurious-cycle}
\blue{\beta^{\rm sh}_{1}} &= \blue{\fs_1}\otimes\blue{\bs_0}, \cr
&\quad \blue{\fs_1}\in H_1(G_x,\FF_2),\ \ \blue{\bs_0}\in H_0(G_y,\Fc_y), \cr
\green{\gamma^{\rm sh}_{1}} &= \green{\bs_0}\otimes\green{\fs_1}, \cr
&\quad \green{\bs_0}\in H_0(G_x,\Fc^{\perp}_x),\ \ \green{\fs_1}\in H_1(G_y,\FF_2).
\end{align}
Here $\fs_1$ is a graph cycle, whose weight is set by the girth of $G$ and can be
$O(1)$ (typically $O(\log n)$ for a random graph), while $\bs_0$ can be a single
vertex; the spurious cycles therefore have weight $O(1)$. Their conjugate
cocycles, in $H^1(G_x,\FF_2)\otimes H^0(G_y,\Fc_y)$ and
$H^0(G_x,\Fc^{\perp}_x)\otimes H^1(G_y,\FF_2)$, are heavy, but a conjugate pair
must be demoted together. When $\Hs$ has full row rank, $H_0(G_y,\Fc_y)=0$ and
the \bl\ copy has no spurious cycles at all, whereas the \gr\ ones are always
present because $H_0(G_x,\Fc^{\perp}_x)\cong H^0(G_x,\Fc^{\perp}_x)^{*}\neq0$.

As in
Refs.~\cite{zhu2025topological,zhu2025transversal, ZhuKobayashiHsin2026} we therefore work with a
\emph{subsystem} code, in which only chosen K\"unneth sectors carry logical
information and the complementary classes are demoted to gauge qubits. The
distance of such a code is a minimum over \emph{all} logical operators that act
non-trivially on the protected register --- not over the weights of chosen basis
representatives, which only give upper bounds. The following lemma, proved in
Appendix~\ref{app:subsystem_distance}, makes this precise.

\begin{lemma}[Subsystem-code distance]\label{lemma:subsystem}
Let a CSS code be defined by three consecutive terms
$C^0\xrightarrow{d^0}C^1\xrightarrow{d^1}C^2$ of a sheaf cochain complex as in
Definition~\ref{def:sheaf-code}, with logical classes $H^1$ ($X$-type) and $H_1$
($Z$-type) paired by Definition~\ref{def:pairing-sheaf}. Choose subspaces
$\mathsf L_X\subseteq H^1$ and $\mathsf L_Z\subseteq H_1$ on which the pairing is
non-degenerate, and declare all logical operators whose classes lie in
\be\label{eq:gauge_classes}
\mathsf G_X:=\mathsf L_Z^{\perp}\subseteq H^1 ,\qquad
\mathsf G_Z:=\mathsf L_X^{\perp}\subseteq H_1
\ee
to be gauge operators. Then:
\begin{enumerate}[label=(\roman*)]
\item $H^1=\mathsf L_X\oplus\mathsf G_X$ and $H_1=\mathsf L_Z\oplus\mathsf G_Z$;
the code space factorizes as $\H_{\rm prot}\otimes\H_{\rm gauge}$ with
$\log_2\dim\H_{\rm prot}=\dim\mathsf L_X$; and a logical Pauli operator
$X(\zeta^1)Z(\eta_1)$ acts trivially on $\H_{\rm prot}$ if and only if
$[\zeta^1]\in\mathsf G_X$ and $[\eta_1]\in\mathsf G_Z$.
\item The subsystem-code distance, i.e.\ the minimum weight of a logical
operator acting non-trivially on $\H_{\rm prot}$, is
\be\label{eq:subsystem_distance}
d=\min\big(d_X^{\rm prot},\,d_Z^{\rm prot}\big),
\ee
with
\begin{align}\label{eq:dXdZ_prot}
d_X^{\rm prot}&=\min\big\{|\zeta^1| : d^1\zeta^1=0,\ [\zeta^1]\notin\mathsf G_X\big\}, \cr
d_Z^{\rm prot}&=\min\big\{|\eta_1| : \partial_1\eta_1=0,\ [\eta_1]\notin\mathsf G_Z\big\}.
\end{align}
\item For a tensor product of decoupled copies, each with its own choice of
$(\mathsf L_X,\mathsf L_Z)$, the distance is the minimum of the distances of the
copies with $\mathsf L_X\neq0$.
\end{enumerate}
\end{lemma}

A chosen pair of dual bases $\{\eta^1_{k'}\}\subset\mathsf L_X$,
$\{\eta_1^k\}\subset\mathsf L_Z$ with
$\langle\eta^1_{k'},\eta_1^k\rangle=\delta_{k,k'}$ only gives the \emph{upper}
bound $d\le\min(\min_k|\eta_1^k|,\min_{k'}|\eta^1_{k'}|)$: a non-zero linear
combination of protected classes, or a different representative of a single
class, may be lighter. A lower bound must control every class with a non-zero
component in the protected sector at once. For K\"unneth sectors this is achieved
by the contraction lemma of Appendix~\ref{app:subsystem_distance}
(Lemma~\ref{lemma:kunneth_contraction}), which upgrades the product weight bound
of Ref.~\cite{Bravyi:2014bq},
\be\label{eq:tensor_weight_bound}
\min\big\{|u\otimes v|\big\} \;\ge\;
\max\big(\min\{|u|\},\ \min\{|v|\}\big) ,
\ee
from elementary tensors $u\otimes v$ to arbitrary classes.

\begin{lemma}[Distance of the untwisted code]\label{lemma:distance_untwisted}
For the untwisted code $\C=\C^{\rd}\otimes\C^{\bl}\otimes\C^{\gr}$ of
Construction~\ref{con:2d-hgp}, choose the protected sectors
\begin{align}\label{eq:protected_sectors}
\rd:\ &\mathsf L_X^{\rd}=H^1(\L,\Fc^{\rd}),\qquad
\mathsf L_Z^{\rd}=H_1(\L,\Fc^{\rd}), \cr
\bl:\ &\mathsf L_X^{\bl}=H^0(G_x,\FF_2)\otimes H^1(G_y,\Fc_y)
=\mathrm{span}\big\{[\blue{\beta^1_{\vc{i_2}}}]\big\}, \cr
&\mathsf L_Z^{\bl}=H_0(G_x,\FF_2)\otimes H_1(G_y,\Fc_y), \cr
\gr:\ &\mathsf L_X^{\gr}=\mathsf L_Z^{\gr}=0 ,
\end{align}
i.e.\ the \rd\ copy is a subspace code, the \bl\ copy is a subsystem code whose
spurious pair $H^1(G_x,\FF_2)\otimes H^0(G_y,\Fc_y)$,
$H_1(G_x,\FF_2)\otimes H_0(G_y,\Fc_y)$ is gauge, and all \gr\ logical qubits are
gauge. Then
\be\label{eq:k_prot}
k_{\rm prot}=\dim H^1(\L,\Fc^{\rd})+\dim H^1(G_y,\Fc_y)=\Theta(n),
\ee
and the subsystem-code distance satisfies
\begin{align}\label{eq:untwisted_distance_bound}
d\ \ge\ \min\big\{&d^0(G_y,\Fc^{\perp}_y),\ d^0(G_x,\Fc_x),\ d_1(G_x,\Fc_x), \cr
&d_1(G_y,\Fc^{\perp}_y),\ |V_x|,\ d_1(G_y,\Fc_y)\big\} \cr
=&\ \Theta(\bar n)=\Theta(\sqrt n),
\end{align}
where $d_1(G,\Fc)$ is the minimum non-zero cycle weight of
Proposition~\ref{prop:subdivision-params} and $d^0(G,\Fc)$ is the minimum weight
of a non-zero $0$-cocycle (with $d^0=\infty$ if $H^0(G,\Fc)=0$). Moreover
$d\le\Delta\,\min_i|\eta^{i}_1|=O(\sqrt n)$ with
$\eta^{i}_1=\text{PD}[\red{\bs^0_i}]\in H_1(G_x,\Fc_x)$, so $d=\Theta(\sqrt n)$.
\end{lemma}

\begin{proof}[Proof sketch]
The full proof is given in Appendix~\ref{app:subsystem_distance}. By
Lemma~\ref{lemma:subsystem}(iii) and~(ii), $d$ is the minimum of the protected
$X$- and $Z$-distances of the \rd\ and \bl\ copies. By the K\"unneth
decomposition \eqref{eq:kunneth_H1} and its homological counterpart, each copy
splits into two sectors that are paired factor by factor, so a class acts
non-trivially on the protected register exactly when its component in a
protected sector is non-zero. The contraction lemma
(Lemma~\ref{lemma:kunneth_contraction}) bounds every such class at once:
contracting a cocycle whose class has a non-zero component in
$H^1(G_x,\cdot)\otimes H^0(G_y,\cdot)$ against a detecting $1$-cycle of the
$x$-factor produces a non-zero $0$-cocycle of the $y$-factor of no larger
weight, and dually for cycles. This yields the six entries of
Eq.~\eqref{eq:untwisted_distance_bound}: for the \rd\ copy the two $X$-sectors
are bounded by $d^0(G_y,\Fc^{\perp}_y)$ and $d^0(G_x,\Fc_x)$, and the two
$Z$-sectors by $d_1(G_x,\Fc_x)$ and $d_1(G_y,\Fc^{\perp}_y)$; for the protected
\bl\ sector, contracting the $y$-factor against a detecting cycle of
$H_1(G_y,\Fc_y)$ leaves the all-ones cocycle $\blue{\cs^0}$ of weight $|V_x|$
[Eq.~\eqref{eq:rep_factor}], and contracting the $x$-factor against
$\blue{\cs^0}$ leaves a non-zero cycle of $H_1(G_y,\Fc_y)$, of weight
$\ge d_1(G_y,\Fc_y)$. All six quantities are $\Theta(\bar n)$ by
Eqs.~\eqref{eq:subdiv_params}, \eqref{eq:dual_cocycle_weight},
\eqref{eq:rep_factor} and the assumption on the transposed code, the entries
$d^0(G_x,\Fc_x)$ and $d_1(G_y,\Fc^{\perp}_y)$ being $\infty$ when $\Hs$ has full
row rank. The spurious \bl\ sector and the whole \gr\ copy are gauge and impose
no constraint. The upper bound is realized by the \rd\ Wilson operator
$\widetilde Z^{\rd}_{\eta}$ of Eq.~\eqref{eq:logical_Z_sheaf} with
$\eta=\eta^{i}_1\otimes\pi_0$, where $\pi_0$ is a single-vertex $0$-chain of
$C_0(G_y,\Fc^{\perp}_y)$.
\end{proof}

% \begin{remark}[Where the subsystem restriction is needed]\label{rem:spurious}
% On the \emph{cycle} side the constant-sheaf factor contributes $\cs_0$, a single
% vertex of weight $1$ by Eq.~\eqref{eq:rep_factor}, so the bound
% \eqref{eq:tensor_weight_bound} for classes involving $\cs_0$ degrades to
% $\Omega(1)$. These are exactly the short classes demoted to gauge qubits by
% Lemma~\ref{lemma:subsystem}; the same is true of the spurious classes
% $\fs_1\otimes\fs'_1$. Without the subsystem restriction the bare distance of the
% \gr\ copy would be only $\Omega(1)$ on the electric side, in agreement with the
% corresponding statement in Ref.~\cite{zhu2025topological}.
% \end{remark}

Passing to the twisted code requires one additional input, because $\tilde{\C}$ is
not a CSS code and its distance is not simply $\min(d_X,d_Z)$.

\begin{definition}[Distance of a Clifford stabilizer code]\label{def:distance_clifford}
For the Clifford stabilizer code $\tilde{\C}$ with stabilizer group
$\tilde{\SS}$, a logical operator is an $L$ with
$P_{\tilde{\C}}\,[L,\tilde{\SS}]\,P_{\tilde{\C}}=\I$ and
$P_{\tilde{\C}} L P_{\tilde{\C}} \neq \lo{I}$, and
\be\label{eq:distance_clifford}
\tilde d = \min\big\{|\mathrm{supp}(L)| \,:\, L \text{ logical}\big\} .
\ee
The subsystem-code distance is defined by restricting to the logical operators
acting non-trivially on the protected sectors, as in Lemma~\ref{lemma:subsystem}.
\end{definition}

\begin{assumption}[Condensation descendants]\label{assump:TQFT}
The topological operators of the gauge  theory are the basic electric and magnetic
operators together with their condensation descendants, the latter being sums of
condensed operators over the support of the condensation. Since a condensation
descendant is supported on a region of higher dimension than the operators being
condensed, it has weight no smaller than the basic operators.  For the
condensation-defect and SymTFT background underlying this statement see
Refs.~\cite{Roumpedakis2022, Kaidi2022, Apruzzi2021}.
\end{assumption}

\begin{lemma}[Distance of the twisted code, conditional]\label{lemma:distance_twisted}
Assuming Assumption~\ref{assump:TQFT}, the twisted code $\tilde{\C}$ with
parameters $[[\tilde n=\Theta(n),\tilde k,\tilde d]]$ has subsystem-code distance
$\tilde d = \Omega(\sqrt{\tilde n}) = \Omega(\sqrt{n})$ with respect to the
protected sectors of Eq.~\eqref{eq:protected_sectors}.
\end{lemma}

\begin{proof}
The logical operators of $\tilde{\C}$ were determined in
Sec.~\ref{sec:logical_operators_sheaf}. The electric ($Z$-type) logical operators
$\widetilde{Z}^{i}_{\eta_1}$ of Eq.~\eqref{eq:logical_Z_sheaf} are \emph{identical}
to those of the untwisted code, since the $Z$-stabilizers are untouched by the
twist; their minimal weights are therefore those of
Lemma~\ref{lemma:distance_untwisted}. The magnetic (dressed-$X$) logical operators
$\cX^{i}(\zeta^1)$ of Eq.~\eqref{eq:magnetic_projected_sheaf} are the bare $X$
strings decorated with $Z$-type projectors, so their support can only grow,
\be\label{eq:support_inequality}
\big|\mathrm{supp}\big(\cX^{i}(\zeta^1)\big)\big|
\;\ge\; \big|\mathrm{supp}\big(\lo{X}^{i}(\zeta^1)\big)\big|
\;=\; |\zeta^1| .
\ee
By Assumption~\ref{assump:TQFT} no other logical operator can be lighter than the
basic electric and magnetic ones. Restricting, as in
Lemma~\ref{lemma:subsystem}, to the same protected sectors
\eqref{eq:protected_sectors} as in Lemma~\ref{lemma:distance_untwisted},
\be\label{eq:twisted_distance}
\tilde d \;\ge\; \min\big(d_X^{\rm prot},\,d_Z^{\rm prot}\big)
\;=\; \Omega(\sqrt{n}),
\ee
which is the bound \eqref{eq:untwisted_distance_bound} of the untwisted case.
\end{proof}

\begin{remark}[Electric and magnetic distances]\label{rem:electric_magnetic}
Because $\tilde{\C}$ is not CSS, the natural refinement of $d_X$ and $d_Z$ is the
pair of \emph{magnetic} and \emph{electric} distances $d_{\cX}$ and $d_Z$,
measuring respectively the errors that fail to commute with the dressed
$X$-stabilizers and with the $Z$-stabilizers. With that reading, the \gr\ copy has subspace distance 
$d^{\gr}_{\cX}=\Omega(\sqrt n)$ but $d^{\gr}_{Z}=\Omega(1)$, again because of the
light cycle $\green{\gamma^{\rm sh}_1}$ of Eq.~\eqref{eq:spurious-cycle}; this is harmless because the \gr\
copy is not used for storage (Remark~\ref{rem:asymmetry_rate}) and the corresponding logicals are treated as gauge qubits.
\end{remark}

Lemma~\ref{lemma:distance_twisted} rests on Assumption~\ref{assump:TQFT}, a
statement about the operator content of the gauge theory that is not proved
here. We now remove it, in the regime where the transposed classical code is
trivial, by verifying the subsystem Knill--Laflamme condition
\cite{Knill1997} for $\tilde{\C}$ directly from the stabilizer group
\eqref{eq:Clifford_group_sheaf}. Two features of the twisted code make this
necessary and non-trivial. First, $\tilde{\C}$ is not a CSS code, so its
distance is not $\min(d_X,d_Z)$ and the logical operators are not all Paulis.
Second, the charge constraint \eqref{eq:constraint_operator_sheaf2} couples the
\rd\ and \bl\ labels: writing a \rd\ logical class as
$\red{a^1}=\sum_{\vc{i_1}}A_{\vc{i_1}}\red{\alpha^1_{\vc{i_1}}}$ with
$A\in\FF_2^{p\times q}$ and a \bl\ class as
$\blue{b^1}=\sum_{j_2}\mathsf b_{j_2}\blue{\beta^1_{\vc{i_2}}}$ with
$\mathsf b\in\FF_2^{q}$, Eq.~\eqref{eq:intersection_structure} turns
\eqref{eq:constraint_operator_sheaf2} into the \emph{bilinear} condition
\be\label{eq:charge_constraint_Ab}
A\,\mathsf b=0 ,
\ee
so the admissible \rd--\bl\ labels do not form a tensor product of independent
qubits (and their number is not a power of two). The natural object is therefore
the operator-algebra form of the Knill--Laflamme condition on the admissible label
space, with the \gr\ copy as gauge.

\begin{theorem}[Subsystem Knill--Laflamme distance of the twisted code]\label{thm:distance_twisted_KL}
Let the parity-check matrix $\bar\Hs$ of the input code have full row rank, so
that $H^0(G_\nu,\Fc_\nu)=0$ for $\nu=x,y$, and let $p=\dim H^1(G_x,\Fc_x)$,
$q=\dim H^1(G_y,\Fc_y)$, both $\Theta(\bar n)$. Then:
\begin{enumerate}[label=(\roman*)]
\item The twisted code space $\tilde{\C}$ of Eq.~\eqref{eq:code_space_Clifford}
has an orthonormal basis labelled by admissible triples
$(A,\mathsf b,\green{g})$ with $A\mathsf b=0$ and $\green{g}\in H^1(\L,\Fc^{\gr})$,
and hence factorizes as
\be\label{eq:twisted_factorization}
\tilde{\C}\;\cong\;\H_{\rd\bl}\otimes\H_{\gr},\qquad
\log_2\dim\H_{\rd\bl}=\Theta(n) ,
\ee
where $\H_{\rd\bl}$ is spanned by $\{\ket{A,\mathsf b}:A\mathsf b=0\}$ and
$\H_{\gr}$ by the \gr\ labels; the \gr\ factor is the gauge register. Let
$V:\H_{\rd\bl}\otimes\H_{\gr}\to\tilde{\C}$ be the corresponding unitary.
\item There is a constant $c>0$, explicitly
\begin{align}\label{eq:KL_threshold}
c\,\bar n:=\min\Big\{&d^0(G_y,\Fc^{\perp}_y),\ d^0(G_x,\Fc^{\perp}_x),\ |V_x|, \cr
&\tfrac{1}{5\Delta}\,d_1(G_x,\Fc_x),\ \tfrac{1}{5\Delta}\,d_1(G_y,\Fc_y)\Big\} \cr
=\Theta(\bar n)&=\Theta(\sqrt n),
\end{align}
such that every operator $E$ supported on fewer than $c\,\bar n$ qubits satisfies
\be\label{eq:KL_twisted}
P_{\tilde{\C}}\,E\,P_{\tilde{\C}}
= V\big(\I_{\rd\bl}\otimes M_E\big)V^{\dagger}
\ee
for some operator $M_E$ on $\H_{\gr}$, diagonal in the \gr\ label basis.
Consequently, for every family of operators $\{E_\alpha\}$ with each product
$E^{\dagger}_\alpha E_\beta$ supported on fewer than $c\,\bar n$ qubits,
$P_{\tilde{\C}}E^{\dagger}_\alpha E_\beta P_{\tilde{\C}}
=V(\I_{\rd\bl}\otimes M_{\alpha\beta})V^{\dagger}$: the subsystem Knill--Laflamme
condition holds for the protected register $\H_{\rd\bl}$ with gauge register
$\H_{\gr}$.
\item The subsystem-code distance of $\tilde{\C}$ with respect to $\H_{\rd\bl}$ is
$\tilde d=\Theta(\sqrt n)$; the upper bound is realized by the \rd\ Wilson
operators of Eq.~\eqref{eq:logical_Z_sheaf}.
\item (Independent registers.) For any fixed $0<\varrho<1$ and any subspace
$U\le\FF_2^{q}$ with $\dim U=r=\lfloor\varrho q\rfloor$, the restriction of
$\H_{\rd\bl}$ to $\{\mathsf b\in U,\ A_{i_1\cdot}\in U^{\perp}\ \forall i_1\}$ is a
tensor product of $k_{\rd}=p(q-r)=\Theta(n)$ \rd\ and $k_{\bl}=r=\Theta(\sqrt n)$
\bl\ logical qubits, and Eq.~\eqref{eq:KL_twisted} continues to hold with $V$
restricted accordingly.
\end{enumerate}
\end{theorem}

\begin{proof}[Proof sketch]
The complete proof is given in Appendix~\ref{app:twisted_KL}; we outline the
steps. (1) \emph{Orbit basis.} On flat configurations
$x=(\red{a^1},\blue{b^1},\green{c^1})$ the dressed generators
\eqref{eq:stabilizer_summary} act as
$\tilde{A}(u)\ket{x}=(-1)^{\theta_x(u)}\ket{x+d^0u}$ for an explicit cubic phase
$\theta_x$, and by Lemma~\ref{lemma:stabilizer_commutation} they form a
representation of the additive group of gauge parameters $u$. The code space is
spanned by the orbit sums $\ket{\psi_x}=\sum_u(-1)^{\theta_x(u)}\ket{x+d^0u}$ over
the admissible configurations, and $\ket{\psi_x}$ depends on the representative
$x$ only through a sign. With full row rank, the only zero modes are the \gr\
$0$-cocycles $\green{\gamma^0_{\vc{i_3}}}$, so admissibility is exactly
Eq.~\eqref{eq:charge_constraint_Ab} and the fibre over every admissible
$(A,\mathsf b)$ is the full set of \gr\ labels; this gives (i).
(2) \emph{No low-weight error changes a label.} With full row rank the
K\"unneth decomposition has no unselected \rd\ sector and no \bl\ graph-short
sector, and the contraction lemma (Corollary~\ref{cor:sector_distance_app})
shows that every non-trivial $X$-type class of every colour has weight at least
$\min\{d^0(G_y,\Fc^\perp_y),|V_x|,d^0(G_x,\Fc^\perp_x)\}$. Hence a Pauli
$E=X(\zeta)Z(\eta)$ below that weight has exact $X$-part $\zeta=d^0\tau$ and
$P_{\tilde\C}EP_{\tilde\C}$ is diagonal in the orbit basis.
(3) \emph{Exact matrix element.} The diagonal entry is an explicit average
$\frac1{|C^0|}\sum_u(-1)^{\Theta_x(u)}$ with
$\Theta_x(u)=\theta_x(u)+\theta_x(u+\tau)+\langle\eta,x+d^0u\rangle$. Expanding
$\Theta_x$ and using the Leibniz rule on $\eta_2$ isolates the terms that depend
on the protected representatives $\red{a^1},\blue{b^1}$; each of them is local to
an $O(|E|)$ ``shadow'' of $1$-cells of $G_x$ or $G_y$, provided the primitives
$\tau$ of the \bl\ and \gr\ parts of $\zeta$ are chosen to occupy $O(|E|)$
columns, which is always possible (Lemma~\ref{lemma:column_primitive}).
(4) \emph{Cleaning.} Because the conjugate cycles of $\red{\as^1}$ and
$\blue{\as^1}$ are heavy [Eq.~\eqref{eq:subdiv_params}], the factor
representatives can be moved off the shadows without changing the
cohomology class; since bra and ket carry the same configuration, the sign by
which $\ket{\psi_x}$ changes cancels. After cleaning every
$(\red{a^1},\blue{b^1})$-dependent term vanishes identically, so the diagonal
entry depends only on the \gr\ label, which is Eq.~\eqref{eq:KL_twisted} for
Paulis; linearity extends it to arbitrary operators on fewer than $c\bar n$
qubits. (5) The upper bound in (iii) is the Wilson operator of the proof of
Lemma~\ref{lemma:distance_untwisted}, which acts on $\ket{\psi_x}$ by the sign
$(-1)^{\langle\eta,\red{a^1}\rangle}$, and (iv) follows because
$\I_{\rd\bl}\otimes M_E$ preserves the restricted register.
\end{proof}

\begin{remark}[Scope of Theorem~\ref{thm:distance_twisted_KL}]\label{rem:KL_scope}
Theorem~\ref{thm:distance_twisted_KL} uses only Eqs.~\eqref{eq:subdiv_params},
\eqref{eq:dual_cocycle_weight}, \eqref{eq:rep_factor}, the full-row-rank
hypothesis, the locality and Leibniz properties of the cup product on $\eta_2$,
and Lemma~\ref{lemma:stabilizer_commutation}; no expansion property of $G$, no
filling inequality, and no classification of topological operators enters. The
full-row-rank hypothesis is what makes the \bl\ copy free of graph-short classes
and the admissibility condition independent of the \gr\ label. If instead the
transposed code is non-trivial but good, the same argument applies after fixing
the \gr\ label and the \bl\ graph-short labels, all of which are heavy; the
intermediate regime of a non-trivial transposed code of small distance is not
covered.
\end{remark}

\subsubsection{Summary}
\label{sec:parameter_summary}

\begin{theorem}[Code parameters]\label{thm:parameters}
Let $\bar{\C}$ be an asymptotically good classical LDPC code and let $\C$ and
$\tilde{\C}$ be the untwisted and twisted codes of
Construction~\ref{con:2d-hgp} built from its subdivided sheaf code. Then both are
qLDPC codes with
\be\label{eq:final_parameters}
\big[[\,n,\ \Theta(n),\ \Theta(\sqrt{n})\,]\big] ,
\ee
i.e.\ constant encoding rate and subsystem-code distance $\Theta(\sqrt{n})$ with
respect to the protected sectors of Eq.~\eqref{eq:protected_sectors}. For $\C$
this is Lemma~\ref{lemma:distance_untwisted}. For $\tilde{\C}$ the distance
statement is Theorem~\ref{thm:distance_twisted_KL}, proved via the subsystem
Knill--Laflamme condition when $\bar\Hs$ has full row rank, and
Lemma~\ref{lemma:distance_twisted}, which holds under
Assumption~\ref{assump:TQFT}, otherwise.
\end{theorem}

\begin{proof}
Combine Lemmas~\ref{lemma:rate_untwisted}, \ref{lemma:rate_twisted},
\ref{lemma:distance_untwisted}, \ref{lemma:distance_twisted} and
Theorem~\ref{thm:distance_twisted_KL}. The qLDPC
property is Eqs.~\eqref{eq:A_stabilizer_untwisted}
and~\eqref{eq:B_stabilizer_untwisted}, whose weights are $O(1)$ for bounded vertex
degree and bounded stalk dimension, together with the fact that the CZ dressing of
Eq.~\eqref{eq:stabilizer_summary} is supported on bounded-size clusters of cells.
\end{proof}

\begin{remark}[Comparison with the high-dimensional CW complex route in Ref.~\cite{ZhuKobayashiHsin2026}]\label{rem:skeleton_comparison}
Theorem~\ref{thm:parameters} reproduces exactly the parameter scaling  obtained for the
skeleton codes pulled back from the high-dimensional CW complex  in Ref.~\cite{ZhuKobayashiHsin2026}. What changes is the derivation. In the CW-complex
route the good classical code is first thickened into a high-dimensional
Poincar\'e complex, the (co)cycles are transported through the mapping between
classical (co)cycles and CW (co)cycles, and the parameters are then pulled back to
the $2$D skeleton; each step must be checked to preserve dimension and weight. In
the sheaf route the classical code is subdivided once
(Proposition~\ref{prop:subdivision-params}), and every quantity above is a
K\"unneth product of two single-graph quantities read off directly from
$\bar{\C}$ and its transpose. The intermediate manifold, and with it the spurious
(co)cycles it introduces in high degree, never appears --- although the subdivision produces short classes of its
own, so the subsystem-code restriction of Lemma~\ref{lemma:subsystem} is still
required.
\end{remark}

%======================================================================
\section{Logical non-Clifford operation via gauging measurement and
spacetime path integral}
\label{sec:logical_operation}
%======================================================================

We now devise
a general fault-tolerant protocol that implements a logical non-Clifford operation on the
constant-rate codes of Construction~\ref{con:2d-hgp}. The protocol is an
\emph{addressable gauging measurement}: one starts from two decoupled untwisted copies, gauges
the addressable transversal CZ gates between them by measuring the twisted \gr-type
$X$-stabilizers, and then ungauges, leaving behind a projection onto an eigenspace
of a product of logical CZ gates.  We then use the constant-rate codes of Construction~\ref{con:2d-hgp}
as an instantiation which can prepare $\Theta(\sqrt{n})$ disjoint logical CZ magic state in parallel. We then derive the result of this protocol in an alternative way with the spacetime path integral.

% The logic is the same as in Ref.~\cite{ZhuKobayashiHsin2026}, but the geometry is
% not. There the classical code was first mapped to a high-dimensional Poincar\'e CW
% complex, the gauged symmetry was a \emph{higher-form} symmetry supported on a
% cycle of codimension $3$, and the protocol had finally to be pulled back to a $2$D
% skeleton in order to be implementable. Here none of these steps occurs. The sheaf
% complex $\L=G_x\times G_y$ of Eq.~\eqref{eq:spatial_HGP} is already
% $\mathsf{d}=2$ dimensional, the symmetry being gauged is the $0$-form
% \emph{subcomplex} symmetry of Sec.~\ref{sec:0form_subcomplex}, labelled by a
% \gr-type sheaf $0$-cocycle, and the addressability that the higher-form
% construction had to manufacture by hand is supplied directly by
% $\dim H^0(\L,\Fc^{\gr})=\Theta(\sqrt{n})$. In this sense the present section is
% the sheaf counterpart of the pullback of Ref.~\cite{ZhuKobayashiHsin2026}, with the
% pullback itself rendered unnecessary.

\subsection{Addressable and parallel gauging measurement of logical CZ}
\label{sec:parallel_gauging}

\subsubsection{The symmetry operator and its logical action}
\label{sec:gauging_idea}

The operators to be measured are the transversal CZ operators of
Sec.~\ref{sec:transversal_CZ} between the \rd\ and \bl\ copies of the untwisted
code $\C=\C^{\rd}\otimes\C^{\bl}$, one for each basis \gr-type $0$-cocycle
$\green{\gamma^{0}_{\vc{i_3}}}=\green{\bs^0_{i_3}}\otimes\green{\bs^0_{j_3}}$ of
Eq.~\eqref{eq:spatial_parts},
\begin{align}\label{eq:measure_subcomplex}
\widetilde{\text{CZ}}^{\rd,\bl}\big(\green{\gamma^{0}_{\vc{i_3}}}\big)
=&\ (-1)^{\int_{\eta_2}\red{\hat{a}^1}\cup\blue{\hat{b}^1}\cup
\green{\gamma^{0}_{\vc{i_3}}}} \cr
\equiv&\ (-1)^{\int_{\eta_2^{\gamma}}\red{\hat{a}^1}\cup\blue{\hat{b}^1}} ,
\qquad
\eta_2^{\gamma} \equiv \eta_2 \frown \green{\gamma^{0}_{\vc{i_3}}},  \cr
\end{align}
where the second line is Eq.~\eqref{eq:cap_support}: capping the twist cycle
$\eta_2$ with the symmetry label produces the $2$-cycle on which the operator is
supported. Since
$\dim\eta_2^{\gamma}=2$, this is a codimension-$0$ operator, i.e.\ a
$0$-form symmetry supported on the proper subcomplex
$\L_{\eta_2^{\gamma}}$ of Eq.~\eqref{eq:subcomplex_support}. 

% What replaces the higher-form structure is multiplicity:
% there is one such operator for every basis class of $H^0(\L,\Fc^{\gr})$, and by
% Eq.~\eqref{eq:rep_factor} and Lemma~\ref{lemma:rate_twisted} there are
% $\Theta(\sqrt{n})$ of them.

The logical action of each of these
operators on the untwisted code is completely explicit,
\begin{align}\label{eq:logical_CZ_instantiated}
\widetilde{\text{CZ}}^{\rd,\bl}\big(\green{\gamma^{0}_{\vc{i_3}}}\big)
\;\widehat{=}\;&
\prod_{\red{\vc{i_1}},\blue{\vc{i_2}}}
\big[\lo{\text{CZ}}^{\rd,\bl}_{\red{\vc{i_1}},\blue{\vc{i_2}}}
\big]^{\Lambda(\red{\vc{i_1}},\blue{\vc{i_2}},\green{\vc{i_3}})} , \cr
\Lambda(\red{\vc{i_1}},\blue{\vc{i_2}},\green{\vc{i_3}})
=&\ \int_{\eta_2} \red{\alpha^1_{\vc{i_1}}} \cup \blue{\beta^1_{\vc{i_2}}} \cup \green{\gamma^0_{\vc{i_3}}} ,
\end{align}
with the multi-indices $\vc{i_1}=(i_1,j_1)$, $\vc{i_2}=(i_2,j_2)$,
$\vc{i_3}=(i_3,j_3)$ of Eq.~\eqref{eq:multi_index}. 

% Recall that the labels $j_2$
% and $i_3$ are frozen, since $\blue{\cs'^0}$ and $\green{\cs^0}$ are the unique
% $0$-cocycle classes of the two constant-sheaf factors; a \bl\ logical qubit is
% therefore labelled by $i_2$ alone and a \gr\ symmetry operator by $j_3$ alone.

The reason this transversal operator can be measured fault-tolerantly at all is
Eq.~\eqref{eq:charge_is_CZ_sheaf}: in the twisted code it coincides with the
\gr-type charge parity operator, which is by construction a product of local
twisted stabilizer generators,
\be\label{eq:CZ_factorizing}
\widetilde{\text{CZ}}^{\rd,\bl}\big(\green{\gamma^{0}_{\vc{i_3}}}\big)
= \mathsf{C}^{\gr}\big(\green{\gamma^{0}_{\vc{i_3}}}\big)
= \prod_{v,t}\big(\tilde{A}^{\gr}_{v,t}\big)^{\gamma_{v,t}} \equiv \prod_{(v,t)\in \green{\gamma^0_{\vc{i_3}}}} \tilde{A}^{\gr}_{v,t},
\ee
where $\green{\gamma^{0}_{\vc{i_3}}}=\sum_{v,t}\gamma_{v,t}\bar{v}_t$ and the
Pauli-$X$ content cancels by the redundancy
\eqref{eq:X_redundancy_sheaf}. Equation~\eqref{eq:CZ_factorizing} is the whole
content of gauging at the operator level: a non-local generalized global symmetry
is traded for a collection of $O(1)$-weight local gauge constraints. Measuring the
$\tilde{A}^{\gr}_{v,t}$ individually --- which is possible, since by
Lemma~\ref{lemma:stabilizer_commutation} twisted $X$-stabilizers of the
\emph{same} colour commute exactly --- and multiplying the outcomes along each
basis $0$-cocycle $\green{\gamma^{0}_{\vc{i_3}}}$ recovers the eigenvalue of the
corresponding symmetry operator. The measurement is therefore \emph{addressable},
one outcome per basis class $[\green{\gamma^{0}_{\vc{i_3}}}]$, and \emph{parallel}, all classes
being read off from the same round of local measurements.

\medskip
\noindent\textit{Gauging and ungauging in spacetime path integral.}
What this measurement does was also described in the path-integral language in
Sec.~\ref{sec:gauging_untwisted}: gauging the $0$-form subcomplex CZ symmetry of
the untwisted code produces the twisted code, Eq.~\eqref{eq:gauged_partition},
and inserting a \gr\ Wilson operator pins the value of the transversal CZ
operator to a chosen sector, Eq.~\eqref{eq:constraint_shifted}. Specializing the
labels used there to the K\"unneth basis at hand,
i.e.\ relabelling $k\to\vc{i_3}$ in $\green{\gamma^{0}_{k}}$ and
$\green{\tilde{\gamma}^{1}_{k}}$ so that they become the basis cocycles
$\green{\gamma^{0}_{\vc{i_3}}}$ and $\green{\tilde{\gamma}^{1}_{\vc{i_3}}}$ of
Eq.~\eqref{eq:spatial_cocycle_decomposition}, the sector label $\green{w_{k}}$
becomes the outcome $\green{\rho_{\vc{i_3}}}$ of the stabilizer measurement,
Eq.~\eqref{eq:outcome_product} below. The protocol realizes the gauging by
measurement rather than by post-selection, so that the sector is random and
recorded rather than chosen; by Eq.~\eqref{eq:charge_measured} the recorded value
is the total \gr-charge parity weighted by $\green{\gamma^{0}_{\vc{i_3}}}$.

Ungauging, the inverse operation, condenses the \gr\ electric charge and returns
the untwisted theory \eqref{eq:Z_untwisted}. Microscopically it is the $Z$-basis
measurement of every \gr\ qubit in step~3 below: it freezes $\green{c^1}$ to a
definite configuration $\green{\theta^1}$, which the constraints of that step
force to be pure gauge, $\green{\theta^1}=d^0\cV^0$, so that after the correction
\eqref{eq:correction} the \gr\ copy decouples. What survives the round trip is
precisely the information gained about the sector, i.e.\ the constraint
\eqref{eq:constraint_shifted}, which reappears as the last generator of $\cS_4$
in Eq.~\eqref{eq:stabilizer_group_4}.

\medskip
\noindent\textit{The spacetime picture.}
Operationally the protocol is therefore a gauging at time $t_i$, mapping
$\C=\C^{\rd}\otimes\C^{\bl}$ to the twisted code $\tilde{\C}$, a dwell of $O(d)$
rounds of error correction inside $\tilde{\C}$, and an ungauging at $t_f$ mapping
back to $\C$. The two transitions are spacetime domain walls $W$ and $W'$
separating the $\red{\ZZ_2}\times\blue{\ZZ_2}$ gauge theory outside the interval
$[t_i,t_f]$ from the twisted
$\red{\ZZ_2}\times\blue{\ZZ_2}\times\green{\ZZ_2}$ theory inside it, so that the
twisted action \eqref{eq:non-Abelian_action} is supported on the slab
$\L\otimes[t_i,t_f]$. The \gr\ charge worldlines are then \emph{relative} cycles
terminating on $W$ and $W'$, and the cycle $\green{\tilde{\gamma}_{1,k}}$ of
Eq.~\eqref{eq:wilson_insertion} may be taken to run across the slab, so that
$\green{\rho_{\vc{i_3}}}$ is the parity of the worldlines it links at any time
between $t_i$ and $t_f$. Since ${\ts}^1$ contributes a single point in the time
direction, a representative of $\green{\tilde{\gamma}^{1}_{\vc{i_3}}}$ may be
chosen with the same spatial support as $\green{\gamma^{0}_{\vc{i_3}}}$ on one
time slice, which is why the spacetime and the equal-time descriptions agree; this
is Lemma~\ref{lemma:temporal_factor} again.

\subsubsection{The gauging measurement protocol}
\label{sec:protocol}

We now list the protocol explicitly. Throughout, $v\in\L(0)$, $e\in\L(1)$ and
$f\in\L(2)$, and $p,q,s$ are the vector indices of the \rd, \bl\ and \gr\ stalks.

\begin{enumerate}

\item \textit{Initialization.} Prepare the \rd\ and \bl\ copies as two decoupled
untwisted sheaf codes $\C=\C^{\rd}\otimes\C^{\bl}$ of
Sec.~\ref{sec:untwisted_code}, with the logical qubits initialized in
$\lo{\ket{+}}$ or $\lo{\ket{0}}$ according to the desired connectivity of the
logical CZ measurement; call the resulting logical state $\lo{\ket{\psi}}_i$.
Initialize every qubit of the \gr\ copy in $\ket{0}$. The stabilizer group is
\be\label{eq:stabilizer_group_1}
\cS_1 = \big\langle\, A^{\rd}_{v,s},\ B^{\rd}_{f},\ A^{\bl}_{v,r},\ B^{\bl}_{f},\
Z^{\gr}_{e,t},\ B^{\gr}_{f},\ \lo{Z}^{\gr}_{\eta_1} \,\big\rangle ,
\ee
where $\lo{Z}^{\gr}_{\eta_1}=\widetilde{Z}^{\gr}_{\eta_1}$ runs over a basis of
$H_1(\L,\Fc^{\gr})$, Eq.~\eqref{eq:logical_Z_sheaf}. The last two families are
redundant given $Z^{\gr}_{e,t}$, but we record them because they survive the later
steps. Since all \gr\ qubits are in $\ket{0}$, every CZ dressing acts trivially
and we may equally well write $\cS_1$ with the twisted generators,
\be\label{eq:stabilizer_group_1_dressed}
\cS_1 = \big\langle\, \tilde{A}^{\rd}_{v,s},\ B^{\rd}_{f},\
\tilde{A}^{\bl}_{v,r},\ B^{\bl}_{f},\ Z^{\gr}_{e,t},\ B^{\gr}_{f},\
\lo{Z}^{\gr}_{\eta_1} \,\big\rangle .
\ee

\item \textit{Gauging (at $t_i$).} Measure all \gr-type twisted $X$-stabilizers
$\tilde{A}^{\gr}_{v,t}$, obtaining outcomes $(-1)^{\mu_{v,t}}=\pm 1$. These
commute exactly with one another by Lemma~\ref{lemma:stabilizer_commutation}, so
the measurement is a single round of commuting local checks. Multiplying the
outcomes along a basis $0$-cocycle and using Eq.~\eqref{eq:CZ_factorizing},
\begin{align}\label{eq:outcome_product}
\widetilde{\text{CZ}}^{\rd,\bl}\big(\green{\gamma^{0}_{\vc{i_3}}}\big)
&= (-1)^{\green{\rho_{\vc{i_3}}}} , \cr
\green{\rho_{\vc{i_3}}} &\equiv \sum_{v,t}\gamma_{v,t}\,\mu_{v,t}
\quad (\text{mod } 2) ,
\end{align}
for every class $[\green{\gamma^{0}_{\vc{i_3}}}]\in H^0(\L,\Fc^{\gr})$. By
Eq.~\eqref{eq:logical_CZ_instantiated} each such outcome is a projective
measurement of a product of logical CZ gates,
$\prod_{\red{\vc{i_1}},\blue{\vc{i_2}}}[\lo{\text{CZ}}^{\rd,\bl}_{\red{\vc{i_1}},
\blue{\vc{i_2}}}]^{\Lambda (\red{\vc{i_1}},\blue{\vc{i_2}},\green{\vc{i_3}})}=(-1)^{\green{\rho_{\vc{i_3}}}}$. The state now lies in
the twisted code, with
\begin{align}\label{eq:stabilizer_group_2}
\cS_2 = \big\langle\, & \tilde{A}^{\rd}_{v,s},\ B^{\rd}_{f},\
\tilde{A}^{\bl}_{v,r},\ B^{\bl}_{f}, \cr
& (-1)^{\mu_{v,t}}\tilde{A}^{\gr}_{v,t},\ B^{\gr}_{f},\
\lo{Z}^{\gr}_{\eta_1} \,\big\rangle ,
\end{align}
the single-qubit operators $Z^{\gr}_{e,t}$ having been removed as they
anticommute with the newly measured generators. One then remains in $\tilde{\C}$
for $O(d)$ rounds of error correction, so that the outcomes
\eqref{eq:outcome_product} are reliable in the presence of measurement errors.

\begin{figure}[t]
    \centering   \includegraphics[width=1\linewidth]{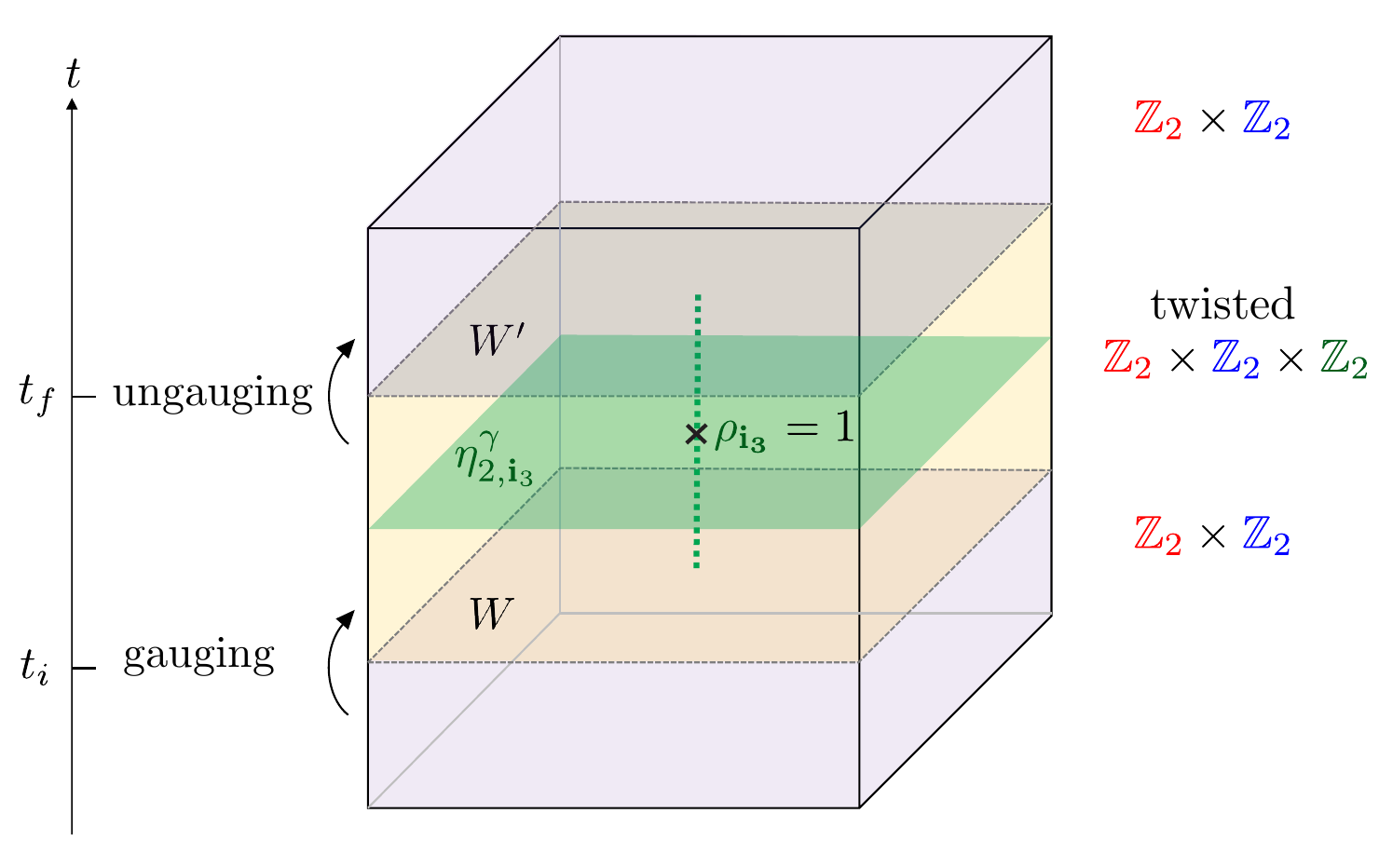}
    \caption{Spacetime picture of the gauging measurement protocol, with
time $t$ running upward. Below $t_i$ and above $t_f$ the system is the untwisted
theory, two decoupled copies $\red{\ZZ_2}\times\blue{\ZZ_2}$ of the
hypergraph-product sheaf code; in the slab $t_i<t<t_f$ it is the twisted
$\red{\ZZ_2}\times\blue{\ZZ_2}\times\green{\ZZ_2}$ theory. The two slab boundaries
are spacetime domain walls: $W$ at $t_i$ implements the gauging (measurement of
all \gr-type dressed $X$-stabilizers $\tilde{A}^{\gr}_{v,t}$), and $W'$ at $t_f$
the ungauging (measurement of every \gr\ qubit in the $Z$ basis); the slab is kept
for $O(d)$ rounds of error correction in between. The green sheet is the capped
subcomplex $\tilde\eta^{\gamma}_{2,\vc{i_3}}$ inserted into the path integral,
whose measurement outcome is recorded as $\green{\rho_{\vc{i_3}}}$; here the
outcome $\green{\rho_{\vc{i_3}}}=1$ is realized by a charge worldline (dotted)
piercing it. Evaluating the path integral on this spacetime reproduces the logical
action \eqref{eq:output_logical_state} of the protocol.}
\label{fig:protocol}
\end{figure}

\item \textit{Ungauging (at $t_f$).} Measure every qubit of the \gr\ copy in the
$Z$ basis, with outcomes $(-1)^{\mu_{e,t}}$, giving
\begin{align}\label{eq:stabilizer_group_3}
\cS_3 = \big\langle\, & \tilde{A}^{\rd}_{v,s},\ B^{\rd}_{f},\
\tilde{A}^{\bl}_{v,r},\ B^{\bl}_{f},\ (-1)^{\mu_{e,t}}Z^{\gr}_{e,t}, \cr
& B^{\gr}_{f},\ \lo{Z}^{\gr}_{\eta_1},\
(-1)^{\green{\rho_{\vc{i_3}}}}
\mathsf{C}^{\gr}\big(\green{\gamma^{0}_{\vc{i_3}}}\big)
\,\big\rangle .
\end{align}
The individual $\tilde{A}^{\gr}_{v,t}$ have been dropped, since they anticommute
with the overlapping single-qubit measurements, but their product
$\mathsf{C}^{\gr}(\green{\gamma^{0}_{\vc{i_3}}})$ is $X$-free by
Eq.~\eqref{eq:CZ_factorizing} and therefore survives. Collect the outcomes into
the \gr-type $1$-cochain
$\green{\theta^1}\equiv\sum_{e,t}\mu_{e,t}\,\bar{e}_t\in C^1(\L,\Fc^{\gr})$. The
redundant generators in $\cS_3$ constrain it twice:
\begin{itemize}
\item $B^{\gr}_{f}=+1$ for every $f$ forces
$(d^1\green{\theta^1})(f)=0$, i.e.\ $\green{\theta^1}\in Z^1(\L,\Fc^{\gr})$ is a
sheaf $1$-cocycle;
\item $\lo{Z}^{\gr}_{\eta_1}=+1$ for every basis cycle forces
$\langle\green{\theta^1},\eta_1\rangle=0$ for all
$[\eta_1]\in H_1(\L,\Fc^{\gr})$, and hence $[\green{\theta^1}]=0$ in
$H^1(\L,\Fc^{\gr})$ by the non-degeneracy of the pairing,
Proposition~\ref{prop:nondegen-sheaf}.
\end{itemize}
Therefore $\green{\theta^1}=d^0\cV^0$ is exact, for some $0$-cochain
$\cV^0=\sum_{v,t}\cV_{v,t}\bar{v}_t\in C^0(\L,\Fc^{\gr})$, and the outcomes can be
returned to $\mu_{e,t}=0$ by the correction
\be\label{eq:correction}
\R = \prod_{v,t}\big(\tilde{A}^{\gr}_{v,t}\big)^{\cV_{v,t}} ,
\ee
i.e.\ by a product of twisted stabilizers over any $\cV^0$ with
$d^0\cV^0=\green{\theta^1}$; the ambiguity in $\cV^0$ is a $0$-cocycle, which by
Eq.~\eqref{eq:constraint_operator_sheaf2} acts trivially. After the correction,
\begin{align}\label{eq:stabilizer_group_4}
\cS_4 = \big\langle\, & A^{\rd}_{v,s},\ B^{\rd}_{f},\ A^{\bl}_{v,r},\
B^{\bl}_{f},\ Z^{\gr}_{e,t}, \cr
& B^{\gr}_{f},\ \lo{Z}^{\gr}_{\eta_1},\
(-1)^{\green{\rho_{\vc{i_3}}}}
\mathsf{C}^{\gr}\big(\green{\gamma^{0}_{\vc{i_3}}}\big)
\,\big\rangle .
\end{align}
The twisted generators $\tilde{A}^{\rd}_{v,s}$, $\tilde{A}^{\bl}_{v,r}$ have been
replaced by the bare ones $A^{\rd}_{v,s}$, $A^{\bl}_{v,r}$, because all \gr\
qubits are back in $\ket{0}$ and the dressings $\text{CZ}^{\bl,\gr}$,
$\text{CZ}^{\rd,\gr}$ are trivialized. The \gr\ copy is thus completely decoupled
and the state has returned to the untwisted code space $\C$. The only difference
between $\cS_4$ and $\cS_1$ is the extra constraint
$\mathsf{C}^{\gr}(\green{\gamma^{0}_{\vc{i_3}}})=
\widetilde{\text{CZ}}^{\rd,\bl}(\green{\gamma^{0}_{\vc{i_3}}})
=(-1)^{\green{\rho_{\vc{i_3}}}}$, i.e., a projection onto an eigenspace of the
transversal CZ operator.

\item \textit{Output.} Writing $\hat{P}_{\pm 1}[\,\cdot\,]$ for the projector onto
the $\pm1$ eigenspace of the indicated operator, the logical output state is
\begin{align}\label{eq:output_logical_state}
\lo{\ket{\psi}}_f
=&\ \prod_{\green{\vc{i_3}}}
\hat{P}_{1-2\green{\rho_{\vc{i_3}}}}
\Big[\widetilde{\text{CZ}}^{\rd,\bl}
\big(\green{\gamma^{0}_{\vc{i_3}}}\big)\Big]\ \lo{\ket{\psi}}_i \cr
=&\ \prod_{\green{\vc{i_3}}}
\hat{P}_{1-2\green{\rho_{\vc{i_3}}}}
\Big[\prod_{\red{\vc{i_1}},\blue{\vc{i_2}}}
\big[\lo{\text{CZ}}^{\rd,\bl}_{\red{\vc{i_1}},\blue{\vc{i_2}}}
\big]^{\Lambda(\red{\vc{i_1}},\blue{\vc{i_2}},\green{\vc{i_3}})}\Big]\
\lo{\ket{\psi}}_i,  \cr
\end{align}
with $1-2\green{\rho_{\vc{i_3}}}\equiv(-1)^{\green{\rho_{\vc{i_3}}}}=\pm1$.

\end{enumerate}

% \begin{remark}[What the sheaf formulation buys]\label{rem:no_pullback}
% Every object in the protocol lives on the $2$D complex $\L$: the measured checks
% $\tilde{A}^{\gr}_{v,t}$ sit on vertices, the measured qubits on edges, and the
% symmetry labels in $H^0(\L,\Fc^{\gr})$. In
% Ref.~\cite{ZhuKobayashiHsin2026} the corresponding operators lived on cells of
% degrees $2$, $3$ and $4$ of a $15$-dimensional complex and had to be transported
% to the $2$D skeleton before the protocol could be run. The dictionary
% \eqref{eq:skeleton_dictionary} identifies the two, but only the sheaf version is
% manifestly local on the physical lattice.
% \end{remark}

\subsection{Instantiation: preparing $\Theta(\sqrt{n})$ disjoint magic states via the magic state
fountain}\label{sec:fountain}

\begin{figure*}[t]
    \centering   \includegraphics[width=0.65 \linewidth]{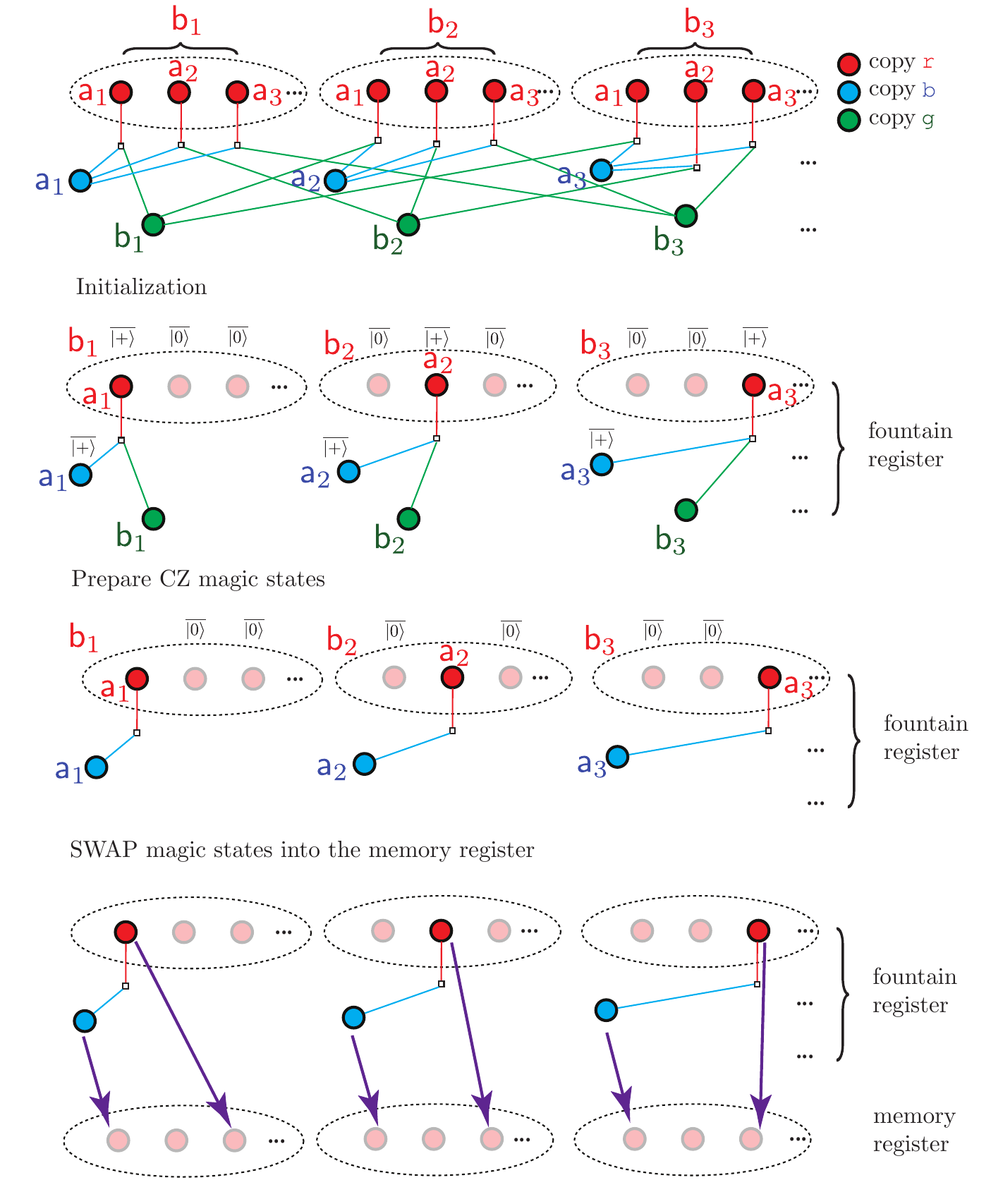}
    \caption{The magic state fountain for
Construction~\ref{con:2d-hgp}, read off from the interaction hypergraph of the
logical gate structure $\Lambda(\red{\vc{i_1}},\blue{\vc{i_2}},\green{\vc{i_3}})$.
\emph{Row 1 (Initialization).} Each logical qubit of the \rd\ copy is a red
vertex, each logical qubit of the \bl\ copy a blue vertex, each symmetry label of
the \gr\ copy a green vertex, and a hyperedge joins a triple whenever
$\Lambda=1$. By Eq.~\eqref{eq:intersection_structure} the \rd\ vertices fall into
$\Theta(\sqrt{n})$ clusters $b_1,b_2,b_3,\dots$, the members $a_1,a_2,a_3,\dots$
of a cluster sharing a single \gr\ label and being distinguished by their \bl\
partner. \emph{Row 2 (Prepare CZ magic states).} Initializing exactly one \rd\
logical qubit per cluster in $\lo{\ket{+}}$ and the rest in $\lo{\ket{0}}$, and
all \bl\ qubits in $\lo{\ket{+}}$, leaves exactly one active hyperedge per
cluster, so the gauging measurement projects each of the $\Theta(\sqrt{n})$
disjoint \rd--\bl\ pairs onto a logical CZ magic state
$\lo{\ket{\text{CZ}}}$. \emph{Rows 3 and 4 (SWAP into the memory register).} The
magic states prepared in the fountain register are moved by logical SWAPs (purple
arrows) into the memory register, where they are consumed by gate teleportation,
leaving the fountain register free for the next run.}
\label{fig:interaction-hypergraph}
\end{figure*}

In this subsection, we consider the instantiation with Construction~\ref{con:2d-hgp} in Sec.~\ref{sec:code_instantiation}.
If all logical qubits are initialized in $\lo{\ket{+}}$, Eq.~\eqref{eq:output_logical_state}
prepares a highly entangled logical resource state in which many logical qubits are
coupled together. For parallel computation it is preferable to prepare magic
states on \emph{disjoint} pairs of logical qubits, which can then be consumed
independently by gate teleportation. This is the \emph{magic state fountain} idea of
Refs.~\cite{zhu2023non,zhu2025topological, ZhuKobayashiHsin2026}, and the explicit form
\eqref{eq:logical_CZ_instantiated} of the gate structure makes the required
initialization transparent.

It is convenient to read $\Lambda(\red{\vc{i_1}},\blue{\vc{i_2}},\green{\vc{i_3}})$ as a hypergraph. Put a \rd\ vertex for each
logical qubit $\red{\vc{i_1}}=(i_1,j_1)$ of the \rd\ copy, a \bl\ vertex for each
logical qubit $\blue{j_2}$ of the \bl\ copy, and a \gr\ vertex for each symmetry
operator $\green{i_3}$, and join a triple whenever
$\Lambda(\red{\vc{i_1}},\blue{\vc{i_2}},\green{\vc{i_3}})=1$. By
Eq.~\eqref{eq:intersection_structure} the \rd\ vertex $(i_1,j_1)$ is joined to
the unique \bl\ vertex $\blue{j_2}=j_1$ and the unique \gr\ vertex
$\green{i_3}=i_1$. The \rd\ vertices therefore fall into $\Theta(\sqrt{n})$
clusters labelled by $i_1$, each cluster containing the $\Theta(\sqrt{n})$
logical qubits with the same $i_1$ and different $j_1$, all sharing a single \gr\
label and distinguished by their \bl\ partner. This is the sheaf translation of
the cluster structure found in Ref.~\cite{ZhuKobayashiHsin2026}.

\begin{theorem}[Magic state fountain]\label{thm:fountain}
There is a gauging measurement protocol on the twisted $2$D hypergraph-product
sheaf code of Construction~\ref{con:2d-hgp} that prepares $\Theta(\sqrt{n})$
logical CZ magic states on disjoint pairs of logical qubits, protected by
subsystem-code distance $\Omega(\sqrt{n})$ (the latter under
Assumption~\ref{assump:TQFT}).
\end{theorem}

\begin{proof}
Fix a bijection $i_1 \mapsto \pi(i_1)$ between the $\Theta(\sqrt{n})$ values of
the $x$-factor label $i_1$ and the $\Theta(\sqrt{n})$ values of the $y$-factor
label $j_1$; both index sets have this cardinality by
Eqs.~\eqref{eq:subdiv_params} and~\eqref{eq:rep_factor}. In each cluster $i_1$
initialize the single \rd\ logical qubit $\red{\vc{i_1}}=(i_1,\pi(i_1))$ in
$\lo{\ket{+}}$ and all remaining \rd\ logical qubits of that cluster in
$\lo{\ket{0}}$; initialize all \bl\ logical qubits in $\lo{\ket{+}}$.

A logical CZ acts as the identity whenever one of its two arguments is in
$\lo{\ket{0}}$, so in
Eq.~\eqref{eq:logical_CZ_instantiated} only the \rd\ qubits prepared in
$\lo{\ket{+}}$ contribute. For the symmetry operator labelled by $\green{i_3}$ the
surviving factors are those with $\red{i_1}=\green{i_3}$ and
$\blue{j_2}=\red{j_1}$, i.e.\ those in the single cluster $i_1=i_3$, which
contains exactly one active \rd\ qubit, the one with $j_1=\pi(i_3)$. Hence
\begin{align}\label{eq:single_pair}
\widetilde{\text{CZ}}^{\rd,\bl}\big(\green{\gamma^{0}_{\vc{i_3}}}\big)
&\;\widehat{=}\;
\lo{\text{CZ}}^{\rd,\bl}_{\red{\vc{i_1}},\blue{j_2}} , \cr
\red{\vc{i_1}} &= \big(\green{i_3},\,\pi(\green{i_3})\big) ,
\qquad \blue{j_2} = \pi(\green{i_3}), \cr
\end{align}
a \emph{single} logical CZ. Distinct $\green{i_3}$ give distinct
$\pi(\green{i_3})$ and hence act on disjoint pairs, so the
$\Theta(\sqrt{n})$ projectors in Eq.~\eqref{eq:output_logical_state} commute and
factorize over pairs,
\be\label{eq:fountain_output}
\lo{\ket{\psi}}_f = \prod_{\green{i_3}}
\hat{P}_{1-2\green{\rho_{\vc{i_3}}}}
\big[\lo{\text{CZ}}^{\rd,\bl}_{\red{\vc{i_1}},\blue{j_2}}\big]\
\lo{\ket{++}}_{\red{\vc{i_1}},\blue{j_2}} .
\ee
On each pair, the outcome $\green{\rho}=0$ yields the logical CZ magic state
\begin{align}\label{eq:CZ_magic}
\hat{P}_{+1}\big[\lo{\text{CZ}}^{\rd,\bl}\big]\lo{\ket{++}}
= \lo{\ket{\text{CZ}}}
\equiv \tfrac{1}{\sqrt{3}}\big(\lo{\ket{00}}+\lo{\ket{01}}+\lo{\ket{10}}\big),
\end{align}
while $\green{\rho}=1$ yields
$\hat{P}_{-1}[\lo{\text{CZ}}^{\rd,\bl}]\lo{\ket{++}}=\lo{\ket{11}}$. Each outcome
occurs with constant probability, so a single execution prepares $\Theta(\sqrt{n})$
disjoint magic states on average. The distance statement is
Theorem~\ref{thm:parameters}, the $O(d)$ rounds of error correction in step~2 of
Sec.~\ref{sec:protocol} ensuring that the outcomes \eqref{eq:outcome_product} are
themselves protected.
\end{proof}

\begin{remark}[Yield of magic state production]\label{rem:fountain_rate}
The $\Theta(\sqrt{n})$ count is set by the constant-sheaf factors: it is the
number of independent \gr-type symmetry labels, $\dim H^0(\L,\Fc^{\gr})$, which by
Eq.~\eqref{eq:rep_factor} is $1\cdot\Theta(\bar n)+\Theta(\bar n)\cdot 1$. It is
the same quantity that appeared as the number of constraints removed by the twist
in Eq.~\eqref{eq:number_of_constraints}, and as the sub-extensive logical
dimension of the \bl\ copy in Eq.~\eqref{eq:k_blue}: one and the same
asymmetry of Construction~\ref{con:2d-hgp} governs all three.
\end{remark}

\subsection{Spacetime path integral as logical action}
\label{sec:spacetime_path_integral}

We now rederive the logical action \eqref{eq:output_logical_state} from the
spacetime path integral \eqref{eq:path_inegral_cup}, which both confirms the
operator-level derivation and exhibits the protocol as a spacetime object.

\subsubsection{General formalism}
\label{sec:general_formalism}

Write the spacetime path integral on the sheaf complex
$\tilde{\L}=\L\otimes I_t$ generically as
\be\label{eq:generic_path_integral}
\cZ = \sum_{\vec{\vc{c}}}\ \prod_{\sigma} w_\sigma\big(\vec{\vc{c}}\big) ,
\qquad
\vec{\vc{c}} \equiv \{\red{a^1},\blue{b^1},\green{c^1}\} ,
\ee
the sum running over sheaf $1$-cochain configurations of the three colours. Its
value is the contraction of the corresponding spacetime tensor network.

The path integral acquires an operator meaning when $\tilde{\L}$ carries
\emph{state boundaries} $\B_I$ and $\B_F$, i.e.\ when it is a spacetime cylinder
$\tilde{\L}=\L\otimes I_t$ whose two ends are copies of the spatial complex $\L$
and at which the cochain variables are left uncontracted. For each boundary
configuration $\vec{\vc{c}}_\B$ the path integral gives an amplitude
$\cZ(\vec{\vc{c}}_\B)$, which defines boundary states via
$\bket{\vec{\vc{c}}_\B}{\psi_\B}=\cZ(\vec{\vc{c}}_\B)$; the bulk then defines an
operator $\T$ between them. This generalizes the usual TQFT association of states
to boundaries of spacetime manifolds: here the ket and bra are associated with
sheaf complexes with state boundaries.

Basis cohomology classes at the two boundaries correspond one-to-one with logical
$Z$-basis states. Denote the $i^{\text{th}}$ input and output basis classes by
$I_i$ and $F_i$, so that a general boundary class reads
$[a_{\text{in}}]=\sum_i u'_i I_i$ and $[a_{\text{out}}]=\sum_i u''_i F_i$, with
$\FF_2$ coefficient vectors $\vec{u}'$ and $\vec{u}''$ labelling the logical basis
states $\ket{\vec{u}'}$, $\ket{\vec{u}''}$. Choosing a bulk cocycle basis
$\{\alpha_j\}$ and expanding $[a_{\text{bulk}}]=\sum_j l_j\alpha_j$, the
restrictions of $\alpha_j$ to the boundaries define $\FF_2$-valued restriction
matrices $\cI_{i,j}$ and $\cF_{i,j}$, and
\be\label{eq:matrix_element}
\boket{\vec{u}''}{\T}{\vec{u}'}
= \sum_{\vec{l}:\ \cI\cdot\vec{l}=\vec{u}',\ \cF\cdot\vec{l}=\vec{u}''}
\cZ\big(\vec{l}\,\big) .
\ee
In our protocol $\T$ is not unitary but a measurement channel: there are
non-trivial \emph{relative} classes $C_i$ connecting the two boundaries --- the
\gr\ charge worldlines --- and an arbitrary such class expands as
$\sum_i\green{\rho_i}C_i$. The path integral therefore depends on the outcome
vector $\green{\vec{\rho}}$, and the logical action is the collection of operators
$\T^{\green{\vec{\rho}}}$ with
\be\label{eq:matrix_element_rho}
\boket{\vec{u}''}{\T^{\green{\vec{\rho}}}}{\vec{u}'}
= \sum_{\vec{l}:\ \cI\cdot\vec{l}=\vec{u}',\ \cF\cdot\vec{l}=\vec{u}''}
\cZ^{\green{\vec{\rho}}}\big(\vec{l}\,\big) .
\ee

\subsubsection{Imaginary-time versus real-time evolution}
\label{sec:imaginary_time}

Before evaluating Eq.~\eqref{eq:matrix_element_rho} we recall the physical status
of the path integral, following Ref.~\cite{Tsui:2019ykk}. The Euclidean path
integral computes matrix elements of $e^{-TH_\infty}$ for a Hamiltonian with
eigenvalues $0$ and $\infty$; $\T$ is therefore a projector onto the ground space,
and can be read as a non-unitary circuit built from local unitaries and stabilizer
projectors $P_\sigma=\frac{1}{2}(\I+S_\sigma)$. It corresponds to a
\emph{post-selected} evolution in which every stabilizer measurement returns $+1$.

Actual error correction is a real-time process in which those projectors are
replaced by measurements with outcomes $\pm 1$; the $-1$ outcomes mark topological
defects, captured by a defect-decorated path integral. If the decoder closes the
defects in a homologically trivial way, topological invariance makes the
defect-decorated path integral equal to the defect-free one, so the logical action
computed below is also the logical action of the error-corrected process. This is
the sheaf-code instance of the spacetime-code perspective developed in
Refs.~\cite{Bauer:2023awl,Bauer:2024qpc,Bauer:2024alh,Davydova:2025ylx}; here it
extends to a Clifford stabilizer code, so that the corresponding circuit is
non-Clifford.

\subsubsection{Derivation of the logical action}
\label{sec:derivation_logical_action}

The input of the protocol is the untwisted $\red{\ZZ_2}\times\blue{\ZZ_2}$ sheaf
gauge theory, whose boundary cocycle bases we denote
$\{\red{I}_{\red{\vc{i_1}}}\}$, $\{\blue{I}_{\blue{\vc{i_2}}}\}$ and
$\{\red{F}_{\red{\vc{i_1}}}\}$, $\{\blue{F}_{\blue{\vc{i_2}}}\}$. In the bulk we
are in the twisted theory, with the spacetime cocycle basis
$\{\red{\tilde{\alpha}^{1}_{\vc{i_1}}}\}$,
$\{\blue{\tilde{\beta}^{1}_{\vc{i_2}}}\}$,
$\{\green{\tilde{\gamma}^{1}_{\vc{i_3}}}\}$ of
Eq.~\eqref{eq:spatial_cocycle_decomposition}. The \rd\ and \bl\ bulk classes
restrict to the boundaries as
$\red{\tilde{\alpha}^{1}_{\vc{i_1}}}|_{\B_I}=\red{I}_{\red{\vc{i_1}}}$ and
$\red{\tilde{\alpha}^{1}_{\vc{i_1}}}|_{\B_F}=\red{F}_{\red{\vc{i_1}}}$, and
likewise for \bl, so the restriction matrices are diagonal,
\begin{align}\label{eq:restriction_matrices}
\cI_{\red{\vc{i_1}'},\red{\vc{i_1}}} &=
\delta_{\red{\vc{i_1}'},\red{\vc{i_1}}} , \cr
\cI_{\blue{\vc{i_2}'},\blue{\vc{i_2}}} &=
\delta_{\blue{\vc{i_2}'},\blue{\vc{i_2}}} ,
\end{align}
all other entries vanishing,
and similarly for $\cF$. The \gr\ sector instead contributes relative classes: the
charge worldlines terminating on $\B_I$ and $\B_F$ are labelled by the classes
dual to $\{\green{\tilde{\gamma}^{1}_{\vc{i_3}}}\}$, with expansion coefficients
$\green{\vec{\rho}}=\{\green{\rho_{\vc{i_3}}}\}$ --- precisely the measurement
outcomes \eqref{eq:outcome_product}.

Parameterizing the bulk classes by $\red{\vec{n}}=\{\red{n_{\vc{i_1}}}\}$,
$\blue{\vec{m}}=\{\blue{m_{\vc{i_2}}}\}$ and
$\green{\vec{l}}=\{\green{l_{\vc{i_3}}}\}$, the path-integral weight conditioned
on the outcomes is, from Eq.~\eqref{eq:path_inegral_cup},
\begin{align}\label{eq:Z_component}
&\cZ^{\green{\vec{\rho}}}\big(\red{\vec{n}},\blue{\vec{m}},\green{\vec{l}}\big) \cr
=& \prod_{\green{\vc{i_3}}}(-1)^{\green{\rho_{\vc{i_3}}}\green{l_{\vc{i_3}}}}
\cdot \prod_{\red{\vc{i_1}},\blue{\vc{i_2}}}
\Big[(-1)^{\red{n_{\vc{i_1}}}\blue{m_{\vc{i_2}}}\green{l_{\vc{i_3}}}}
\Big]^{\int_{\tilde{\eta}_3}\red{\tilde{\alpha}^{1}_{\vc{i_1}}}\cup
\blue{\tilde{\beta}^{1}_{\vc{i_2}}}\cup
\green{\tilde{\gamma}^{1}_{\vc{i_3}}}}. \cr
\end{align}
The exponent is exactly the invariant of Sec.~\ref{sec:condition_cohomology}, and
by Lemma~\ref{lemma:factorization} the temporal factor drops out,
\be\label{eq:exponent_is_Lambda}
\int_{\tilde{\eta}_3}\red{\tilde{\alpha}^{1}_{\vc{i_1}}}\cup
\blue{\tilde{\beta}^{1}_{\vc{i_2}}}\cup\green{\tilde{\gamma}^{1}_{\vc{i_3}}}
= \int_{\eta_2}\red{\alpha^{1}_{\vc{i_1}}}\cup\blue{\beta^{1}_{\vc{i_2}}}\cup
\green{\gamma^{0}_{\vc{i_3}}}
= \Lambda\big(\red{\vc{i_1}},\blue{\vc{i_2}},\green{\vc{i_3}}\big) ,
\ee
by Eq.~\eqref{eq:spacetime_equals_spatial}. The spacetime derivation and the
operator derivation are thus controlled by the same coefficient matrix; this
identity is what the pullback to the skeleton accomplished in
Ref.~\cite{ZhuKobayashiHsin2026}, and here it is a one-line consequence of the
K\"unneth factorization.

By Eq.~\eqref{eq:restriction_matrices} the matrix
\eqref{eq:matrix_element_rho} is diagonal: if
$(\red{\vec{u}'},\blue{\vec{v}'})\neq(\red{\vec{u}''},\blue{\vec{v}''})$ no
$(\red{\vec{n}},\blue{\vec{m}})$ satisfies both boundary conditions and the
element vanishes. On the diagonal,
$(\red{\vec{n}},\blue{\vec{m}})=(\red{\vec{u}'},\blue{\vec{v}'})$ is fixed and only
$\green{\vec{l}}$ is summed, so the sum factorizes over $\green{\vc{i_3}}$,
\be\label{eq:T_factorization}
\T^{\green{\vec{\rho}}} = \prod_{\green{\vc{i_3}}}
\tilde{\T}^{\green{\rho_{\vc{i_3}}}}_{\green{\vc{i_3}}} .
\ee
Each factor is evaluated by summing $\green{l_{\vc{i_3}}}\in\{0,1\}$. For
$\green{\rho_{\vc{i_3}}}=0$,
\begin{align}\label{eq:T_rho_zero}
&\boket{\red{\vec{n}},\blue{\vec{m}}}
{\tilde{\T}^{\green{0}}_{\green{\vc{i_3}}}}
{\red{\vec{n}},\blue{\vec{m}}} \cr
=&\ 1 + \prod_{\red{\vc{i_1}},\blue{\vc{i_2}}}
\big[(-1)^{\red{n_{\vc{i_1}}}\blue{m_{\vc{i_2}}}}\big]^{\Lambda} \cr
\propto&\ \boket{\red{\vec{n}},\blue{\vec{m}}}
{\hat{P}_{+1}\Big[\prod_{\red{\vc{i_1}},\blue{\vc{i_2}}}
\big[\lo{\text{CZ}}^{\rd,\bl}_{\red{\vc{i_1}},\blue{\vc{i_2}}}\big]^{\Lambda}
\Big]}{\red{\vec{n}},\blue{\vec{m}}} ,
\end{align}
where we used
$(-1)^{\red{n_{\vc{i_1}}}\blue{m_{\vc{i_2}}}}\equiv
\lo{\text{CZ}}^{\rd,\bl}_{\red{\vc{i_1}},\blue{\vc{i_2}}}$ on the logical
$Z$-basis and dropped an overall factor of $\frac{1}{2}$. For
$\green{\rho_{\vc{i_3}}}=1$ the sum carries the extra sign
$(-1)^{\green{l_{\vc{i_3}}}}$ and gives
$1-\prod[\,\cdot\,]^{\Lambda}\propto\hat{P}_{-1}[\,\cdot\,]$ instead. Combining
with Eq.~\eqref{eq:T_factorization},
\begin{align}\label{eq:logical_action_form}
\T^{\green{\vec{\rho}}}
=&\ \prod_{\green{\vc{i_3}}}
\hat{P}_{1-2\green{\rho_{\vc{i_3}}}}
\Big[\prod_{\red{\vc{i_1}},\blue{\vc{i_2}}}
\big[\lo{\text{CZ}}^{\rd,\bl}_{\red{\vc{i_1}},\blue{\vc{i_2}}}
\big]^{\Lambda(\red{\vc{i_1}},\blue{\vc{i_2}},\green{\vc{i_3}})}\Big] \cr
\equiv&\ \prod_{\green{\vc{i_3}}}
\hat{P}_{1-2\green{\rho_{\vc{i_3}}}}
\Big[\widetilde{\text{CZ}}^{\rd,\bl}
\big(\green{\gamma^{0}_{\vc{i_3}}}\big)\Big] ,
\end{align}
which reproduces Eq.~\eqref{eq:output_logical_state}, now derived from the
spacetime path integral rather than from the stabilizer bookkeeping of
Sec.~\ref{sec:protocol}. 

In the instantiation with Construction \ref{con:2d-hgp} and with the explicit evaluation
$\Lambda=\delta_{\red{i_1},\green{i_3}}\delta_{\red{j_1},\blue{j_2}}$ of
Eq.~\eqref{eq:logical_CZ_instantiated} and the initialization of
Theorem~\ref{thm:fountain}, each factor of
Eq.~\eqref{eq:logical_action_form} is a projector onto a single logical CZ acting
on a disjoint pair.

\begin{remark}[Comparison with the CW-complex route]\label{rem:pullback_comparison}
In Ref.~\cite{ZhuKobayashiHsin2026} the corresponding logical action was first
derived on the high-dimensional spacetime complex, in terms of cocycles
$\red{\tilde{\alpha}^{8}}$, $\blue{\tilde{\beta}^{6}}$,
$\green{\tilde{\gamma}^{3}}$ and their Poincar\'e duals, and then \emph{pulled
back} to the $2$D skeleton, where it became a statement about
$\red{\bar{\alpha}^{1}}$, $\blue{\bar{\beta}^{1}}$, $\green{\bar{\gamma}^{0}}$.
Here the degrees $1,1,0$ appear from the outset, being fixed by
$\mathsf{d}=2$ and the K\"unneth decomposition
\eqref{eq:cocycle_decomposition}; the pullback has no analogue because there is
nothing to pull back from. The Poincar\'e duals of the CW route are likewise
replaced by intrinsic sheaf objects, the capped cycle
$\eta_2^{\gamma}=\eta_2\frown\green{\gamma^{0}}$ of
Eq.~\eqref{eq:measure_subcomplex} on the operator side and the pairing
$\mathcal{P}$ of Eq.~\eqref{eq:P-def} on the logical side.
\end{remark}

%======================================================================
\section{Almost-constant magic rate from local codes with a multiplication
property}
\label{sec:high_yield}
%======================================================================

It is convenient to measure the usefulness of a gauging-measurement protocol not
by the absolute number of magic states it produces but by that number per
physical qudit.

\begin{definition}[\textit{Magic rate}]\label{def:magic_rate}
For a code of block length $n$ admitting a gauging measurement that prepares
$s$ logical CZ magic states on disjoint pairs of logical qudits in a single run,
the \emph{magic rate} is
\be\label{eq:magic_rate}
\varrho \;\equiv\; \frac{s}{n} .
\ee
We say the magic rate is \emph{almost constant} if
$\varrho \ge n^{-\epsilon}$ for arbitrarily small $\epsilon>0$.
\end{definition}

The quantity $s$ is what Ref.~\cite{GolowichTamoZhu2026} calls the \emph{logical
yield} of a transversal gate; we use ``yield'' only when quoting results from
that reference, and ``magic rate'' for the corresponding statements here.

Construction~\ref{con:2d-hgp} produces a constant-rate code, but only
$\Theta(\sqrt{n})$ independently addressable transversal logical CZ gates
(Remark~\ref{rem:asymmetry_rate}, Lemma~\ref{lemma:rate_twisted}), hence only
$\Theta(\sqrt{n})$ magic states per run of the gauging measurement
(Theorem~\ref{thm:fountain}), i.e.\ a magic rate
$\varrho=\Theta(n^{-1/2})$ that decays with system size. The reason is visible in
Eq.~\eqref{eq:rep_factor}: one tensor factor of the \bl\ and \gr\ sheaves is the
constant sheaf $\Fc^0$, for which $\dim H^0 = 1$. The constant sheaf was chosen
only because it makes the local cycle condition
\eqref{eq:multiplication_property_local} hold for free,
$\cC^{\perp}_{v} * \cC_{v} \subseteq \text{Rep}^{\perp}$ being a restatement of
duality. Removing it requires classical LDPC codes whose local codes satisfy a
genuine multiplication property while retaining large dimension and distance.

Such codes have recently been constructed by Golowich, Tamo and
Zhu~\cite{GolowichTamoZhu2026}, over an alphabet $\Fq$ of large constant size ---
which is exactly why the $\Fq$ generalization of
Sec.~\ref{sec:Fq_generalization} was needed. In this section we feed them into
Construction~\ref{con:2d-hgp}. The outcome is a family of twisted sheaf codes
with $n^{1-\epsilon}$ independently addressable transversal logical CZ gates on
\emph{disjoint} pairs of logical qudits, hence an \emph{almost constant} magic
rate $\varrho \ge n^{-\epsilon}$, at essentially unchanged distance.

\subsection{Sheaves and systems of local codes: a dictionary}
\label{sec:dictionary}

The objects of Ref.~\cite{GolowichTamoZhu2026} are \emph{systems of local codes},
which are precisely the sheaves of Definition~\ref{def:presheaf} on a graph. We
record the translation, which is the content of \cite[Fact
2.21]{GolowichTamoZhu2026}.

\begin{lemma}[Dictionary]\label{lemma:dictionary}
Let $G=(V,E)$ be a bipartite graph with all edges oriented from left to right, and
let $\Fc$ be a sheaf on $G$ over $\Fq$ with local parity-check matrices
$\Hs_v$, local codes $\cC_v=\ker\Hs_v$ and assembled co-restrictions
$\Fc_{v \shortrightarrow E(v)} = \Hs_v^{T}$, so that
$\im \Fc_{v\shortrightarrow E(v)} = \cC^{\perp}_{v}$. Define the associated
\emph{Tanner code}
\begin{align}\label{eq:tanner_from_sheaf}
\T^{\dag}(G,\Fc) = \Big\{ c \in \Fq^{E} \ :\ & c|_{E(v)} \in
\im \Fc_{v \shortrightarrow E(v)} \cr
& = \cC^{\perp}_{v}
\ \ \forall v \in V \Big\} .
\end{align}
Then
\be\label{eq:dictionary}
H_1\big(G,\Fc^{\perp}\big) \;=\; \T^{\dag}(G,\Fc) \;\cong\; H^0\big(G,\Fc\big) ,
\ee
the isomorphism sending $c \in \T^{\dag}(G,\Fc) $ to the $0$-cocycle $z$ determined by
$c|_{E(v)} = \Fc_{v\shortrightarrow E(v)}(z_v)$. Moreover, if $G$ has maximum
degree $\Delta$,
\be\label{eq:dictionary_distance}
d_1\big(G,\Fc^{\perp}\big) = d\big(\T^{\dag}(G,\Fc) \big)
\;\leq\; \tfrac{\Delta}{2}\, d^{0}\big(G,\Fc\big) .
\ee
\end{lemma}

\begin{proof}
The first equality of Eq.~\eqref{eq:dictionary} is
Proposition~\ref{Prop:sheaf_homology} applied to $\Fc^{\perp}$: a $1$-chain is a
cycle iff its restriction to each $E(v)$ lies in
$\ker \Hs^{\perp}_{v} = (\cC^{\perp}_{v})^{\perp\perp}$, i.e.\ in
$\cC^{\perp}_{v}$. The bipartite orientation guarantees that every edge of $E(v)$
enters the boundary with the same sign, so no relative signs appear inside the
local constraint; for $\mathrm{char}\,\Fq=2$ the hypothesis is unnecessary
(Remark~\ref{rem:signs}). The isomorphism with $H^0(G,\Fc)$ is
Proposition~\ref{Prop:sheaf_cohomology}: a $0$-cochain $z$ is a cocycle iff
$\Fc_{v\shortrightarrow e}(z_v)=\Fc_{v'\shortrightarrow e}(z_{v'})$ on every edge
$e=(v,v')$, i.e.\ iff the common value assembles into an element of $\T^{\dag}(G,\Fc) $,
and the map is invertible because each $\Fc_{v\shortrightarrow E(v)}$ is
injective. For Eq.~\eqref{eq:dictionary_distance}, each non-zero edge of
$c$ forces both of its endpoints to be non-zero in $z$, and each vertex meets at
most $\Delta$ edges, so $|z|\ge 2|c|/\Delta$.
\end{proof}

Under Eq.~\eqref{eq:dictionary} the classical sheaf codes of
Propositions~\ref{Prop:sheaf_homology} and~\ref{Prop:sheaf_cohomology} and the
Tanner codes of Ref.~\cite{GolowichTamoZhu2026} are the same objects, and the
Schur product of local codes, Definition~\ref{def:schur_product}, matches their
product $\Fc * \Fc'$ of systems of local codes, defined by
$\im(\Fc * \Fc')_{v\shortrightarrow E(v)} =
\im\Fc_{v\shortrightarrow E(v)} * \im\Fc'_{v\shortrightarrow E(v)}$. We write
$\Fc^{*2} = \Fc * \Fc$.

\subsection{The classical input}
\label{sec:classical_input}

We use the following, which collects \cite[Theorem 4.1 and Lemma
5.5]{GolowichTamoZhu2026} in the parameter regime of \cite[Corollary
5.1]{GolowichTamoZhu2026}, i.e.\ the regime with \emph{constant} check weight.

\begin{theorem}[Golowich--Tamo--Zhu \cite{GolowichTamoZhu2026}]
\label{thm:GTZ_classical}
For every $\epsilon>0$ and every prime $p$ there are constants
$u = O_{\epsilon}(1)$ and $\Delta = O_{\epsilon}(1)$, with $q=p^{u}$, and an
infinite family of bipartite graphs $\Gamma$ of maximum degree $\Delta$ carrying a
sheaf $\Fc$ over $\Fq$ with the following properties, where $N=|E(\Gamma)|$:
\begin{enumerate}[label=(\roman*)]
\item \emph{(multiplication property)} at every vertex $v$,
$\cC^{\perp}_{v} * \cC^{\perp}_{v} \subseteq \im \Fc^{*2}_{v\shortrightarrow
E(v)}$, by definition of $\Fc^{*2}$;
\item \emph{(dimension)} $\dim \T^{\dag}(\Gamma,\Fc) \ \ge\ N^{1-\epsilon}$, and there is
an \emph{extendable} set $M \subseteq E(\Gamma)$ --- one contained in an
information set --- with $|M| = \dim \T^{\dag}(\Gamma,\Fc) \ge N^{1-\epsilon}$;
\item \emph{(distance)} $d\big(\T^{\dag}(\Gamma,\Fc)\big) \ge N^{1-\epsilon}$ and
$d\big(\T^{\dag}(\Gamma,\Fc^{*2})\big) \ge N^{1-\epsilon}$;
\item \emph{(decoding)} both Tanner codes admit radius-$N^{1-\epsilon}$
$\mathrm{poly}(N)$-time decoders.
\end{enumerate}
\end{theorem}

The codes are punctured $t$-fold tensor powers of a constant-sized
algebraic-geometry code $C$ with $C^{*2}$ of large distance, arranged on a
$t$-partite graph whose vertices are the axis-parallel lines of a hypercube
$[2n]^{t}$; the puncturing is what makes a tensor code expressible as a Tanner
code with constant check weight. Property~(i) is inherited from the multiplication
property of the algebraic base code, since the product of two polynomials of
degree $<k$ has degree $<2k$.

\begin{remark}[Why $\Fq$ is unavoidable]\label{rem:why_large_q}
A local code of length $\Delta$ whose Schur square is still a proper code must
have rate $\le 1/2$, and to have both large dimension and large distance at rate
$\le 1/2$ it must be close to MDS, which over $\FF_2$ is impossible for
$\Delta>2$. This is the precise sense in which the constant sheaf of
Construction~\ref{con:2d-hgp} was forced upon us at $q=2$, and it is removed by
the $\Fq$ theory of Sec.~\ref{sec:Fq_generalization}.
\end{remark}

\subsection{The construction}
\label{sec:high_yield_construction}

Besides $\Fc$ and $\Fc^{*2}$ we shall use the \emph{dual} of the Schur square,
\be\label{eq:F_diamond}
\Fc^{\diamond} := \big(\Fc^{*2}\big)^{\perp} ,
\qquad
\im\Fc^{\diamond}_{v\shortrightarrow E(v)}
= \big(\cC^{\perp}_{v}*\cC^{\perp}_{v}\big)^{\perp} ,
\ee
for which the dictionary \eqref{eq:dictionary}, applied to $\Fc^{*2}$, reads
\begin{align}\label{eq:F_diamond_cycles}
H_1\big(\Gamma,\Fc^{\diamond}\big)&=\T^{\dag}\big(\Gamma,\Fc^{*2}\big), \cr
d_1\big(\Gamma,\Fc^{\diamond}\big)&=d\big(\T^{\dag}(\Gamma,\Fc^{*2})\big)\ \ge\ N^{1-\epsilon} .
\end{align}
Thus the $1$-cycles of $\Fc^{\diamond}$ are the codewords of the second ``good''
Tanner code of Theorem~\ref{thm:GTZ_classical}(iii), while the $0$-cocycles of
$\Fc$ are the codewords of the first. This is the property that decides in which
K\"unneth slot each sheaf must sit.

\begin{construction}[High-magic-rate $2$D HGP sheaf complexes]\label{con:high_yield}
Let $\Gamma$ and $\Fc$ be as in Theorem~\ref{thm:GTZ_classical}, and take
$G_x = G_y = \Gamma$. Equip the square complex $\L = G_x \times G_y$ with the
three product sheaves
\begin{align}\label{eq:high_yield_sheaves}
\Fc^{\rd} &= \Fc^{\diamond} \boxtimes \Fc , \cr
\Fc^{\bl} &= \Fc \boxtimes \Fc^{\diamond} , \cr
\Fc^{\gr} &= \Fc \boxtimes \Fc ,
\end{align}
and take the twist to be the fundamental class
$\eta_2 = \zeta\otimes\zeta' = [\Gamma]\otimes[\Gamma] = [\L]$. The twisted code
$\tilde{\C}$ is then defined by Eq.~\eqref{eq:Clifford_group_Fq}, with
$\tilde{\eta}_3 = \eta_2\otimes I_t$. Here $\Fc_1\boxtimes\Fc_2$ denotes the
external tensor product of a sheaf on $G_x$ with a sheaf on $G_y$, i.e.\ the
product sheaf on $\L=G_x\times G_y$ with stalks $\Fc_1(\sigma)\otimes\Fc_2(\tau)$,
whose cochain complex is $C^\bullet(G_x,\Fc_1)\otimes C^\bullet(G_y,\Fc_2)$.
\end{construction}

The assignment \eqref{eq:high_yield_sheaves} is dictated by the K\"unneth
labelling of Eq.~\eqref{eq:spatial_cocycle_decomposition}: the \rd\ logical
classes that will be used carry their degree $1$ on the $x$-factor and the \bl\
ones on the $y$-factor, while the \gr\ symmetry labels have degree $0$ on both
factors. Each dynamical colour therefore carries $\Fc^{\diamond}$ in its
degree-one slot and $\Fc$ in its degree-zero slot, and the symmetry colour
carries $\Fc$ twice. This is exactly the slot pattern of the product complexes
$\mathcal A^{(i)}$ of Ref.~\cite[Sec.~3.1]{GolowichTamoZhu2026}, in which the dual
system $\Fc^{(r)\perp}$ occupies the degree-one slot: the conjugate ($Z$-type)
logical operators of a degree-one class are $1$-cycles of the sheaf in that slot,
and by Eq.~\eqref{eq:F_diamond_cycles} these are heavy for $\Fc^{\diamond}$,
whereas for $\Fc$ itself they would be the codewords of
$\T^{\dag}(\Gamma,\Fc^{\perp})$, which Theorem~\ref{thm:GTZ_classical} does not
control. Remark~\ref{rem:slot_choice} below shows that this is not a technicality.

\begin{lemma}[Admissibility]\label{lemma:admissible}
The fundamental class $[\L]$ is a $2$-cycle of
$\Fc^{\rd}\otimes\Fc^{\bl}\otimes\Fc^{\gr}$, so
Construction~\ref{con:high_yield} defines a twisted sheaf code with a non-trivial
cohomology invariant.
\end{lemma}

\begin{proof}
By Lemma~\ref{lemma:fundamental_class} and
Proposition~\ref{prop:multiplication_property_equivalence} it suffices to check the
local cycle condition of Lemma~\ref{lemma:local_cycle} on each spatial factor,
i.e.\ that at every vertex $v$ the Schur product of any three words taken from the
three local codes $\im\Fc^{i}_{\nu,v\shortrightarrow E(v)}$, $i=\rd,\bl,\gr$, has
coordinate sum zero. On the $x$-factor the three local codes are
$(\cC^{\perp}_{v}*\cC^{\perp}_{v})^{\perp}$, $\cC^{\perp}_{v}$ and
$\cC^{\perp}_{v}$ (for \rd, \bl, \gr), and for $c$ in the first and $a,a'$ in the
other two, $\sum_{e\in E(v)}c_e a_e a'_e=\langle c, a*a'\rangle=0$ by definition of
$\Fc^{\diamond}$; this is Theorem~\ref{thm:GTZ_classical}(i) read as a triple
orthogonality, Eq.~\eqref{eq:triple_orthogonality}. On the $y$-factor the local
codes are $\cC^{\perp}_v$, $(\cC^{\perp}_v*\cC^{\perp}_v)^{\perp}$, $\cC^{\perp}_v$
and the same computation applies. Hence $\zeta=[\Gamma]$ and $\zeta'=[\Gamma]$ are
$1$-cycles of the respective tensor sheaves and $\eta_2=\zeta\otimes\zeta'$ is a
$2$-cycle by Lemma~\ref{lemma:local_cycle}.
\end{proof}

\subsection{Logical CZ structure: a perfect matching}
\label{sec:high_yield_CZ}

We now compute the coefficient $\Lambda$ of Eq.~\eqref{eq:logical_CZ_Fq} for
Construction~\ref{con:high_yield}. The outcome replaces the Kronecker-delta
structure $\Lambda = \delta_{i_1 i_2}\delta_{j_1 j_3}$ of
Eq.~\eqref{eq:logical_CZ_instantiated} by a \emph{perfect matching}: each \gr\
symmetry label picks out exactly one \rd--\bl\ pair, and distinct labels pick out
disjoint pairs.

For $e \in M$ let $\bs_{e} \in \T^{\dag}(\Gamma,\Fc)$ be a codeword with
$\bs_{e}|_{M} = \delta_{e}$, which exists because $M$ is extendable, and let
$\bs^0_{e} \in H^0(\Gamma,\Fc)$ be the corresponding $0$-cocycle under
Lemma~\ref{lemma:dictionary}. Let $\bar{e}\in C^1(\Gamma,\Fc^{\diamond})=\Fq^{E(\Gamma)}$
be the elementary $1$-cochain of Eq.~\eqref{eq:cochain_expansion_untwisted}
supported on the single edge $e$ with unit value (every edge stalk is $\Fq$, so
$\bar e$ makes sense for any of the sheaves), and let
\be\label{eq:square_indicator}
\varpi^{2}_{(e,f)} := \bar e\otimes\bar f
\in C^2\big(\L,\Fc^{\rd}\otimes\Fc^{\bl}\big)
\ee
be the indicator of the square $e\times f$. We write $J:=M\times M$ and
$s:=|J|=|M|^2$.

\begin{lemma}[Matching structure]\label{lemma:matching}
Adopt the K\"unneth labelling of Eq.~\eqref{eq:spatial_cocycle_decomposition}, and
for $(e,f) \in J$ choose
\begin{align}\label{eq:high_yield_basis}
\red{\alpha^{1}_{(e,f)}} &= \red{\bar{e}} \otimes \red{\bs^0_{f}} , \cr
\blue{\beta^{1}_{(e,f)}} &= \blue{\bs^0_{e}} \otimes \blue{\bar{f}} , \cr
\green{\gamma^{0}_{(e,f)}} &= \green{\bs^0_{e}} \otimes \green{\bs^0_{f}} .
\end{align}
Then, at the level of cochains,
\begin{align}\label{eq:cochain_matching}
\red{\alpha^1_{j}}\cup\blue{\beta^1_{k}}&=\delta_{jk}\,\varpi^2_{j}, \cr
\ell_j(\green{\gamma^0})&:=\int_{\eta_2}\varpi^2_{j}\cup\green{\gamma^0},\quad
\ell_j\big(\green{\gamma^0_{k}}\big)=\delta_{jk},
\end{align}
for $j,k\in J$ and $\green{\gamma^0}\in Z^0(\L,\Fc^{\gr})$; in particular, with $\eta_2=[\L]$,
\begin{align}\label{eq:Lambda_matching}
\Lambda\big(\red{(e_1,f_1)},\blue{(e_2,f_2)},\green{(e_3,f_3)}\big)
\qquad\qquad& \cr
= \delta_{e_1 e_2}\,\delta_{e_1 e_3}\ \delta_{f_1 f_2}\,\delta_{f_2 f_3}& .
\end{align}
The classes in Eq.~\eqref{eq:high_yield_basis} are linearly independent, so
they label $s$ logical qudits of the \rd\ copy, $s$ of the \bl\ copy, and $s$
independent symmetry generators, and $\ell=(\ell_j)_{j\in J}:Z^0(\L,\Fc^{\gr})\to\Fq^{J}$
is surjective.
\end{lemma}

\begin{proof}
By Lemma~\ref{lemma:factorization} cup products factorize over the two spatial
graphs, with the K\"unneth sign $(-1)^{|y||x'|}$ for
$(x\otimes y)\cup(x'\otimes y')$, which is $+1$ here because the second factor of
$\red{\alpha^1}$ and the first factor of $\blue{\beta^1}$ have degree $0$. On a
graph the cup product of a $1$-cochain with a $0$-cocycle is evaluated edgewise:
for a $0$-cocycle $x^0$ with associated Tanner codeword $x=\Phi(x^0)$ of
Lemma~\ref{lemma:dictionary} one has $(x^0\cup y^1)_{e}=x_{e}\,y^1_{e}$ and
$(y^1\cup x^0)_{e}=y^1_{e}\,x_{e}$, since
$\Fc_{v\shortrightarrow e}(x^0_v)=\Fc_{v'\shortrightarrow e}(x^0_{v'})=x_e$ for a
cocycle. Hence, for $j=(e_1,f_1)$ and $k=(e_2,f_2)$,
\begin{align}\label{eq:Lambda_x_eval}
\red{\alpha^1_{j}}\cup\blue{\beta^1_{k}}
&=\big(\red{\bar e_1}\cup\blue{\bs^0_{e_2}}\big)\otimes
\big(\red{\bs^0_{f_1}}\cup\blue{\bar f_2}\big) \cr
&=\bs_{e_2}(e_1)\,\bs_{f_1}(f_2)\ \bar e_1\otimes\bar f_2
=\delta_{e_1e_2}\delta_{f_1f_2}\,\varpi^2_{j}, \cr
\end{align}
using $e_1,f_2\in M$ together with $\bs_{e_2}|_M=\delta_{e_2}$ and
$\bs_{f_1}|_M=\delta_{f_1}$. Likewise
$\varpi^2_{j}\cup\green{\gamma^0_{(e_3,f_3)}}
=(\bar e_1\cup\green{\bs^0_{e_3}})\otimes(\bar f_1\cup\green{\bs^0_{f_3}})
=\delta_{e_1e_3}\delta_{f_1f_3}\,\bar e_1\otimes\bar f_1$, and since the edge
stalks of the tensor sheaf are $\Fq\otimes\Fq\otimes\Fq\cong\Fq$ and $[\L]$ is the
all-ones $2$-chain, $\int_{\eta_2}\bar e_1\otimes\bar f_1=1$. This gives
Eq.~\eqref{eq:cochain_matching} and hence Eq.~\eqref{eq:Lambda_matching}.

That $\Lambda$ depends only on the cohomology classes is
Proposition~\ref{prop:cup-properties}(c) together with
Lemma~\ref{lemma:admissible}; concretely, replacing $\red{\bar{e}_1}$ by
$\red{\bar{e}_1}+d^0\chi$ with $\chi\in C^0(\Gamma,\Fc^{\diamond})$ changes the
first factor of $\Lambda$ by
$\sum_v \big\langle \Fc^{\diamond}_{v\shortrightarrow E(v)}(\chi_v),\,
(\blue{\bs_{e_2}}*\green{\bs_{e_3}})|_{E(v)}\big\rangle$, which vanishes because
$\im\Fc^{\diamond}_{v\shortrightarrow E(v)}=(\cC^\perp_v*\cC^\perp_v)^{\perp}$
while $\bs_{e_2}*\bs_{e_3}$ restricts to $\cC^{\perp}_{v}*\cC^{\perp}_{v}$.
Independence follows from Eq.~\eqref{eq:Lambda_matching}: if
$\sum_{j}\lambda_{j}\red{\alpha^1_{j}}$ were a coboundary, pairing it against
$\blue{\beta^1_{k}}$ and $\green{\gamma^0_{k}}$ would give $\lambda_{k}=0$ for
every $k$, and similarly for the other two families. Surjectivity of $\ell$ is
$\ell(\green{\gamma^0_k})=\delta_k$.
\end{proof}

\begin{corollary}[Addressable disjoint logical CZ]\label{cor:disjoint_CZ}
In Construction~\ref{con:high_yield}, for every $(e,f)\in M\times M$ the
transversal operator $\widetilde{\text{CZ}}^{\rd,\bl}(\green{\gamma^0_{(e,f)}})$
of Eq.~\eqref{eq:CZ_Fq} implements a \emph{single} logical CZ,
\be\label{eq:single_logical_CZ}
\widetilde{\text{CZ}}^{\rd,\bl}\big(\green{\gamma^{0}_{(e,f)}}\big)
\ \widehat{=}\
\lo{\text{CZ}}^{\rd,\bl}_{\red{(e,f)},\,\blue{(e,f)}}(1) ,
\ee
acting on the \rd\ and \bl\ logical qudits carrying the same label. Distinct
labels act on disjoint pairs, so all $|M|^2$ operators commute and may be applied,
or measured, in parallel.
\end{corollary}

\begin{proof}
Immediate from Eq.~\eqref{eq:logical_CZ_Fq} and
Eq.~\eqref{eq:Lambda_matching}: the only non-vanishing multiplier for
$\green{\vc{i_3}}=(e,f)$ is the one with $\red{\vc{i_1}}=\blue{\vc{i_2}}=(e,f)$,
and its value is $1$. Disjointness is the statement that the map
$(e,f)\mapsto\big(\red{(e,f)},\blue{(e,f)}\big)$ is injective.
\end{proof}

Corollary~\ref{cor:disjoint_CZ} should be compared with
Sec.~\ref{sec:fountain}, where the disjointness had to be engineered by
initializing all but one \rd\ logical qudit per cluster in $\lo{\ket{0}}$, at a
cost of a factor $\sqrt{n}$ in the magic rate. Here the matching is built into
the code:
\emph{no} gauge-fixing of the selected \rd\ or \bl\ logical qudits is needed to
make the logical CZ gates disjoint. Gauge-fixing is still required for the
logical classes outside the chosen bases --- the complementary K\"unneth sectors
$H^0(\Gamma,\Fc^{\diamond})\otimes H^1(\Gamma,\Fc)$ of the \rd\ copy and
$H^1(\Gamma,\Fc)\otimes H^0(\Gamma,\Fc^{\diamond})$ of the \bl\ copy, which contain
short classes, and the unselected classes inside the chosen sectors --- exactly
as in Lemma~\ref{lemma:subsystem}; see Theorem~\ref{thm:GTZ_KL}.

\begin{remark}[Why the degree-one slot must carry $\Fc^{\diamond}$]
\label{rem:slot_choice}
One might be tempted to use the simpler assignment $\Fc^{\rd}=\Fc^{\bl}=\Fc\boxtimes\Fc$,
$\Fc^{\gr}=\Fc^{\diamond}\boxtimes\Fc^{\diamond}$, which also satisfies the local
cycle condition and admits the same perfect matching with
$\green{\gamma^0}$ built from $\T^{\dag}(\Gamma,\Fc^{\diamond})$. It does not have
the distance one would like. The conjugate ($Z$-type) cycles of the selected
classes $\red{\bar e}\otimes\red{\bs^0_f}$ would then lie in
$H_1(\Gamma,\Fc)\otimes H_0(\Gamma,\Fc)$, and by Eq.~\eqref{eq:dictionary}
$H_1(\Gamma,\Fc)=\T^{\dag}(\Gamma,\Fc^{\perp})$ is the Tanner code of the
\emph{dual} local codes $\cC_v$. On the GTZ graph this code has codewords of
weight at most $n_0^2=O_\epsilon(1)$, where $2n_0$ is the length of the base code
(every point of the hypercube lies on only two of the $t$ lines that are vertices
of $\Gamma$, so a dual-type Tanner codeword is constrained in two directions
only), and such a codeword tensored with a single vertex is a $Z$-type logical
operator of weight $O_\epsilon(1)$ acting non-trivially on a selected logical
qudit (Proposition~\ref{prop:short_dual_cycles_app} and
Corollary~\ref{cor:short_logical_app} in Appendix~\ref{app:GTZ_KL}). The
assignment \eqref{eq:high_yield_sheaves} avoids this by placing the degree-one
slot on $\Fc^{\diamond}$, whose $1$-cycles form $\T^{\dag}(\Gamma,\Fc^{*2})$. Note
also that the $\Fc\boxtimes\Fc$ assignment would require the inclusion
$\T^{\dag}(\Gamma,\Fc)\subseteq\T^{\dag}(\Gamma,\Fc^{\diamond})$ to produce the \gr\
labels, i.e.\ the triple orthogonality $\mathbf 1\perp\cC^\perp_v*\cC^\perp_v*\cC^\perp_v$
of the base code, which is not part of Theorem~\ref{thm:GTZ_classical}; with
$\Fc^{\gr}=\Fc\boxtimes\Fc$ the labels $\green{\gamma^0_{(e,f)}}$ exist by
Theorem~\ref{thm:GTZ_classical}(ii) alone.
\end{remark}

\subsection{Parameters}
\label{sec:high_yield_parameters}

We take as protected sectors the spans of the selected classes of
Eq.~\eqref{eq:high_yield_basis} and of their conjugate cycles,
\begin{align}\label{eq:protected_sectors_GTZ}
\mathsf L_X^{\rd}&=\mathrm{span}\{[\red{\alpha^1_{j}}]:j\in J\}, \cr
\mathsf L_Z^{\rd}&=\mathrm{span}\{[\bs^{*2}_{e}\otimes\hs_{f}]:(e,f)\in J\}, \cr
\mathsf L_X^{\bl}&=\mathrm{span}\{[\blue{\beta^1_{j}}]:j\in J\}, \cr
\mathsf L_Z^{\bl}&=\mathrm{span}\{[\hs_{e}\otimes\bs^{*2}_{f}]:(e,f)\in J\}, \cr
\mathsf L_X^{\gr}&=\mathsf L_Z^{\gr}=0 ,
\end{align}
where $\bs^{*2}_{e}:=\bs_e*\bs_e\in\T^{\dag}(\Gamma,\Fc^{*2})=H_1(\Gamma,\Fc^{\diamond})$
satisfies $\bs^{*2}_{e}|_M=\delta_e$ and $\hs_f\in C_0(\Gamma,\Fc)$ is a $0$-chain dual
to the $\bs^0_{f'}$, $\langle\bs^0_{f'},\hs_f\rangle=\delta_{ff'}$
(Lemma~\ref{lemma:conjugate_app}); the pairing is $\langle\red{\alpha^1_j},
\bs^{*2}_{e}\otimes\hs_f\rangle=\delta_{j,(e,f)}$, so it is non-degenerate on
$\mathsf L_X\times\mathsf L_Z$ as required by Lemma~\ref{lemma:subsystem}. All
other logical classes are gauge.

\begin{theorem}[Code parameters and magic rate]\label{thm:high_yield_parameters}
For every $\epsilon>0$ and every prime $p$, Construction~\ref{con:high_yield}
yields an infinite family of twisted sheaf codes over $\Fq$, $q=p^{O_\epsilon(1)}$,
with constant stabilizer weight, block length $n$, and
\be\label{eq:high_yield_params}
\Big[\big[\, n,\ \ k \ \ge\ n^{1-\epsilon},\ \
d \ \ge\ n^{(1-\epsilon)/2} \,\big]\Big]_q ,
\ee
supporting $s = |M|^{2}$ independently addressable transversal logical CZ gates
on disjoint pairs of logical qudits, hence an almost constant magic rate
\be\label{eq:high_yield_rate}
\varrho \;=\; \frac{s}{n} \;=\; \frac{|M|^{2}}{n} \;\ge\; n^{-\epsilon} .
\ee
Here $d$ is the subsystem-code distance of Lemma~\ref{lemma:subsystem} and
Definition~\ref{def:distance_clifford} with respect to the protected sectors
\eqref{eq:protected_sectors_GTZ}, established for the twisted code by
Theorem~\ref{thm:GTZ_KL} below without Assumption~\ref{assump:TQFT}.
\end{theorem}

\begin{proof}
\emph{Length.} The qudits sit on the $1$-cells of $\L=\Gamma\times\Gamma$, of
which there are $2|V(\Gamma)||E(\Gamma)| = \Theta(N^2)$, with stalk dimensions
$O(\Delta)=O(1)$; hence $n=\Theta(N^{2})$ and $N=\Theta(\sqrt{n})$.

\emph{Magic rate.} By Lemma~\ref{lemma:matching} and
Corollary~\ref{cor:disjoint_CZ} there are $|M|^{2}$ disjoint pairs, and
$|M|\ge N^{1-\epsilon'}$ by Theorem~\ref{thm:GTZ_classical}(ii), so
$s \ge N^{2(1-\epsilon')} = n^{1-\epsilon'}$ and
$\varrho = s/n \ge n^{-\epsilon'}$; rename $\epsilon'\to\epsilon$.

\emph{Dimension.} The protected register has $k_{\rm prot}=\dim\mathsf L^{\rd}_X
+\dim\mathsf L^{\bl}_X=2s$ logical qudits in the untwisted code
(Lemma~\ref{lemma:subsystem}(i)); in the twisted code the $s$ charge constraints
of Lemma~\ref{lemma:constraints_Fq} bind the pairs, and
Theorem~\ref{thm:GTZ_KL}(i) gives $\log_q\dim\tilde{\C}_0=s\log_q(2q-1)\ge s$ for
the gauge-fixed twisted code. Hence $k\ge s=\varrho n\ge n^{1-\epsilon}$. (The full
code space is much larger: by Eq.~\eqref{eq:kunneth_H1} the \rd\ copy alone has
$\dim H^1(\Gamma,\Fc^{\diamond})\dim H^0(\Gamma,\Fc)
+\dim H^0(\Gamma,\Fc^{\diamond})\dim H^1(\Gamma,\Fc)=\Omega(N^{2-\epsilon})$ logical
qudits, but the complementary ones are gauge.)

\emph{Distance.} For the untwisted code we use Lemma~\ref{lemma:subsystem}(ii)
with the protected sectors \eqref{eq:protected_sectors_GTZ}. A $Z$-type logical
operator acting non-trivially on $\mathsf L^{\rd}_X$ is a $1$-cycle of
$\Fc^{\diamond}\boxtimes\Fc$ pairing non-trivially with some
$\red{\bar e}\otimes\red{\bs^0_f}$; contracting it against the $0$-cocycle
$\red{\bs^0_f}$ (Lemma~\ref{lemma:kunneth_contraction}(b)) leaves a non-zero
$1$-cycle of $\Fc^{\diamond}$ on $\Gamma$, i.e.\ a non-zero codeword of
$\T^{\dag}(\Gamma,\Fc^{*2})$, of weight $\ge N^{1-\epsilon}$ by
Eq.~\eqref{eq:F_diamond_cycles}. An $X$-type logical operator with a non-zero
component in $\mathsf L^{\rd}_X$ is a $1$-cocycle pairing non-trivially with some
conjugate cycle $\bs^{*2}_{e}\otimes\hs_f$; contracting it against $\bs^{*2}_{e}$
(Lemma~\ref{lemma:kunneth_contraction}(a)) leaves a non-zero $0$-cocycle of $\Fc$
on $\Gamma$, of weight $\ge d^0(\Gamma,\Fc)\ge\frac{2}{\Delta}d(\T^{\dag}(\Gamma,\Fc))
\ge\frac{2}{\Delta}N^{1-\epsilon}$ by Lemma~\ref{lemma:dictionary}. The \bl\ copy
is identical with the two factors interchanged (Proposition~\ref{prop:reldist_app}).
Hence $d=\Omega(N^{1-\epsilon})=\Omega(n^{(1-\epsilon)/2})$. For the twisted code
the same bound is Theorem~\ref{thm:GTZ_KL}(iii) below.

\emph{Locality.} The graph has constant degree $\Delta=O_\epsilon(1)$ and the
stalks are of constant dimension, so the $A$ and $B$ generators of
Eq.~\eqref{eq:AB_Fq} have weight $O(1)$; the CZ dressing of
Eq.~\eqref{eq:dressing_Fq} is supported on bounded clusters of cells, so
$\tilde{\SS}$ is generated by constant-weight Clifford operators and the code is
qLDPC.
\end{proof}

\paragraph*{Distance of the twisted code without Assumption~\ref{assump:TQFT}.}
As in Sec.~\ref{sec:distance}, the quickest route to the distance of the twisted
code is Assumption~\ref{assump:TQFT}: the electric logical operators are
untouched by the twist and the magnetic ones are only dressed,
Eq.~\eqref{eq:support_inequality}, so Lemma~\ref{lemma:distance_twisted} applied
to the bounds just derived gives $\tilde d=\Omega(n^{(1-\epsilon)/2})$ at once.
This shortcut based on TQFT argument is not completely necessary here. The following theorem verifies the
Knill--Laflamme condition \cite{Knill1997} directly from the stabilizer group
\eqref{eq:Clifford_group_Fq}, in the convention of
Theorem~\ref{thm:high_yield_fountain} below in which every logical qudit outside
the selected pairs is fixed to $\lo{\ket 0}$ and the \gr\ copy is initialized in
$\ket 0$; this is also the convention of Ref.~\cite{GolowichTamoZhu2026}, whose
subsystem codes fix the unused logical qudits before the transversal gate.
It covers not only the twisted code itself but every outcome sector of the gauging
measurement, which is what the fountain protocol requires.

For a \gr\ $0$-chain $\sigma\in C_0(\L,\Fc^{\gr})$ --- the record of the
measurement of the \gr\ twisted generators $\tilde A^{\gr}(\lambda^0)$, with
outcome $\omega^{\mathrm{Tr}\langle\lambda^0,\sigma\rangle}$ --- let
$\tilde{\C}_\sigma$ be the joint eigenspace of (a) all flux operators
$B^{i}_f(\beta)$ with eigenvalue $1$; (b) all \gr\ electric operators
$\widetilde Z^{\gr}_{\eta_1}(\beta)$, $\eta_1\in Z_1(\L,\Fc^{\gr})$, with eigenvalue
$1$ (trivial \gr\ class); (c) all \rd\ and \bl\ electric operators
$\widetilde Z^{i}_{\eta_1}(\beta)$ with $[\eta_1]\in\mathsf G^{i}_Z=(\mathsf L^{i}_X)^{\perp}$,
with eigenvalue $1$ (gauge logical qudits fixed to $\lo{\ket 0}$); (d) all
$\tilde A^{\rd}(\lambda^0)$, $\tilde A^{\bl}(\lambda^0)$ with eigenvalue $1$; and (e)
all $\tilde A^{\gr}(\lambda^0)$ with eigenvalue
$\omega^{\mathrm{Tr}\langle\lambda^0,\sigma\rangle}$. Thus $\tilde{\C}_0$ is the
twisted code $\tilde{\C}$ with the gauge register fixed, and $\tilde{\C}_\sigma$
for $\sigma\ne0$ is the state space after a gauging measurement with a non-trivial
record; its states carry the \gr\ $X$-syndrome $\sigma$ and are not twisted-code
states. Call $\sigma$ \emph{admissible} if there is $\vc w\in\Fq^{J}$ with
\be\label{eq:admissible_record}
\langle\green{\gamma^0},\sigma\rangle=\sum_{j\in J}w_j\,\ell_j(\green{\gamma^0})
\qquad\forall\,\green{\gamma^0}\in Z^0(\L,\Fc^{\gr}),
\ee
with $\ell_j$ as in Eq.~\eqref{eq:cochain_matching}; $\vc w=\vc w(\sigma)$ is
then unique by surjectivity of $\ell$, and $w_j$ is the outcome of the logical CZ
observable of the pair $j$. Put
$A_{\vc w}:=\{(\vc a,\vc b)\in\Fq^{J}\times\Fq^{J}:a_jb_j=w_j\ \forall j\}$.

\begin{theorem}[Knill--Laflamme distance of the high-magic-rate twisted code]
\label{thm:GTZ_KL}
For Construction~\ref{con:high_yield} with the protected sectors
\eqref{eq:protected_sectors_GTZ}:
\begin{enumerate}[label=(\roman*)]
\item If $\sigma$ is not admissible then $\tilde{\C}_\sigma=0$. If $\sigma$ is
admissible with $\vc w=\vc w(\sigma)$, then $\tilde{\C}_\sigma$ has an orthonormal
basis $\{\ket{\psi^{\sigma}_{\vc a,\vc b}}:(\vc a,\vc b)\in A_{\vc w}\}$, labelled by
the selected \rd\ and \bl\ logical labels, on which the selected electric
operators act as $\widetilde Z^{\rd}_{\bs^{*2}_{e}\otimes\hs_f}(\beta)=\omega^{\mathrm{Tr}(\beta a_{(e,f)})}$
and $\widetilde Z^{\bl}_{\hs_e\otimes\bs^{*2}_{f}}(\beta)=\omega^{\mathrm{Tr}(\beta b_{(e,f)})}$.
Hence
\be\label{eq:sector_dimension}
\dim\tilde{\C}_\sigma=|A_{\vc w}|=\prod_{j\in J}
\begin{cases}2q-1,& w_j=0,\\ q-1,& w_j\neq0,\end{cases}
\ee
and in particular $\log_q\dim\tilde{\C}_0=s\log_q(2q-1)\ge s$.
\item Let
\begin{align}\label{eq:KL_threshold_GTZ}
c\,N^{1-\epsilon}&:=\min\Big\{d^0(\Gamma,\Fc),\ \tfrac{1}{3\Delta}\,d_1(\Gamma,\Fc^{\diamond})\Big\} \cr
&\ \ge\ \tfrac{1}{3\Delta}\,N^{1-\epsilon} .
\end{align}
For every admissible $\sigma$ and every operator $E$ supported on fewer than
$c\,N^{1-\epsilon}$ qudits there is a scalar $m_{E,\sigma}$ with
\be\label{eq:KL_GTZ}
P_{\tilde{\C}_\sigma}\,E\,P_{\tilde{\C}_\sigma}=m_{E,\sigma}\,P_{\tilde{\C}_\sigma} .
\ee
Consequently every $\tilde{\C}_\sigma$ detects all errors on fewer than
$cN^{1-\epsilon}$ qudits and corrects all errors on fewer than $cN^{1-\epsilon}/2$
qudits.
\item In particular no operator on fewer than $cN^{1-\epsilon}$ qudits acts
non-trivially on the selected logical qudits of the gauge-fixed twisted code
$\tilde{\C}_0$: the subsystem-code distance of Definition~\ref{def:distance_clifford}
with respect to the protected sectors \eqref{eq:protected_sectors_GTZ}, in the
gauge-fixed convention, is
\be\label{eq:twisted_distance_GTZ}
\tilde d\ \ge\ c\,N^{1-\epsilon}=\Omega\big(n^{(1-\epsilon)/2}\big),
\ee
and the same holds in every outcome sector $\tilde{\C}_\sigma$.
\end{enumerate}
\end{theorem}

\begin{proof}[Proof sketch]
The complete proof is given in Appendix~\ref{app:GTZ_KL}; it follows the pattern
of Appendix~\ref{app:twisted_KL} with three changes forced by the present
construction. (1) \emph{Orbit basis.} The dressed generators act on flat
configurations by $\tilde A(u)\ket x=\omega^{\mathrm{Tr}\,\theta_x(u)}\ket{x+d^0u}$
and form a representation of the gauge group $C^0$ on each flat orbit
(Lemma~\ref{lemma:commutation_Fq}). On an orbit with selected labels
$(\vc a,\vc b)$ and trivial \gr\ class, the zero modes act by the character
$\varsigma\mapsto\omega^{\mathrm{Tr}\int_{\eta_2}a_{\vc a}\cup b_{\vc b}\cup\varsigma_{\gr}}
=\omega^{\mathrm{Tr}\sum_j a_jb_j\ell_j(\varsigma_{\gr})}$, by the cochain-level
matching \eqref{eq:cochain_matching}. The eigenvalue pattern (e) prescribes the
character $\omega^{\mathrm{Tr}\langle\varsigma_{\gr},\sigma\rangle}$, and a joint
eigenvector exists on the orbit iff the two characters agree on $Z^0$, i.e.\ iff
$\sigma$ is admissible and $a_jb_j=w_j$; it is then unique and of uniform
magnitude. This gives (i), with the \gr\ record $\sigma$ built into the orbit
state as a linear phase $\omega^{-\mathrm{Tr}\langle u_{\gr},\sigma\rangle}$.
(2) \emph{No low-weight error changes a label.} A change of a selected label
requires an $X$-type cocycle with a non-zero class in $\mathsf L^{\rd}_X$ or
$\mathsf L^{\bl}_X$, of weight $\ge d^0(\Gamma,\Fc)$ by the contraction argument
above; a change of a gauge or \gr\ label leaves $\tilde{\C}_\sigma$. So for
$|E|<d^0(\Gamma,\Fc)$ the matrix of $P E P$ in the orbit basis is diagonal, and
its entries are explicit averages $\omega^{\mathrm{Tr}\langle\tau_{\gr},\sigma\rangle}
|C^0|^{-1}\sum_u\omega^{\mathrm{Tr}\,\Theta_x(u)}$ over the gauge group, with
$\Theta_x(u)$ a cubic polynomial whose only non-local label-dependent term is
$\int_{\eta_2}a_{\vc a}\cup b_{\vc b}\cup\tau_{\gr}$, $\tau_{\gr}$ a primitive of the
\gr\ $X$-part of $E$. (3) \emph{Product-form cleaning.} Because the degree-one
factors $\bar e$, $\bar f$ of the selected representatives live on the two graph
factors separately, they can be cleaned off the $O(\Delta|E|)$ edges of $\Gamma$
shadowed by $E$ by coboundaries of $\Fc^{\diamond}$ on $\Gamma$ (Lemma~\ref{lemma:clean_app}
with $d_1(\Gamma,\Fc^{\diamond})$), keeping the representatives in product form.
All local label-dependent terms then vanish identically, and the non-local term
becomes $\sum_{j,k}a_jb_k\int_{\eta_2}\red{\alpha'^1_j}\cup\blue{\beta'^1_k}\cup\tau_{\gr}$
in which, by the matching \eqref{eq:cochain_matching} and the Leibniz rule, every
cross term $j\neq k$ is the integral of an exact cochain whose primitive is
supported away from $E$, hence vanishes, while the diagonal terms contribute
$\sum_jw_j(\cdots)$, a function of the record alone. The diagonal entries are
therefore independent of $(\vc a,\vc b)\in A_{\vc w}$, which is \eqref{eq:KL_GTZ}
for Paulis, and linearity extends it to arbitrary operators. Part (iii) is a
restatement of (ii).
\end{proof}

\begin{remark}[Scope of Theorem~\ref{thm:GTZ_KL}]\label{rem:GTZ_KL_scope}
Theorem~\ref{thm:GTZ_KL} uses only Theorem~\ref{thm:GTZ_classical}(ii)--(iii),
the dictionary of Lemma~\ref{lemma:dictionary}, the locality and Leibniz
properties of the cup product on $\eta_2=[\L]$ and Lemma~\ref{lemma:commutation_Fq};
Assumption~\ref{assump:TQFT} is not used, nor is any property of the GTZ graph
beyond its bounded degree. It is a statement about the gauge-fixed code
$\tilde{\C}_\sigma$. Unlike Theorem~\ref{thm:distance_twisted_KL}(ii), it does not
assert a factorization of the full twisted code into protected and gauge
registers with $P E P=\I\otimes M_E$: the charge constraints
\eqref{eq:constraint_Fq} couple the selected \rd\ classes to \emph{unselected}
\bl\ classes lying in the same K\"unneth sector
$H^0(\Gamma,\Fc)\otimes H^1(\Gamma,\Fc^{\diamond})$ (and symmetrically), so the
admissible labels do not form a product of a protected and a gauge factor. This
is harmless for the protocol, which fixes the gauge register by preparation and
never changes it (Lemma~\ref{lemma:protocol_app}), and it is the same convention
under which the subsystem codes of Ref.~\cite{GolowichTamoZhu2026} are used.
The complementary K\"unneth sectors contain logical operators of weight
$O_\epsilon(1)$ (both $X$-type, from $\T^{\dag}(\Gamma,\Fc^{\diamond})$, and $Z$-type,
from $\T^{\dag}(\Gamma,\Fc^{\perp})$; cf.\ Remark~\ref{rem:slot_choice}), exactly as
the parent codes of Ref.~\cite{GolowichTamoZhu2026} do; they never act on the
selected labels, and errors that flip the fixed gauge labels are detected in the
sense of Eq.~\eqref{eq:KL_GTZ} ($m_{E,\sigma}=0$).
\end{remark}

\begin{remark}[Comparison with Construction~\ref{con:2d-hgp}]
\label{rem:yield_comparison}
The two instantiations trade encoding rate against magic rate:
Construction~\ref{con:2d-hgp} gives $[[n,\Theta(n),\Omega(\sqrt{n})]]_2$ with
magic rate $\varrho=\Theta(n^{-1/2})$, whereas
Construction~\ref{con:high_yield} gives
$[[n, n^{1-\epsilon}, n^{(1-\epsilon)/2}]]_q$ with
$\varrho \ge n^{-\epsilon}$. The distance is the same up to the $n^{\epsilon}$
factor and the encoding rate drops only by $n^{-\epsilon}$, while the magic rate
improves quadratically, from $n^{-1/2}$ to almost constant. In the language of
Ref.~\cite{GolowichTamoZhu2026}, whose logical yield $s$ is our
$\varrho\,n$, Construction~\ref{con:2d-hgp} sits at
$d\cdot s = \Theta(n)$ --- the barrier that all previous constant-weight
constructions obeyed --- while Construction~\ref{con:high_yield} achieves
$d\cdot s = n^{3(1-\epsilon)/2} \gg n$.
\end{remark}

\subsection{Gauging measurement at almost-constant magic rate}
\label{sec:high_yield_fountain}

The protocol of Sec.~\ref{sec:protocol} applies verbatim, with the $\Fq$
modifications of Sec.~\ref{sec:subcomplex_Fq}: the measured checks are the \gr-type
twisted stabilizers $\tilde{A}^{\gr}(\lambda^0)$, which commute exactly among
themselves by Lemma~\ref{lemma:commutation_Fq}, and the outcome attached to the
label $\green{\gamma^0_{(e,f)}}$ is now an element $\green{w_{(e,f)}}\in\Fq$
rather than a bit.

\begin{theorem}[Magic state fountain at almost-constant rate]\label{thm:high_yield_fountain}
There is a gauging measurement protocol on the twisted sheaf code of
Construction~\ref{con:high_yield} that, in a single run, measures the $s=|M|^2$
logical CZ observables of the disjoint selected pairs and prepares the logical CZ
magic state on every pair whose outcome is $\green{w}=0$. The outcomes are
independent across pairs with $\Pr[\green{w_{(e,f)}}=0]=(2q-1)/q^{2}$, so with
probability $1-e^{-\Omega(n^{1-\epsilon})}$ a single run produces
$\Theta(s)=\Omega(n^{1-\epsilon})$ logical CZ magic states on disjoint pairs ---
an almost constant magic rate $\varrho\ge n^{-\epsilon}$ (after renaming
$\epsilon$) --- and every intermediate state lies in an outcome sector
$\tilde{\C}_\sigma$ protected by Knill--Laflamme distance
$\Omega(n^{(1-\epsilon)/2})$ (Theorem~\ref{thm:GTZ_KL}), without
Assumption~\ref{assump:TQFT}.
\end{theorem}

\begin{proof}
Initialize the selected \rd\ and \bl\ logical qudits of
Eq.~\eqref{eq:high_yield_basis} in $\lo{\ket{+}}$ and every other \rd\ and \bl\
logical qudit --- the gauge register of Eq.~\eqref{eq:protected_sectors_GTZ} --- in
$\lo{\ket{0}}$; initialize the \gr\ qudits in $\ket{0}$. Run steps 2--4 of
Sec.~\ref{sec:protocol}. Measuring the \gr\ twisted generators returns a record
$\sigma$, which by Lemma~\ref{lemma:protocol_app} is admissible with logical
outcome $\vc w=\vc w(\sigma)$, $w_{(e,f)}$ being the eigenvalue label of the
observable $\widetilde{\mathrm{CZ}}^{\rd,\bl}(\green{\gamma^0_{(e,f)}})$ of
Corollary~\ref{cor:disjoint_CZ}; the post-measurement state lies in
$\tilde{\C}_\sigma$ and, in the basis of Theorem~\ref{thm:GTZ_KL}(i), equals
\begin{align}\label{eq:high_yield_output}
\lo{\ket{\psi}}_f = \bigotimes_{(e,f)\in M\times M}
&\hat{P}_{\green{w_{(e,f)}}}
\Big[\lo{\text{CZ}}^{\rd,\bl}_{\red{(e,f)},\blue{(e,f)}}\Big] \cr
&\times\ \lo{\ket{++}}_{\red{(e,f)},\blue{(e,f)}} ,
\end{align}
up to normalization, where $\hat{P}_{w}=\frac1q\sum_{\lambda\in\Fq}
\omega^{-\mathrm{Tr}(\lambda w)}\lo{\mathrm{CZ}}(\lambda)$ projects onto the
$\omega^{\mathrm{Tr}(\lambda w)}$ eigenspace, i.e.\ onto
$\mathrm{span}\{\lo{\ket{a,b}}:ab=w\}$ (Lemma~\ref{lemma:protocol_app}). Each
projector acts on a distinct pair, so no post-selection couples different pairs,
the outcomes are independent with $\Pr[w=0]=(2q-1)/q^{2}$ --- the number of
solutions of $ab=0$ in $\Fq^2$ divided by $q^2$ --- and the outcome $w=0$ yields
the logical CZ magic state $\lo{\ket{\psi_0}}\propto\sum_{ab=0}\lo{\ket{a,b}}$ on
that pair. A Chernoff bound gives $\Theta(s)$ such pairs with probability
$1-e^{-\Omega(s)}$. The distance statement is Theorem~\ref{thm:GTZ_KL}(ii)
applied to the sector $\tilde{\C}_\sigma$, and the $O(d)$ rounds of error
correction of step~2 protect the outcomes.
\end{proof}

\begin{remark}[Relation to transversal $C^{r-1}Z$]\label{rem:relation_CrZ}
Ref.~\cite{GolowichTamoZhu2026} uses the multiplication property to build
$r$-fold products of Tanner complexes supporting transversal $C^{r-1}Z$ gates
of large \emph{logical yield} directly, i.e.\ non-Clifford gates on an
$r$-dimensional complex. Here the same
classical input is used differently: the complex stays $\mathsf{d}=2$
dimensional, the transversal gate is the \emph{Clifford} gate CZ, and the
non-Clifford logical operation is obtained not from the gate itself but from
\emph{gauging} it, Sec.~\ref{sec:gauging_untwisted}. The $r=3$ multiplication
property is nevertheless exactly what is needed, because the twist
$\int_{\eta_2}\red{a^1}\cup\blue{b^1}\cup\green{\gamma^0}$ is a triple cup product
with one slot in degree $0$ --- two dynamical gauge fields and one symmetry
label --- rather than the three degree-$1$ slots of a $C^2Z$ gate. The advantage of
staying in two dimensions is that the distance scales as $n^{1/2}$ rather than
$n^{1/3}$, which is an improvement.
\end{remark}

\begin{remark}[Realization on qubits]\label{rem:qubit_realization}
Construction~\ref{con:high_yield} is stated over $\Fq$ with $q=p^{u}$ and
$u=O_{\epsilon}(1)$, because the classical input of
Theorem~\ref{thm:GTZ_classical} is algebraic
(Remark~\ref{rem:why_large_q}). For $p=2$ this costs nothing in practice: each
$\Fq$-qudit is a block of $u$ qubits, the twisted stabilizers
\eqref{eq:twisted_generators_Fq} become qubit Clifford operators of weight
$O(u\Delta)=O(1)$, and by
Proposition~\ref{prop:cz_descent_app} of Appendix~\ref{app:alphabet} each
$\Fq$-CZ appearing in the dressing \eqref{eq:dressing_Fq} and in the transversal
operator \eqref{eq:CZ_Fq} is a depth-one circuit of $u$ ordinary qubit CZ gates.
The code parameters transform as
$[[n,k,d]]_{q}\rightsquigarrow[[un,\,uk,\,\ge d]]_{2}$, so the magic rate
$\varrho$ changes only by the constant factor $1/u$ and remains almost constant.
Crucially, no concatenation with a multiplication-friendly code is required:
that machinery is needed only to descend a genuinely \emph{trilinear} transversal
gate, whereas the operators here are bilinear in the dynamical fields, the third
cup-product slot being a classical label. See Appendix~\ref{app:alphabet} for the
statement and for the general alphabet-reduction results of
Refs.~\cite{GolowichLin2025,GolowichGuruswami2025,Nguyen2024,Wills2024}.
\end{remark}

%======================================================================
\section{Discussion and Outlook}
\label{sec:discussion}
%======================================================================

We have shown that the combinatorial machinery normally reserved for
triangulated manifolds --- cup products, cap products, a chain-cochain pairing,
and with them the path integral of a twisted Dijkgraaf-Witten gauge theory --- is
already present inside a simplicial (cell) complex with sheaf structure. The
consequences are a family of non-Abelian (Clifford stabilizer) qLDPC codes built
directly on arbitrary sheaf complexes including the $2$D hypergraph-product complexes, a $0$-form subcomplex symmetry that
supplies addressable transversal logical CZ gates without any transversal CCZ, and
a gauging measurement protocol that turns those gates into a magic state fountain
at an almost constant magic rate. 

The main difference of the TQFT-type gauging in Ref.~\cite{ZhuKobayashiHsin2026} and present work and the surgery-type gauging in Ref.~\cite{christos2026non}  is that the later  needs to introduce an ancillary system which will modify the underlying complex from $\L$ to a merged complex $\L'$ and typically may not preserve the distance of the code. Currently it is only able to measure a global transversal gate instead of addressable logical gates required for parallel magic state production, and potential parallel measurement in surgery typically requires the space overhead of the ancillary system growing proportionally to the support of the logical operators \cite{cohen22, cross2024improved, williamson2024low,
swaroop2024universal}. On the other hand, the twisted (gauged) code in \cite{ZhuKobayashiHsin2026} and present paper  is supported on the same complex $\L$ as the untwisted code which is more in the spirit of standard gauging in TQFT and hence preserves the distance and permits parallel gauging measurement without the need for ancillary system and the associated spatial overhead.

We close with what we see as the main open
directions.

\subsection*{Code parameters}

The distance in both instantiations scales as $n^{1/2}$ up to the
$n^{-\epsilon}$ loss in Sec.~\ref{sec:high_yield}, which is the generic ceiling
for a $2$D product of two graph codes; the trade is that the construction stays
$2$D, whereas the transversal-gate route to non-Clifford operations requires a
$3$D product and the more demanding connectivity that comes with it. Whether the
twisted sheaf codes can be pushed to linear distance, or made asymptotically good,
is the obvious question. Two routes seem worth trying. The first is to replace
the hypergraph product by a balanced product or a lifted product: in
App.~\ref{sec:balanced} we have shown that the twist, the twisted Clifford
stabilizer code and the gauging measurement all descend to such quotients, so
what remains is to combine the expansion-based distance arguments for
balanced-product codes with the twisted structure (Remark~\ref{rem:balanced_distance}),
and to find symmetric inputs for which all $H$-sectors contribute to the fountain
(Remark~\ref{rem:balanced_rate}). The second is to use the higher-dimensional
expansion of the sheaf directly rather than the K\"unneth factorization of
Sec.~\ref{sec:kunneth_factorization}. 

A second parameter question is the magic rate itself.
Theorem~\ref{thm:high_yield_parameters} gives $\varrho\ge n^{-\epsilon}$ for every
$\epsilon>0$, with an alphabet size $q=p^{O_\epsilon(1)}$ that blows up as
$\epsilon\to0$; the obstruction is the Singleton-type bound on codes with a
multiplication property~\cite{randriambololona2013upper, mirandola2015critical},
which forces a rate loss for any fixed alphabet. A genuinely \emph{constant}
magic rate would require either a different mechanism for satisfying the local
cycle condition of Lemma~\ref{lemma:local_cycle}.

\subsection*{Decoding}

Decoding is largely untouched here. The twisted code has the same $Z$-checks as
the untwisted one and dressed $X$-checks with the same supports, so the syndrome
graph is unchanged and any qLDPC decoder applies at the level of the Pauli part;
what is new is that the stabilizer group is non-Pauli, so error propagation
through the CZ dressing has to be tracked. Establishing a phenomenological noise
threshold taking into account measurement errors for the $O(d)$ rounds of dressed $X$-stabilizer measurement in
Sec.~\ref{sec:protocol} are the natural next steps. An important step would be to generalize the just-in-time decoder 
in the topological-code case~\cite{Bauer:2023awl, Bauer:2024qpc, Bauer:2024alh,
Davydova:2025ylx} to the sheaf code setting. 

\subsection*{Beyond CZ, and beyond $\ZZ_2$}

The twist used throughout is the type-III Dijkgraaf-Witten cocycle, which produces
logical CZ magic states and hence, by gate teleportation, logical CZ gates. Other
group cohomology classes, and other finite gauge groups, give different twisted
sheaf gauge theories; type-I and type-II twists, non-Abelian gauge groups, and
fermionic (spin) versions are all available in the manifold setting and should
have sheaf analogues.  More broadly, the subdivision of
Sec.~\ref{sec:subdivision} makes \emph{any} classical LDPC code into a simplicial
object, so any state-sum invariant that a simplicial complex supports --- Turaev-Viro
among them~\cite{Turaev:1992hq, walker2021universalstatesum} --- can in principle be
written on a qLDPC code. Whether those invariants give useful codes, and whether
the resulting anyon theories admit universal gate sets, is open.

On the algebraic side, the $\Fq$ theory of Sec.~\ref{sec:Fq_generalization} is a
$(\ZZ_p)^u$ gauge theory whose twist is selected by the Frobenius tensor
$T_{ijk}=\mathrm{Tr}(\theta_i\theta_j\theta_k)$. It would be interesting to ask
which twisted $(\ZZ_p)^u$ theories arise this way as $\Fq$ varies --- i.e.\ to
characterize the image of the map from field structures to group cohomology
classes --- and whether the ones that do are distinguished physically. The
alphabet reduction of Appendix~\ref{app:alphabet} is easy for us because our
operators are bilinear in the dynamical fields; constructing binary classical codes
with a genuine multiplication property, which would remove the detour through
$\Fq$ altogether, remains the outstanding problem inherited from
Refs.~\cite{GolowichTamoZhu2026, GolowichLin2025, GolowichGuruswami2025,
Nguyen2024, Wills2024}.

\subsection*{Physics of non-Abelian qLDPC codes}

The twisted codes constructed here are non-Abelian topological orders that live on
complexes with no manifold structure, and the defect-network reading of
Sec.~\ref{sec:untwisted_code} (Figs.~\ref{fig:gapped_boundary_picture}
and~\ref{fig:gapped_boundary_dual}) suggests that the right physical picture is a
network of TQFT patches glued along gapped interfaces, in the spirit of
Ref.~\cite{Aasen2020}. This raises concrete questions. What is the anyon content
of a sheaf gauge theory --- is there a sensible fusion category attached to a
sheaf complex, as opposed to a manifold? The non-Abelian fusion rules and the
Borromean-ring braiding statistics derived in the CW-complex
setting~\cite{ZhuKobayashiHsin2026} should have intrinsic sheaf derivations, with
the triple intersection of Fig.~\ref{fig:triple_intersection} again providing the
mechanism. More speculatively, the abundance of top-dimensional cycles that makes
subcomplex symmetries possible is also what distinguishes these complexes from
manifolds; a classification of the generalized symmetries of a sheaf gauge theory,
in the spirit of the topological-operator perspective of
Refs.~\cite{Roumpedakis2022, Kaidi2022, Apruzzi2021}, would put the $0$-form
subcomplex symmetry of Sec.~\ref{sec:0form_subcomplex} into a systematic
framework.

\section{Acknowledgments}

\textit{Note added.---} Around the time of submitting this paper, we have noticed another related work  \cite{li2026nonabelian}. The difference is that the present paper introduces not only the non-Abelian sheaf qLDPC codes but also the associated twisted gauge theory, and focuses on achieving almost-constant magic rate.  

G.Z. is supported by the U.S. Department of Energy, Office of Science, National Quantum Information Science Research Centers, Co-design Center for Quantum Advantage (C2QA) under contract number DE-SC0012704. G.Z. thanks Louis Golowich
and Itzhak Tamo for the collaboration on a related project. S.J.S.T. acknowledges funding from the NUS Development Grant and is grateful for the support from IBM during his internship including the companionship from his fellow interns Keller Blackwell, Anqi Gong, Shubham Jain, and Adam Wills.
P.-S. H. is supported by Department of Mathematics at King's College London. G.Z. and P.-S. H. were supported in part by grant no. NSF PHY-2309135 to the Kavli Institute for Theoretical Physics (KITP).

\textbf{AI disclosure:} All ideas and proofs are the work of the human authors. Claude and ChatGPT were used to polish the writing and improve the structure of the paper, and to find references.

\begin{appendix}

\section{Local acyclicity of the subdivided sheaf}\label{app:acyclic}

\begin{definition}[Locally acyclic sheaf]\label{def:locally-acyclic-sheaf}
    Suppose $X$ is a $t$-dimensional cell complex with a sheaf $\Fc$, we say $\Fc$ is \emph{locally acyclic} if for each $i$-cell $\sigma$, the cohomology $H^j(X_{\geq \sigma}, \Fc) = 0$ for $i \leq j \leq t-1$, i.e., the following sequence is exact
\begin{widetext}   
    \begin{equation}
        0 \xrightarrow{} \Fc(\sigma) \xrightarrow{} \prod_{\sigma_{i+1} \in X_{\geq \sigma}(i+1)} \Fc(\sigma_{i+1}) \xrightarrow{} \prod_{\sigma_{i+2} \in X_{\geq \sigma}(i+2)} \Fc(\sigma_{i+2}) \xrightarrow{} \ldots \xrightarrow{} \prod_{\sigma_y \in X_{\geq \sigma}(t)} \Fc(\sigma_t)
    \end{equation}
    \end{widetext} 
\end{definition}

We now proceed to state and prove a proposition regarding local acyclicity on a graph.

\begin{proposition}[Criterion for local acyclicity on a graph]
\label{prop:graph-local-acyclicity}
    Let $G$ be a graph and let $\Fc$ be the sheaf defined from local parity-check matrices $\{\mathsf{H}_v\}_{v \in V}$ as in   Eqs.~\eqref{eq:local_system} and~\eqref{eq:co-restriction_map}. If each
    $\mathsf{H}_v$ has full row rank, then $\Fc$ is locally acyclic.
\end{proposition}

\begin{proof}
    Since $G$ is 1-dimensional, the only case we have to show is for $i =0$ when $\sigma$ is a vertex. Fix a vertex $v \in V$. Then, $G_{\geq u}$ consists of the vertex $v$ along with the incident edges $E(v)$. Therefore, its local cochain complex is simply
    \begin{equation}
        0 \longrightarrow \Fc(v) \xrightarrow{\delta^0_u} \prod_{e \in E(v)} \Fc(e)  \longrightarrow 0,
    \end{equation}
    where $\delta^0_u$ is the map whose components are the co-restriction maps $\Fc_{v \shortrightarrow e}$. By Eq.~\eqref{eq:co-restriction_map}, this is
    exactly the transpose of the local parity-check matrix:
    \begin{equation}
    \delta^0_u = \mathsf{H}_v^\top : \FF_2^{m(v)} \longrightarrow
    \FF_2^{E(v)}.
    \end{equation}
    Hence local acyclicity at $u$ is equivalent to injectivity of
    $\mathsf{H}_v^\top$. Since $\mathsf{H}_v$ has full row rank,
    \[
    \rk(\mathsf{H}_v^\top)=\rk(\mathsf{H}_v)=m(v)=\dim \Fc(v),
    \]
    so $\mathsf{H}_v^\top$ is injective. Thus
    $H^0(G_{\ge u},\Fc)=\ker \delta^0_u=0$ for every vertex $v$, which is exactly the
    required local acyclicity condition.
\end{proof}

Now, we state a corollary regarding the local acyclicity of the subdivided sheaf.

\begin{corollary}[Local acyclicity of the subdivided sheaf]
\label{cor:subdivided-locally-acyclic-sheaf}
The subdivided sheaf $\Fc$ of Section~\ref{sec:subdivision} is locally acyclic. The same statement also holds for the dual sheaf $\Fc^\perp$.
\end{corollary}

\begin{proof}
    For an original check vertex $v_h \in V_h$ of degree $w_r(v_h)$, the local code is $\cC_{v_h} = \mathrm{SPC}(d)$. We can choose the $1 \times d$ parity-check matrix
    \[\mathsf{H}_{v_h} = \begin{bmatrix} 1 & 1 & \cdots & 1 \end{bmatrix},\]
    which has full row rank.
    For a new vertex introduced $v' \in V'$ of degree $w_c(v')$, the local code is a length-$w_c(v')$ repetition code. We can also choose the $\left(w_c(v') - 1\right) \times w_c(v')$ parity-check matrix so that it has full row rank. 
    Thus, every local parity-check matrix appearing in our subdivided sheaf has full row rank and Proposition~\ref{prop:graph-local-acyclicity} tells us that $\Fc$ is locally acyclic.
    Because the dual sheaf $\Fc^\perp$ simply swaps the local codes, the same argument holds and $\Fc^\perp$ is also locally acyclic.
\end{proof}

%======================================================================
\section{Alphabet reduction from $\Fq$ to $\FF_2$}
\label{app:alphabet}
%======================================================================

The high-magic-rate construction of Sec.~\ref{sec:high_yield} is stated over an
alphabet $\Fq$ with $q=p^{u}$ and $u=O_{\epsilon}(1)$, because the classical
codes with a non-trivial multiplication property supplied by
Theorem~\ref{thm:GTZ_classical} are algebraic and do not exist over $\FF_2$
(Remark~\ref{rem:why_large_q}). Hardware, and the standard magic states, are
binary. This appendix records what is needed to pass from $\Fq$ to $\FF_2$, and
in particular proves the statement invoked in
Remark~\ref{rem:qubit_realization}. Throughout $\omega=e^{2\pi i/p}$ and
$\mathrm{Tr}=\mathrm{Tr}_{\Fq/\FF_p}$ as in Eq.~\eqref{eq:omega_trace}; for
concreteness the reader may take $p=2$, so that $\omega=-1$ and an $\Fq$-qudit is
a block of $u$ qubits.

\subsection{Three layers of the problem}
\label{app:three_layers}

It is useful to separate three questions that are often conflated.

\medskip
\noindent\textit{The Hilbert space.} An $\Fq$-qudit \emph{is} $u$ qudits of
dimension $p$. Fixing an $\FF_p$-basis $\{\theta_1,\ldots,\theta_u\}$ of $\Fq$
and writing $x=\sum_i x_i\theta_i$ identifies
$\ket{x} \leftrightarrow \ket{x_1\cdots x_u}$. Nothing to prove.

\medskip
\noindent\textit{The code.} An $\Fq$-linear code $C\subseteq\Fq^{n}$ is in
particular an $\FF_p$-linear code of length $un$ and $\FF_p$-dimension
$u\dim_{\Fq}C$, and its weights obey
\be\label{eq:app_weights}
|c|_{\Fq} \;\le\; |c|_{\FF_p} \;\le\; u\,|c|_{\Fq} ,
\ee
since a non-zero $\Fq$-symbol contributes between $1$ and $u$ non-zero
$\FF_p$-symbols. Hence $[n,k,d]_{q}\rightsquigarrow[un,uk,\ge d]_{p}$: the rate
is preserved exactly, the absolute distance is preserved, and only the
\emph{relative} distance falls, by the factor $u$. Moreover the generalized
Paulis of Eq.~\eqref{eq:qudit_pauli} are generated by single-$\FF_p$-qudit
Paulis: with the trace-dual basis $\{\theta^{*}_j\}$ determined by
$\mathrm{Tr}(\theta_i\theta^{*}_j)=\delta_{ij}$,
\be\label{eq:app_pauli_descent}
X(\theta_i) = X_i , \qquad Z(\theta^{*}_j) = Z_j ,
\ee
so an $\Fq$ stabilizer code is already an $\FF_p$ stabilizer code, with check
weights multiplied by at most $u$. For constant $u$ this layer is free.

\medskip
\noindent\textit{The gate.} This is where the content lies. Expanding the
diagonal gates of Eqs.~\eqref{eq:qudit_CZ} in the basis $\{\theta_i\}$,
\begin{align}
\mathrm{Tr}(\lambda x y) &= \sum_{i,j} B^{\lambda}_{ij}\, x_i y_j ,
&B^{\lambda}_{ij} &= \mathrm{Tr}\big(\lambda\,\theta_i\theta_j\big) ,
\label{eq:app_bilinear}\\
\mathrm{Tr}(\lambda x y z) &= \sum_{i,j,k} T^{\lambda}_{ijk}\, x_i y_j z_k ,
&T^{\lambda}_{ijk} &= \mathrm{Tr}\big(\lambda\,\theta_i\theta_j\theta_k\big) ,
\label{eq:app_trilinear}
\end{align}
so an $\Fq$-CZ becomes a pattern of $\FF_p$-CZ gates governed by the bilinear
form $B^{\lambda}$, and an $\Fq$-CCZ a pattern of $\FF_p$-CCZ gates governed by
the trilinear form $T^{\lambda}$. The bilinear case is elementary; the trilinear
case is the whole difficulty.

\subsection{The bilinear case: a change of basis}
\label{app:bilinear}

\begin{proposition}[$\Fq$-CZ is a depth-one $\FF_p$-CZ circuit]
\label{prop:cz_descent_app}
Let $\lambda\in\Fq^{\times}$. There exist $\FF_p$-bases $\{\theta_i\}$ and
$\{\phi_j\}$ of $\Fq$, used respectively on the first and second qudit, such that
\be\label{eq:app_cz_diag}
\text{CZ}(\lambda) \;=\; \prod_{i=1}^{u} \text{CZ}_{\,i,i} ,
\ee
where $\text{CZ}_{i,i}$ is the ordinary $\FF_p$-qudit CZ between the $i$-th
component of the first block and the $i$-th component of the second. Thus
$\text{CZ}(\lambda)$ is a depth-one circuit of $u$ elementary CZ gates. If the
same basis is used in both blocks, $\text{CZ}(\lambda)$ is a circuit of at most
$u(u+1)/2$ elementary CZ gates in depth $O(u)$.
\end{proposition}

\begin{proof}
The pairing $(x,y)\mapsto\mathrm{Tr}(\lambda x y)$ is $\FF_p$-bilinear and
non-degenerate: the trace form is non-degenerate on $\Fq$ by
Eq.~\eqref{eq:trace_delta}, and multiplication by $\lambda\neq0$ is invertible.
Choose any basis $\{\theta_i\}$ for the first block and let $\{\phi_j\}$ be its
dual with respect to this pairing, $\mathrm{Tr}(\lambda\theta_i\phi_j) =
\delta_{ij}$, which exists by non-degeneracy. Then $B^{\lambda}_{ij}=\delta_{ij}$
in Eq.~\eqref{eq:app_bilinear}, so $\mathrm{Tr}(\lambda xy)=\sum_i x_i y_i$ and
the gate factorizes as Eq.~\eqref{eq:app_cz_diag}. For the same-basis statement,
$B^{\lambda}$ is a symmetric $u\times u$ matrix, so the associated quadratic form
has at most $u(u+1)/2$ non-zero coefficients, and the corresponding CZ circuit
can be scheduled in $O(u)$ layers by edge-colouring the graph of $B^{\lambda}$.
\end{proof}

\begin{example}[$\Fq=\FF_4$]\label{ex:F4_cz_app}
Let $\FF_4=\FF_2[\alpha]/(\alpha^2+\alpha+1)$, so that $\mathrm{Tr}(x)=x+x^2$ and
$\mathrm{Tr}(0)=\mathrm{Tr}(1)=0$,
$\mathrm{Tr}(\alpha)=\mathrm{Tr}(\alpha^2)=1$. The basis
$\{\theta_1,\theta_2\}=\{\alpha,\alpha^2\}$ is \emph{self-dual}:
\begin{align}\label{eq:app_F4_selfdual}
\mathrm{Tr}(\alpha\cdot\alpha) &= \mathrm{Tr}(\alpha^2) = 1 , \cr
\mathrm{Tr}(\alpha\cdot\alpha^2) &= \mathrm{Tr}(1) = 0 , \cr
\mathrm{Tr}(\alpha^2\cdot\alpha^2) &= \mathrm{Tr}(\alpha) = 1 ,
\end{align}
so $B^{1}=\I$ and $\text{CZ}(1)=\text{CZ}_{1,1}\,\text{CZ}_{2,2}$: the $\FF_4$ CZ
gate is literally two qubit CZ gates applied in parallel.
\end{example}

Proposition~\ref{prop:cz_descent_app} is the reason that
Remark~\ref{rem:qubit_realization} holds. Both operators of the construction that
carry the twist --- the transversal CZ operator $\widetilde{\text{CZ}}^{\rd,\bl}
(\green{\gamma^0})$ of Eq.~\eqref{eq:CZ_Fq} and the stabilizer dressing
$U^{\rd}(\lambda^0)$ of Eq.~\eqref{eq:dressing_Fq} --- are triple cup products in
which \emph{one slot is a fixed classical cochain}, the symmetry label $\green{\gamma^0}$
respectively the stabilizer label $\lambda^0$. They are therefore bilinear in the
dynamical gauge fields, with an $\Fq$-multiplier read off from the classical
label, and Proposition~\ref{prop:cz_descent_app} applies to each two-qudit factor
separately. Blocking $u$ qudits together does not merge distinct blocks, so the
disjointness of the logical CZ pairs established in
Corollary~\ref{cor:disjoint_CZ} survives, and within a matched pair of blocks the
induced gate is itself a perfect matching of the $u$ components by
Eq.~\eqref{eq:app_cz_diag}.

\subsection{The trilinear case}
\label{app:trilinear}

For a genuinely trilinear transversal gate the previous argument fails: symmetric
trilinear forms cannot in general be diagonalized, so there is no change of basis
making $\mathrm{Tr}(\lambda xyz)=\sum_i x_iy_iz_i$.

\begin{example}[$\FF_4$ again]\label{ex:F4_ccz_app}
With $\theta_1=\alpha$, $\theta_2=\alpha^2$ and $\lambda=1$ one finds
$T_{111}=\mathrm{Tr}(\alpha^3)=0$, $T_{112}=\mathrm{Tr}(\alpha^4)=1$,
$T_{122}=\mathrm{Tr}(\alpha^5)=1$, $T_{222}=\mathrm{Tr}(\alpha^6)=0$, so that,
using the full symmetry of $T$,
\begin{align}\label{eq:app_F4_ccz}
\mathrm{Tr}(xyz) =\ & x_1y_1z_2+x_1y_2z_1+x_2y_1z_1 \cr
& + x_1y_2z_2+x_2y_1z_2+x_2y_2z_1 ,
\end{align}
i.e.\ six qubit CCZ gates, with the diagonal terms $x_1y_1z_1$ and $x_2y_2z_2$
\emph{absent}. No relabelling turns Eq.~\eqref{eq:app_F4_ccz} into
$x_1y_1z_1+x_2y_2z_2$: the two tensors have different ranks.
\end{example}

Two things go wrong at once: the gate count grows as $O(u^{3})$ rather than
$O(u)$, and the induced logical action is no longer a product of CCZ gates on
disjoint triples, so a transversal $\Fq$-CCZ of logical yield $s$ does not
automatically descend to a binary transversal CCZ of yield $us$. The classical
tool that repairs this is a \emph{multiplication-friendly embedding}: a pair of
$\FF_p$-linear maps $\pi:\Fq\to\FF_p^{m}$, $\sigma:\FF_p^{m}\to\Fq$ with
$xy=\sigma(\pi(x)*\pi(y))$, where $*$ is the componentwise product. The least
such $m$ is the bilinear complexity of $\Fq/\FF_p$, which is $O(u)$ by the
algebraic-geometry construction of Chudnovsky and Chudnovsky; in the cryptography
literature the same object is an \emph{arithmetic secret sharing} scheme
\cite{cramer2000general}. For $\FF_4/\FF_2$ the Karatsuba identity gives $m=3$
explicitly, with $\pi(x)=(x_0,x_1,x_0+x_1)$ and
$\sigma(u_1,u_2,u_3)=(u_1+u_2)+(u_1+u_3)\alpha$.

Four recent works use this machinery to reduce the alphabet of quantum codes with
transversal non-Clifford gates, in four different places:
Ref.~\cite{Wills2024} fixes the alphabet at $q=2^{10}$ and converts the resulting
$10$-qubit magic states to conventional CCZ and $T$ magic states with constant
overhead; Ref.~\cite{Nguyen2024} applies arithmetic secret sharing to the
\emph{classical} codes before the quantum construction, obtaining asymptotically
good binary codes with transversal CCZ; Ref.~\cite{GolowichGuruswami2025}
introduces a quantum analogue of multiplication-friendly codes and applies it as
a concatenation \emph{after} the quantum construction, which makes the reduction
iterable; and Ref.~\cite{GolowichLin2025} packages the result as two lemmas of
general applicability, reducing the alphabet from $q=p^{u}$ to any power of $p$
with a multiplicative $u^{O(1)}$ loss in the other parameters
\cite[Lemma 3.44]{GolowichLin2025}, and reducing the transversal-gate sparsity
from $w$ to $1$ with a multiplicative $w^{O(1)}$ loss
\cite[Lemma 3.45]{GolowichLin2025}. It is these two lemmas that
Ref.~\cite{GolowichTamoZhu2026} invokes to pass from its
$\Fq$-alphabet statement to a statement valid for every prime power including
$q=2$.

We emphasize that Construction~\ref{con:high_yield} does \emph{not} need any of
this. The non-Clifford content of the protocol does not reside in a transversal
gate at all: it resides in the gauging measurement, i.e.\ in measuring the
\gr-type twisted stabilizers and recording the outcome
(Sec.~\ref{sec:high_yield_fountain}). The trilinear form
$\int_{\eta_2}\red{a^1}\cup\blue{b^1}\cup\green{c^1}$ is still present --- it is
the Dijkgraaf--Witten twist --- but it is never executed as a physical
three-qudit gate; it is executed as a CZ-dressed stabilizer measurement whose
third slot is classical. Converting a non-Clifford \emph{gate} problem into a
Clifford \emph{measurement} problem is precisely the trade the gauging
construction makes, and alphabet reduction is easy on the Clifford side. What
does still require a large alphabet is the classical input itself, and nothing
above changes that.

%======================================================================
\section{Subsystem-code distance: proofs of Lemmas~\ref{lemma:subsystem}
and~\ref{lemma:distance_untwisted}}
\label{app:subsystem_distance}
%======================================================================

Throughout this appendix all vector spaces are over $\FF_2$. For a (co)chain
$\chi$ on the square complex $\L=G_x\times G_y$ we write $|\chi|$ for the number
of $1$-cells on which $\chi$ is non-zero. Since every stalk has dimension at most
$\Delta=O(1)$, the number of qubits in the support of the corresponding Pauli
operator lies between $|\chi|$ and $\Delta|\chi|$; in particular a lower bound on
$|\chi|$ is a lower bound on the qubit weight. For a one-dimensional sheaf
complex $Y=C^\bullet(G,\Fc)$ we use $d_1(Y)=\min\{|h|:0\neq h\in Z_1(Y)\}$
(Proposition~\ref{prop:subdivision-params}) and
$d^0(Y)=\min\{|\beta|:0\neq\beta\in Z^0(Y)\}$, with $\min\emptyset=\infty$; note
that $H_1(Y)=Z_1(Y)$ and $H^0(Y)=Z^0(Y)$ because $G$ has no $2$-cells, whereas
$H^1(Y)$ and $H_0(Y)$ are quotients whose classes always admit light
representatives (a single edge, a single vertex).

\subsection{CSS subsystem codes from demoted logical qubits}
\label{app:demoted}

Let $C^0\xrightarrow{d^0}C^1\xrightarrow{d^1}C^2$ be a three-term sheaf cochain
complex, $\partial_i=(d^{i-1})^T$ the dual boundary maps, and let the qubits be
the coordinates of $C^1$, with $X$-stabilizers $X(d^0\bar v_s)$ and
$Z$-stabilizers $Z(\partial_2\bar f)$ as in
Eqs.~\eqref{eq:A_stabilizer_untwisted} and~\eqref{eq:B_stabilizer_untwisted}.
A Pauli operator $X(\zeta^1)Z(\eta_1)$ commutes with all $Z$-stabilizers iff
$\langle\zeta^1,\partial_2\bar f\rangle=\langle d^1\zeta^1,\bar f\rangle=0$ for all
$f$, i.e.\ iff $d^1\zeta^1=0$, and with all $X$-stabilizers iff
$\langle d^0\bar v_s,\eta_1\rangle=\langle\bar v_s,\partial_1\eta_1\rangle=0$ for
all $(v,s)$, i.e.\ iff $\partial_1\eta_1=0$. Hence the Pauli normalizer of the
stabilizer group $\SS$ is
\be\label{eq:normalizer_app}
\N(\SS)=\big\{\pm X(\zeta^1)Z(\eta_1):\ d^1\zeta^1=0,\ \partial_1\eta_1=0\big\},
\ee
its elements act on the code space through the classes
$([\zeta^1],[\eta_1])\in H^1\oplus H_1$ (coboundaries and boundaries are
stabilizers), and $X(\zeta^1)$, $Z(\eta_1)$ commute iff
$\langle\zeta^1,\eta_1\rangle=0$, the pairing being well defined and
non-degenerate on classes by Proposition~\ref{prop:nondegen-sheaf}.

Fix $\mathsf L_X\subseteq H^1$, $\mathsf L_Z\subseteq H_1$ with the pairing
non-degenerate on $\mathsf L_X\times\mathsf L_Z$ (so
$\dim\mathsf L_X=\dim\mathsf L_Z=:k_{\rm prot}$), define $\mathsf G_X$,
$\mathsf G_Z$ by Eq.~\eqref{eq:gauge_classes}, and let the gauge group be
\be\label{eq:gauge_group_app}
\cG:=\big\langle\, i\I,\ \SS,\ X(\zeta^1)\ ([\zeta^1]\in\mathsf G_X),\
Z(\eta_1)\ ([\eta_1]\in\mathsf G_Z)\,\big\rangle .
\ee

\begin{lemma}[Structure of the demoted-logical code]\label{lemma:structure_app}
$H^1=\mathsf L_X\oplus\mathsf G_X$ and $H_1=\mathsf L_Z\oplus\mathsf G_Z$;
$\langle\mathsf G_X,\mathsf L_Z\rangle=\langle\mathsf L_X,\mathsf G_Z\rangle=0$;
and the pairing is non-degenerate on $\mathsf G_X\times\mathsf G_Z$. Consequently
the code space factorizes as $\H_{\rm prot}\otimes\H_{\rm gauge}$ with
$\dim\H_{\rm prot}=2^{k_{\rm prot}}$, every $P=X(\zeta^1)Z(\eta_1)\in\N(\SS)$ acts
as $\lo P_{\rm prot}\otimes\lo P_{\rm gauge}$, and $P$ acts trivially on
$\H_{\rm prot}$ (equivalently $P\in\cG$) if and only if $[\zeta^1]\in\mathsf G_X$
and $[\eta_1]\in\mathsf G_Z$.
\end{lemma}

\begin{proof}
$\mathsf L_X\cap\mathsf G_X=\{\ell\in\mathsf L_X:\langle\ell,\mathsf L_Z\rangle=0\}=0$
by non-degeneracy on $\mathsf L_X\times\mathsf L_Z$, and
$\dim\mathsf G_X=\dim H^1-\dim\mathsf L_Z$, so $H^1=\mathsf L_X\oplus\mathsf G_X$;
likewise for $H_1$. The two orthogonality relations are the definitions
\eqref{eq:gauge_classes}. If $g\in\mathsf G_X$ pairs to zero with $\mathsf G_Z$,
it pairs to zero with $\mathsf L_Z\oplus\mathsf G_Z=H_1$, hence $g=0$ by
Proposition~\ref{prop:nondegen-sheaf}. Thus $\mathsf L_X\oplus\mathsf L_Z$ and
$\mathsf G_X\oplus\mathsf G_Z$ are two mutually commuting symplectic subspaces of
the logical Pauli group $\N(\SS)/\SS$ whose direct sum is everything; choosing
symplectic bases of each (symplectic Gram--Schmidt) exhibits the logical Pauli
algebra as $\mathrm{End}(\H_{\rm prot})\otimes\mathrm{End}(\H_{\rm gauge})$.
Writing $[\zeta^1]=\ell_X+g_X$ and $[\eta_1]=\ell_Z+g_Z$ along the direct sums,
$\lo P_{\rm prot}$ is the logical Pauli operator of $(\ell_X,\ell_Z)$, which is
trivial iff $\ell_X=\ell_Z=0$.
\end{proof}

\begin{proof}[Proof of Lemma~\ref{lemma:subsystem}]
Part~(i) is Lemma~\ref{lemma:structure_app}. For part~(ii), let
$P=X(\zeta^1)Z(\eta_1)\in\N(\SS)\setminus\cG$. By Lemma~\ref{lemma:structure_app},
$[\zeta^1]\notin\mathsf G_X$ or $[\eta_1]\notin\mathsf G_Z$, and
$|P|=|\mathrm{supp}\,\zeta^1\cup\mathrm{supp}\,\eta_1|\ge\max(|\zeta^1|,|\eta_1|)
\ge\min(d^{\rm prot}_X,d^{\rm prot}_Z)$. Conversely a minimizer $\zeta^1$ of
$d_X^{\rm prot}$ in Eq.~\eqref{eq:dXdZ_prot} gives $X(\zeta^1)\in\N(\SS)\setminus\cG$
of weight $d_X^{\rm prot}$, and similarly for $Z$. Note that
$[\zeta^1]\notin\mathsf G_X$ iff $\langle\zeta^1,\eta_1\rangle=1$ for some
$[\eta_1]\in\mathsf L_Z$ iff the component of $[\zeta^1]$ in $\mathsf L_X$ along
$H^1=\mathsf L_X\oplus\mathsf G_X$ is non-zero. For part~(iii), every Pauli
operator on a tensor product of codes on disjoint qubit sets factorizes as
$P=\bigotimes_iP_i$; $P\in\N(\SS)$ iff each $P_i\in\N(\SS^i)$, and $P\in\cG$ iff
each $P_i\in\cG^i$. If $P\notin\cG$ then some $P_i\notin\cG^i$, which requires
$\mathsf L^i_X\neq0$, and $|P|\ge|P_i|\ge d(\C^i)$. Conversely a minimizer of
$\C^i$ tensored with identities lies in $\N(\SS)\setminus\cG$.
\end{proof}

\subsection{The K\"unneth contraction lemma}
\label{app:contraction}

Let $Y=Y_1\otimes Y_2$ be the tensor product of two one-dimensional sheaf
complexes on graphs $G_1$ and $G_2$ (in the application $Y_1$, $Y_2$ are the
$x$- and $y$-factors of one colour in Eq.~\eqref{eq:three_copies}). The
$1$-cells of $G_1\times G_2$ are the horizontal cells $(e,w)$, $e\in E(G_1)$,
$w\in V(G_2)$, of bidegree $(1,0)$, and the vertical cells $(v,f)$, $v\in V(G_1)$,
$f\in E(G_2)$, of bidegree $(0,1)$; accordingly a $1$-(co)chain splits as
$\chi=\chi^{(1,0)}+\chi^{(0,1)}$. The K\"unneth theorem gives
\begin{align}\label{eq:kunneth_both_app}
H^1(Y)&=\big[H^1(Y_1)\otimes H^0(Y_2)\big]\oplus\big[H^0(Y_1)\otimes H^1(Y_2)\big], \cr
H_1(Y)&=\big[H_1(Y_1)\otimes H_0(Y_2)\big]\oplus\big[H_0(Y_1)\otimes H_1(Y_2)\big],
\end{align}
and the pairing $H^1(Y)\times H_1(Y)\to\FF_2$ is the tensor product of the
factor pairings on matching summands and vanishes between the two summands
(cf.\ Lemma~\ref{lemma:factorization}). In particular, if $\mathsf L_X$ is one
summand of $H^1(Y)$ and $\mathsf L_Z$ the matching summand of $H_1(Y)$, then
$\mathsf G_X=\mathsf L_Z^{\perp}$ is the other summand of $H^1(Y)$, $\mathsf G_Z$
the other summand of $H_1(Y)$, and $[\zeta^1]\notin\mathsf G_X$ iff the K\"unneth
component of $[\zeta^1]$ in $\mathsf L_X$ is non-zero.

\begin{lemma}[Contraction]\label{lemma:kunneth_contraction}
\begin{enumerate}[label=(\alph*)]
\item For a cycle $h\in Z_1(Y_1)$ and $\zeta\in C^1(Y)$ let
$\iota_h\zeta:=(h\otimes\mathrm{id})\,\zeta^{(1,0)}\in C^0(Y_2)$, i.e.\
$(\iota_h\zeta)(w)=\sum_{e}\langle h(e),\zeta(e,w)\rangle$. Then $\iota_h$ maps
cocycles to cocycles and coboundaries to $0$, so it descends to
$H^1(Y)\to H^0(Y_2)$, where it equals $\langle h,\cdot\rangle\otimes\mathrm{id}$ on
$H^1(Y_1)\otimes H^0(Y_2)$ and $0$ on $H^0(Y_1)\otimes H^1(Y_2)$. Moreover
$|\iota_h\zeta|\le|\zeta|$.
\item For a $0$-cocycle $\beta\in Z^0(Y_2)$ and $\eta\in C_1(Y)$ let
$\iota^\beta\eta:=(\mathrm{id}\otimes\beta)\,\eta^{(1,0)}\in C_1(Y_1)$. Then
$\iota^\beta$ maps cycles to cycles and boundaries to $0$, so it descends to
$H_1(Y)\to H_1(Y_1)$, where it equals $\mathrm{id}\otimes\langle\beta,\cdot\rangle$
on $H_1(Y_1)\otimes H_0(Y_2)$ and $0$ on $H_0(Y_1)\otimes H_1(Y_2)$. Moreover
$|\iota^\beta\eta|\le|\eta|$.
\end{enumerate}
The same statements hold with the roles of the two factors interchanged.
\end{lemma}

\begin{proof}
(a) On a square $(e,f)$ the cocycle condition $d\zeta=0$ reads
$(d_2\zeta^{(1,0)})(e,f)=(d_1\zeta^{(0,1)})(e,f)$, where $d_1$, $d_2$ are the
coboundaries of the two factors. Applying $h\otimes\mathrm{id}$ and using
$\partial_1=(d_1)^T$,
\begin{align*}
\big(d_2\,\iota_h\zeta\big)(f)
&=\sum_e\Big\langle h(e),\sum_{v\prec e}\Fc_{v,e}\,\zeta(v,f)\Big\rangle \cr
&=\sum_v\big\langle(\partial_1h)(v),\zeta(v,f)\big\rangle=0 ,
\end{align*}
and likewise $\iota_h(d\sigma)=\sum_v\langle(\partial_1h)(v),\sigma(v,\cdot)\rangle=0$
for $\sigma\in C^0(Y)$. The formula on classes follows by evaluating on
elementary tensors $\as^1\otimes\bs^0$ and $\bs^0\otimes\as^1$. If
$(\iota_h\zeta)(w)\neq0$ then $\zeta(e,w)\neq0$ for some $e$, so distinct non-zero
vertices of $\iota_h\zeta$ come from distinct non-zero cells of $\zeta$, giving
$|\iota_h\zeta|\le|\zeta|$.
(b) The cycle condition $\partial\eta=0$ reads
$(\partial_1\otimes\mathrm{id})\eta^{(1,0)}=(\mathrm{id}\otimes\partial_2)\eta^{(0,1)}$;
applying $\mathrm{id}\otimes\beta$ and using $\beta\circ\partial_2=(d_2\beta)^T=0$
gives $\partial_1(\iota^\beta\eta)=0$. For a boundary $\eta=\partial\rho$ with
$\rho\in C_1(Y_1)\otimes C_1(Y_2)$ one has
$\eta^{(1,0)}=(\mathrm{id}\otimes\partial_2)\rho$, so
$\iota^\beta\eta=(\mathrm{id}\otimes\beta\partial_2)\rho=0$. The rest is as in (a).
\end{proof}

\begin{corollary}[Distance of a K\"unneth sector]\label{cor:sector_distance_app}
Let $\zeta\in Z^1(Y)$. If the component of $[\zeta]$ in $H^1(Y_1)\otimes H^0(Y_2)$
is non-zero then $|\zeta|\ge d^0(Y_2)$; if its component in
$H^0(Y_1)\otimes H^1(Y_2)$ is non-zero then $|\zeta|\ge d^0(Y_1)$. Let
$\eta\in Z_1(Y)$. If the component of $[\eta]$ in $H_1(Y_1)\otimes H_0(Y_2)$ is
non-zero then $|\eta|\ge d_1(Y_1)$; if its component in $H_0(Y_1)\otimes H_1(Y_2)$
is non-zero then $|\eta|\ge d_1(Y_2)$.
\end{corollary}

\begin{proof}
Write the component as $\sum_j[\as_j]\otimes[\bs_j]$ with $\{[\bs_j]\}$ a basis of
$H^0(Y_2)$ and some $[\as_j]\neq0$. By Proposition~\ref{prop:nondegen-sheaf}
there is $h\in Z_1(Y_1)$ with $\langle h,\as_j\rangle=1$. By
Lemma~\ref{lemma:kunneth_contraction}(a), $\iota_h\zeta$ is a $0$-cocycle of $Y_2$
with class $\sum_{j'}\langle h,\as_{j'}\rangle[\bs_{j'}]\neq0$, hence a non-zero
element of $Z^0(Y_2)$, and $|\zeta|\ge|\iota_h\zeta|\ge d^0(Y_2)$. The second
statement is the same with the factors interchanged. For the homological
statements use Lemma~\ref{lemma:kunneth_contraction}(b) with a $0$-cocycle
$\beta$ of the $H_0$-factor detecting the class, which exists by
non-degeneracy of $H^0\times H_0\to\FF_2$.
\end{proof}

For an elementary tensor $\zeta=u\otimes v$ the corollary reduces to the product
weight bound \eqref{eq:tensor_weight_bound}; its content is that the same bound
holds for every class with a non-zero component in the sector, in particular for
all linear combinations of basis classes.

\subsection{Proof of Lemma~\ref{lemma:distance_untwisted}}
\label{app:proof_untwisted}

Table~\ref{tab:sectors} lists the six K\"unneth sectors of the three copies in the
convention of Eq.~\eqref{eq:three_copies}, together with the factor quantity that
bounds each sector by Corollary~\ref{cor:sector_distance_app}. The symbol
``$\cdot$'' marks a sector whose classes admit $O(1)$ representatives (a single
vertex times a short graph cycle); these are exactly the spurious cycles of
Eq.~\eqref{eq:spurious-cycle}, and the corresponding $X$-sectors are their
conjugates.

\begin{table*}[t]
\centering
\begin{tabular}{@{}lllll@{}}
\toprule
sector & $X$-type ($H^1$) & bound & $Z$-type ($H_1$) & bound \\
\midrule
\rd\,1   & $H^1(G_x,\Fc_x)\otimes H^0(G_y,\Fc^{\perp}_y)$ & $d^0(G_y,\Fc^{\perp}_y)$ & $H_1(G_x,\Fc_x)\otimes H_0(G_y,\Fc^{\perp}_y)$ & $d_1(G_x,\Fc_x)$ \\
\rd\,2   & $H^0(G_x,\Fc_x)\otimes H^1(G_y,\Fc^{\perp}_y)$ & $d^0(G_x,\Fc_x)$ & $H_0(G_x,\Fc_x)\otimes H_1(G_y,\Fc^{\perp}_y)$ & $d_1(G_y,\Fc^{\perp}_y)$ \\
\bl\,1   & $H^0(G_x,\FF_2)\otimes H^1(G_y,\Fc_y)$ & $|V_x|$ & $H_0(G_x,\FF_2)\otimes H_1(G_y,\Fc_y)$ & $d_1(G_y,\Fc_y)$ \\
\bl\,2   & $H^1(G_x,\FF_2)\otimes H^0(G_y,\Fc_y)$ & $d^0(G_y,\Fc_y)$ & $H_1(G_x,\FF_2)\otimes H_0(G_y,\Fc_y)$ & $\cdot$ \\
\gr\,1   & $H^1(G_x,\Fc^{\perp}_x)\otimes H^0(G_y,\FF_2)$ & $|V_y|$ & $H_1(G_x,\Fc^{\perp}_x)\otimes H_0(G_y,\FF_2)$ & $d_1(G_x,\Fc^{\perp}_x)$ \\
\gr\,2   & $H^0(G_x,\Fc^{\perp}_x)\otimes H^1(G_y,\FF_2)$ & $d^0(G_x,\Fc^{\perp}_x)$ & $H_0(G_x,\Fc^{\perp}_x)\otimes H_1(G_y,\FF_2)$ & $\cdot$ \\
\bottomrule
\end{tabular}
\caption{K\"unneth sectors of the three copies of Eq.~\eqref{eq:three_copies}
and the lower bound on the weight of every (co)cycle whose class has a non-zero
component in the sector (Corollary~\ref{cor:sector_distance_app}). Sectors
\rd\,1, \rd\,2 and \bl\,1 are protected in Lemma~\ref{lemma:distance_untwisted};
\bl\,2 and the whole \gr\ copy are gauge. The entries $d^0(G_x,\Fc_x)$,
$d_1(G_y,\Fc^{\perp}_y)$ and $d_1(G_x,\Fc^{\perp}_x)$ are $\infty$ when $\Hs$ has
full row rank (the corresponding sectors are then absent) and $\Theta(\bar n)$
when the transposed code is good.}
\label{tab:sectors}
\end{table*}

\begin{proof}[Proof of Lemma~\ref{lemma:distance_untwisted}]
\emph{Protected pairs.} The \rd\ choice $\mathsf L^{\rd}_X=H^1(\L,\Fc^{\rd})$,
$\mathsf L^{\rd}_Z=H_1(\L,\Fc^{\rd})$ is non-degenerately paired by
Proposition~\ref{prop:nondegen-sheaf}, with $\mathsf G^{\rd}_X=\mathsf G^{\rd}_Z=0$.
The \bl\ choice is the sector \bl\,1 of Table~\ref{tab:sectors}; its pairing is
$\langle\blue{\cs^0},\cs_0\rangle\cdot\langle\cdot,\cdot\rangle_{G_y}
=\langle\cdot,\cdot\rangle_{G_y}$, non-degenerate, and by the remark after
Eq.~\eqref{eq:kunneth_both_app} the gauge classes $\mathsf G^{\bl}_X$,
$\mathsf G^{\bl}_Z$ are the sector \bl\,2, i.e.\ the spurious pair. The \gr\
copy has $\mathsf L^{\gr}=0$. Hence Lemma~\ref{lemma:subsystem} applies, and by
its part~(iii), $d=\min\big(d(\C^{\rd}),d(\C^{\bl})\big)$.

\emph{Dimension.} By Proposition~\ref{prop:subdivision-params} and Poincar\'e
duality, $\dim H^1(G_\nu,\Fc_\nu)=\dim H^0(G_\nu,\Fc^{\perp}_\nu)=\bar k_\nu$ and
$\dim H^0(G_\nu,\Fc_\nu)=\dim H^1(G_\nu,\Fc^{\perp}_\nu)=\bar k^{\top}_\nu$, where
$\bar k_\nu$, $\bar k^{\top}_\nu$ are the dimensions of the classical code and its
transpose in the $\nu$-factor. Inserting into Eq.~\eqref{eq:kunneth_H1} gives
$k_{\rm prot}=\bar k_x\bar k_y+\bar k^{\top}_x\bar k^{\top}_y+\bar k_y=\Theta(n)$
by Eqs.~\eqref{eq:good_code} and~\eqref{eq:n_scaling}. (The gauge register has
$\bar k^{\top}_y\,\beta_1(G_x)+\dim H^1(\L,\Fc^{\gr})$ qubits, with
$\beta_1(G)=|E|-|V|+1$; it is $\Theta(n)$ as well.)

\emph{Lower bound, \rd\ copy.} Since $\mathsf G^{\rd}_X=0$, every cocycle
$\zeta\in Z^1(\L,\Fc^{\rd})$ with $[\zeta]\neq0$ has a non-zero component in
\rd\,1 or \rd\,2. By Corollary~\ref{cor:sector_distance_app} with
$Y_1=C^\bullet(G_x,\Fc_x)$, $Y_2=C^\bullet(G_y,\Fc^{\perp}_y)$, a non-zero \rd\,1
component forces $|\zeta|\ge d^0(G_y,\Fc^{\perp}_y)$ and a non-zero \rd\,2
component forces $|\zeta|\ge d^0(G_x,\Fc_x)$. Dually, every cycle
$\eta\in Z_1(\L,\Fc^{\rd})$ with $[\eta]\neq0$ has $|\eta|\ge d_1(G_x,\Fc_x)$
(\rd\,1) or $|\eta|\ge d_1(G_y,\Fc^{\perp}_y)$ (\rd\,2). If $\Hs$ has full row
rank the \rd\,2 sectors are absent and the two corresponding bounds are vacuous;
otherwise they are $\Theta(\bar n)$ by the assumption that the transposed code is
good, via Proposition~\ref{prop:subdivision-params} and
Eq.~\eqref{eq:dual_cocycle_weight} with $\Fc$ and $\Fc^{\perp}$ interchanged.

\emph{Lower bound, \bl\ copy.} $[\zeta]\notin\mathsf G^{\bl}_X$ iff the \bl\,1
component of $[\zeta]$ is non-zero. With $Y_1=C^\bullet(G_x,\FF_2)$,
$Y_2=C^\bullet(G_y,\Fc_y)$ and the factors interchanged in
Corollary~\ref{cor:sector_distance_app}, $|\zeta|\ge d^0(G_x,\FF_2)=|V_x|$:
concretely, contracting the $y$-factor against a cycle $h'\in Z_1(G_y,\Fc_y)$
detecting the class leaves a non-zero multiple of the all-ones cocycle
$\blue{\cs^0}$ of Eq.~\eqref{eq:rep_factor}. Dually, $[\eta]\notin\mathsf G^{\bl}_Z$
iff the \bl\,1 component of $[\eta]$ is non-zero, and contracting the $x$-factor
against $\blue{\cs^0}\in Z^0(G_x,\FF_2)$, which detects the generator of
$H_0(G_x,\FF_2)$, leaves a non-zero cycle of $C^\bullet(G_y,\Fc_y)$, so
$|\eta|\ge d_1(G_y,\Fc_y)$. No bound on the sector \bl\,2 is needed: it is gauge.

\emph{Lower bound, \gr\ copy.} $\mathsf L^{\gr}=0$, so the \gr\ copy does not
enter Lemma~\ref{lemma:subsystem}(iii).

Collecting the bounds and using Eq.~\eqref{eq:subdiv_params} for
$d_1(G_\nu,\Fc_\nu)$, Eq.~\eqref{eq:dual_cocycle_weight} for
$d^0(G_y,\Fc^{\perp}_y)$ and Eq.~\eqref{eq:rep_factor} for $|V_x|$ proves
Eq.~\eqref{eq:untwisted_distance_bound}.

\emph{Upper bound.} Fix $w_0\in V_y$ and a $0$-chain $\pi_0\in C_0(G_y,\Fc^{\perp}_y)$
supported on $w_0$ such that the vector $v(\pi_0):=(\langle\pi_0,\red{\bs^0_j}\rangle)_j$
is non-zero; such a $\pi_0$ exists because the global sections $\red{\bs^0_j}$ are
linearly independent, so their evaluation vectors at the vertices span the dual
space. Put $\eta:=\eta^{i}_1\otimes\pi_0$ with $\eta^{i}_1=\text{PD}[\red{\bs^0_i}]\in H_1(G_x,\Fc_x)$.
Then $\partial\eta=(\partial_1\eta^{i}_1)\otimes\pi_0=0$, so
$\widetilde Z^{\rd}_{\eta}$ commutes with all stabilizers, and
$\langle\red{\alpha^1_{\vc{i_1}}},\eta\rangle
=\langle\red{\as^1_{i_1}},\eta^{i}_1\rangle\,\langle\red{\bs^0_{j_1}},\pi_0\rangle
=\delta_{i_1,i}\,v(\pi_0)_{j_1}$ by Corollary~\ref{cor:kronecker-dual}, which is
non-zero for some $\vc{i_1}$. Hence $[\eta]\notin\mathsf G^{\rd}_Z=0$ and
$\widetilde Z^{\rd}_{\eta}$ is a logical operator acting non-trivially on the
protected register, of qubit weight at most $\Delta|\eta^{i}_1|\le\Delta^2|E_x|=O(\sqrt n)$.
\end{proof}

\begin{remark}[One-sided variant]\label{rem:one_sided_app}
If no assumption at all is made on the transposed code (redundant rows, light
$H^0(G,\Fc)$), the \rd\,2 pair may contain light classes of both types and the
\rd\ copy is no longer a subspace code. Demoting \rd\,2 to gauge as well, i.e.\
taking $\mathsf L^{\rd}_X=H^1(G_x,\Fc_x)\otimes H^0(G_y,\Fc^{\perp}_y)
=\mathrm{span}\{[\red{\alpha^1_{\vc{i_1}}}]\}$, the same proof gives
$k_{\rm prot}=\bar k_x\bar k_y+\bar k_y=\Theta(n)$ and
$d\ge\min\{d^0(G_y,\Fc^{\perp}_y),d_1(G_x,\Fc_x),|V_x|,d_1(G_y,\Fc_y)\}=\Theta(\sqrt n)$
from Eqs.~\eqref{eq:subdiv_params}, \eqref{eq:dual_cocycle_weight} and
\eqref{eq:rep_factor} alone. The protected register is then exactly the set of
classes $\{[\red{\alpha^1_{\vc{i_1}}}],[\blue{\beta^1_{\vc{i_2}}}]\}$ of
Eq.~\eqref{eq:spatial_cocycle_decomposition} that carries the logical CZ
structure of Eq.~\eqref{eq:intersection_structure}.
\end{remark}

%======================================================================
\section{Subsystem Knill--Laflamme distance of the twisted code: proof of
Theorem~\ref{thm:distance_twisted_KL}}
\label{app:twisted_KL}
%======================================================================

Throughout this appendix $\bar\Hs$ has full row rank, so that by
Proposition~\ref{prop:subdivision-params} and Poincar\'e duality
\begin{align}\label{eq:fullrank_vanishing}
H^0(G_\nu,\Fc_\nu)&=H_0(G_\nu,\Fc_\nu)=0, \cr
H^1(G_\nu,\Fc^{\perp}_\nu)&=H_1(G_\nu,\Fc^{\perp}_\nu)=0 ,
\end{align}
for $\nu=x,y$. We use the weight conventions and the quantities $d_1$, $d^0$ of
Appendix~\ref{app:subsystem_distance}, and we drop colours and degree
superscripts on cochains when no confusion can arise: a flat configuration is
$x=(a,b,c)$ with $a=\red{a^1}\in Z^1(\L,\Fc^{\rd})$, $b=\blue{b^1}\in Z^1(\L,\Fc^{\bl})$,
$c=\green{c^1}\in Z^1(\L,\Fc^{\gr})$, and a gauge parameter is
$u=(\mu,\nu,\rho)\in C^0:=C^0(\L,\Fc^{\rd})\oplus C^0(\L,\Fc^{\bl})\oplus C^0(\L,\Fc^{\gr})$,
the $0$-cocycles among them being the $\mu^0,\nu^0,\green{\gamma^0}$ of
Eqs.~\eqref{eq:constraint_operator_sheaf}--\eqref{eq:constraint_operator_sheaf2}.
$1$-cells of $\L$ are horizontal, $(e,w)$ with $e\in E_x$, $w\in V_y$ (bidegree
$(1,0)$), or vertical, $(v,f)$ with $v\in V_x$, $f\in E_y$ (bidegree $(0,1)$).

\subsection{K\"unneth structure under full row rank}
\label{app:KL_kunneth}

By \eqref{eq:fullrank_vanishing}, \eqref{eq:kunneth_H1} and
\eqref{eq:rep_factor}, the three copies have
\begin{align}
H^1(\L,\Fc^{\rd})&=H^1(G_x,\Fc_x)\otimes H^0(G_y,\Fc^{\perp}_y), \cr
H^1(\L,\Fc^{\bl})&=\FF_2\blue{\cs^0}\otimes H^1(G_y,\Fc_y), \cr
H^1(\L,\Fc^{\gr})&=H^0(G_x,\Fc^{\perp}_x)\otimes H^1(G_y,\FF_2), \label{eq:kunneth_fullrank}
\end{align}
with the homological counterparts obtained by lowering all indices, and the
zero modes are
\begin{align}
Z^0(\L,\Fc^{\rd})&=H^0(G_x,\Fc_x)\otimes H^0(G_y,\Fc^\perp_y)=0, \cr
Z^0(\L,\Fc^{\bl})&=\FF_2\blue{\cs^0}\otimes H^0(G_y,\Fc_y)=0, \cr
Z^0(\L,\Fc^{\gr})&=H^0(G_x,\Fc^{\perp}_x)\otimes\FF_2\green{\cs'^0}
=\mathrm{span}\{\green{\gamma^0_{\vc{i_3}}}\}. \label{eq:zeromodes_fullrank}
\end{align}
Thus the \rd\,2 and \bl\,2 sectors of Table~\ref{tab:sectors} are absent: the
\rd\ copy has only its protected sector, the \bl\ copy has no graph-short
classes, and the only light sector is the $Z$-type \gr\,2 sector
$H_0(G_x,\Fc^{\perp}_x)\otimes H_1(G_y,\FF_2)$, which is gauge.

Fix the bases of Eq.~\eqref{eq:spatial_cocycle_decomposition}, with
$\{\red{\as^1_{i_1}}\}_{i_1\le p}$ and $\{\green{\bs^0_{i_3}}\}_{i_3\le p}$ dual
bases of $H^1(G_x,\Fc_x)$ and $H^0(G_x,\Fc^{\perp}_x)$ and
$\{\blue{\as^1_{j_2}}\}_{j_2\le q}$, $\{\red{\bs^0_{j_1}}\}_{j_1\le q}$ dual bases
of $H^1(G_y,\Fc_y)$ and $H^0(G_y,\Fc^{\perp}_y)$ (Corollary~\ref{cor:kronecker-dual}),
so that $\Lambda(\vc{i_1},\vc{i_2},\vc{i_3})=\delta_{i_1,i_3}\delta_{j_1,j_2}$ by
Eq.~\eqref{eq:intersection_structure}. For every \gr\ class
$\green{g}\in H^1(\L,\Fc^{\gr})$ fix a representative $c_{\green g}$ of pure
bidegree $(0,1)$ (a sum of $\beta\otimes\fs$ with $\beta\in Z^0(G_x,\Fc^\perp_x)$,
$\fs\in C^1(G_y,\FF_2)$). For $A\in\FF_2^{p\times q}$ and $\mathsf b\in\FF_2^{q}$
define the \emph{reference representatives}
\begin{align}\label{eq:refreps_app}
a_A&:=\sum_{i_1,j_1}A_{i_1j_1}\,\red{\as^1_{i_1}}\otimes\red{\bs^0_{j_1}}
=\sum_{j_1}\as^x_{A,j_1}\otimes\red{\bs^0_{j_1}}, \cr
b_{\mathsf b}&:=\blue{\cs^0}\otimes\sum_{j_2}\mathsf b_{j_2}\blue{\as^1_{j_2}}
=\blue{\cs^0}\otimes\as^y_{\mathsf b},
\end{align}
with $\as^x_{A,j_1}:=\sum_{i_1}A_{i_1j_1}\red{\as^1_{i_1}}\in Z^1(G_x,\Fc_x)$ and
$\as^y_{\mathsf b}\in Z^1(G_y,\Fc_y)$; $a_A$ has pure bidegree $(1,0)$ and
$b_{\mathsf b}$ pure bidegree $(0,1)$.

\subsection{Dressed gauge transformations and orbit states}
\label{app:KL_orbit}

For a flat configuration $x=(a,b,c)$ and $u=(\mu,\nu,\rho)\in C^0$ put
\begin{align}
T_{\rd}(\mu;b,c)&:=\int_{\eta_2}\mu\cup b\cup c, \cr
T_{\bl}(a;\nu,c)&:=\int_{\eta_2}a\cup\nu\cup c, \cr
T_{\gr}(a,b;\rho)&:=\int_{\eta_2}a\cup b\cup\rho, \label{eq:T_forms}\\
\theta_x(u)&:=T_{\rd}(\mu;b,c)+T_{\bl}(a+d\mu;\nu,c) \cr
&\qquad+T_{\gr}(a+d\mu,\,b+d\nu;\rho). \label{eq:theta_def}
\end{align}
These are the $\FF_2$-trilinear forms of Definition~\ref{def:invariant-form}.
We shall use three properties of them. (P1) \emph{Locality}: by
Definition~\ref{def:cup-simplicial}, $\int_{\eta_2}\alpha\cup\beta\cup\gamma$ is a
sum over the $2$-cells $\sigma$ of $\L$ of a trilinear function of the values of
$\alpha,\beta,\gamma$ on the faces of $\sigma$. (P2) \emph{Bigrading}: with the
product cup structure on $\L=G_x\times G_y$ underlying
Lemma~\ref{lemma:factorization}, a triple product of cochains of total degree $2$
vanishes identically unless the bidegrees add up to $(1,1)$; in particular two
factors of bidegree $(0,1)$ (or two of bidegree $(1,0)$) give zero. (P3)
\emph{Leibniz on the cycle}: by Proposition~\ref{prop:cup-properties}(a) and
$\partial\eta_2=0$, for cochains of total degree $1$,
\be\label{eq:leibniz_cycle}
\int_{\eta_2}\!d\alpha\cup\beta\cup\gamma+\int_{\eta_2}\!\alpha\cup d\beta\cup\gamma
+\int_{\eta_2}\!\alpha\cup\beta\cup d\gamma=0,
\ee
exactly as in Eq.~\eqref{eq:AA_commutator}.

Let
\be\label{eq:A_u_def}
\tilde{A}(u):=\prod_{v,t}\big(\tilde A^{\gr}_{v,t}\big)^{\gamma_{v,t}}
\prod_{v,r}\big(\tilde A^{\bl}_{v,r}\big)^{\nu_{v,r}}
\prod_{v,s}\big(\tilde A^{\rd}_{v,s}\big)^{\mu_{v,s}}
\ee
be the product of dressed generators \eqref{eq:stabilizer_summary} labelled by
$u$. Applying the three factors in turn to a flat computational-basis state and
reading off the CZ phases as in Eq.~\eqref{eq:constraint_derivation},
\be\label{eq:A_action_app}
\tilde{A}(u)\ket{x}=(-1)^{\theta_x(u)}\ket{x+d^0u} .
\ee

\begin{lemma}[Group law]\label{lemma:grouplaw_app}
On flat configurations, $\tilde A(u)\tilde A(u')=\tilde A(u+u')$; equivalently
\be\label{eq:grouplaw_app}
\theta_x(u+u')=\theta_x(u)+\theta_{x+d^0u}(u')\qquad\forall\,u,u'\in C^0 .
\ee
\end{lemma}
\begin{proof}
Within one colour the Pauli parts compose additively and the phases
\eqref{eq:T_forms} are linear in the gauge parameter, so
$\tilde A^{\rd}(\mu)\tilde A^{\rd}(\mu')=\tilde A^{\rd}(\mu+\mu')$ and likewise for
\bl, \gr. Generators of different colours commute on the zero-flux subspace by
Lemma~\ref{lemma:stabilizer_commutation}, and flat configurations lie in that
subspace. Hence the ordered product $\tilde A(u)$ is independent of the ordering
on flat states and multiplicative in $u$. Comparing phases in
\eqref{eq:A_action_app} gives \eqref{eq:grouplaw_app}.
\end{proof}

For flat $x$ define the (unnormalized) \emph{orbit state} and the
\emph{charge-parity character} of the zero modes $Z^0:=\ker d^0\subset C^0$,
\begin{align}\label{eq:orbit_state_app}
\ket{\psi_x}&:=\sum_{u\in C^0}(-1)^{\theta_x(u)}\ket{x+d^0u}, \cr
\mathsf Q_x(\varsigma)&:=\theta_x(\varsigma)\qquad(\varsigma\in Z^0).
\end{align}
By \eqref{eq:zeromodes_fullrank} and \eqref{eq:theta_def},
$\mathsf Q_x(\varsigma)=T_{\gr}(a,b;\varsigma_{\gr})$ with $\varsigma_{\gr}\in\mathrm{span}\{\green{\gamma^0_{\vc{i_3}}}\}$:
this is the charge parity of Eq.~\eqref{eq:constraint_operator_sheaf2}.

\begin{lemma}\label{lemma:orbit_app}
Let $x$ be flat.
\begin{enumerate}[label=(\roman*)]
\item $\theta_{x+d^0u}(\varsigma)=\mathsf Q_x(\varsigma)$ for all $u\in C^0$,
$\varsigma\in Z^0$; in particular $\mathsf Q_x$ depends only on the cohomology
classes of $x$.
\item $\langle\psi_x|\psi_x\rangle=|C^0||Z^0|$ if $\mathsf Q_x\equiv0$ (we then call
$x$ \emph{admissible}) and $0$ otherwise.
\item $\ket{\psi_{x+d^0\lambda}}=(-1)^{\theta_x(\lambda)}\ket{\psi_x}$ for every
$\lambda\in C^0$.
\item Orbit states of configurations with different cohomology classes are
orthogonal.
\end{enumerate}
\end{lemma}
\begin{proof}
(i) Apply \eqref{eq:grouplaw_app} to $u+\varsigma$ and to $\varsigma+u$:
$\theta_x(u+\varsigma)=\theta_x(u)+\theta_{x+d^0u}(\varsigma)$ and
$\theta_x(\varsigma+u)=\theta_x(\varsigma)+\theta_{x+d^0\varsigma}(u)
=\theta_x(\varsigma)+\theta_x(u)$; since $C^0$ is abelian,
$\theta_{x+d^0u}(\varsigma)=\theta_x(\varsigma)$.
(ii) Group $\sum_{u,u'}$ by $d^0u=d^0u'$, i.e.\ $u'=u+\varsigma$, and use
$\theta_x(u+\varsigma)=\theta_x(u)+\mathsf Q_x(\varsigma)$; the inner sum is
$\sum_\varsigma(-1)^{\mathsf Q_x(\varsigma)}$.
(iii) Substitute $u\mapsto u+\lambda$ in \eqref{eq:orbit_state_app} and use
\eqref{eq:grouplaw_app}. (iv) Distinct classes have disjoint supports in the
computational basis.
\end{proof}

By Eq.~\eqref{eq:intersection_structure}, $x=(a_A,b_{\mathsf b},c_{\green g})$ is
admissible iff $T_{\gr}(a_A,b_{\mathsf b};\green{\gamma^0_{\vc{i_3}}})
=\sum_{j}A_{i_3j}\mathsf b_j=0$ for all $i_3$, i.e.\ iff $A\mathsf b=0$,
Eq.~\eqref{eq:charge_constraint_Ab}, independently of $\green g$. Here the
full-row-rank hypothesis enters for the first time: with $H^0(G_y,\Fc_y)\ne0$
the \bl\ zero modes $\blue{\cs^0}\otimes h$ of Eq.~\eqref{eq:constraint_operator_sheaf2}
would impose $T_{\bl}(a;\blue{\cs^0}\otimes h,c)=0$, which has bidegree
$(1,0)+(0,0)+(0,1)=(1,1)$ and couples $A$ to the \gr\ label.

\begin{proposition}[Orbit basis of $\tilde{\C}$]\label{prop:orbit_basis_app}
For every flat gauge orbit $\mathcal O_x=\{x+d^0u\}$, the subspace of
$\tilde{\C}$ supported on $\mathcal O_x$ is one-dimensional and spanned by
$\ket{\psi_x}$ if $\mathsf Q_x\equiv0$, and zero otherwise. Hence the normalized
$\ket{\psi_{(a_A,b_{\mathsf b},c_{\green g})}}$ with $A\mathsf b=0$ form an
orthonormal basis of $\tilde{\C}$, and
\begin{align}\label{eq:dim_RB}
\dim\H_{\rd\bl}&=2^{pq}+(2^{q}-1)\,2^{p(q-1)}, \cr
\log_2\dim\H_{\rd\bl}&=\Theta(pq)=\Theta(n).
\end{align}
\end{proposition}
\begin{proof}
The $B$-stabilizers restrict the computational basis to flat configurations,
on which $C^0$ acts by \eqref{eq:A_action_app}. The stabilizer of a point under
$x\mapsto x+d^0u$ is $Z^0$, acting on $\ket{x+d^0u}$ by the character
$(-1)^{\mathsf Q_x}$, the same at every point of the orbit by
Lemma~\ref{lemma:orbit_app}(i). If $\mathsf Q_x\not\equiv0$ some $\tilde A(\varsigma)$
acts as $-\I$ on the span of the orbit, so there is no invariant vector. If
$\mathsf Q_x\equiv0$ the rescaled basis $\{(-1)^{\theta_x(u)}\ket{x+d^0u}\}$ is
permuted by $\tilde A(u')$ [Eq.~\eqref{eq:grouplaw_app}], so the action is the
regular representation of $C^0/Z^0$, whose invariant subspace is one-dimensional
and spanned by $\ket{\psi_x}\ne0$. Orthogonality is Lemma~\ref{lemma:orbit_app}(iv),
and \eqref{eq:kunneth_fullrank} with \eqref{eq:charge_constraint_Ab} enumerate the
admissible classes: $2^{pq}$ labels with $\mathsf b=0$ and $2^{p(q-1)}$ matrices
$A$ for each of the $2^q-1$ non-zero $\mathsf b$.
\end{proof}

\begin{remark}\label{rem:not_power_of_two}
$\dim\H_{\rd\bl}$ is not a power of two, because the constraint
\eqref{eq:charge_constraint_Ab} is bilinear in the labels. The statement
``one $\FF_2$ constraint per charge-parity operator'' of
Lemma~\ref{lemma:rate_twisted} is therefore to be read at the level of
$\log_2\dim\tilde{\C}=\Theta(n)$; a genuine qubit register is extracted in
part~(iv) of Theorem~\ref{thm:distance_twisted_KL}.
\end{remark}

\begin{definition}\label{def:V_app}
Let $\H_{\rd\bl}$ have orthonormal basis $\{\ket{A,\mathsf b}:A\mathsf b=0\}$ and
$\H_{\gr}$ orthonormal basis $\{\ket{\green g}:\green g\in H^1(\L,\Fc^{\gr})\}$, and let
$V\ket{A,\mathsf b}\otimes\ket{\green g}:=(|C^0||Z^0|)^{-1/2}\ket{\psi_{(a_A,b_{\mathsf b},c_{\green g})}}$.
By Proposition~\ref{prop:orbit_basis_app}, $V$ is a unitary onto $\tilde{\C}$ and
$P_{\tilde{\C}}=VV^\dagger$. This proves Theorem~\ref{thm:distance_twisted_KL}(i).
\end{definition}

\subsection{No low-weight error changes a label}
\label{app:KL_heavy}

\begin{proposition}\label{prop:heavy_app}
Let $d_\star:=\min\{d^0(G_y,\Fc^{\perp}_y),\,|V_x|,\,d^0(G_x,\Fc^{\perp}_x)\}=\Theta(\bar n)$.
If $\zeta_{\rd}\in Z^1(\L,\Fc^{\rd})$, $\zeta_{\bl}\in Z^1(\L,\Fc^{\bl})$,
$\zeta_{\gr}\in Z^1(\L,\Fc^{\gr})$ have $|\zeta_i|<d_\star$, then all three are
coboundaries.
\end{proposition}
\begin{proof}
By \eqref{eq:kunneth_fullrank} each copy has a single K\"unneth summand, so a
non-zero class has a non-zero component there, and Corollary~\ref{cor:sector_distance_app}
applies: for \rd, $|\zeta_{\rd}|\ge d^0(G_y,\Fc^{\perp}_y)$; for \bl, contracting the
$y$-factor against a detecting cycle leaves $\blue{\cs^0}$, so $|\zeta_{\bl}|\ge|V_x|$;
for \gr, contracting the $y$-factor against a detecting graph cycle leaves a
non-zero element of $Z^0(G_x,\Fc^\perp_x)$, so $|\zeta_{\gr}|\ge d^0(G_x,\Fc^\perp_x)$.
The three quantities are $\Theta(\bar n)$ by Eqs.~\eqref{eq:dual_cocycle_weight}
and~\eqref{eq:rep_factor}.
\end{proof}

\begin{corollary}\label{cor:blockdiag_app}
Let $E=X(\zeta)Z(\eta)$ be a Pauli operator with $|E|<d_\star$. Then
$V^\dagger P_{\tilde{\C}}EP_{\tilde{\C}}V$ is diagonal in the basis
$\{\ket{A,\mathsf b}\otimes\ket{\green g}\}$, with entries
$m_E(A,\mathsf b,\green g)=\langle\psi_x|E|\psi_x\rangle/\langle\psi_x|\psi_x\rangle$,
$x=(a_A,b_{\mathsf b},c_{\green g})$.
\end{corollary}
\begin{proof}
$\langle\psi_{x'}|E|\psi_x\rangle\neq0$ requires $x'+d^0u'=x+d^0u+\zeta$ for some
$u,u'$, so each $\zeta_i$ is a cocycle with $[\zeta_i]=[x'_i]-[x_i]$; by
Proposition~\ref{prop:heavy_app} these classes vanish.
\end{proof}

\subsection{The exact matrix element}
\label{app:KL_matrix_element}

From now on $E=X(\zeta)Z(\eta)$ with $|E|<d_\star$, so $\zeta_i=d^0\tau_i$ for
some $\tau=(\tau_{\rd},\tau_{\bl},\tau_{\gr})\in C^0$; by
\eqref{eq:zeromodes_fullrank} $\tau_{\rd}$ and $\tau_{\bl}$ are unique and
$\tau_{\gr}$ is unique up to $Z^0(\L,\Fc^{\gr})$. (In the language of
Eq.~\eqref{eq:magnetic_contractible_sheaf}, $X(\zeta)$ is the product of dressed
stabilizers $\tilde A(\tau)$ up to the diagonal CZ-type factor
$(-1)^{\theta_x(\tau)}$.)

\begin{lemma}[Exact matrix element]\label{lemma:exact_app}
For admissible $x$,
\begin{align}\label{eq:exact_app}
\frac{\langle\psi_x|E|\psi_x\rangle}{\langle\psi_x|\psi_x\rangle}
&=\frac1{|C^0|}\sum_{u\in C^0}(-1)^{\Theta_x(u)}, \cr
\Theta_x(u)&:=\theta_x(u)+\theta_x(u+\tau)+\langle\eta,x+d^0u\rangle .
\end{align}
\end{lemma}
\begin{proof}
$E\ket{x+d^0u}=(-1)^{\langle\eta,x+d^0u\rangle}\ket{x+d^0(u+\tau)}$. In
$\sum_{u,u'}(-1)^{\theta_x(u)+\theta_x(u')}\langle x+d^0u'|E|x+d^0u\rangle$ only
$u'=u+\tau+\varsigma$ with $\varsigma\in Z^0$ contribute, and
$\theta_x(u+\tau+\varsigma)=\theta_x(u+\tau)+\mathsf Q_x(\varsigma)=\theta_x(u+\tau)$
by \eqref{eq:grouplaw_app}, Lemma~\ref{lemma:orbit_app}(i) and admissibility. The
sum over $\varsigma$ gives $|Z^0|$, which cancels against the norm.
\end{proof}

\begin{lemma}[Column-localized primitives]\label{lemma:column_primitive}
Let $R=C^\bullet(G,\FF_2)$ for a connected graph $G=(V,E)$ of degree $\le\Delta$
and let $B=(B^0\xrightarrow{d_B}B^1)$ be a one-dimensional sheaf complex with
$d^0(B)=\min\{|\beta|:0\ne\beta\in Z^0(B)\}$ ($=\infty$ if $Z^0(B)=0$). Let
$\zeta=d\sigma$ be an exact $1$-cochain of $R\otimes B$ with $|\zeta|<d^0(B)$, and
let $\Pi\subseteq V$ be the set of $w\in V$ such that $\zeta$ is non-zero on a cell
$(w,\cdot)$ or on a cell $(f,\cdot)$ with $w\in f$, so $|\Pi|\le2|\zeta|$. Then
$\zeta$ has a primitive $\tau\in C^0(R\otimes B)=\bigoplus_{w\in V}B^0$ with
$\tau_w=0$ for all $w\notin\Pi$. The same holds for $B\otimes R$.
\end{lemma}
\begin{proof}
Write $\sigma=(\sigma_w)_{w\in V}$; the components of $d\sigma$ are
$(d\sigma)(w,\cdot)=d_B\sigma_w$ and $(d\sigma)(f,\cdot)=\sigma_w+\sigma_{w'}$ for
$f=(w,w')$, the constant sheaf having identity co-restriction maps. Let
$W:=V\setminus\Pi$. For $w\in W$, $d_B\sigma_w=0$, so $\sigma_w\in Z^0(B)$; for an
edge $f=(w,w')$ with $w,w'\in W$, $\sigma_w=\sigma_{w'}$. Hence $\sigma$ is constant
on every connected component $K$ of the induced subgraph $G[W]$, with value
$\sigma_K\in Z^0(B)$. All these values coincide: otherwise take components
$K\ne K'$ with $\sigma_K\ne\sigma_{K'}$ and a shortest path $w_0,\dots,w_m$ in $G$
from $K$ to $K'$ among all such pairs; no interior vertex lies in $W$ (it would
belong to a component differing from $K$ or $K'$ and give a shorter path), so
$\sigma_K+\sigma_{K'}=\sum_{i<m}(\sigma_{w_i}+\sigma_{w_{i+1}})=\sum_{i<m}\zeta(f_i,\cdot)$
over distinct edges $f_i$, whence $|\sigma_K+\sigma_{K'}|\le|\zeta|<d^0(B)$ although
$\sigma_K+\sigma_{K'}$ is a non-zero element of $Z^0(B)$---a contradiction. Let
$\sigma_\star$ be the common value; $\mathbf 1\otimes\sigma_\star\in Z^0(R\otimes B)$,
so $\tau:=\sigma-\mathbf 1\otimes\sigma_\star$ is a primitive of $\zeta$ vanishing
on $W$.
\end{proof}

Applied to $\L$ with $\Fc^{\bl}=\FF_2\otimes\Fc_y$ (columns $V_x$, $B=C^\bullet(G_y,\Fc_y)$,
$Z^0(B)=0$) and $\Fc^{\gr}=\Fc^\perp_x\otimes\FF_2$ (columns $V_y$,
$B=C^\bullet(G_x,\Fc^\perp_x)$, $d^0(B)\ge d_\star>|\zeta_{\gr}|$), we fix primitives
with
\begin{align}\label{eq:columns_app}
\mathrm{supp}_{V_x}(\tau_{\bl})&\subseteq\Pi_{\bl},\qquad |\Pi_{\bl}|\le2|\zeta_{\bl}|, \cr
\mathrm{supp}_{V_y}(\tau_{\gr})&\subseteq\Pi_{\gr},\qquad |\Pi_{\gr}|\le2|\zeta_{\gr}| .
\end{align}
The individual columns $\tau_{\bl}(v)\in C^0(G_y,\Fc_y)$, $\tau_{\gr}(w)\in C^0(G_x,\Fc^\perp_x)$
may be heavy; only their column supports will be used.

\begin{lemma}[Protected-label dependence]\label{lemma:expansion_app}
For admissible $x=(a,b,c)$ with $a$ of pure bidegree $(1,0)$ and $b$, $c$ of pure
bidegree $(0,1)$,
\be\label{eq:split_app}
\Theta_x(u)=\theta^{\circ}_{c}(u)+D_{a,b}(u),
\ee
where $\theta^{\circ}_{c}$ does not involve $a$ or $b$,
\begin{align}
\theta^{\circ}_{c}(u)&=T_{\bl}(\zeta_{\rd};\nu,c)+T_{\bl}(d\mu+\zeta_{\rd};\tau_{\bl},c) \cr
&\quad+T_{\gr}(\zeta_{\rd},d\nu;\rho)+T_{\gr}(d\mu,\zeta_{\bl};\rho) \cr
&\quad+T_{\gr}(\zeta_{\rd},\zeta_{\bl};\rho)
+T_{\gr}(d\mu+\zeta_{\rd},\,d\nu+\zeta_{\bl};\tau_{\gr}) \cr
&\quad+\langle\eta_{\gr},c\rangle+\langle\eta,d^0u\rangle, \label{eq:theta0_app}
\end{align}
and $D_{a,b}$ is a sum of terms each linear in $a$ or in $b$,
\begin{align}
D_{a,b}(u)&=T_{\bl}(a;\tau_{\bl},c)+T_{\gr}(a,\zeta_{\bl};\rho) \cr
&\quad+T_{\bl}(a;\nu+\tau_{\bl},\zeta_{\gr})+\langle\eta_{\rd},a\rangle \cr
&\quad+T_{\gr}(\zeta_{\rd},b;\rho)+T_{\rd}(\mu+\tau_{\rd};b,\zeta_{\gr}) \cr
&\quad+\langle\eta_{\bl},b\rangle+T_{\gr}(a,b;\tau_{\gr}). \label{eq:D_app}
\end{align}
\end{lemma}
\begin{proof}
Put $a_\mu:=a+d\mu$, $b_\nu:=b+d\nu$. From \eqref{eq:theta_def}, by trilinearity,
\begin{align*}
&\theta_x(u)+\theta_x(u+\tau) \cr
&=T_{\rd}(\tau_{\rd};b,c)+T_{\bl}(\zeta_{\rd};\nu,c)+T_{\bl}(a_\mu+\zeta_{\rd};\tau_{\bl},c) \cr
&\quad+T_{\gr}(\zeta_{\rd},b_\nu;\rho)+T_{\gr}(a_\mu,\zeta_{\bl};\rho) \cr
&\quad+T_{\gr}(\zeta_{\rd},\zeta_{\bl};\rho)
+T_{\gr}(a_\mu+\zeta_{\rd},\,b_\nu+\zeta_{\bl};\tau_{\gr}).
\end{align*}
$T_{\rd}(\tau_{\rd};b,c)=0$ by (P2), since $b$ and $c$ both have bidegree $(0,1)$.
Expanding $a_\mu,b_\nu$ and collecting the pieces containing $a$ or $b$ gives
\eqref{eq:theta0_app} and the first, second, fifth and last terms of
\eqref{eq:D_app} together with $T_{\gr}(a,d\nu+\zeta_{\bl};\tau_{\gr})$ and
$T_{\gr}(d\mu+\zeta_{\rd},b;\tau_{\gr})$. Since $d\nu+\zeta_{\bl}=d(\nu+\tau_{\bl})$ and
$da=0$, (P3) gives $T_{\gr}(a,d(\nu+\tau_{\bl});\tau_{\gr})=T_{\bl}(a;\nu+\tau_{\bl},d\tau_{\gr})
=T_{\bl}(a;\nu+\tau_{\bl},\zeta_{\gr})$, and likewise
$T_{\gr}(d(\mu+\tau_{\rd}),b;\tau_{\gr})=T_{\rd}(\mu+\tau_{\rd};b,\zeta_{\gr})$. Finally
$\langle\eta,x+d^0u\rangle=\langle\eta_{\rd},a\rangle+\langle\eta_{\bl},b\rangle
+\langle\eta_{\gr},c\rangle+\langle\eta,d^0u\rangle$.
\end{proof}

The Leibniz step is essential: before it, $T_{\gr}(a,d\nu;\tau_{\gr})$ probes $a$ on
every $x$-edge of the squares in the $\tau_{\gr}$ columns, which is not local in
$x$; after it the same number is expressed through $\zeta_{\gr}$ and is local.

\subsection{Shadows and cleaning}
\label{app:KL_cleaning}

For a set $T$ of $1$-cells of $\L$ let $\pi_x(T)\subseteq E_x$ and
$\pi_y(T)\subseteq E_y$ be the $x$- and $y$-edges of the squares whose closure
meets $T$; $|\pi_x(T)|,|\pi_y(T)|\le\Delta|T|$. For $\Pi\subseteq V_\nu$ let
$N_\nu(\Pi)\subseteq E_\nu$ be the edges incident to $\Pi$; $|N_\nu(\Pi)|\le\Delta|\Pi|$.
Define the \emph{shadows}
\begin{align}
S_x&:=\pi_x(\mathrm{supp}\,\zeta_{\bl})\cup\pi_x(\mathrm{supp}\,\zeta_{\gr}) \cr
&\qquad\cup N_x(\Pi_{\bl})\cup\pi_x(\mathrm{supp}\,\eta_{\rd}), \cr
S_y&:=\pi_y(\mathrm{supp}\,\zeta_{\rd})\cup\pi_y(\mathrm{supp}\,\zeta_{\gr}) \cr
&\qquad\cup N_y(\Pi_{\gr})\cup\pi_y(\mathrm{supp}\,\eta_{\bl}). \label{eq:shadows_app}
\end{align}
By \eqref{eq:columns_app}, $|S_x|,|S_y|\le\kappa|E|$ with $\kappa:=5\Delta$.

\begin{lemma}[Shadow]\label{lemma:shadow_app}
Let $a\in Z^1(\L,\Fc^{\rd})$ have pure bidegree $(1,0)$ and vanish on every
horizontal cell $(e,w)$ with $e\in S_x$, and let $b\in Z^1(\L,\Fc^{\bl})$ have pure
bidegree $(0,1)$ and vanish on every vertical cell $(v,f)$ with $f\in S_y$. Then
$D_{a,b}(u)=0$ for all $u\in C^0$.
\end{lemma}
\begin{proof}
By (P1) each term of \eqref{eq:D_app} is a sum over squares of a trilinear local
form, or a pairing $\langle\eta,\cdot\rangle$, and in each term every square (or
cell) that could contribute has its $a$- or $b$-argument equal to zero:
for $T_{\bl}(a;\tau_{\bl},c)$ and the $\tau_{\bl}$-part of $T_{\bl}(a;\nu+\tau_{\bl},\zeta_{\gr})$,
a square $(e,f)$ contributes only if $\tau_{\bl}$ is non-zero at a corner $(v,w)$,
so $v\in\Pi_{\bl}$, $e\in N_x(\Pi_{\bl})\subseteq S_x$, and $a$ vanishes on the
square ($a$ has no vertical component); for $T_{\gr}(a,\zeta_{\bl};\rho)$,
$T_{\bl}(a;\nu,\zeta_{\gr})$ and $\langle\eta_{\rd},a\rangle$ the contributing squares
or cells have $x$-edges in $\pi_x(\mathrm{supp}\,\zeta_{\bl})$,
$\pi_x(\mathrm{supp}\,\zeta_{\gr})$, $\pi_x(\mathrm{supp}\,\eta_{\rd})\subseteq S_x$;
for $T_{\gr}(\zeta_{\rd},b;\rho)$, $T_{\rd}(\mu+\tau_{\rd};b,\zeta_{\gr})$ and
$\langle\eta_{\bl},b\rangle$ they have $y$-edges in $S_y$, where $b$ vanishes ($b$
has no horizontal component); for $T_{\gr}(a,b;\tau_{\gr})$ a square contributes
only if $\tau_{\gr}$ is non-zero at a corner $(v,w)$, so $w\in\Pi_{\gr}$ and
$f\in N_y(\Pi_{\gr})\subseteq S_y$, where $b$ vanishes.
\end{proof}

\begin{lemma}[Factor cleaning]\label{lemma:clean_app}
Let $C^0\xrightarrow{d}C^1$ be a one-dimensional sheaf complex with minimum
non-zero cycle weight $d_1$, and let $S$ be a set of $1$-cells with $|S|<d_1$.
Then for every $\as\in C^1$ there is $\lambda\in C^0$ with $(\as+d\lambda)|_S=0$.
\end{lemma}
\begin{proof}
If the restricted equation $\pi_S(d\lambda)=\pi_S(\as)$ had no solution,
finite-dimensional duality would produce a functional on $C^1(S)$ annihilating
$\mathrm{im}(\pi_S d)$ with value $1$ on $\pi_S\as$; extended by zero it is a
non-zero chain $h\in C_1$ with $d^Th=0$ and $\mathrm{supp}\,h\subseteq S$, so
$|h|\le|S|<d_1$, contradicting the definition of $d_1$.
\end{proof}

\begin{proposition}[Cleaned representatives]\label{prop:cleaned_app}
Assume $\kappa|E|<\min\big(d_1(G_x,\Fc_x),d_1(G_y,\Fc_y)\big)$. For every admissible
$(A,\mathsf b)$ there are $\lambda_{\rd}\in C^0(\L,\Fc^{\rd})$,
$\lambda_{\bl}\in C^0(\L,\Fc^{\bl})$ such that $a':=a_A+d^0\lambda_{\rd}$ has pure
bidegree $(1,0)$ and vanishes on all $(e,w)$ with $e\in S_x$, and
$b':=b_{\mathsf b}+d^0\lambda_{\bl}$ has pure bidegree $(0,1)$ and vanishes on all
$(v,f)$ with $f\in S_y$.
\end{proposition}
\begin{proof}
By Lemma~\ref{lemma:clean_app} in $C^\bullet(G_x,\Fc_x)$ choose $\lambda^x_{j_1}$
with $(\as^x_{A,j_1}+d_x\lambda^x_{j_1})|_{S_x}=0$ and put
$\lambda_{\rd}:=\sum_{j_1}\lambda^x_{j_1}\otimes\red{\bs^0_{j_1}}$; since
$d_y\red{\bs^0_{j_1}}=0$, $d^0\lambda_{\rd}=\sum_{j_1}d_x\lambda^x_{j_1}\otimes\red{\bs^0_{j_1}}$
has pure bidegree $(1,0)$ and $a'=\sum_{j_1}(\as^x_{A,j_1}+d_x\lambda^x_{j_1})\otimes\red{\bs^0_{j_1}}$
vanishes as required. Similarly choose $\lambda^y\in C^0(G_y,\Fc_y)$ with
$(\as^y_{\mathsf b}+d_y\lambda^y)|_{S_y}=0$ and put $\lambda_{\bl}:=\mathbf 1_x\otimes\lambda^y$;
since $d_x\mathbf 1_x=0$, $b'=\blue{\cs^0}\otimes(\as^y_{\mathsf b}+d_y\lambda^y)$.
\end{proof}

\subsection{Proof of Theorem~\ref{thm:distance_twisted_KL}}
\label{app:KL_proof}

Part~(i) is Definition~\ref{def:V_app}. For part~(ii) let first
$E=X(\zeta)Z(\eta)$ be a Pauli operator with $|E|<c\bar n$, where $c\bar n$ is
defined in Eq.~\eqref{eq:KL_threshold}; then $|E|<d_\star$ and
$\kappa|E|<\min(d_1(G_x,\Fc_x),d_1(G_y,\Fc_y))$. By Corollary~\ref{cor:blockdiag_app},
$V^\dagger P_{\tilde\C}EP_{\tilde\C}V$ is diagonal with entries
$m_E(A,\mathsf b,\green g)$. Fix $\green g$ and an admissible $(A,\mathsf b)$, and let
$x=(a_A,b_{\mathsf b},c_{\green g})$. By Lemma~\ref{lemma:orbit_app}(iii) the
normalized matrix element depends only on the classes, so we may replace $x$ by
$x'=(a',b',c_{\green g})$ with the cleaned representatives of
Proposition~\ref{prop:cleaned_app}: the sign $(-1)^{\theta_x(\lambda)}$ appears once
in the bra and once in the ket and cancels. By Lemmas~\ref{lemma:exact_app},
\ref{lemma:expansion_app} and~\ref{lemma:shadow_app},
\begin{align}
m_E(A,\mathsf b,\green g)
&=\frac1{|C^0|}\sum_{u}(-1)^{\theta^{\circ}_{c_{\green g}}(u)+D_{a',b'}(u)} \cr
&=\frac1{|C^0|}\sum_{u}(-1)^{\theta^{\circ}_{c_{\green g}}(u)},
\end{align}
which does not depend on $(A,\mathsf b)$. Writing $m_E(\green g)$ for the common
value, $V^\dagger P_{\tilde\C}EP_{\tilde\C}V=\I_{\rd\bl}\otimes\sum_{\green g}m_E(\green g)\ket{\green g}\!\bra{\green g}$,
which is Eq.~\eqref{eq:KL_twisted} for Paulis. A general operator supported on a
set $R$ of fewer than $c\bar n$ qubits is a linear combination of Pauli operators
supported in $R$, so \eqref{eq:KL_twisted} follows by linearity; applied to the
products $E^\dagger_\alpha E_\beta$ this is the subsystem Knill--Laflamme
condition.

For part~(iii), the lower bound $\tilde d\ge c\bar n$ is \eqref{eq:KL_twisted}:
no operator on fewer than $c\bar n$ qubits acts non-trivially on $\H_{\rd\bl}$.
For the upper bound take, as in the proof of Lemma~\ref{lemma:distance_untwisted}
(Appendix~\ref{app:proof_untwisted}), $\eta=\eta^{i}_1\otimes\pi_0$ with
$\eta^{i}_1=\text{PD}[\red{\bs^0_i}]\in H_1(G_x,\Fc_x)$ and $\pi_0$ a single-vertex
$0$-chain of $C_0(G_y,\Fc^\perp_y)$ with $v(\pi_0)=(\langle\pi_0,\red{\bs^0_{j_1}}\rangle)_{j_1}\ne0$.
The Wilson operator $\widetilde Z^{\rd}_\eta$ of Eq.~\eqref{eq:logical_Z_sheaf}
commutes with every generator of $\tilde{\SS}$ (it is diagonal, and
$\langle\eta,d^0u\rangle=\langle\partial_1\eta,u\rangle=0$), and acts on
$\ket{\psi_x}$ as $(-1)^{\langle\eta,a_A\rangle}=(-1)^{A_{i\cdot}\cdot v(\pi_0)}$, a
non-constant function of the admissible label $A$; its weight is at most
$\Delta|\eta^i_1|=O(\sqrt n)$. Hence $\tilde d=\Theta(\sqrt n)$.

For part~(iv), $\I_{\rd\bl}\otimes M_E$ preserves the subspace
$\H'_{\rd\bl}\otimes\H_{\gr}$ with
$\H'_{\rd\bl}=\mathrm{span}\{\ket{A,\mathsf b}:\mathsf b\in U,\ A_{i_1\cdot}\in U^\perp\}$,
on which $A\mathsf b=0$ holds automatically and which is a tensor product of
$p(q-r)$ \rd\ and $r$ \bl\ qubits; compressing \eqref{eq:KL_twisted} with the
restricted $V'$ gives $P'EP'=V'(\I\otimes M_E)V'^\dagger$, and the upper bound is
obtained by choosing $\pi_0$ with $v(\pi_0)\notin U$, which exists because the
vectors $v(\pi_0)$ span $\FF_2^{q}$ and $\dim U<q$. \hfill$\square$

\begin{remark}[Where the hypotheses enter]\label{rem:KL_hypotheses}
Eq.~\eqref{eq:subdiv_params} (one-sided goodness) enters only through the
cleaning step, Lemma~\ref{lemma:clean_app}. Eq.~\eqref{eq:dual_cocycle_weight}
enters through Proposition~\ref{prop:heavy_app} and
Lemma~\ref{lemma:column_primitive}. Full row rank enters three times: it removes
the \rd\ and \bl\ zero modes, so that admissibility is exactly
\eqref{eq:charge_constraint_Ab} and the fibre over every $(A,\mathsf b)$ is all of
$\H_{\gr}$; it removes the \bl\ graph-short sector, so that \bl\ labels cannot
change; and it makes Lemma~\ref{lemma:column_primitive} for the \bl\ copy trivial.
The light \gr\ Wilson operators of Eq.~\eqref{eq:spurious-cycle} act on
$\ket{\psi_x}$ as $(-1)^{\langle\eta,c_{\green g}\rangle}$, i.e.\ as a non-trivial
$M_E$ on the gauge register, which is exactly what Eq.~\eqref{eq:KL_twisted}
permits; since no low-weight error changes $\green g$, one may alternatively fix
$\green g=\green g_0$, obtaining a \emph{subspace} code
$V(\H_{\rd\bl}\otimes\ket{\green g_0})$ of distance $\ge c\bar n$ with
$\log_2\dim=\Theta(n)$.
\end{remark}

%======================================================================
\section{Subsystem Knill--Laflamme distance of the high-magic-rate twisted code: proof of
Theorem~\ref{thm:GTZ_KL}}
\label{app:GTZ_KL}
%======================================================================

Throughout this appendix we are in the setting of Construction~\ref{con:high_yield}:
$\L=\Gamma\times\Gamma$ with $\Gamma$ the GTZ graph of
Theorem~\ref{thm:GTZ_classical}, of maximum degree $\Delta=O_\epsilon(1)$, and the
three sheaves \eqref{eq:high_yield_sheaves} over $\Fq$. Phases are written with
the additive character $\chi(x):=\omega^{\mathrm{Tr}(x)}$ of
Sec.~\ref{sec:Fq_generalization}, and we use repeatedly that $\chi$ is non-trivial
and that the trace form is non-degenerate, Eq.~\eqref{eq:trace_delta}. We use the
weight conventions of Appendix~\ref{app:subsystem_distance}: $|\chi|$ is the
number of $1$-cells of $\L$ on which a (co)chain $\chi$ is non-zero, which is a
lower bound on the number of qudits in the support of the corresponding Pauli
operator; $d_1(Y)$ and $d^0(Y)$ are the minimum non-zero cycle and $0$-cocycle
weights of a one-dimensional sheaf complex $Y$. For the graph $\Gamma$,
Lemma~\ref{lemma:dictionary} and Eq.~\eqref{eq:F_diamond_cycles} give
\begin{align}\label{eq:GTZ_app_distances}
d_1(\Gamma,\Fc^{\diamond})&=d\big(\T^{\dag}(\Gamma,\Fc^{*2})\big)\ \ge\ N^{1-\epsilon}, \cr
d^0(\Gamma,\Fc)&\ \ge\ \tfrac{2}{\Delta}\,d\big(\T^{\dag}(\Gamma,\Fc)\big)\ \ge\
\tfrac{2}{\Delta}\,N^{1-\epsilon}.
\end{align}
The K\"unneth contraction lemma, Lemma~\ref{lemma:kunneth_contraction}, and its
Corollary~\ref{cor:sector_distance_app} were stated over $\FF_2$; their proofs use
only linearity, $\partial=d^{T}$ and the non-degeneracy of
Proposition~\ref{prop:nondegen-sheaf}, so they hold verbatim over $\Fq$, the
graded signs of Eq.~\eqref{eq:leibniz_signed} affecting only the sign of the
contraction. Colours and degree superscripts are dropped when no confusion can
arise: a flat configuration is $x=(a,b,c)$ with
$a\in Z^1(\L,\Fc^{\rd})$, $b\in Z^1(\L,\Fc^{\bl})$, $c\in Z^1(\L,\Fc^{\gr})$, a gauge
parameter is $u=(\mu,\nu,\rho)\in C^0:=C^0(\L,\Fc^{\rd})\oplus C^0(\L,\Fc^{\bl})\oplus C^0(\L,\Fc^{\gr})$,
and $Z^0:=\ker d^0\subset C^0$. Horizontal $1$-cells are $(e,w)$, $e\in E(\Gamma)$,
$w\in V(\Gamma)$, of bidegree $(1,0)$; vertical ones are $(v,f)$, of bidegree
$(0,1)$.

\subsection{K\"unneth sectors and selected classes}
\label{app:GTZ_kunneth}

By Eq.~\eqref{eq:kunneth_H1} the three copies have
\begin{align}
H^1(\L,\Fc^{\rd})&=\big[H^1(\Gamma,\Fc^{\diamond})\otimes H^0(\Gamma,\Fc)\big] \cr
&\quad\oplus\big[H^0(\Gamma,\Fc^{\diamond})\otimes H^1(\Gamma,\Fc)\big], \cr
H^1(\L,\Fc^{\bl})&=\big[H^1(\Gamma,\Fc)\otimes H^0(\Gamma,\Fc^{\diamond})\big] \cr
&\quad\oplus\big[H^0(\Gamma,\Fc)\otimes H^1(\Gamma,\Fc^{\diamond})\big], \cr
H^1(\L,\Fc^{\gr})&=\big[H^1(\Gamma,\Fc)\otimes H^0(\Gamma,\Fc)\big] \cr
&\quad\oplus\big[H^0(\Gamma,\Fc)\otimes H^1(\Gamma,\Fc)\big], \label{eq:kunneth_GTZ_app}
\end{align}
with homological counterparts obtained by lowering indices, and the zero modes
are $Z^0(\L,\Fc^{\rd})=H^0(\Gamma,\Fc^{\diamond})\otimes H^0(\Gamma,\Fc)$,
$Z^0(\L,\Fc^{\bl})=H^0(\Gamma,\Fc)\otimes H^0(\Gamma,\Fc^{\diamond})$ and
$Z^0(\L,\Fc^{\gr})=H^0(\Gamma,\Fc)\otimes H^0(\Gamma,\Fc)\ni\green{\gamma^0_{(e,f)}}$.
The selected classes $[\red{\alpha^1_{(e,f)}}]$ lie in the first \rd\ summand and
the $[\blue{\beta^1_{(e,f)}}]$ in the second \bl\ summand; the protected sectors
$\mathsf L_X^{i}$, $\mathsf L_Z^{i}$ of Eq.~\eqref{eq:protected_sectors_GTZ} are
proper subspaces of these summands, and everything else is gauge. For label
vectors $\vc a,\vc b\in\Fq^{J}$ we use the \emph{reference representatives}
\begin{align}\label{eq:refreps_GTZ}
a_{\vc a}&:=\sum_{j\in J}a_j\,\red{\alpha^1_{j}}
=\sum_{(e,f)\in J}a_{(e,f)}\,\bar e\otimes\red{\bs^0_f}, \cr
b_{\vc b}&:=\sum_{j\in J}b_j\,\blue{\beta^1_{j}}
=\sum_{(e,f)\in J}b_{(e,f)}\,\blue{\bs^0_e}\otimes\bar f,
\end{align}
of pure bidegree $(1,0)$ and $(0,1)$ respectively. By
Eq.~\eqref{eq:cochain_matching},
\begin{align}\label{eq:ab_cup_app}
a_{\vc a}\cup b_{\vc b}&=\sum_{j\in J}a_jb_j\,\varpi^2_{j}, \cr
\int_{\eta_2}a_{\vc a}\cup b_{\vc b}\cup\green{\gamma^0}&=\sum_{j\in J}a_jb_j\,\ell_j(\green{\gamma^0}),
\end{align}
for every $\green{\gamma^0}\in Z^0(\L,\Fc^{\gr})$.

\begin{lemma}[Conjugate cycles]\label{lemma:conjugate_app}
For $e\in M$ let $\bs^{*2}_{e}:=\bs_e*\bs_e\in\Fq^{E(\Gamma)}$ (Schur square), and for
$f\in M$ let $\hs_f\in C_0(\Gamma,\Fc)$ satisfy
$\langle\bs^0_{f'},\hs_f\rangle=\delta_{ff'}$ for all $f'\in M$. Then
$\bs^{*2}_{e}\in\T^{\dag}(\Gamma,\Fc^{*2})=H_1(\Gamma,\Fc^{\diamond})$, $\bs^{*2}_{e}|_M=\delta_e$,
$\bs^{*2}_{e}\otimes\hs_f\in Z_1(\L,\Fc^{\rd})$ and $\hs_e\otimes\bs^{*2}_{f}\in Z_1(\L,\Fc^{\bl})$,
and
\be\label{eq:conjugate_pairing_app}
\langle\red{\alpha^1_{(e',f')}},\bs^{*2}_{e}\otimes\hs_f\rangle
=\langle\blue{\beta^1_{(e',f')}},\hs_e\otimes\bs^{*2}_{f}\rangle=\delta_{ee'}\delta_{ff'}.
\ee
Hence the pairing is non-degenerate on $\mathsf L_X^{i}\times\mathsf L_Z^{i}$
for $i=\rd,\bl$.
\end{lemma}
\begin{proof}
$\bs^{*2}_{e}|_{E(v)}=\bs_e|_{E(v)}*\bs_e|_{E(v)}\in\cC^{\perp}_v*\cC^{\perp}_v$, so $\bs^{*2}_{e}$
lies in the Tanner code of the local codes $\im\Fc^{*2}_{v\shortrightarrow E(v)}$, which
is $H_1(\Gamma,\Fc^{\diamond})$ by Eq.~\eqref{eq:dictionary} applied to $\Fc^{*2}$;
$\bs^{*2}_{e}(e')=\bs_e(e')^2=\delta_{ee'}$ for $e'\in M$. The $\hs_f$ exist because the
$\bs^0_f$, $f\in M$, are linearly independent in $C^0(\Gamma,\Fc)$ ($\Phi$ is
injective and $\bs_f|_M=\delta_f$). A $1$-cycle tensored with a $0$-chain is a
$1$-cycle, and $\langle\bar e'\otimes\bs^0_{f'},\bs^{*2}_{e}\otimes\hs_f\rangle
=\bs^{*2}_{e}(e')\langle\bs^0_{f'},\hs_f\rangle$.
\end{proof}

\begin{proposition}[Selected relative distances]\label{prop:reldist_app}
Let $\eta\in Z_1(\L,\Fc^{\rd})$ with $\langle\mathsf L_X^{\rd},\eta\rangle\neq0$.
Then $|\eta|\ge d_1(\Gamma,\Fc^{\diamond})$. Let $\zeta\in Z^1(\L,\Fc^{\rd})$ with
$0\ne[\zeta]\in\mathsf L^{\rd}_X$, or more generally with
$\langle\zeta,\mathsf L^{\rd}_Z\rangle\ne0$. Then $|\zeta|\ge d^0(\Gamma,\Fc)$. The
same bounds hold for the \bl\ copy.
\end{proposition}
\begin{proof}
If $\langle\mathsf L_X^{\rd},\eta\rangle\ne0$ then $\langle\bar e\otimes\bs^0_f,\eta\rangle\ne0$
for some $(e,f)\in J$. Contract against the $0$-cocycle $\bs^0_f$ of the
$y$-factor, Lemma~\ref{lemma:kunneth_contraction}(b): $\iota^{\bs^0_f}\eta$ is a
$1$-cycle of $C^\bullet(\Gamma,\Fc^{\diamond})$ with
$\langle\bar e,\iota^{\bs^0_f}\eta\rangle=\langle\bar e\otimes\bs^0_f,\eta\rangle\ne0$,
hence non-zero, and $|\eta|\ge|\iota^{\bs^0_f}\eta|\ge d_1(\Gamma,\Fc^{\diamond})$.
If $\langle\zeta,\mathsf L^{\rd}_Z\rangle\ne0$ then
$\langle\zeta,\bs^{*2}_{e}\otimes\hs_f\rangle\ne0$ for some $(e,f)$; by
Lemma~\ref{lemma:conjugate_app} this is the case whenever $0\neq[\zeta]\in\mathsf L^{\rd}_X$.
Contract against the $1$-cycle $\bs^{*2}_{e}$ of the $x$-factor,
Lemma~\ref{lemma:kunneth_contraction}(a): $\iota_{\bs^{*2}_{e}}\zeta$ is a $0$-cocycle
of $C^\bullet(\Gamma,\Fc)$ with
$\langle\iota_{\bs^{*2}_{e}}\zeta,\hs_f\rangle=\langle\zeta,\bs^{*2}_{e}\otimes\hs_f\rangle\ne0$,
hence non-zero, and $|\zeta|\ge|\iota_{\bs^{*2}_{e}}\zeta|\ge d^0(\Gamma,\Fc)$. For \bl\
interchange the two factors.
\end{proof}

With Eq.~\eqref{eq:GTZ_app_distances} this proves the distance statement of
Theorem~\ref{thm:high_yield_parameters} for the untwisted code, by
Lemma~\ref{lemma:subsystem}(ii).

\subsection{Why the degree-one slot must carry $\Fc^{\diamond}$}
\label{app:GTZ_obstruction}

We now justify Remark~\ref{rem:slot_choice}. We need the structure of the GTZ
graph, which we quote from Ref.~\cite[Theorem 4.1 and Claims 4.2--4.3]{GolowichTamoZhu2026}.
Let $C\subseteq\Fq^{2n_0}$ be the base code, of dimension $k_0$ and distance
$d(C)\ge3n_0/2$, and let $t\ge2$. The vertices of $\Gamma=\Gamma(n_0,t)$ are the
axis-parallel lines of the hypercube $[2n_0]^{t}$ in the families
$V_i=[n_0{+}1{:}2n_0]^{[1:i-1]}\times[1{:}n_0]^{[i+1:t]}$, $i\in[t]$ (direction-$i$
lines), and its edges are the points of
$\mathsf E=\bigsqcup_{j=1}^{t-1}\mathsf E_j$,
$\mathsf E_j=[n_0{+}1{:}2n_0]^{[1:j]}\times[1{:}n_0]^{[j+1:t]}$, a point of
$\mathsf E_j$ being the edge between the direction-$j$ and the direction-$(j{+}1)$
line through it and lying on no other vertex of $\Gamma$. The local code at a
line $v\in V_i$ is the puncturing $\im\Fc_{v\shortrightarrow E(v)}=\cC^{\perp}_v=C|_{P_v}$
of $C$ to the set $P_v\subseteq[2n_0]$ of coordinates of $v$ that are edges,
which is $[2n_0]$ for $2\le i\le t-1$, $[n_0{+}1{:}2n_0]$ for $i=1$ and $[1{:}n_0]$
for $i=t$; $N=|\mathsf E|=(t-1)n_0^{t}$ and $\Delta=2n_0$. For $P\subseteq[2n_0]$ let
$\Short(C^{\perp};P):=\{y\in\Fq^{P}:\text{the zero-extension of $y$ lies in }C^{\perp}\}$,
so that $(C|_P)^{\perp}=\Short(C^\perp;P)$, and put
$S_{+}:=\Short(C^\perp;[n_0{+}1{:}2n_0])$, $S_{-}:=\Short(C^\perp;[1{:}n_0])$, both of
dimension $\ge(2n_0-k_0)-n_0>0$.

\begin{lemma}[Full support of the half-shortened duals]\label{lemma:fullsupport_app}
If $n_0>2$, every coordinate of $[n_0{+}1{:}2n_0]$ is in the support of some word of
$S_+$, and every coordinate of $[1{:}n_0]$ is in the support of some word of $S_-$.
\end{lemma}
\begin{proof}
Let $P$ be either half and suppose the coordinate $i\in P$ vanishes on every word
of $(C|_P)^{\perp}$. Then the unit vector $e_i$ lies in $C|_P$, so some $c\in C$
has $c|_P=e_i$ and weight $\le1+n_0<3n_0/2\le d(C)$, a contradiction.
\end{proof}

\begin{proposition}[Light cycles of $\Fc$]\label{prop:short_dual_cycles_app}
Fix $j\in[t-1]$, values of the $t-2$ coordinates other than the $j$-th and
$(j{+}1)$-st compatible with $\mathsf E_j$, and $y\in S_+$, $y'\in S_-$. Let
$z\in\Fq^{E(\Gamma)}$ be supported on the points of $\mathsf E_j$ with these fixed
coordinates, with value $z(e)=y(e_j)\,y'(e_{j+1})$. Then
\be
z\in H_1(\Gamma,\Fc)=\T^{\dag}(\Gamma,\Fc^{\perp}),\qquad |z|=|y|\,|y'|\le n_0^2 ,
\ee
so $d_1(\Gamma,\Fc)\le n_0^2=O_\epsilon(1)$ for every $t$, and for every edge
$e\in\mathsf E_j$ there are $y,y'$ with $z(e)\neq0$.
\end{proposition}
\begin{proof}
By Eq.~\eqref{eq:dictionary} (with $\Fc$ replaced by $\Fc^\perp$) we must check
$z|_{E(v)}\in(\cC^{\perp}_v)^{\perp}=\Short(C^\perp;P_v)$ at every vertex. A point
of $\mathsf E_j$ lies only on vertices of $V_j$ and $V_{j+1}$, so $z|_{E(v)}=0$
unless $v\in V_j\sqcup V_{j+1}$. For a direction-$j$ line $v\in V_j$, the
$2n_0$ points of $v$ have $j$-th coordinate in $[1{:}n_0]$ (points of
$\mathsf E_{j-1}$ or outside $\mathsf E$, where $z=0$) or in $[n_0{+}1{:}2n_0]$
(points of $\mathsf E_j$), so the zero-extension of $z|_{E(v)}$ to the line is
either $0$ or $y'(e_{j+1})\cdot(0,y)$, which lies in $C^\perp$ because $y\in S_+$.
For $v\in V_{j+1}$ the zero-extension is a multiple of $(y',0)\in C^\perp$. The
last statement is Lemma~\ref{lemma:fullsupport_app}.
\end{proof}

\begin{corollary}[An $O_\epsilon(1)$-weight logical operator for $\Fc\boxtimes\Fc$]
\label{cor:short_logical_app}
Suppose the \rd\ sheaf were $\Fc\boxtimes\Fc$, with selected classes
$\bar e\otimes\bs^0_f$, $(e,f)\in J$. For $z$ as in
Proposition~\ref{prop:short_dual_cycles_app} with $z(e)\ne0$ and a single-vertex
$0$-chain $\hs$ with $\langle\bs^0_f,\hs\rangle\ne0$, the Wilson operator
$\widetilde Z^{\rd}_{z\otimes\hs}$ of Eq.~\eqref{eq:logical_Z_Fq} has weight
$|z|\le n_0^2$, commutes with all generators of $\tilde{\SS}$ (and of the untwisted
stabilizer group), and acts on every code state with \rd\ class $[a]$ as the
scalar $\chi(\langle a,z\otimes\hs\rangle)$, with
$\langle\bar e\otimes\bs^0_f,z\otimes\hs\rangle=z(e)\langle\bs^0_f,\hs\rangle\ne0$. It is
therefore an undetectable logical operator of weight $O_\epsilon(1)$ acting
non-trivially on every selected logical qudit, and the selected-subsystem
distance of that assignment is $O_\epsilon(1)$.
\end{corollary}
\begin{proof}
$z\otimes\hs$ is a $1$-cycle since $\partial(z\otimes\hs)=\partial z\otimes\hs
-z\otimes\partial\hs=0$. A \rd\ $Z$-type operator commutes with the flux
operators, with the diagonal dressings $U^{i}$, with the \bl\ and \gr\ $X$-parts,
and with $A^{\rd}(\lambda^0)$ iff $\langle d^0\lambda^0,z\otimes\hs\rangle
=\langle\lambda^0,\partial(z\otimes\hs)\rangle=0$. On a flat configuration
$\red{\hat a^1}=a+d^0\mu$ it acts as $\chi(\langle a+d^0\mu,z\otimes\hs\rangle)
=\chi(\langle a,z\otimes\hs\rangle)$.
\end{proof}

For the assignment \eqref{eq:high_yield_sheaves} the same cycles $z\otimes\hs$
exist in the \emph{gauge} sector $H_0(\Gamma,\Fc^{\diamond})\otimes H_1(\Gamma,\Fc)$
of the \rd\ copy, where they pair trivially with $\mathsf L_X^{\rd}$; the
selected conjugate sector is governed by $d_1(\Gamma,\Fc^{\diamond})$ instead,
Proposition~\ref{prop:reldist_app}.

\subsection{Dressed gauge transformations and outcome sectors over $\Fq$}
\label{app:GTZ_orbit}

For a flat configuration $x=(a,b,c)$ and $u=(\mu,\nu,\rho)\in C^0$ define, as in
Eq.~\eqref{eq:T_forms}, the $\Fq$-valued trilinear forms
$T_{\rd}(\mu;b,c)=\int_{\eta_2}\mu\cup b\cup c$,
$T_{\bl}(a;\nu,c)=\int_{\eta_2}a\cup\nu\cup c$,
$T_{\gr}(a,b;\rho)=\int_{\eta_2}a\cup b\cup\rho$, and
\begin{align}\label{eq:theta_GTZ_app}
\theta_x(u):=T_{\rd}(\mu;b,c)&+T_{\bl}(a+d\mu;\nu,c) \cr
&+T_{\gr}(a+d\mu,b+d\nu;\rho).
\end{align}
They have the properties (P1) \emph{locality} and (P2) \emph{bigrading} of
Appendix~\ref{app:KL_orbit} unchanged, while (P3) acquires the graded signs of
Eq.~\eqref{eq:leibniz_signed}: for cochains $\alpha,\beta,\gamma$ of total degree $1$,
\begin{align}\label{eq:leibniz_cycle_Fq}
\int_{\eta_2}\!d\alpha\cup\beta\cup\gamma&+(-1)^{|\alpha|}\!\int_{\eta_2}\!\alpha\cup d\beta\cup\gamma \cr
&+(-1)^{|\alpha|+|\beta|}\!\int_{\eta_2}\!\alpha\cup\beta\cup d\gamma=0,
\end{align}
because $\partial\eta_2=0$ (Lemma~\ref{lemma:admissible}). Let
$\tilde A(u):=\tilde A^{\gr}(\rho)\tilde A^{\bl}(\nu)\tilde A^{\rd}(\mu)$ be the
product of the twisted generators \eqref{eq:twisted_generators_Fq}. Applying the
three factors in turn to a flat computational-basis state,
\be\label{eq:A_action_GTZ_app}
\tilde A(u)\ket{x}=\chi\big(\theta_x(u)\big)\ket{x+d^0u},
\ee
and exactly as in Lemma~\ref{lemma:grouplaw_app}, using
Lemma~\ref{lemma:commutation_Fq} in place of Lemma~\ref{lemma:stabilizer_commutation},
$\tilde A(u)\tilde A(u')=\tilde A(u+u')$ on flat configurations, i.e.
\be\label{eq:grouplaw_GTZ_app}
\theta_x(u+u')=\theta_x(u)+\theta_{x+d^0u}(u')\qquad\forall\,u,u'\in C^0 .
\ee
Following Appendix~\ref{app:KL_orbit} we set
$\mathsf Q_x(\varsigma):=\theta_x(\varsigma)$ for $\varsigma\in Z^0$; the proof of
Lemma~\ref{lemma:orbit_app}(i) gives $\theta_{x+d^0u}(\varsigma)=\mathsf Q_x(\varsigma)$
for all $u$, so $\mathsf Q_x$ depends only on the classes of $x$.

We restrict to the configurations of the gauge-fixed code. For
$(\vc a,\vc b)\in\Fq^J\times\Fq^J$ let $x_{\vc a,\vc b}:=(a_{\vc a},b_{\vc b},0)$ and
let $\mathcal O_{\vc a,\vc b}:=\{x_{\vc a,\vc b}+d^0u:u\in C^0\}$ be its gauge
orbit, the set of flat configurations with \rd\ class in $\mathsf L^{\rd}_X$ with
label $\vc a$, \bl\ class with label $\vc b$, and trivial \gr\ class; distinct
$(\vc a,\vc b)$ give disjoint orbits by Lemma~\ref{lemma:matching}. For
$\varsigma=(\varsigma_{\rd},\varsigma_{\bl},\varsigma_{\gr})\in Z^0$ all
coboundaries in \eqref{eq:theta_GTZ_app} vanish and $c=0$, so by
Eq.~\eqref{eq:ab_cup_app}
\be\label{eq:charge_GTZ_app}
\mathsf Q_{x_{\vc a,\vc b}}(\varsigma)=T_{\gr}(a_{\vc a},b_{\vc b};\varsigma_{\gr})
=\sum_{j\in J}a_jb_j\,\ell_j(\varsigma_{\gr}) .
\ee
For a \gr\ $0$-chain $\sigma\in C_0(\L,\Fc^{\gr})$ write
$\langle u,\sigma\rangle:=\langle\rho,\sigma\rangle$ for $u=(\mu,\nu,\rho)$, and define
the (unnormalized) \emph{sector orbit state}
\be\label{eq:sector_state_app}
\ket{\psi^{\sigma}_{x}}:=\sum_{u\in C^0}\chi\big(\theta_x(u)-\langle u,\sigma\rangle\big)\ket{x+d^0u} .
\ee

\begin{lemma}\label{lemma:sector_orbit_app}
Let $x=x_{\vc a,\vc b}$ and $\sigma\in C_0(\L,\Fc^{\gr})$.
\begin{enumerate}[label=(\roman*)]
\item $\tilde A(u')\ket{\psi^\sigma_x}=\chi(\langle u',\sigma\rangle)\ket{\psi^\sigma_x}$
for every $u'\in C^0$; in particular $\ket{\psi^\sigma_x}$ satisfies conditions
(d) and (e) of $\tilde{\C}_\sigma$, and conditions (a)--(c) hold because every
configuration in $\mathcal O_{\vc a,\vc b}$ is flat with the prescribed classes.
\item $\langle\psi^\sigma_x|\psi^\sigma_x\rangle=|C^0||Z^0|$ if
$\mathsf Q_x(\varsigma)=\langle\varsigma,\sigma\rangle$ for all $\varsigma\in Z^0$
(we then call $x$ \emph{$\sigma$-admissible}), and $0$ otherwise.
\item $\ket{\psi^\sigma_{x+d^0\lambda}}=\chi\big(\langle\lambda,\sigma\rangle-\theta_x(\lambda)\big)\ket{\psi^\sigma_x}$
for every $\lambda\in C^0$.
\item Within the span of $\{\ket{y}:y\in\mathcal O_{\vc a,\vc b}\}$, the joint
eigenspace defined by (d)--(e) is spanned by $\ket{\psi^\sigma_x}$ if $x$ is
$\sigma$-admissible and is zero otherwise.
\end{enumerate}
\end{lemma}
\begin{proof}
(i) By \eqref{eq:A_action_GTZ_app} and \eqref{eq:grouplaw_GTZ_app},
$\tilde A(u')\ket{\psi^\sigma_x}=\sum_u\chi(\theta_x(u+u')-\langle u,\sigma\rangle)\ket{x+d^0(u+u')}$,
and the substitution $u\mapsto u-u'$ produces the factor $\chi(\langle u',\sigma\rangle)$.
(ii) Group $\sum_{u,u'}$ by $d^0u'=d^0u$, i.e.\ $u'=u+\varsigma$ with $\varsigma\in Z^0$,
and use $\theta_x(u+\varsigma)=\theta_x(u)+\mathsf Q_x(\varsigma)$; the inner sum is
$\sum_{\varsigma\in Z^0}\chi(\mathsf Q_x(\varsigma)-\langle\varsigma,\sigma\rangle)$,
which is $|Z^0|$ if the $\Fp$-character $\varsigma\mapsto\chi(\mathsf Q_x(\varsigma)-\langle\varsigma,\sigma\rangle)$
of $Z^0$ is trivial and $0$ otherwise; since both $\mathsf Q_x$ and
$\langle\cdot,\sigma\rangle$ are $\Fq$-linear on $Z^0$, replacing $\varsigma$ by
$\alpha\varsigma$ and using Eq.~\eqref{eq:trace_delta} shows that triviality is the
$\Fq$-identity stated. (iii) Substitute $u\mapsto u-\lambda$ and use
\eqref{eq:grouplaw_GTZ_app}. (iv) For $\varsigma\in Z^0$, $\tilde A(\varsigma)$ is
diagonal and acts on the span of the orbit as the scalar $\chi(\mathsf Q_x(\varsigma))$,
by \eqref{eq:A_action_GTZ_app} and $\theta_{x+d^0u}(\varsigma)=\mathsf Q_x(\varsigma)$;
condition (e) requires the eigenvalue $\chi(\langle\varsigma,\sigma\rangle)$, so a
non-zero joint eigenvector exists only if $x$ is $\sigma$-admissible. If it is, the
rescaled vectors $\chi(\theta_x(u)-\langle u,\sigma\rangle)\ket{x+d^0u}$ depend only
on the coset $u+Z^0$ (by (ii)), and by \eqref{eq:grouplaw_GTZ_app} the operators
$\chi(-\langle u',\sigma\rangle)\tilde A(u')$ permute them regularly as $C^0/Z^0$; the
joint eigenspace is the invariant subspace of this regular representation, which
is one-dimensional and spanned by $\ket{\psi^\sigma_x}$.
\end{proof}

\begin{proposition}[Sector basis]\label{prop:sector_basis_app}
Let $\sigma\in C_0(\L,\Fc^{\gr})$. If $\sigma$ is not admissible in the sense of
Eq.~\eqref{eq:admissible_record} then $\tilde{\C}_\sigma=0$. If $\sigma$ is
admissible with $\vc w=\vc w(\sigma)$, then
$\big\{\ket{\psi^\sigma_{\vc a,\vc b}}:=(|C^0||Z^0|)^{-1/2}\ket{\psi^\sigma_{x_{\vc a,\vc b}}}
:(\vc a,\vc b)\in A_{\vc w}\big\}$ is an orthonormal basis of $\tilde{\C}_\sigma$, on
which $\widetilde Z^{\rd}_{\bs^{*2}_{e}\otimes\hs_f}(\beta)$ and
$\widetilde Z^{\bl}_{\hs_e\otimes\bs^{*2}_{f}}(\beta)$ act as $\chi(\beta a_{(e,f)})$ and
$\chi(\beta b_{(e,f)})$. This proves Theorem~\ref{thm:GTZ_KL}(i).
\end{proposition}
\begin{proof}
Conditions (a)--(c) are diagonal in the computational basis and restrict the
support of any $\ket\Psi\in\tilde{\C}_\sigma$ to flat configurations with trivial
\gr\ class (by (b) and Proposition~\ref{prop:nondegen-sheaf}) and with \rd\ and
\bl\ classes in $\mathsf L^{\rd}_X$ and $\mathsf L^{\bl}_X$ (by (c), since
$(\mathsf G^{i}_Z)^{\perp}=\mathsf L^{i}_X$ by Lemma~\ref{lemma:structure_app}),
i.e.\ to $\bigsqcup_{\vc a,\vc b}\mathcal O_{\vc a,\vc b}$. Each $\tilde A(u)$
preserves the span of each orbit, so $\tilde{\C}_\sigma$ is the direct sum of its
intersections with these spans, which Lemma~\ref{lemma:sector_orbit_app}(iv)
identifies. By Eq.~\eqref{eq:charge_GTZ_app}, $x_{\vc a,\vc b}$ is
$\sigma$-admissible iff $\sum_ja_jb_j\ell_j(\varsigma_{\gr})=\langle\varsigma_{\gr},\sigma\rangle$
for all $\varsigma_{\gr}\in Z^0(\L,\Fc^{\gr})$; if this holds for some
$(\vc a,\vc b)$ then $\sigma$ is admissible with $w_j=a_jb_j$, and if $\sigma$ is
admissible with record $\vc w$ it holds iff $\sum_j(a_jb_j-w_j)\ell_j\equiv0$, i.e.\
iff $(\vc a,\vc b)\in A_{\vc w}$ by surjectivity of $\ell$. Orthogonality is the
disjointness of the orbits, and the electric operators are diagonal with the
orbit-constant values $\chi(\beta\langle a_{\vc a},\bs^{*2}_{e}\otimes\hs_f\rangle)=\chi(\beta a_{(e,f)})$
by Eq.~\eqref{eq:conjugate_pairing_app}. The count \eqref{eq:sector_dimension} is
the number of solutions of $ab=w_j$ in $\Fq^2$ for each $j$.
\end{proof}

\begin{lemma}[The ideal protocol]\label{lemma:protocol_app}
Let $\ket{\Psi^0_{\vc a,\vc b}}$ be the normalized uniform superposition over
$(a_{\vc a}+B^1(\L,\Fc^{\rd}))\times(b_{\vc b}+B^1(\L,\Fc^{\bl}))$, i.e.\ the state of
the untwisted \rd\ and \bl\ codes with selected labels $(\vc a,\vc b)$ and all
gauge labels fixed to $\lo{\ket0}$, and let
$\ket\Phi:=q^{-s}\sum_{\vc a,\vc b}\ket{\Psi^0_{\vc a,\vc b}}\otimes\ket{0}_{\gr}$ be the
initial state of Theorem~\ref{thm:high_yield_fountain}. Measuring the \gr\
twisted generators realizes the projectors
$P^{\gr}_\sigma:=|C^0(\L,\Fc^{\gr})|^{-1}\sum_{\lambda^0}\chi(-\langle\lambda^0,\sigma\rangle)\tilde A^{\gr}(\lambda^0)$,
and
\begin{enumerate}[label=(\roman*)]
\item $P^{\gr}_\sigma\big(\ket{\Psi^0_{\vc a,\vc b}}\otimes\ket0_{\gr}\big)$ is
$\big(|Z^0(\L,\Fc^{\gr})|/|C^0(\L,\Fc^{\gr})|\big)^{1/2}\ket{\psi^\sigma_{\vc a,\vc b}}$
if $\sigma$ is admissible and $(\vc a,\vc b)\in A_{\vc w(\sigma)}$, and $0$
otherwise;
\item hence $\Pr[\sigma]=\|P^{\gr}_\sigma\Phi\|^2=q^{-2s}|A_{\vc w(\sigma)}|\,|Z^0(\L,\Fc^{\gr})|/|C^0(\L,\Fc^{\gr})|$
for admissible $\sigma$, the post-measurement state is
$|A_{\vc w}|^{-1/2}\sum_{(\vc a,\vc b)\in A_{\vc w}}\ket{\psi^\sigma_{\vc a,\vc b}}\in\tilde{\C}_\sigma$,
which in the logical basis of Proposition~\ref{prop:sector_basis_app} is the
product state \eqref{eq:high_yield_output}; and the logical outcome has the
product distribution $\Pr[\vc w]=\prod_{j}N_{w_j}/q^2$ with $N_0=2q-1$ and
$N_w=q-1$ for $w\ne0$.
\end{enumerate}
\end{lemma}
\begin{proof}
$\lambda^0\mapsto\tilde A^{\gr}(\lambda^0)$ is a unitary representation of the
additive group $C^0(\L,\Fc^{\gr})$ (Sec.~\ref{sec:Clifford_Fq}), so $P^{\gr}_\sigma$
is the projector onto the eigenvalue pattern (e). On the initial state the \rd\
and \bl\ dressings are trivial ($\green{\hat c^1}=0$), so it satisfies (a)--(d),
and $P^{\gr}_\sigma$ preserves (a)--(d) because $\tilde A^{\gr}$ commutes with the
flux and electric operators and, on the zero-flux subspace, with
$\tilde A^{\rd},\tilde A^{\bl}$ (Lemma~\ref{lemma:commutation_Fq}). Thus
$P^{\gr}_\sigma(\ket{\Psi^0_{\vc a,\vc b}}\otimes\ket0_{\gr})\in\tilde{\C}_\sigma$ is
supported on $\mathcal O_{\vc a,\vc b}$, hence is a multiple of
$\ket{\psi^\sigma_{\vc a,\vc b}}$ by Proposition~\ref{prop:sector_basis_app}. Its
squared norm is $\langle\Psi^0_{\vc a,\vc b},0|P^{\gr}_\sigma|\Psi^0_{\vc a,\vc b},0\rangle$;
the term $\lambda^0$ contributes only if $d^0\lambda^0=0$ (the \gr\ field must stay
$0$), and for $\lambda^0\in Z^0(\L,\Fc^{\gr})$ the generator reduces to the diagonal
dressing $\chi(\int_{\eta_2}\red{\hat a^1}\cup\blue{\hat b^1}\cup\lambda^0)
=\chi(\sum_ja_jb_j\ell_j(\lambda^0))$ by Eq.~\eqref{eq:ab_cup_app} and Stokes.
Hence the squared norm is $|C^0(\L,\Fc^{\gr})|^{-1}\sum_{\lambda^0\in Z^0}
\chi(\sum_ja_jb_j\ell_j(\lambda^0)-\langle\lambda^0,\sigma\rangle)$, which equals
$|Z^0(\L,\Fc^{\gr})|/|C^0(\L,\Fc^{\gr})|$ if $\langle\cdot,\sigma\rangle=\sum_ja_jb_j\ell_j$
on $Z^0(\L,\Fc^{\gr})$ and $0$ otherwise, exactly as in
Lemma~\ref{lemma:sector_orbit_app}(ii). This is (i), with the phase of
$\ket{\psi^\sigma_{\vc a,\vc b}}$ fixed so that the multiple is positive. The
vectors $P^{\gr}_\sigma(\ket{\Psi^0_{\vc a,\vc b}}\otimes\ket0)$ for distinct
$(\vc a,\vc b)$ are orthogonal (disjoint orbits), which gives $\Pr[\sigma]$ and the
post-measurement state; since $A_{\vc w}=\prod_j\{(a_j,b_j):a_jb_j=w_j\}$, the
uniform superposition over $A_{\vc w}$ is $\bigotimes_j N_{w_j}^{-1/2}\sum_{a_jb_j=w_j}\lo{\ket{a_j,b_j}}$,
which is Eq.~\eqref{eq:high_yield_output}. The admissible records with a given
$\vc w$ are the $\sigma$ with prescribed restriction $\langle\cdot,\sigma\rangle|_{Z^0(\L,\Fc^{\gr})}$,
a coset of the annihilator of $Z^0(\L,\Fc^{\gr})$ in $C_0(\L,\Fc^{\gr})$ of
cardinality $|C^0(\L,\Fc^{\gr})|/|Z^0(\L,\Fc^{\gr})|$; summing $\Pr[\sigma]$ over it
gives $\Pr[\vc w]=q^{-2s}|A_{\vc w}|$.
\end{proof}

\subsection{No low-weight error changes a label}
\label{app:GTZ_heavy}

\begin{proposition}\label{prop:heavy_GTZ_app}
Let $\sigma$ be admissible, $(\vc a,\vc b),(\vc a',\vc b')\in A_{\vc w(\sigma)}$, and
let $E=X(\zeta)Z(\eta)$ be a Pauli operator with $|E|<d^0(\Gamma,\Fc)$. If
$\langle\psi^\sigma_{\vc a',\vc b'}|E|\psi^\sigma_{\vc a,\vc b}\rangle\ne0$ then
$(\vc a',\vc b')=(\vc a,\vc b)$ and every component of $\zeta$ is exact,
$\zeta_i=d^0\tau_i$ with $\tau=(\tau_{\rd},\tau_{\bl},\tau_{\gr})\in C^0$.
Consequently $P_{\tilde{\C}_\sigma}EP_{\tilde{\C}_\sigma}$ is diagonal in the basis
of Proposition~\ref{prop:sector_basis_app}, with entries
$m_E(\vc a,\vc b):=\langle\psi^\sigma_x|E|\psi^\sigma_x\rangle/\langle\psi^\sigma_x|\psi^\sigma_x\rangle$,
$x=x_{\vc a,\vc b}$.
\end{proposition}
\begin{proof}
$E\ket y$ is a multiple of $\ket{y+\zeta}$, so a non-zero matrix element requires
$x_{\vc a',\vc b'}+d^0u'=x_{\vc a,\vc b}+d^0u+\zeta$ for some $u,u'$; hence each
$\zeta_i$ is a cocycle with $[\zeta_{\rd}]=[a_{\vc a'}]-[a_{\vc a}]\in\mathsf L^{\rd}_X$,
$[\zeta_{\bl}]=[b_{\vc b'}]-[b_{\vc b}]\in\mathsf L^{\bl}_X$ and $[\zeta_{\gr}]=0$. If
$[\zeta_{\rd}]\ne0$ then $|E|\ge|\zeta_{\rd}|\ge d^0(\Gamma,\Fc)$ by
Proposition~\ref{prop:reldist_app}, a contradiction; so $[\zeta_{\rd}]=0$, and
likewise $[\zeta_{\bl}]=0$, whence $\vc a'=\vc a$, $\vc b'=\vc b$ by linear
independence (Lemma~\ref{lemma:matching}).
\end{proof}

\subsection{The exact matrix element and its label dependence}
\label{app:GTZ_matrix_element}

From now on $E=X(\zeta)Z(\eta)$ with $|E|<d^0(\Gamma,\Fc)$ and $\zeta=d^0\tau$,
$\eta=(\eta_{\rd},\eta_{\bl},\eta_{\gr})$. We write $\zeta_{\rd}=d\tau_{\rd}$ etc.

\begin{lemma}[Exact matrix element]\label{lemma:exact_GTZ_app}
For $\sigma$-admissible $x=x_{\vc a,\vc b}$,
\begin{align}\label{eq:exact_GTZ_app}
m_E(\vc a,\vc b)&=\chi\big(\langle\tau,\sigma\rangle\big)\,
\frac1{|C^0|}\sum_{u\in C^0}\chi\big(\Theta_x(u)\big), \cr
\Theta_x(u)&:=\theta_x(u)-\theta_x(u+\tau)+\langle\eta,x+d^0u\rangle .
\end{align}
\end{lemma}
\begin{proof}
$E\ket{x+d^0u}=\chi(\langle\eta,x+d^0u\rangle)\ket{x+d^0(u+\tau)}$. In
$\sum_{u,u'}\overline{\chi(\theta_x(u')-\langle u',\sigma\rangle)}\,
\chi(\theta_x(u)-\langle u,\sigma\rangle)\langle x+d^0u'|E|x+d^0u\rangle$ only
$u'=u+\tau+\varsigma$ with $\varsigma\in Z^0$ contribute, and
$\theta_x(u+\tau+\varsigma)=\theta_x(u+\tau)+\mathsf Q_x(\varsigma)$ by
\eqref{eq:grouplaw_GTZ_app} and Lemma~\ref{lemma:orbit_app}(i), while
$\langle u+\tau+\varsigma,\sigma\rangle=\langle u,\sigma\rangle+\langle\tau,\sigma\rangle+\langle\varsigma,\sigma\rangle$.
By admissibility $\mathsf Q_x(\varsigma)=\langle\varsigma,\sigma\rangle$, so the
$\varsigma$-dependence cancels and the sum over $\varsigma$ gives $|Z^0|$, which
cancels against the norm $|C^0||Z^0|$.
\end{proof}

\begin{lemma}[Label dependence]\label{lemma:expansion_GTZ_app}
For $x=(a,b,0)$ with $a$ of pure bidegree $(1,0)$ and $b$ of pure bidegree $(0,1)$,
\be\label{eq:split_GTZ_app}
\Theta_x(u)=\theta^\circ(u)+D_{a,b}(u),
\ee
where $\theta^\circ$ does not involve $a$ or $b$,
\begin{align}
\theta^\circ(u)&=-T_{\gr}(d\mu,d\nu;\tau_{\gr})-T_{\gr}(\zeta_{\rd},d\nu;\rho+\tau_{\gr}) \cr
&\quad-T_{\gr}(d\mu,\zeta_{\bl};\rho+\tau_{\gr})-T_{\gr}(\zeta_{\rd},\zeta_{\bl};\rho+\tau_{\gr}) \cr
&\quad+\langle\eta,d^0u\rangle, \label{eq:theta0_GTZ_app}
\end{align}
and
\begin{align}
D_{a,b}(u)&=-T_{\gr}(a,b;\tau_{\gr})+T_{\bl}(a;\nu,\zeta_{\gr})-T_{\rd}(\mu;b,\zeta_{\gr}) \cr
&\quad-T_{\gr}(a,\zeta_{\bl};\rho+\tau_{\gr})-T_{\gr}(\zeta_{\rd},b;\rho+\tau_{\gr}) \cr
&\quad+\langle\eta_{\rd},a\rangle+\langle\eta_{\bl},b\rangle . \label{eq:D_GTZ_app}
\end{align}
\end{lemma}
\begin{proof}
Since $c=0$, the first two terms of \eqref{eq:theta_GTZ_app} vanish and
$\theta_x(u)=T_{\gr}(a+d\mu,b+d\nu;\rho)$. By trilinearity,
\begin{align*}
\theta_x(u)-\theta_x(u+\tau)
&=-T_{\gr}(a+d\mu,\,b+d\nu;\tau_{\gr}) \cr
&\quad-T_{\gr}(\zeta_{\rd},\,b+d\nu;\rho+\tau_{\gr}) \cr
&\quad-T_{\gr}(a+d\mu,\,\zeta_{\bl};\rho+\tau_{\gr}) \cr
&\quad-T_{\gr}(\zeta_{\rd},\zeta_{\bl};\rho+\tau_{\gr}) .
\end{align*}
Expanding and collecting the pieces containing $a$ or $b$ gives
\eqref{eq:theta0_GTZ_app} and all terms of \eqref{eq:D_GTZ_app} except the
second and third, in place of which one has $-T_{\gr}(a,d\nu;\tau_{\gr})-T_{\gr}(d\mu,b;\tau_{\gr})$.
By \eqref{eq:leibniz_cycle_Fq} with $da=0$, $|a|=1$, $|\nu|=0$:
$-T_{\gr}(a,d\nu;\tau_{\gr})-T_{\bl}(a;\nu,d\tau_{\gr})=0$, i.e.\
$-T_{\gr}(a,d\nu;\tau_{\gr})=T_{\bl}(a;\nu,\zeta_{\gr})$; and with $|\mu|=0$, $db=0$,
$|b|=1$: $T_{\gr}(d\mu,b;\tau_{\gr})-T_{\rd}(\mu;b,d\tau_{\gr})=0$, i.e.\
$-T_{\gr}(d\mu,b;\tau_{\gr})=-T_{\rd}(\mu;b,\zeta_{\gr})$. Finally
$\langle\eta,x+d^0u\rangle=\langle\eta_{\rd},a\rangle+\langle\eta_{\bl},b\rangle+\langle\eta,d^0u\rangle$.
\end{proof}

As in Appendix~\ref{app:KL_matrix_element}, the Leibniz step turns the terms in
which a label representative is paired with a \emph{gauge} variable and the
macroscopic primitive $\tau_{\gr}$ into terms paired with the local error
$\zeta_{\gr}$. The one term in $D_{a,b}$ that remains non-local is
$T_{\gr}(a,b;\tau_{\gr})$, which pairs both label representatives with the primitive.

\subsection{Shadows and product-form cleaning}
\label{app:GTZ_cleaning}

For a set $T$ of $1$-cells of $\L$ let $\pi_x(T)\subseteq E(\Gamma)$ be the
$x$-edges and $\pi_y(T)\subseteq E(\Gamma)$ the $y$-edges of the squares whose
closure meets $T$; $|\pi_x(T)|,|\pi_y(T)|\le\Delta|T|$. Define the shadows
\begin{align}
S_x&:=\pi_x(\mathrm{supp}\,\zeta_{\bl})\cup\pi_x(\mathrm{supp}\,\zeta_{\gr})\cup\pi_x(\mathrm{supp}\,\eta_{\rd}), \cr
S_y&:=\pi_y(\mathrm{supp}\,\zeta_{\rd})\cup\pi_y(\mathrm{supp}\,\zeta_{\gr})\cup\pi_y(\mathrm{supp}\,\eta_{\bl}), \label{eq:shadows_GTZ_app}
\end{align}
so that $|S_x|,|S_y|\le3\Delta|E|$.

\begin{lemma}[Shadow]\label{lemma:shadow_GTZ_app}
Let $a\in Z^1(\L,\Fc^{\rd})$ have pure bidegree $(1,0)$ and vanish on every
horizontal cell $(e,w)$ with $e\in S_x$, and let $b\in Z^1(\L,\Fc^{\bl})$ have pure
bidegree $(0,1)$ and vanish on every vertical cell $(v,f)$ with $f\in S_y$. Then
\be\label{eq:D_reduced_app}
D_{a,b}(u)=-T_{\gr}(a,b;\tau_{\gr})\qquad\forall\,u\in C^0 .
\ee
\end{lemma}
\begin{proof}
By (P1) each remaining term of \eqref{eq:D_GTZ_app} is a sum over squares (or
cells) of a local form, and in every term the square or cell that could
contribute has its $a$- or $b$-argument equal to zero: $T_{\bl}(a;\nu,\zeta_{\gr})$
and $T_{\gr}(a,\zeta_{\bl};\rho+\tau_{\gr})$ receive contributions only from squares
on whose closure $\zeta_{\gr}$, respectively $\zeta_{\bl}$, is non-zero, whose
$x$-edges lie in $S_x$, where $a$ vanishes ($a$ has no vertical component);
$\langle\eta_{\rd},a\rangle$ receives contributions only from horizontal cells
$(e,w)$ with $e\in\pi_x(\mathrm{supp}\,\eta_{\rd})\subseteq S_x$; and symmetrically
$T_{\rd}(\mu;b,\zeta_{\gr})$, $T_{\gr}(\zeta_{\rd},b;\rho+\tau_{\gr})$ and
$\langle\eta_{\bl},b\rangle$ see $b$ only on vertical cells with $y$-edge in $S_y$.
\end{proof}

\begin{proposition}[Cleaned representatives and cross terms]\label{prop:cleaned_GTZ_app}
Assume $3\Delta|E|<d_1(\Gamma,\Fc^{\diamond})$. For every $e\in M$ there are
$\chi_e,\chi'_e\in C^0(\Gamma,\Fc^{\diamond})$ such that $\bar e':=\bar e+d\chi_e$
vanishes on $S_x$ and $\bar e'':=\bar e+d\chi'_e$ vanishes on $S_y$. Put, for
$j=(e,f)\in J$,
\begin{align}\label{eq:cleaned_basis_app}
\red{\alpha'^1_{j}}&:=\bar e'\otimes\red{\bs^0_f}=\red{\alpha^1_j}+d^0(\chi_e\otimes\red{\bs^0_f}), \cr
\blue{\beta'^1_{j}}&:=\blue{\bs^0_e}\otimes\bar f''=\blue{\beta^1_j}+d^0(\blue{\bs^0_e}\otimes\chi'_f),
\end{align}
and $a'_{\vc a}:=\sum_ja_j\red{\alpha'^1_j}$, $b'_{\vc b}:=\sum_jb_j\blue{\beta'^1_j}$. Then:
\begin{enumerate}[label=(\roman*)]
\item $a'_{\vc a}$ and $b'_{\vc b}$ are cohomologous to $a_{\vc a}$, $b_{\vc b}$, have
pure bidegrees $(1,0)$ and $(0,1)$, and satisfy the hypotheses of
Lemma~\ref{lemma:shadow_GTZ_app}.
\item For $j\ne k$ in $J$ there is a $1$-cochain $\Psi_{jk}$ of $\Fc^{\rd}\otimes\Fc^{\bl}$,
vanishing on every cell of every square on whose closure $\zeta_{\gr}$ is
non-zero, with $\red{\alpha'^1_j}\cup\blue{\beta'^1_k}=d\Psi_{jk}$. Consequently
\begin{align}\label{eq:cross_terms_app}
T_{\gr}(a'_{\vc a},b'_{\vc b};\tau_{\gr})&=\sum_{j\in J}a_jb_j\,F_{j}, \cr
F_j&:=\int_{\eta_2}\red{\alpha'^1_j}\cup\blue{\beta'^1_j}\cup\tau_{\gr},
\end{align}
which equals $\sum_jw_jF_j$ for all $(\vc a,\vc b)\in A_{\vc w}$.
\end{enumerate}
\end{proposition}
\begin{proof}
The cleaning potentials exist by Lemma~\ref{lemma:clean_app} applied to
$C^\bullet(\Gamma,\Fc^{\diamond})$, whose minimum non-zero cycle weight is
$d_1(\Gamma,\Fc^{\diamond})>|S_x|,|S_y|$. (i) $d^0(\chi_e\otimes\bs^0_f)=d\chi_e\otimes\bs^0_f$
and $d^0(\bs^0_e\otimes\chi'_f)=\bs^0_e\otimes d\chi'_f$ because $d\bs^0=0$ and the
K\"unneth sign is $+1$ for a degree-$0$ first factor; $\alpha'^1_j$ is supported on
horizontal cells $(e'',w)$ with $\bar e'(e'')\ne0$, hence $e''\notin S_x$, and
similarly for $\beta'^1_j$.
(ii) Let $j=(e_1,f_1)$, $k=(e_2,f_2)$. By the K\"unneth rule and the edgewise rule
of Lemma~\ref{lemma:matching},
$\alpha'^1_j\cup\beta'^1_k=(\bar e'_1\cup\bs^0_{e_2})\otimes(\bs^0_{f_1}\cup\bar f''_2)$,
and
\begin{align*}
\bar e'_1\cup\bs^0_{e_2}&=\bar e_1\cup\bs^0_{e_2}+d\chi_{e_1}\cup\bs^0_{e_2} \cr
&=\delta_{e_1e_2}\,\bar e_1+d(\chi_{e_1}\cup\bs^0_{e_2}), \cr
\bs^0_{f_1}\cup\bar f''_2&=\delta_{f_1f_2}\,\bar f_2+d(\bs^0_{f_1}\cup\chi'_{f_2}),
\end{align*}
by the Leibniz rule with $d\bs^0=0$, since $\bs_{e_2}(e_1)=\delta_{e_1e_2}$ and
$\bs_{f_1}(f_2)=\delta_{f_1f_2}$ for $e_1,f_2\in M$. If $e_1\ne e_2$, put
$\phi:=\chi_{e_1}\cup\bs^0_{e_2}\in C^0(\Gamma,\Fc^{\diamond}\otimes\Fc)$; then
$\alpha'^1_j\cup\beta'^1_k=d\phi\otimes y$ with $y:=\bs^0_{f_1}\cup\bar f''_2$ a
$1$-cochain on $\Gamma$, and $d(\phi\otimes y)=d\phi\otimes y+\phi\otimes dy=d\phi\otimes y$
because $\Gamma$ has no $2$-cells; so $\Psi_{jk}:=\phi\otimes y$, a vertical cochain
supported on cells $(v,f)$ with $\bar f''_2(f)\ne0$, i.e.\ $f\notin S_y\supseteq\pi_y(\mathrm{supp}\,\zeta_{\gr})$.
If $e_1=e_2$ then $f_1\ne f_2$, and with $\phi':=\bs^0_{f_1}\cup\chi'_{f_2}$ and
$y':=\bar e'_1\cup\bs^0_{e_2}$ one has $\alpha'^1_j\cup\beta'^1_k=y'\otimes d\phi'=-d(y'\otimes\phi')$,
since $d(y'\otimes\phi')=dy'\otimes\phi'-y'\otimes d\phi'$ and $dy'=0$; so
$\Psi_{jk}:=-y'\otimes\phi'$, a horizontal cochain supported on cells $(e,w)$ with
$\bar e'_1(e)\ne0$, i.e.\ $e\notin S_x\supseteq\pi_x(\mathrm{supp}\,\zeta_{\gr})$. In
both cases $\Psi_{jk}$ vanishes on the squares on whose closure $\zeta_{\gr}$ is
non-zero. By bilinearity $T_{\gr}(a'_{\vc a},b'_{\vc b};\tau_{\gr})=\sum_{j,k}a_jb_k
\int_{\eta_2}\alpha'^1_j\cup\beta'^1_k\cup\tau_{\gr}$, and for $j\ne k$, by
\eqref{eq:leibniz_cycle_Fq} with $|\Psi_{jk}|=1$, $|\tau_{\gr}|=0$,
\[
\int_{\eta_2}d\Psi_{jk}\cup\tau_{\gr}=\int_{\eta_2}\Psi_{jk}\cup d\tau_{\gr}
=\int_{\eta_2}\Psi_{jk}\cup\zeta_{\gr}=0
\]
by (P1), since every square that could contribute has $\zeta_{\gr}\ne0$ on its
closure and $\Psi_{jk}=0$ on its faces. The last statement is $a_jb_j=w_j$ on $A_{\vc w}$.
\end{proof}

\subsection{Proof of Theorem~\ref{thm:GTZ_KL}}
\label{app:GTZ_proof}

Part (i) is Proposition~\ref{prop:sector_basis_app}. For part (ii) let first
$E=X(\zeta)Z(\eta)$ be a Pauli operator with $|E|<cN^{1-\epsilon}$, so that
$|E|<d^0(\Gamma,\Fc)$ and $3\Delta|E|<d_1(\Gamma,\Fc^{\diamond})$. By
Proposition~\ref{prop:heavy_GTZ_app}, $P_{\tilde{\C}_\sigma}EP_{\tilde{\C}_\sigma}$ is
diagonal in the basis $\{\ket{\psi^\sigma_{\vc a,\vc b}}\}$ with entries
$m_E(\vc a,\vc b)$, and $\zeta=d^0\tau$. Fix $(\vc a,\vc b)\in A_{\vc w}$. By
Lemma~\ref{lemma:sector_orbit_app}(iii) the normalized diagonal entry depends
only on the orbit, so we may compute it with the cleaned representative
$x'=(a'_{\vc a},b'_{\vc b},0)$ of Proposition~\ref{prop:cleaned_GTZ_app}, the phase
appearing once in the bra and once in the ket. By
Lemmas~\ref{lemma:exact_GTZ_app}, \ref{lemma:expansion_GTZ_app}, \ref{lemma:shadow_GTZ_app}
and Proposition~\ref{prop:cleaned_GTZ_app}(ii),
\begin{align}
m_E(\vc a,\vc b)&=\frac{\chi(\langle\tau,\sigma\rangle)}{|C^0|}\sum_{u}
\chi\big(\theta^\circ(u)+D_{a'_{\vc a},b'_{\vc b}}(u)\big) \cr
&=\frac{\chi\big(\langle\tau,\sigma\rangle-\sum_jw_jF_j\big)}{|C^0|}\sum_{u}\chi\big(\theta^\circ(u)\big),
\end{align}
which does not depend on $(\vc a,\vc b)$. Writing $m_{E,\sigma}$ for the common
value, $P_{\tilde{\C}_\sigma}EP_{\tilde{\C}_\sigma}=m_{E,\sigma}P_{\tilde{\C}_\sigma}$.
A general operator supported on a set $R$ of fewer than $cN^{1-\epsilon}$ qudits is
a linear combination of Pauli operators supported in $R$, so
Eq.~\eqref{eq:KL_GTZ} follows by linearity; applied to the products
$E^\dagger_\alpha E_\beta$ of errors on fewer than $cN^{1-\epsilon}/2$ qudits, this is
the Knill--Laflamme condition \cite{Knill1997}, which gives the detection and
correction statements. The numerical bound \eqref{eq:KL_threshold_GTZ} is
Eq.~\eqref{eq:GTZ_app_distances}. Part (iii) restates (ii) for the gauge-fixed
twisted code $\tilde{\C}_0$ and the sectors $\tilde{\C}_\sigma$: a logical operator in
the sense of Definition~\ref{def:distance_clifford} acting non-trivially on the
selected labels must have $P_{\tilde{\C}_\sigma}LP_{\tilde{\C}_\sigma}\not\propto P_{\tilde{\C}_\sigma}$,
hence weight $\ge cN^{1-\epsilon}$. \hfill$\square$

\begin{remark}[Where the hypotheses enter]\label{rem:GTZ_hypotheses}
Theorem~\ref{thm:GTZ_classical}(iii) enters twice: through
$d^0(\Gamma,\Fc)$ in Proposition~\ref{prop:heavy_GTZ_app} (no light $X$-type
operator changes a selected label) and through $d_1(\Gamma,\Fc^{\diamond})$ in the
cleaning step, Proposition~\ref{prop:cleaned_GTZ_app}. Theorem~\ref{thm:GTZ_classical}(ii)
enters through the matching \eqref{eq:cochain_matching}, which is used twice: to
identify the charge \eqref{eq:charge_GTZ_app} and to make the cross terms in
\eqref{eq:cross_terms_app} exact. The gauge-fixing conditions (b)--(c) enter
only in Proposition~\ref{prop:sector_basis_app}, where they restrict the support
to the orbits $\mathcal O_{\vc a,\vc b}$ and make the \gr\ field exact, which is
what kills the first two terms of \eqref{eq:theta_GTZ_app} and removes the need
for the column-localized primitives of Lemma~\ref{lemma:column_primitive}. In
contrast to Appendix~\ref{app:twisted_KL}, a plain cleaning of
$a_{\vc a}$ and $b_{\vc b}$ off the shadows would not suffice: the non-local term
$T_{\gr}(a',b';\tau_{\gr})$ would then contain cross terms $a_jb_k$, $j\ne k$, which
are not constant on $A_{\vc w}$. It is the product form of the cleaned
representatives, inherited from the product form of the selected classes, that
makes the cross terms exact with a primitive supported away from the error.
\end{remark}

%======================================================================
\section{Generalization to balanced-product sheaf complexes}
\label{sec:balanced}
%======================================================================

Both instantiations above live on the hypergraph product $\L=G_x\times G_y$, and
both inherit its $\sqrt{n}$ distance ceiling. The known way past this ceiling is
to replace the tensor product by a \emph{balanced product}
\cite{9490244} or a \emph{lifted product} \cite{PK:almost_linear,pkldpc22}: one
takes two classical codes with a common free action of a finite group $H$ and
quotients the product by the diagonal action, $\L_H=(G_x\times G_y)/H$. The twisted
sheaf gauge theory of Secs.~\ref{sec:non-Abelian_LDPC}--\ref{sec:Fq_generalization}
used only three properties of the underlying sheaf complex: a cellular cup product
with the Leibniz rule, a $2$-cycle $\eta_2$ of the tensor sheaf, and bounded local
structure. In this section we show that all three descend to balanced products,
so that the whole construction --- the twisted Clifford stabilizer code, the
$0$-form subcomplex symmetry, and the gauging-measurement protocol --- goes
through verbatim. We then instantiate it with balanced products of the subdivided
sheaf codes of Sec.~\ref{sec:code_instantiation}, and finally, when the order of
$H$ is coprime to the characteristic, we use the K\"unneth theorem over the group
algebra to obtain an explicit basis of the logical qudits and the logical CZ
structure, which acquires a \emph{selection rule} on the $H$-characters. We
concentrate on the code construction; the distance of balanced-product codes is a
separate problem, discussed at the end.

\subsection{Equivariant sheaf codes and the balanced product}
\label{sec:balanced_formalism}

\begin{definition}[Symmetric sheaf code]\label{def:symmetric_sheaf}
Let $G=(V,E)$ be a bipartite graph with all edges oriented from left to right and
let $\Fc$ be a sheaf on $G$ with local codes $\cC^{\perp}_v=\im\Fc_{v\shortrightarrow E(v)}$
(Lemma~\ref{lemma:dictionary}). A finite group $H$ \emph{acts on $(G,\Fc)$} if
$H$ acts on $G$ by graph automorphisms that preserve the bipartition and act
\emph{freely} on $V$, together with linear isomorphisms
$\varphi_h:\Fc(\sigma)\to\Fc(h\sigma)$ for every cell $\sigma$ and $h\in H$, equal
to the identity on the edge stalks $\Fc(e)=\Fq$, such that
$\varphi_{h'}\varphi_h=\varphi_{h'h}$ and
\be\label{eq:equivariant_sheaf}
\varphi_h\circ\Fc_{v\shortrightarrow e}=\Fc_{hv\shortrightarrow he}\circ\varphi_h
\qquad(v\prec e).
\ee
In terms of the Tanner code, $H$ is a group of automorphisms of the Tanner graph
of $\T^{\dag}(G,\Fc)$ that permutes the coordinates and maps every local code
onto the local code of the image vertex, $P_h\,\cC^{\perp}_v=\cC^{\perp}_{hv}$.
Since the action is free, the stalk bases can be chosen \emph{equivariantly},
by fixing them on one representative of each orbit of vertices and transporting
them with $\varphi_h$; we always do so, and then every $\varphi_h$ is the
identity matrix and $H$ simply permutes the coordinates of $C^\bullet(G,\Fc)$ and
$C_\bullet(G,\Fc)$.
\end{definition}

Freeness on vertices implies freeness on edges, since a bipartition-preserving
automorphism fixing an edge fixes its two endpoints. The standard examples are
\emph{quasi-cyclic} and, more generally, \emph{lifted} LDPC codes: a parity-check
matrix $\bar\Hs\in\Fq^{m\times n}$ composed of $\ell\times\ell$ circulant blocks
is invariant under $H=\ZZ_\ell$ acting by simultaneous cyclic shifts within each
block of bits and each block of checks, and the action is free on both; replacing
circulants by the regular representation of any finite group $H$ gives the
$H$-lifted codes of Refs.~\cite{PK:almost_linear,pkldpc22}. The subdivision of
Sec.~\ref{sec:subdivision} is compatible with such a symmetry: the subdivided
graph $G$ is the Tanner graph itself, with bit vertices, check vertices and one
edge per non-zero entry, and the subdivision sheaf (repetition codes at bit
vertices, SPC codes at check vertices) is mapped to itself by every Tanner-graph
automorphism, as is its dual $\Fc^{\perp}$ and the constant sheaf $\Fc^0$.

The group acts on cochains and chains by
\begin{align}\label{eq:H_action_cochains}
(h\cdot x)(\sigma)&=\varphi_h\big(x(h^{-1}\sigma)\big), \cr
(h\cdot z)(\sigma)&=\varphi_h\big(z(h^{-1}\sigma)\big),
\end{align}
commuting with $d$ and $\partial$ by Eq.~\eqref{eq:equivariant_sheaf} (the
incidence signs are preserved because the orientation is), and preserving the
canonical pairing, $\langle h\cdot x,h\cdot z\rangle=\langle x,z\rangle$. Hence
$H$ acts on $H^\bullet(G,\Fc)$ and $H_\bullet(G,\Fc)$, and the Tanner-code
isomorphisms of Lemma~\ref{lemma:dictionary} and the Poincar\'e pairing
$\mathcal P$ of Eq.~\eqref{eq:P-def} are $H$-equivariant.

\begin{definition}[Balanced product of sheaf complexes]\label{def:balanced_product}
Let $H$ act on $(G_x,\Fc_x)$ and on $(G_y,\Fc_y)$ as in
Definition~\ref{def:symmetric_sheaf}, and let it act diagonally on the square
complex $\L=G_x\times G_y$, $h(\sigma,\tau)=(h\sigma,h\tau)$. The
\emph{balanced product} is the quotient cell complex
\be\label{eq:balanced_complex}
\L_H:=G_x\times_H G_y:=(G_x\times G_y)/H ,
\ee
whose $i$-cells are the $H$-orbits of $i$-cells of $\L$, with the induced face
relation, and the \emph{balanced product sheaf} $\Fc_x\boxtimes_H\Fc_y$ assigns to
the orbit of $(\sigma,\tau)$ the stalk $\Fc_x(\sigma)\otimes\Fc_y(\tau)$ of any
representative, different representatives being identified through
$\varphi_h\otimes\varphi_h$, with the induced co-restriction maps. We write
$\pi:\L\to\L_H$ for the quotient map.
\end{definition}

\begin{lemma}[Cell structure]\label{lemma:balanced_cells}
The diagonal action of $H$ on the cells of $\L$ is free, and no non-trivial
element of $H$ maps a vertex, edge or square of a closed square of $\L$ to another
cell of the same closed square. Consequently $\L_H$ is a square complex with
$|\L_H(i)|=|\L(i)|/|H|$ cells in each degree, every square of $\L_H$ has four
distinct vertices and four distinct edges, the bidegree of $1$-cells (horizontal
$(e,w)$ or vertical $(v,f)$) is well defined on $\L_H$, and the number of cells
incident to a given cell is the same as in $\L$. In particular $\L_H$ is a graded
incidence poset in the sense of Definition~\ref{def:poset}.
\end{lemma}
\begin{proof}
Freeness on the vertices of $\L$ is freeness on $V_x$. If $h\neq1$ mapped the
vertex $(v,w)$ of a square $(e,f)$, $e=(v,v_1)$, $f=(w,w_1)$, to another vertex of
the same square, then $hv\in\{v,v_1\}$ with $hv$ on the same side of the
bipartition as $v$, forcing $hv=v$ and $h=1$; the same argument applies to edges
and to the square itself. Hence $\pi$ is injective on every closed square, the
quotient inherits the local incidence structure of $\L$, and the even-incidence
condition of Definition~\ref{def:poset} holds because it holds in each closed
square.
\end{proof}

\begin{lemma}[Cochains, chains and pairing on the quotient]\label{lemma:balanced_cochains}
Pull-back along $\pi$ identifies $C^\bullet(\L_H,\Fc_x\boxtimes_H\Fc_y)$ with the
$H$-invariant cochains,
\begin{align}\label{eq:invariant_cochains}
\pi^*:\ C^\bullet(\L_H)\ &\xrightarrow{\ \cong\ }\ C^\bullet(\L)^{H}, \cr
C^\bullet(\L)^{H}&=\big(C^\bullet(G_x,\Fc_x)\otimes C^\bullet(G_y,\Fc_y)\big)^{H}, \cr
\end{align}
as cochain complexes. Dually, push-forward identifies $C_\bullet(\L_H)$ with the
coinvariants $C_\bullet(\L)_H$, and the norm map $z\mapsto\sum_hh\cdot z$ identifies
the coinvariants with the invariants $C_\bullet(\L)^H$ (the action being free);
write $\pi^!:C_\bullet(\L_H)\to C_\bullet(\L)^H$ for the resulting isomorphism of
chain complexes, which sends a cell of $\L_H$ to the sum of the cells in its
orbit. For $x\in C^\bullet(\L_H)$ and $z\in C_\bullet(\L_H)$,
\be\label{eq:pairing_lift}
\langle\pi^*x,\ \pi^!z\rangle_{\L}=|H|\,\langle x,z\rangle_{\L_H} .
\ee
\end{lemma}
\begin{proof}
An invariant cochain is determined by its values on orbit representatives, and
$d$ commutes with the action; this gives \eqref{eq:invariant_cochains}. For a
free action the norm map is an isomorphism from coinvariants to invariants of a
permutation module, and it commutes with $\partial$. Finally
$\langle\pi^*x,\pi^!z\rangle_\L=\sum_{\sigma\in\L}\langle x(\pi\sigma),
z(\pi\sigma)\rangle$, and every cell of $\L_H$ has exactly $|H|$ preimages.
\end{proof}

\begin{lemma}[Descent of the cup product]\label{lemma:balanced_cup}
The cup product of $\L$ (Definition~\ref{def:cup-simplicial}, in its product
form of Lemma~\ref{lemma:factorization}) is $H$-equivariant,
$h\cdot(x\cup y)=(h\cdot x)\cup(h\cdot y)$. Hence it restricts to invariant
cochains and defines a cup product on $\L_H$,
\be\label{eq:balanced_cup}
\pi^*(x\cup_{H}y):=\pi^*x\cup\pi^*y ,
\ee
which satisfies the Leibniz rule and associativity of
Proposition~\ref{prop:cup-properties}, descends to cohomology, is local in the
sense of (P1) of Appendix~\ref{app:KL_orbit}, and respects the bigrading (P2).
\end{lemma}
\begin{proof}
$H$ acts by orientation-preserving cellular automorphisms, so it maps front faces
to front faces and back faces to back faces, and by \eqref{eq:equivariant_sheaf}
it intertwines the co-restriction maps and their tensor products; equivariance
follows from Eq.~\eqref{eq:cup-simplicial}. Everything else is inherited from
$\L$ through \eqref{eq:invariant_cochains}.
\end{proof}

\begin{lemma}[Fundamental class of the balanced product]\label{lemma:balanced_fundamental}
Let $\Fc^{\rd},\Fc^{\bl},\Fc^{\gr}$ be balanced-product sheaves on $\L_H$ built from
$H$-symmetric sheaves on $G_x$ and $G_y$, and let $[\L_H]\in C_2(\L_H,\Fc^{\rd}\otimes\Fc^{\bl}\otimes\Fc^{\gr})$
be the sum of all squares with coefficient $1$ (the square stalks are scalar).
Then $\pi^![\L_H]=[\L]$, and $[\L_H]$ is a $2$-cycle if and only if $[\L]$ is,
i.e.\ if and only if the local cycle condition of Lemma~\ref{lemma:local_cycle}
--- the multiplication property \eqref{eq:multiplication_property_local} ---
holds at every vertex of $G_x$ and of $G_y$.
\end{lemma}
\begin{proof}
$\pi^!$ sends the sum of the squares of $\L_H$ to the sum of all squares of $\L$,
and it is an isomorphism of chain complexes (Lemma~\ref{lemma:balanced_cochains}),
so $\partial[\L_H]=0$ iff $\partial[\L]=0$. The second statement is
Lemma~\ref{lemma:fundamental_class} and
Proposition~\ref{prop:multiplication_property_equivalence}, which are conditions
at single vertices and are therefore insensitive to the quotient.
\end{proof}

With a cup product and a $2$-cycle in hand, the twisted theory is defined
exactly as before.

\begin{proposition}[Twisted balanced-product code]\label{prop:twisted_balanced}
Let $\Fc^{\rd},\Fc^{\bl},\Fc^{\gr}$ be as in Lemma~\ref{lemma:balanced_fundamental},
satisfying the local cycle condition, and put $\eta_2=[\L_H]$. Then the twisted
stabilizer group \eqref{eq:Clifford_group_Fq}, with the dressing
$U^{\rd}(\lambda^0)=\omega^{\mathrm{Tr}\int_{[\L_H]}\lambda^0\cup_H\blue{\hat b^1}\cup_H\green{\hat c^1}}$
and its colour permutations, defines a twisted sheaf code $\tilde{\C}_H$ on the
qudits of the $1$-cells of $\L_H$, and all of
Secs.~\ref{sec:non-Abelian_LDPC}--\ref{sec:Fq_generalization} hold for it
verbatim: the generators commute on the zero-flux subspace
(Lemma~\ref{lemma:commutation_Fq}), the code-space constraints
\eqref{eq:constraint_Fq} hold, the transversal operators
$\widetilde{\mathrm{CZ}}^{\rd,\bl}(\green{\gamma^0})=\omega^{\mathrm{Tr}\int_{[\L_H]}\red{\hat a^1}\cup_H\blue{\hat b^1}\cup_H\green{\gamma^0}}$,
$\green{\gamma^0}\in Z^0(\L_H,\Fc^{\gr})$, are $0$-form subcomplex symmetries of the
untwisted balanced-product code with logical action \eqref{eq:logical_CZ_Fq},
and the gauging-measurement protocol of Sec.~\ref{sec:protocol} applies. The code
has $n_H=n_{\L}/|H|$ physical qudits, where $n_\L$ is the number of qudits of the
corresponding construction on $\L$, and its generators have the same weights as
those on $\L$; in particular it is qLDPC whenever the construction on $\L$ is.
\end{proposition}
\begin{proof}
The derivations of Secs.~\ref{sec:non-Abelian_LDPC}--\ref{sec:Fq_generalization}
use only the Leibniz rule and associativity of the cup product, Stokes' identity
$\langle d\varpi,\eta_2\rangle=\langle\varpi,\partial\eta_2\rangle=0$, and the fact
that the stalks of the tensor sheaf on $2$-cells are scalars; all three hold on
$\L_H$ by Lemmas~\ref{lemma:balanced_cup} and~\ref{lemma:balanced_fundamental}.
The qudit count and the generator weights follow from
Lemma~\ref{lemma:balanced_cells}, since the stalks on an orbit of $1$-cells are
those of any representative and a cell of $\L_H$ has the same incident cells as
its lifts.
\end{proof}

It is often convenient to compute on $\L$ rather than on $\L_H$. By
\eqref{eq:invariant_cochains} and \eqref{eq:pairing_lift}, if $|H|$ is invertible
in $\Fq$ (i.e.\ $p\nmid|H|$) then for every $2$-cochain $\Omega$ of $\L_H$
\be\label{eq:integral_lift}
\int_{[\L_H]}\Omega=\frac{1}{|H|}\int_{[\L]}\pi^*\Omega ,
\ee
so that invariant forms on the balanced product are the invariant forms on the
product, evaluated on $H$-invariant cochains and divided by $|H|$.

\subsection{Instantiation: balanced product of subdivided codes}
\label{sec:balanced_instantiation}

\begin{construction}[Balanced-product subdivided sheaf complexes]\label{con:balanced}
Let $\bar\C=\ker\bar\Hs$ be a classical LDPC code whose Tanner graph admits a
free, bipartition-preserving action of a finite group $H$ (for instance a
quasi-cyclic or $H$-lifted code), let $G$ be its subdivision with the sheaf $\Fc$
of Proposition~\ref{prop:subdivision-params}, so that $H$ acts on $(G,\Fc)$,
$(G,\Fc^{\perp})$ and $(G,\Fc^0)$ in the sense of
Definition~\ref{def:symmetric_sheaf}, and take $G_x=G_y=G$ with the same action.
Equip $\L_H=G\times_H G$ with the three balanced-product sheaves of
Eq.~\eqref{eq:three_copies_identical_factors},
\begin{align}\label{eq:balanced_three_copies}
\Fc^{\rd}&=\Fc\boxtimes_H\Fc^{\perp}, \cr
\Fc^{\bl}&=\Fc^0\boxtimes_H\Fc, \cr
\Fc^{\gr}&=\Fc^{\perp}\boxtimes_H\Fc^0 ,
\end{align}
and the twist $\eta_2=[\L_H]$. The twisted code $\tilde{\C}_H$ is the one of
Proposition~\ref{prop:twisted_balanced}.
\end{construction}

\begin{lemma}[Admissibility]\label{lemma:balanced_admissible}
$[\L_H]$ is a $2$-cycle of $\Fc^{\rd}\otimes\Fc^{\bl}\otimes\Fc^{\gr}$, so
Construction~\ref{con:balanced} defines a twisted sheaf code with a non-trivial
cohomology invariant whenever $H^0(G,\Fc^{\perp})^{H}\neq0$ (see
Proposition~\ref{prop:balanced_CZ} below).
\end{lemma}
\begin{proof}
By Lemma~\ref{lemma:balanced_fundamental} it suffices to check the local cycle
condition at the vertices of $G$, which is the same computation as for
Construction~\ref{con:2d-hgp}: $\cC^{\perp}_v*\cC_v\subseteq\mathrm{Rep}^{\perp}$,
i.e.\ $\sum_{e\in E(v)}(\vc a*\vc b*\vc 1)_e=\langle\vc a,\vc b\rangle=0$ for
$\vc a\in\cC^{\perp}_v$, $\vc b\in\cC_v$.
\end{proof}

The special case in which $H$ is trivial is Construction~\ref{con:2d-hgp}. When
$\bar\Hs$ is an $H$-lifted matrix over the group algebra $\Fq[H]$, the
untwisted \rd\ code of Construction~\ref{con:balanced} is a (subdivided) lifted
product code in the sense of Refs.~\cite{PK:almost_linear,pkldpc22}, and
$\L_H$ is the quotient of the hypergraph-product complex by the diagonal action,
with $|H|$ times fewer qudits:
\be\label{eq:balanced_length}
n_H=\frac{n_\L}{|H|}=\Theta\Big(\frac{\bar n^{2}}{|H|}\Big)
\ee
for a constant-rate input code of length $\bar n$.

\subsection{K\"unneth structure for $p\nmid|H|$ and the logical CZ structure}
\label{sec:balanced_kunneth}

A balanced product has no K\"unneth formula in general: the cohomology of
$C^\bullet(G_x)\otimes_{\Fq[H]}C^\bullet(G_y)$ involves $\mathrm{Tor}$ terms over the
group algebra. When $|H|$ is invertible in $\Fq$, however, $\Fq[H]$ is semisimple
(Maschke's theorem) and the formula survives, with the tensor product replaced by
its $H$-invariant part.

\begin{proposition}[{K\"unneth theorem over the group algebra $\Fq[H]$}]\label{prop:balanced_kunneth}
Assume $p\nmid|H|$. Then for balanced-product sheaves
\begin{align}\label{eq:balanced_kunneth}
&H^{i}\big(\L_H,\Fc_x\boxtimes_H\Fc_y\big) \cr
&\qquad\cong\
\Big[\bigoplus_{a+b=i}H^{a}(G_x,\Fc_x)\otimes H^{b}(G_y,\Fc_y)\Big]^{H}, \cr
\end{align}
and likewise for homology. If moreover $H$ is abelian and $\Fq$ contains the
$\exp(H)$-th roots of unity, let $\hat H$ be the group of characters
$\chi:H\to\Fq^{\times}$ and $M_\chi:=\{m\in M:h\cdot m=\chi(h)m\}$ the isotypic
components of an $H$-module $M$; then
\begin{align}\label{eq:balanced_kunneth_chars}
H^{1}\big(\L_H,\Fc_x\boxtimes_H\Fc_y\big)
\cong\bigoplus_{\chi\in\hat H}
\Big[&H^1(G_x)_\chi\otimes H^0(G_y)_{\bar\chi} \cr
&\oplus H^0(G_x)_\chi\otimes H^1(G_y)_{\bar\chi}\Big],
\end{align}
with $\bar\chi=\chi^{-1}$, and the pure tensors $x_\chi\otimes y_{\bar\chi}$ of
eigen-cocycles are themselves $H$-invariant cocycles of $\L$, i.e.\ cocycles of
$\L_H$, whose classes form a basis of $H^1(\L_H)$ (the \emph{K\"unneth basis}).
Dually, $Z^0(\L_H,\Fc_x\boxtimes_H\Fc_y)=\bigoplus_\chi H^0(G_x)_\chi\otimes H^0(G_y)_{\bar\chi}$.
\end{proposition}
\begin{proof}
Let $\pi_H:=|H|^{-1}\sum_{h\in H}h$, an idempotent on $C^\bullet(\L)$ commuting
with $d$ whose image is $C^\bullet(\L)^H\cong C^\bullet(\L_H)$ by
\eqref{eq:invariant_cochains}. Taking cohomology commutes with the idempotent,
so $H^\bullet(C^\bullet(\L)^H)=\pi_HH^\bullet(\L)=H^\bullet(\L)^H$, and
$H^\bullet(\L)=H^\bullet(G_x)\otimes H^\bullet(G_y)$ by the K\"unneth theorem over
the field $\Fq$, the isomorphism being natural and hence $H$-equivariant. For
semisimple $\Fq[H]$ with all irreducible representations one-dimensional, the
invariants of a tensor product of two modules are
$(M\otimes N)^H=\bigoplus_\chi M_\chi\otimes N_{\bar\chi}$, which gives
\eqref{eq:balanced_kunneth_chars}; a pure tensor $x_\chi\otimes y_{\bar\chi}$ is
invariant because $h\cdot(x\otimes y)=\chi(h)\bar\chi(h)\,x\otimes y$. The
homological statement is identical with coinvariants, which coincide with
invariants here.
\end{proof}

For $H=\ZZ_\ell$ with $\ell$ odd and $q=2^{u}$ such that $\ell\mid q-1$, the
characters are the Fourier modes $\chi_k(j)=\zeta^{jk}$, $\zeta$ a primitive
$\ell$-th root of unity, and \eqref{eq:balanced_kunneth_chars} says that a
logical qudit of the balanced product is a pair of factor classes with
\emph{opposite} Fourier momenta $(k,-k)$: the diagonal quotient keeps exactly the
momentum-conserving part of the product. If $\Fq$ does not contain the roots of
unity, the same holds with irreducible $\Fq[H]$-modules $V$ in place of characters
and $(M\otimes N)^H=\bigoplus_V(M_V\otimes N_{V^*})^H$, where the invariants are
contractions rather than single pure tensors; the trivial-character sector, which
is all that the instantiation below uses, is unaffected.

\begin{proposition}[Logical CZ structure on a balanced product: selection rule]
\label{prop:balanced_CZ}
Assume $p\nmid|H|$, $H$ abelian with $\Fq\supseteq\mu_{\exp(H)}$, and let
\begin{align}\label{eq:balanced_basis}
\red{\alpha^1}&=\red{\as^1_{\chi}}\otimes\red{\bs^0_{\bar\chi}}, \cr
&\red{\as^1_{\chi}}\in H^1(G_x,\Fc^{\rd}_x)_\chi,\ \ \red{\bs^0_{\bar\chi}}\in H^0(G_y,\Fc^{\rd}_y)_{\bar\chi}, \cr
\blue{\beta^1}&=\blue{\bs^0_{\psi}}\otimes\blue{\as^1_{\bar\psi}}, \cr
&\blue{\bs^0_{\psi}}\in H^0(G_x,\Fc^{\bl}_x)_\psi,\ \ \blue{\as^1_{\bar\psi}}\in H^1(G_y,\Fc^{\bl}_y)_{\bar\psi}, \cr
\green{\gamma^0}&=\green{\bs^0_{\phi}}\otimes\green{\bs'^0_{\bar\phi}}, \cr
&\green{\bs^0_{\phi}}\in H^0(G_x,\Fc^{\gr}_x)_\phi,\ \ \green{\bs'^0_{\bar\phi}}\in H^0(G_y,\Fc^{\gr}_y)_{\bar\phi}, \cr
\end{align}
be K\"unneth basis elements of the three copies. Then the coefficient of the
logical CZ action \eqref{eq:logical_CZ_Fq} on $\L_H$ is
\begin{align}\label{eq:Lambda_balanced}
\Lambda_H\big(\red{\alpha^1},\blue{\beta^1},\green{\gamma^0}\big)
&:=\int_{[\L_H]}\red{\alpha^1}\cup_H\blue{\beta^1}\cup_H\green{\gamma^0} \cr
&=\frac{\delta_{\chi\psi\phi,\,1}}{|H|}\
\Big\langle[G_x],\,\red{\as^1_\chi}\cup\blue{\bs^0_\psi}\cup\green{\bs^0_\phi}\Big\rangle \cr
&\qquad\times\Big\langle[G_y],\,\red{\bs^0_{\bar\chi}}\cup\blue{\as^1_{\bar\psi}}\cup\green{\bs'^0_{\bar\phi}}\Big\rangle .
\end{align}
In words: the logical CZ gates of the balanced product are those of the
hypergraph product, restricted to triples of factor classes whose characters
multiply to one, $\chi\psi\phi=1$, and rescaled by $|H|^{-1}$.
\end{proposition}
\begin{proof}
By Eq.~\eqref{eq:integral_lift} and Lemma~\ref{lemma:balanced_cup},
$\Lambda_H=|H|^{-1}\int_{[\L]}\red{\alpha^1}\cup\blue{\beta^1}\cup\green{\gamma^0}$,
the three cochains being viewed as invariant cochains on $\L$, and
Lemma~\ref{lemma:factorization} factorizes the right-hand side over the two graph
factors (with K\"unneth sign $+1$, as in the proof of Lemma~\ref{lemma:matching}). The
$x$-factor integrand $\red{\as^1_\chi}\cup\blue{\bs^0_\psi}\cup\green{\bs^0_\phi}$ is an
eigen-cochain of the tensor sheaf with character $\theta:=\chi\psi\phi$, because
the cup product is equivariant. The fundamental chain $[G_x]$ is $H$-invariant and
the pairing is $H$-invariant, so $\langle x,[G_x]\rangle=\langle h\cdot x,h\cdot[G_x]\rangle
=\theta(h)\langle x,[G_x]\rangle$ for all $h$, which forces $\langle x,[G_x]\rangle=0$
unless $\theta=1$. The $y$-factor has character $\bar\chi\bar\psi\bar\phi=\bar\theta$
and gives the same condition.
\end{proof}

For the instantiation of Construction~\ref{con:balanced} the degree-$0$ slots of
the \bl\ and \gr\ copies carry the constant sheaf, whose only $0$-cocycle is the
all-ones cocycle $\cs^0$, of trivial character. The selection rule then forces
the \rd\ class to be $H$-invariant as well.

\begin{corollary}[CZ structure and parameters of Construction~\ref{con:balanced}]
\label{cor:balanced_subdivided_CZ}
Assume $p\nmid|H|$ and that $\Fq$ contains the $\exp(H)$-th roots of unity. In
Construction~\ref{con:balanced} write $k_\chi:=\dim H^1(G,\Fc)_\chi=\dim\bar\C_{\bar\chi}$
and $k^{\top}_\chi:=\dim H^1(G,\Fc^{\perp})_\chi=\dim\bar\C^{\top}_{\bar\chi}$ for the
isotypic dimensions of the input code $\bar\C\cong H_1(G,\Fc)$ and its transpose
$\bar\C^{\top}\cong H_1(G,\Fc^{\perp})$, and $k_1$, $k_1^{\top}$ for the
$H$-invariant parts.
\begin{enumerate}[label=(\roman*)]
\item The \rd\ copy has
\be\label{eq:k_red_balanced}
k^{\rd}_H=\sum_{\chi\in\hat H}\big(k_\chi^2+(k^{\top}_\chi)^2\big)
\ \ge\ \frac{\bar k^{2}+(\bar k^{\top})^{2}}{|H|},
\ee
and in particular constant rate, $k_H/n_H=\Theta(1)$, for a constant-rate input
code; the \bl\ copy has the protected sector $\blue{\cs^0}\otimes H^1(G,\Fc)^{H}$
of dimension $k_1$, and the \gr\ copy has $k_1=\dim H^0(G,\Fc^{\perp})^{H}$
independent symmetry generators $\green{\gamma^0_{i_3}}=\green{\bs^0_{1,i_3}}\otimes\green{\cs'^0}$.
\item Choose bases $\{\red{\as^1_{\chi,i}}\}$ of $H^1(G,\Fc)_\chi$ and
$\{\bs^0_{\bar\chi,j}\}$ of $H^0(G,\Fc^{\perp})_{\bar\chi}$ which, for $\chi=1$, are
dual under the Poincar\'e pairing $\mathcal P$ of Eq.~\eqref{eq:P-def} restricted
to the invariant sectors, and label the K\"unneth basis by
$\red{\alpha^1_{(\chi;i,j)}}=\red{\as^1_{\chi,i}}\otimes\red{\bs^0_{\bar\chi,j}}$,
$\blue{\beta^1_{j'}}=\blue{\cs^0}\otimes\blue{\as^1_{1,j'}}$,
$\green{\gamma^0_{i'}}=\green{\bs^0_{1,i'}}\otimes\green{\cs'^0}$. Then
\be\label{eq:Lambda_balanced_subdivided}
\Lambda_H\big(\red{(\chi;i,j)},\blue{j'},\green{i'}\big)
=\frac{1}{|H|}\ \delta_{\chi,1}\ \delta_{i,i'}\ \delta_{j,j'} .
\ee
Thus the logical CZ structure of Eq.~\eqref{eq:intersection_structure} is
reproduced on the $H$-invariant labels, with the multiplier $|H|^{-1}$ (which can
be absorbed by rescaling $\green{\gamma^0}$ by $|H|$), while the \rd\ logical
qudits with $\chi\neq1$ are untouched by every symmetry operator.
\item The gauging-measurement protocol of Sec.~\ref{sec:fountain}, with one \rd\
qudit per cluster $i'$ active and the others in $\lo{\ket0}$, prepares $k_1$
logical CZ magic states on disjoint pairs per run, so the magic rate is
$\varrho_H=k_1/n_H$.
\end{enumerate}
\end{corollary}
\begin{proof}
(i) By \eqref{eq:balanced_kunneth_chars}, $k^{\rd}_H=\sum_\chi[h^1_{\Fc}(\chi)h^0_{\Fc^\perp}(\bar\chi)
+h^0_{\Fc}(\chi)h^1_{\Fc^\perp}(\bar\chi)]$ with $h^a_{\cdot}(\chi)$ the isotypic
dimensions. The Poincar\'e pairing $\mathcal P:H^1(G,\Fc)\times H^0(G,\Fc^{\perp})\to\Fq$
is non-degenerate (Corollary~\ref{cor:kronecker-dual}) and $H$-invariant, so it
identifies $H^0(G,\Fc^{\perp})$ with the dual module of $H^1(G,\Fc)$ and
$h^0_{\Fc^\perp}(\bar\chi)=h^1_\Fc(\chi)=k_\chi$; likewise $h^0_{\Fc}(\chi)=h^1_{\Fc^\perp}(\bar\chi)=k^{\top}_{\bar\chi}$.
Summing over $\chi$ gives the equality in \eqref{eq:k_red_balanced}, and
$\sum_\chi k_\chi^2\ge(\sum_\chi k_\chi)^2/|\hat H|=\bar k^2/|H|$ is
Cauchy--Schwarz. Constant rate follows from \eqref{eq:balanced_length}. For the
\bl\ copy, $H^0(G,\Fc^0)=\Fq\cs^0$ ($G$ being connected) has trivial character, so the sector
$H^0(G,\Fc^0)_\chi\otimes H^1(G,\Fc)_{\bar\chi}$ is non-zero only for $\chi=1$;
similarly $Z^0(\L_H,\Fc^{\gr})=H^0(G,\Fc^\perp)^{H}\otimes\cs'^0$, of dimension
$h^0_{\Fc^\perp}(1)=k_1$. (ii) In Proposition~\ref{prop:balanced_CZ} the \bl\ and
\gr\ characters are $\psi=\phi=1$, so the selection rule reads $\chi=1$, and the
two factor integrals reduce to the Poincar\'e pairings of
Eq.~\eqref{eq:intersection_structure} because $\cs^0$ is the unit of the cup
product: $\langle[G],\red{\as^1_{1,i}}\cup\blue{\cs^0}\cup\green{\bs^0_{1,i'}}\rangle
=\mathcal P(\red{\as^1_{1,i}},\green{\bs^0_{1,i'}})=\delta_{ii'}$ and
$\langle[G],\red{\bs^0_{1,j}}\cup\blue{\as^1_{1,j'}}\cup\green{\cs'^0}\rangle
=\mathcal P(\blue{\as^1_{1,j'}},\red{\bs^0_{1,j}})=\delta_{jj'}$; the restriction of
$\mathcal P$ to the invariant sectors is non-degenerate because isotypic
components pair only with their inverses. (iii) is the proof of
Theorem~\ref{thm:fountain} with the cluster index $i'$ running over the $k_1$
invariant \gr\ labels.
\end{proof}

\begin{remark}[What the quotient costs, and how to avoid paying it]
\label{rem:balanced_rate}
For a free action the isotypic dimensions are typically balanced,
$k_\chi\approx\bar k/|H|$, so the balanced product keeps the encoding rate
(Eq.~\eqref{eq:k_red_balanced}) but reduces the number of symmetry generators,
and hence of magic states per run, from $\Theta(\bar k)$ to
$k_1\approx\bar k/|H|$; relative to the smaller block length $n_H=n_\L/|H|$ the
magic rate becomes $\varrho_H\approx\bar k/\bar n^{2}=\Theta\big((|H|\,n_H)^{-1/2}\big)$,
a loss of $|H|^{1/2}$ with respect to the hypergraph product. The culprit is the
constant sheaf in the degree-$0$ slots of the \bl\ and \gr\ copies, whose
$0$-cocycle is $H$-invariant and therefore, through the selection rule, admits
only $H$-invariant partners. If instead all three copies carry non-constant
$H$-symmetric sheaves in both slots --- for example the GTZ-type configuration of
Construction~\ref{con:high_yield} built from an $H$-symmetric multiplication-friendly
Tanner code --- every triple of characters with $\chi\psi\phi=1$ contributes to
\eqref{eq:Lambda_balanced}, the three-index coupling becomes ``momentum
conserving'' rather than ``zero momentum'', and all $|H|$ Fourier sectors take
part in the fountain. Proposition~\ref{prop:balanced_CZ} is the general formula
for this case; working it out for a specific symmetric multiplication-friendly
input is left for future work.
\end{remark}

\begin{remark}[Modular case]\label{rem:balanced_modular}
When $p\mid|H|$ (for instance $H=\ZZ_2^{m}$ over $\FF_2$, the setting of several
lifted-product constructions) the averaging idempotent does not exist,
Eqs.~\eqref{eq:balanced_kunneth}--\eqref{eq:Lambda_balanced} fail, and
$H^\bullet(\L_H)$ must be computed from the balanced tensor product
$C^\bullet(G_x)\otimes_{\Fq[H]}C^\bullet(G_y)$ directly, as in
Ref.~\cite{9490244}. Nothing in Sec.~\ref{sec:balanced_formalism} depends on
$p\nmid|H|$: the twisted code, its symmetry operators and the protocol are
defined for every $H$; only the explicit K\"unneth basis and the closed formula
for $\Lambda_H$ are lost.
\end{remark}

\begin{remark}[Distance]\label{rem:balanced_distance}
We have deliberately said nothing about the distance of $\tilde{\C}_H$. The
contraction and cleaning arguments of Appendices~\ref{app:subsystem_distance},
\ref{app:twisted_KL} and~\ref{app:GTZ_KL} work factor by factor on $\L$, and
transporting them to $\L_H$ through the lifts \eqref{eq:invariant_cochains} loses
a factor $|H|$ in every weight, because the lift of a cochain supported on $m$
cells of $\L_H$ is supported on $|H|m$ cells of $\L$; for the regime of interest,
$|H|$ polynomial in $\bar n$, this is not useful. The distance of balanced- and
lifted-product CSS codes is instead controlled by expansion
\cite{9490244,PK:almost_linear,pkldpc22,Quantum_tanner}, and whether those
arguments can be combined with the twist is exactly the question raised in
Sec.~\ref{sec:discussion}. Two things can nevertheless be said. First, the
twist does not change the electric operators and only dresses the magnetic ones,
so under Assumption~\ref{assump:TQFT} (Route~1 of
Sec.~\ref{sec:summary_of_results}) $\tilde{\C}_H$ has the subsystem distance of the
untwisted balanced-product code with the protected sectors of
Corollary~\ref{cor:balanced_subdivided_CZ}. Second, the Knill--Laflamme
arguments of Appendices~\ref{app:twisted_KL} and~\ref{app:GTZ_KL} use the product
structure only through two inputs --- relative distance bounds for the protected
K\"unneth sectors and a cleaning lemma for the degree-one factor representatives
--- and both are statements about the untwisted balanced-product CSS code and its
K\"unneth sectors; given them, the rest of the argument (orbit basis, group law,
label independence of the diagonal matrix elements) is unchanged on $\L_H$,
because it only uses locality, the Leibniz rule and
Lemma~\ref{lemma:commutation_Fq}, all of which descend.
\end{remark}

\end{appendix}

\bibliographystyle{apsrev4-2}
\bibliography{mybib_merge, biblio, updated}

\end{document}